\documentclass[10pt,a4paper]{scrartcl}
\KOMAoptions{DIV=12}

\usepackage[utf8]{inputenc} 
\usepackage[T1]{fontenc}
\usepackage{microtype}

\usepackage{amsmath,amssymb,amsfonts,amsthm}
\usepackage{mathtools}  
\usepackage{bbold}      

\usepackage{mathabx}

\usepackage{cancel}
\usepackage{slashed}

\usepackage{graphicx}   
\usepackage{tikz}
\usetikzlibrary{arrows.meta,positioning,decorations.pathreplacing,calc,fit,backgrounds,patterns,shapes}
\usepackage{rotating}

\usepackage[font=small,format=plain,indention=0pt]{caption}
\usepackage{array}
\usepackage{booktabs}
\usepackage{multirow}
\usepackage{makecell}
\usepackage{tabularx,adjustbox}
\usepackage{threeparttable}
\usepackage{longtable}

\usepackage{pifont}

\usepackage{hyperref}
\usepackage{cleveref}   

\usepackage{comment}
\usepackage{enumitem}

\usepackage{etoc}

\usepackage[giveninits=true,isbn=false,sorting=none]{biblatex}
\usepackage{appendix} 

\newtheorem{theorem}{Theorem}[section]
\newtheorem{lemma}[theorem]{Lemma}
\newtheorem{corollary}[theorem]{Corollary}

\theoremstyle{definition}
\newtheorem{remark}[theorem]{Remark}
\newtheorem{example}[theorem]{Example}

\DeclareMathOperator{\ard}{ARD}
\DeclareMathOperator{\sci}{SCI}
\DeclareMathOperator*{\argmax}{arg\,max}

\DeclareMathOperator{\EE}{\mathbb{E}}

\newcommand{\appropto}{\mathrel{\vcenter{
  \offinterlineskip\halign{\hfil$##$\cr
    \propto\cr\noalign{\kern2pt}\sim\cr}}}}

\newcommand{\DKL}[2]{D_{\mathrm{KL}}\!\left(#1 \,\middle\|\, #2\right)}
\newcommand{\DF}[2]{D_{f}\!\left(#1 \,\middle\|\, #2\right)}

\newcommand{\logg}[1]{\log\!\left(#1\right)}

\newcommand{\BS}{\mathcal{B}} 
\newcommand{\GS}{\mathcal{P}} 
\newcommand{\CS}{\mathcal{C}} 
\newcommand{\RS}{\mathcal{R}} 

\newcommand{\cat}[1]{\emph{#1}} 

\newcommand{\PP}[1]{\mathbb{P}\!\left(#1\right)}

\newcommand{\PPS}[3]{%
  \mathbb{P}%
  \if\relax\detokenize{#2}\relax\else_{#2}\fi%
  \if\relax\detokenize{#3}\relax\else^{#3}\fi%
  \!\left(#1\right)}

\newcommand{\PS}[3]{%
  P%
  \if\relax\detokenize{#2}\relax\else_{#2}\fi%
  \if\relax\detokenize{#3}\relax\else^{#3}\fi%
  \!\left(#1\right)}

\newcolumntype{M}{>{\(\displaystyle}l<{\)}} 

\usepackage{authblk}

\newif\ifincludesupplement

\includesupplementtrue      

\title{Unified framework for measuring segregation resolves how social and geographical space jointly shape connections}

\author[1]{Johannes Happenhofer}
\author[1,2,5]{Sahil Loomba}
\author[3,1]{Till Hoffmann}
\author[4]{Sumeet Agarwal}
\author[1,5]{Nick S. Jones}

\affil[1]{Department of Mathematics, Imperial College London, South Kensington Campus, London SW7 2AZ, UK}
\affil[5]{I-X Centre for AI in Science, Imperial College London, White City Campus, London W12 7TA, UK}
\affil[2]{Institute for Data, Systems, and Society, Massachusetts Institute of Technology, Cambridge, MA 02139, USA}
\affil[3]{Department of Biostatistics, Harvard T.H. Chan School of Public Health, Boston, MA 02115, USA}
\affil[4]{Yardi School of Artificial Intelligence and Department of Electrical Engineering, Indian Institute of Technology Delhi, New Delhi 110016, India}

\date{\today}

\begin{document}
\maketitle

\begin{abstract}
Our understanding of how geographical and social segregation interact remains limited, as relatively few studies investigate them jointly, and existing approaches often lack a framework distinguishing geographical, social, and total segregation. Additionally, large-scale individually resolved geo-social network data are rarely publicly available. We address both. Conceptually, we develop a unified framework that measures segregation in geosocial networks by comparing network models to appropriate null models and recovers the Theil index, dissimilarity index, and network modularity as special cases. Empirically, we turn to privacy-preserving aggregated relational data (ARD): we combine the Facebook Social Connectedness Index for the US with US Census and Pew data, and introduce an ARD-compatible joint geosocial intervening-opportunities model to infer link probabilities between region--group cell pairs. Applying our segregation framework, we find that social segregation predominates over geographical segregation, with notable separation for White--Black, college-degree--no-degree, and high-income--low/middle-income across both segregation types. We find increasing social homophily with geographical distance and group-specific geographical connectivity patterns, suggesting that geographical segregation may affect cross-group connectivity not only directly but also by amplifying social segregation.
\end{abstract}

\section{Introduction}
The separation of social groups has long been a central problem across demography, sociology, and human geography, giving rise to a rich literature on the measurement, causes, and consequences of segregation \cite{duncan1955segregation, schelling1971segregation, massey1988segregation}. 
Traditionally, this separation has been studied independently along two dimensions: geographical segregation, the differential distribution of sociodemographic groups across regions, and social segregation, the differential preference to form ties across sociodemographic groups. 

Geographical segregation measures are typically based on comparisons between the local distribution over sociodemographic groups in pre-defined regions and the marginal (countrywide) distribution over sociodemographic groups, and include measures like the Theil Information Theory index (or, relatedly, the Mutual Information Index), the index of Dissimilarity, a squared coefficient of variation, a Gini index, the Relative Diversity index, and others \cite{reardon2002measures, reardon2004measures}. While these measures are useful in their own right, the arbitrary character of defining geographical segregation based on pre-defined regions --- usually as defined in available census data --- has been widely noticed \cite{wong1997spatial, openshaw1984modifiable}. One typical problem associated with this approach is the modifiable areal unit problem (MAUP), i.e., individuals within a region are assumed to have the same geospatial distance to each other and the same distance to individuals in another region --- which can be quite misleading if there is high population density at the boundaries of the regions. Another is the checkerboard problem, i.e., homogeneous regions with only one group present in each region arranged like a checkerboard yield the same segregation index as arranging all regions composed of one group, say, in the west and all other regions composed of the other group in the east.  As a remedy, egocentric (or geospatially weighted) extensions of the traditional segregation measures have been proposed. For example, ref.~\cite{reardon2004measures} suggests the use of geospatial proximity kernels, but clearly, the right choice of a geospatial proximity kernel is crucial. We argue that simply building a proximity kernel based on the Euclidean distance between geographical locations is insufficient, and that the population density at different geographical locations interacts with how ``distant'' two individuals are from each other. On the other hand, simply considering the average chance that individuals at different geographical locations interact with each other is insufficient for capturing geographical segregation as well, as it would conflate social and geographical segregation. 

Social segregation has typically been studied within the framework of homophily \cite{mcpherson2001birds} or assortative mixing in the social network literature, that is, the preference of individuals to connect to others who have similar sociodemographic characteristics. Moody \cite{moody2001race}, for instance, introduced an odds-ratio index of friendship segregation in school networks, which compares the likelihood of forming ties within versus across groups. Another line of work employs stochastic block models (SBMs) with predefined blocks corresponding to sociodemographic categories \cite{holland1983stochastic, karrer2011stochastic}. In this setting, segregation is assessed by the estimated probabilities of ties within and between groups (mixing patterns), offering a likelihood-based approach. An overview of social segregation measures can be found in \cite{bojanowski_measuring_2014}. Perhaps the most popular measure is network modularity or, when normalized, the nominal network assortativity coefficient \cite{newman2003mixing, newman2018networks}, which compares the observed number of within-group edges to the expected number of such edges under the degree configuration model, which matches the observed degrees in expectation but is otherwise random. For continuous variables, the assortativity coefficient is the Pearson correlation of sociodemographics between pairs of connected individuals. We note that similar to traditional geographical segregation, a comparison to a counterfactual null model without segregation determines the value of network modularity.

Recently, geographical and social segregation have been studied together in ref.~\cite{kazmina2024socio} for the Netherlands, with network data built upon population registers curated by Statistics Netherlands. They compute network assortativity for income, education, ethnic group and migrant generation for a social network (family, school, work and neighborhood ties) and a spatial network (complete graph within an administrative neighborhood, no edges between such regions), with one of their main findings being that the amount of social segregation is twice the amount of geographical segregation. While this approach roughly separates geographical and social segregation, the social network might itself be partly impacted by geographical segregation, and the spatial network has the arbitrary character of using pre-defined regions.
Ref.~\cite{xu2019quantifying} uses large-scale mobile phone and socio-economic datasets for Singapore to analyze geographical and social segregation along socio-economic status differences. They define a Communication Segregation Index by the similarity-weighted share of a person’s communications with socio-economically similar contacts, and a Physical Segregation Index by similarity-weighted co-location probabilities across the city. Again, while those two indices approximately separate social and geographical segregation, it might be that communication shares to different groups are partly determined by the share of those groups in the geographical proximity, and co-location probabilities partly determined by social preferences (e.g., attending specific cultural events).  

Quantifying homophily based on the number of links between different sociodemographic groups does not differentiate between social and geographical segregation, but rather measures the \emph{combined} impact of the two. We call the combined impact of geographical and social segregation \emph{total segregation}, and propose a principled approach that considers both geographical and social segregation as latent constructs.
As an example, consider three hypothetical societies A--C composed of blue and red individuals, such that the joint distribution of geographical coordinates and group membership is identical across all three societies (Fig.~\ref{fig:hypothetical_societies}). Therefore, any traditional geographical segregation measure assigns the same value to all three societies. In both A and B, individuals form friendships purely based on geospatial homophily: conditional on geographical coordinates, friendship formation is independent of group membership. However, geospatial homophily is much stronger in B than in A. Because the two groups are geographically clustered, individuals in B consequently have fewer friends in the other group and more friends within their own group than individuals in A. Measures such as nominal network assortativity \cite{newman2003mixing}, applied to the marginal group-to-group mixing matrix, would therefore report stronger homophily in B than in A. Interpreting this difference as a difference in social segregation would generate an apparent social-homophily signal even though social homophily is absent in both societies. We would instead argue that B is more geographically segregated than A --- even though any traditional geographical segregation measure would assign the same value to A and B --- and that their social segregation is the same, namely absent.
Now compare societies B and C. Their marginal group-to-group mixing matrices are identical and nominal network assortativity $\rho$ is also the same. However, whereas B has strong geospatial homophily and no social homophily, C has much weaker geospatial homophily but positive social homophily: conditional on geographical coordinates, individuals preferentially connect to members of their own group. We would therefore argue that B is more geographically segregated than C, while C is more socially segregated than B. By contrast, conventional approaches based on traditional geographical segregation measures and marginal group-to-group connectivity would not distinguish between B and C. We argue that this distinction matters: for example, a different housing policy might be much more effective in B than in C if the goal were to reduce segregation, while distinguishing between the two is also important for understanding the social processes responsible for segregation, as individuals in C exhibit stronger social homophily.
\begin{figure}[htbp]
    \centering
    \includegraphics[width=0.85\textwidth]{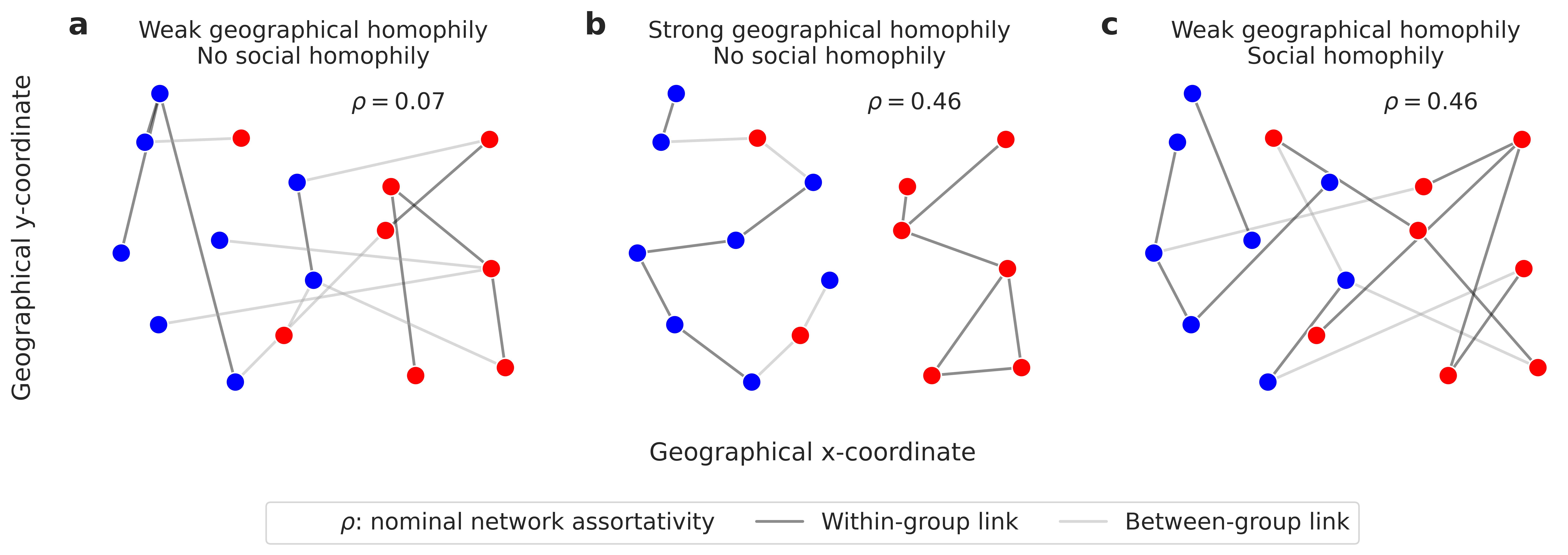}
\caption{ \emph{Traditional geographical segregation and total segregation do not determine geographical and social segregation.} The three hypothetical societies have identical joint distributions of geographical coordinates and group membership, represented by blue and red nodes, so any traditional geographical segregation measure would assign the same value to all three. \textbf{a}, Individuals connect with weak geospatial homophily and no social homophily: conditional on geographical coordinates, friendship formation is independent of group membership. \textbf{b}, Geospatial homophily is stronger, while social homophily remains absent, resulting in more within-group and fewer between-group links and therefore higher nominal network assortativity $\rho$. \textbf{c}, Geospatial homophily is weaker than in \textbf{b}, but individuals preferentially connect to members of their own group conditional on geographical coordinates. The marginal group-to-group mixing matrices, and hence $\rho$, are identical in \textbf{b} and \textbf{c}, despite their different geographical and social homophily. Thus, \textbf{a} and \textbf{b} show that differences in marginal group-to-group connectivity need not reflect differences in social segregation, while \textbf{b} and \textbf{c} show that identical marginal group-to-group connectivity and identical traditional geographical segregation can conceal very different levels of geographical and social segregation. }
    \label{fig:hypothetical_societies}
\end{figure}

We therefore define geographical segregation as the impact of geographical homophily --- the decay of connectivity between pairs of individuals with geographical distance --- on the chance that different sociodemographic groups connect, and define social segregation as the impact of differential preferences towards different sociodemographic groups on the connectivity between these groups. This distinction closely parallels Blau's \cite{blau1977inequality,blau1994structural} distinction between structural opportunities and individual preferences for intergroup association: in our setting, the former corresponds to geographical and the latter to social segregation.
Hence, measuring geographical and social segregation requires a suitable conceptual framework and typically a model-based approach: a fully observed large-scale geosocial network or a statistical estimate of relevant key features of the network (how connected various geosocial groups are) uncovers only total segregation, which has to be further resolved into geographical and social segregation by defining suitable counterfactual models relative to the observed network or estimated ``true'' model.

Studying geographical and social segregation jointly requires large-scale social network data. However, large-scale individual-level network data are typically not publicly available due to privacy concerns, as is the case for the data underlying the above-mentioned studies for the Netherlands \cite{kazmina2024socio} and Singapore \cite{xu2019quantifying}, as well as for the individual-level Facebook friendship graph underlying the investigation of social connectedness (fraction of individuals' friends who are high-SES) in ref.~\cite{chetty2022social}. Moreover, even if such data were publicly accessible, the sheer size of large-scale geosocial networks with millions of nodes and billions of edges would pose serious challenges for statistical inference. We therefore adopt a different approach based on Aggregated Relational Data (ARD), which in our setting refers to aggregated connection volumes between sets of nodes. This definition differs from the classical notion of ARD, which consists of ego-reported counts of ties to demographic groups \cite{mccormick2012latent, maltiel2015estimating, feehan2016generalizing}, although classical ARD can be understood as a special case of our notion of ARD. In addition to preserving privacy, such data is also highly compressed and therefore more suitable for statistical inference on large-scale geosocial networks. The specific ARD data we use for this study is the Social Connectedness Index (SCI) \cite{bailey2018sci, bailey2020social}, which is, upon rescaling, given by the aggregated number of Facebook (FB) friendships between geographical regions, divided by the product of the number of Facebook users in those regions.

Since the SCI provides only aggregated connection volumes between regions, rather than individual-level ties across geosocial groups, we develop a probabilistic geosocial network model that specifies tie probabilities conditional on individuals' geosocial coordinates. This in turn allows the specification of a likelihood for the observed SCI data, conditional on the sociodemographic distribution across regions. We estimate the sociodemographic distribution of Facebook users across regions by combining U.S. Census data \cite{acs_5y_2017_2021_techdoc,census_2020_dhc_techdoc,census_mdat_cps_asec2022}, Pew Research Center surveys \cite{pew_2021_core_trends,pew_atp_wave_112}, and Meta Marketing API reach estimates \cite{meta_reach_estimate_guide,meta_ad_account_reach_estimate_ref}.

It is well established that geographical distance is one of the dominant predictors of connectivity in large-scale geosocial networks \cite{takhteyev2012geography, scellato2011socio, bailey2020social}. Mapping distances directly to connection probabilities implies that individuals in regions with high population density will have much larger degrees than individuals in low density regions, which is not corroborated by empirical data.
Rather, the marginal connection probability of a random pair of individuals within one contiguous region scales as the reciprocal of the population size of this region, as opposed to the area of this region. This population-size scaling and related phenomena have been discussed in \cite{stouffer1940intervening, liben2005geographic, kumar2006navigating, backstrom2010find, simini2012universal, kotsubo2021kernel} under the notions of rank-based network formation or intervening opportunities (IOs), where, for two individuals $i$ and $j$, from the perspective of $i$, all individuals $k$ with a smaller distance to $i$ than $j$ constitute intervening opportunities. Our novel contribution is a probabilistic network model based upon intervening opportunities defined by a latent joint distance in geosocial space, as opposed to intervening opportunities in geographical space alone. The model specifies a distribution of the $\ard$, conditional on the geospatial distribution across regions (or the entities upon which the ARD data is defined) and the latent distance in geosocial space. We fit this model to the SCI data for zip code tabulation areas (ZCTAs) in the United States, separately for each state. By extensive model comparison, we demonstrate that this joint geosocial IO model outperforms simpler models, e.g., models which factorize into a geospatial connectivity kernel (built either upon distance or IOs in geographical space) and a Stochastic Block Model for social connectivity. The model fits enable us to apply our general framework for geographical and social segregation. In accordance with \cite{kazmina2024socio}, we find that social segregation is stronger than geographical segregation. Other findings include  
strong homophily along race, education and income, and interaction patterns between racial homophily and membership to specific age, education, or income categories. Specifically, racial homophily appears to decline with educational attainment and is weaker within the highest income group. We also find evidence of non-trivial interaction patterns between geographical distance and social homophily, and correlations between geographical and social segregation.

\section{Results}
\subsection{Measuring geographical, social, and total segregation}
We distinguish four forms of segregation in geosocial networks---total, traditional geographical, social, and geographical segregation---and operationalize each through a corresponding network model, as illustrated in Fig.~\ref{fig:tot_trad_soc_geo_seg_w_formulas}. Precise definitions of these models are provided in Methods, Section~\ref{sec:segregation_via_counterfactual_models}.

\begin{figure}[htbp]
    \centering
    \includegraphics[width=1.0\textwidth]{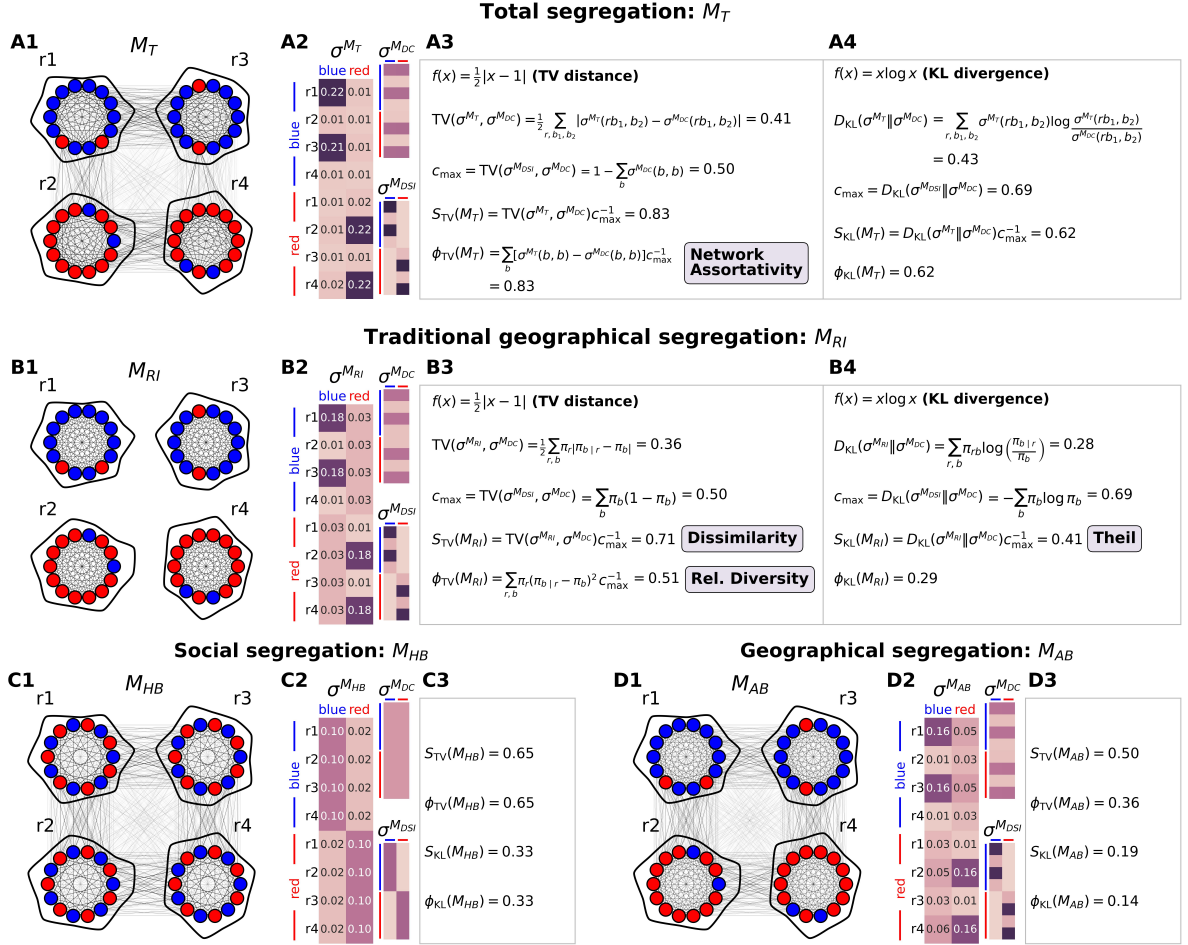}
\caption{\emph{Proposed unified segregation measurement framework recovers network modularity, dissimilarity, and Theil index as special cases, and separates total, social, and geographical segregation.}
Toy edge-probability models over individuals nested in four geographical regions ($r_1,\dots,r_4$) and two social groups (blue and red). The models are visualised as networks with edge thickness proportional to the probability of a tie between the corresponding individuals. For each model $M$, these probabilities induce a conditional distribution $\sigma^M$ over pairs of geosocial coordinates, conditional on there being a tie. Segregation is measured by comparing $\sigma^M$ with a degree-preserving random-mixing reference $\sigma^{M_{DC}}$ (``Degree Configuration Model'') using an $f$-divergence, normalised by the $f$-divergence of the maximal-segregation reference  $\sigma^{M_{DSI}}$ (``Degree-constrained Social Isolation'') from the random mixing reference $\sigma^{M_{DC}};$ the resulting normalization constants are denoted by $c_{\max}$ in panels \textbf{A3--A4} and \textbf{B3--B4}. 
\textbf{A1--A4,} In an assumed true data-generating model $M_T$ (or an inferred model as an estimate thereof), both geographical structure and group-based tie preferences contribute to the overall connectivity pattern, so that blue individuals connect mainly to blue individuals and red mainly to red. 
\textbf{B1--B4,} Traditional geographical segregation corresponds to the ``Regional Isolation Model'' $M_{RI}$, where individuals only connect within their region irrespective of group membership; as a result, overall connectivity is determined by regional group composition. With total variation distance this recovers the dissimilarity index, and with Kullback--Leibler divergence it recovers the Theil index. 
\textbf{C1--C3,} By counterfactually reallocating individuals such that each region has the same share of blue and red individuals (homogeneous social space model $M_{HB}$), we isolate group-based tie preferences from geographical patterns; we call this social segregation.  
\textbf{D1--D3,} By counterfactually changing the model such that only geography, but not group membership, affects tie probabilities (social ambiphily model $M_{AB}$), we isolate geographical segregation. 
The formula panels show two choices of $f$, corresponding to Total Variation distance and Kullback--Leibler divergence, and report the resulting normalised segregation index $S_f(M)$ together with our homophily index $\phi_f(M)$; for Total Variation, $\phi_f$ recovers nominal network assortativity (i.e., normalized network modularity), and in the special case of the Regional Isolation Model $M_{RI}$ it coincides with Relative Diversity.}
\label{fig:tot_trad_soc_geo_seg_w_formulas}
\end{figure}
Total segregation corresponds to the true geosocial network model (or an estimate thereof), as illustrated with a toy example in panel \textbf{A1}. 
Traditional geographical segregation assumes counterfactually that individuals only connect within predefined regions; we refer to this as the \emph{regional isolation model} $M_{RI}$, illustrated in panel \textbf{B1}.
Social segregation isolates separation due to group-based connectivity preferences. It is represented by a counterfactual network model $M_{HB}$ in which individuals are reallocated so that each region has the same sociodemographic composition, while the connection probabilities conditional on geosocial coordinates are retained, see panel \textbf{C1}.
Geographical segregation isolates separation induced by the decay of connectivity with geographical distance and the geographical distribution of groups. It is represented by a counterfactual \emph{social ambiphily} model $M_{AB}$, in which the decay of connectivity with geographical distance is retained, but sociodemographic group membership does not affect tie probabilities conditional on geographical coordinates, see panel \textbf{D1}.

We compress the connectivity information contained in a given geosocial network by the distribution of links over geosocial coordinates: for a model $M$ --- whether the true network, an estimate thereof, or one of the counterfactual models above --- let $\sigma^M(r_1b_1,b_2)$ denote the probability that a randomly selected pair of friends consists of one individual in region $r_1$ and group $b_1$ and the other individual in group $b_2$. Equivalently, $\sigma^M(r_1b_1,b_2)$ is the fraction of edges from geosocial coordinate $r_1b_1$ to social coordinate $b_2$. These edge distributions are depicted by the larger heatmaps in panels \textbf{A2}, \textbf{B2}, \textbf{C2}, and \textbf{D2}.

We measure segregation by comparing $\sigma^M$ with the corresponding edge distribution under a degree-preserving random-mixing null model $M_{DC}$ (\emph{degree configuration model}) via an $f$-divergence. We normalize by the $f$-divergence between a maximal segregation reference --- that is, the \emph{degree-constrained social isolation model} $M_{DSI}$ --- and the same random-mixing null. The resulting normalized index $S_f$ lies in $[0,1]$, with $0$ indicating no segregation and $1$ indicating maximal segregation, where distinct sociodemographic groups do not connect. If this normalization is omitted, we refer to the resulting index as unnormalized.

The random mixing and maximal segregation null models corresponding to the different segregation-specific network models are depicted, for the same toy example, as small heatmaps in panels \textbf{A2}, \textbf{B2}, \textbf{C2}, and \textbf{D2}. Note that these null models are always defined relative to a given network model $M,$ as they preserve the row and column marginals of $M.$ 

$S_f$ measures the magnitude of deviations from random mixing, but does not by itself distinguish homophilous from heterophilous deviations. A deviation for a coordinate pair $(r b_1,b_2)$ is homophilous if connectivity is higher than under the null and $b_1=b_2$, or lower than under the null and $b_1\neq b_2$; deviations in the opposite directions are heterophilous. We decompose $S_f$ into a homophilous contribution $\alpha_f$ and a heterophilous contribution $\beta_f$, with $S_f=\alpha_f+\beta_f$, and define the homophily index $\phi_f\coloneqq\alpha_f-\beta_f$. If all deviations are homophilous, then $S_f=\phi_f$.

We focus on two choices of $S_f$: a Total Variation distance (TV) version, based on $f(x)=\frac{1}{2}|x-1|$, and a Kullback--Leibler divergence (KL) version, based on $f(x)=x\log x$. For the TV version, our homophily index $\phi_{\mathrm{TV}}$ recovers nominal network assortativity, equivalently normalized network modularity, when applied to the empirical edge distribution of an observed network. Under the regional isolation model $M_{RI}$ (panel \textbf{B1}), $S_{\mathrm{TV}}$ recovers the Dissimilarity index and $S_{\mathrm{KL}}$ recovers the Theil index, while $\phi_{\mathrm{TV}}$ recovers the Relative Diversity index. In each case, the normalization constant $c_{\max}$ implied by the maximal-segregation reference recovers the standard normalization used for the corresponding measure \cite{reardon2002measures, newman2006modularity}, see panels \textbf{A3}, \textbf{B3} and \textbf{B4}. Moreover, applying $S_{\mathrm{KL}}$ beyond the regional isolation model extends the well-known group and regional decomposability properties of the Theil index to arbitrary geosocial network models.

Thus, the framework separates total, social, geographical, and traditional geographical segregation; distinguishes homophilous from heterophilous contributions; and contains existing social-network and traditional geographical segregation measures as special cases.

\subsection{Inferring geosocial connectivity from aggregated relational data}
The edge distributions underlying the proposed indices are normalized connection volumes between geosocial and social coordinates. The geographical coordinate need not in principle be a predefined region: with sufficiently detailed data, the same framework could be applied at individual-level spatial resolution. In practice, however, privacy-preserving public network data are usually spatially aggregated, so we work at regional resolution.

At this regional resolution, $\sigma^M(r_1b_1,b_2)$ could be computed directly if ARD reported connection volumes from each source region--group coordinate $r_1b_1$ to each target social group $b_2$, as in the toy examples in Fig.~\ref{fig:tot_trad_soc_geo_seg_w_formulas}. The Facebook Social Connectedness Index provides a coarser form of ARD: aggregate connection volumes between geographical regions, without the sociodemographic coordinates of the linked individuals. We therefore infer expected geosocial--geosocial connection volumes from region--region ARD using a statistical geosocial network model and variation in sociodemographic composition across regions, as illustrated in Fig.~\ref{fig:data_pipeline}. Marginalizing these inferred volumes over the target region yields $\sigma^M(r_1b_1,b_2)$ and hence allows us to apply the segregation framework above, including the construction of the counterfactual models defining social and geographical segregation.

\begin{figure}[htbp]
    \centering
    \includegraphics[width=1.0\textwidth]{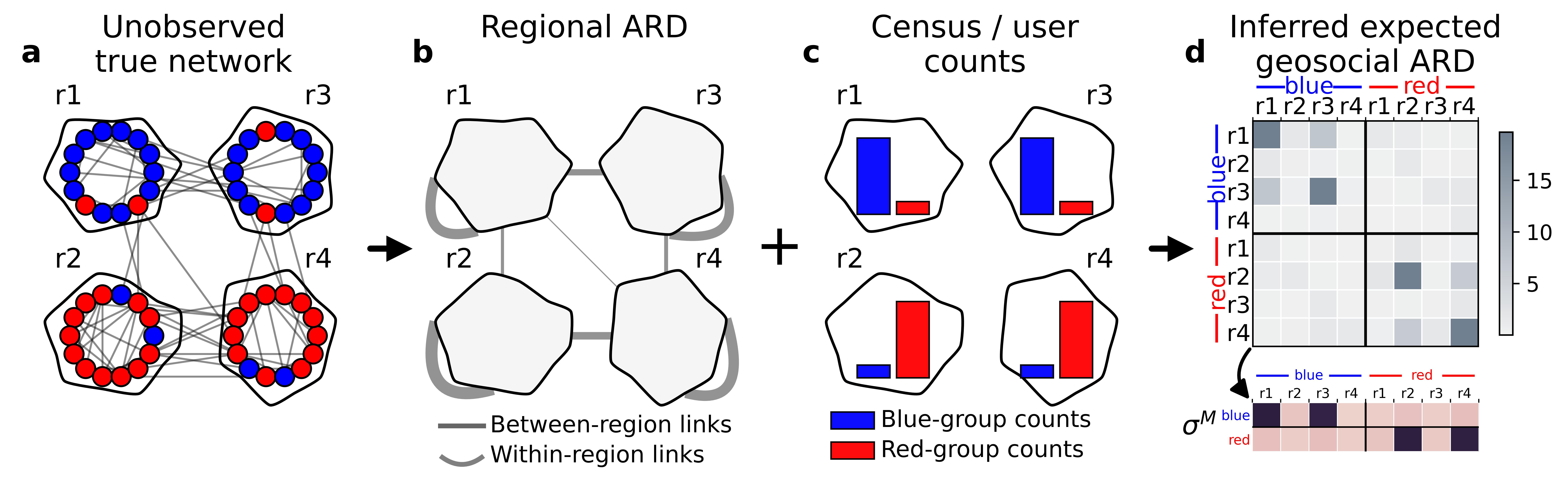}
\caption{
\emph{Schematic data pipeline for inferring geosocial connectivity.}
\textbf{a}, The individual-level geosocial Facebook network is unobserved.
\textbf{b}, The observed Social Connectedness Index (SCI) is derived from regional aggregated relational data (ARD), which record total connectivity between regions but do not include sociodemographic information.
\textbf{c}, Census and estimated Facebook population counts provide information on the distribution of social groups within each region.
\textbf{d}, Variation in the sociodemographic composition across regions in the census and estimated FB populations, together with a suitable statistical network model, enables the statistical estimation of regional-social ARD counts from regional ARD counts. Darker cells indicate larger expected link volumes. The regional-social ARD directly yields the conditional coordinate distribution $\sigma^M$, which underlies our segregation measurement framework, see Fig.~\ref{fig:tot_trad_soc_geo_seg_w_formulas}.
}
\label{fig:data_pipeline}
\end{figure}

\subsection{Joint geosocial intervening opportunities best explain SCI connectivity}
To infer geosocial connection volumes from region--region SCI, we model connectivity as a function of individuals' geographical and sociodemographic coordinates. We compare three model classes, all built around a notion of distance in geosocial space. Here, distance denotes observed geographical distance, inferred sociodemographic dissimilarity, or their combination, and is not assumed to satisfy the axioms of a metric.
The first, distance-based class can be understood as mapping distances directly to connection probabilities, parameterized on the logit scale as a decay of connectivity with geographical distance and group-pair-specific effects---these group-pair-specific effects can be viewed as social dissimilarity or distance.\footnote{Under sparse connectivity, this model approximately factorizes connection probabilities into a geographical kernel and a social kernel, where the social kernel is a stochastic block model (SBM), parameterized on the logit scale.} In a hybrid class, we substitute geographical intervening opportunities (IOs) for geographical distance, but continue to model the effect of sociodemographics via a group-pair-specific dissimilarity structure.
In the joint IO class, geography and sociodemographics no longer affect connection probabilities through separate components. Instead, a distance in joint geosocial space is inferred; for a potential tie from individual $i$ to individual $j$, the relevant IOs are all other individuals who are closer to $i$ in joint geosocial space, and these joint IOs are mapped to connection probabilities.
Accordingly, geographical and social space interact in a non-trivial manner, as illustrated for categorical social space in Fig.~\ref{fig:IOs}, panels \textbf{C1}--\textbf{C4}: a social-group-pair-specific rescaling factor $\kappa$ determines the geographical range over which individuals of each group count as IOs; under social homophily, socially closer groups contribute over a larger geographical radius and socially more distant groups over a smaller radius.

\begin{figure}[htbp]
    \centering
    \includegraphics[width=1.0\textwidth]{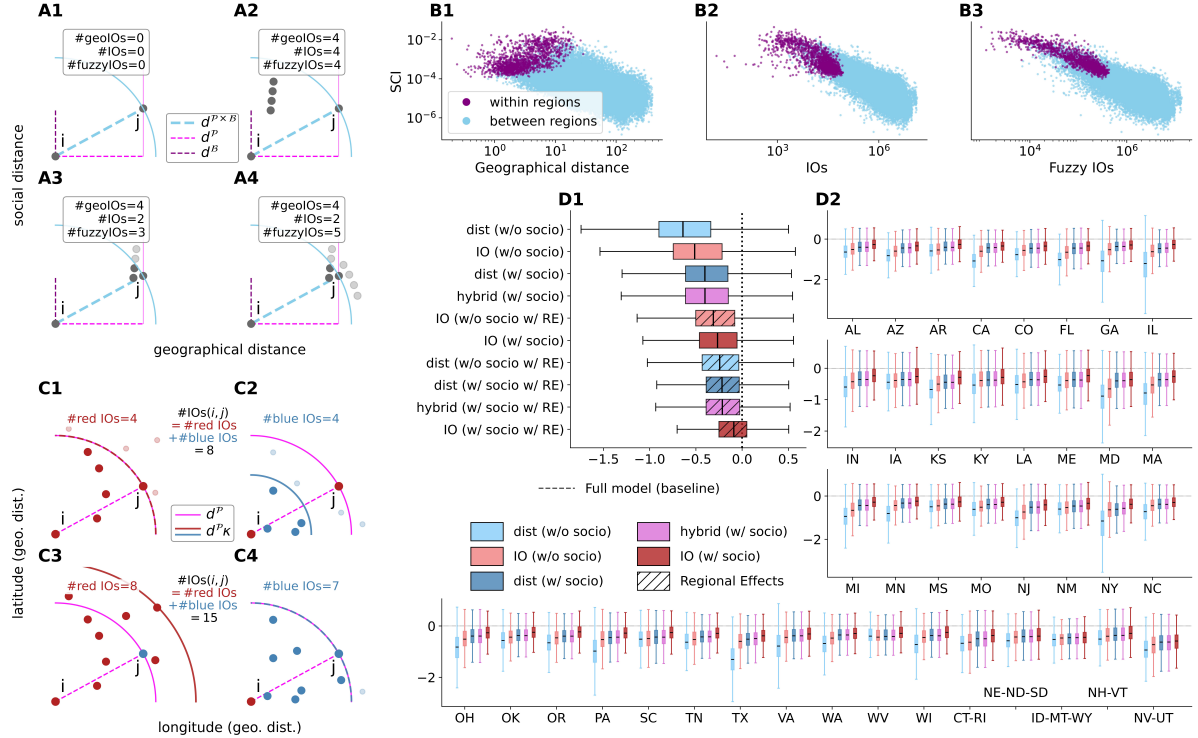}
\caption{
\emph{Joint geosocial intervening opportunities explain Facebook Social Connectedness Index (SCI) better than distance-based or hybrid models.}
\textbf{A1}--\textbf{A4} illustrate the models considered by showing how individuals $i$ and $j$ with a fixed geosocial distance $d^{\GS\times\BS}$ connect. If distances are mapped directly to connectivity, \textbf{A1}--\textbf{A4} are equivalent; if intervening opportunities (IOs) are considered only for geographical space (marginal IOs) based upon geographical distance $d^{\GS}$, and social distance $d^{\BS}$ is considered for social space (hybrid model), \textbf{A2}--\textbf{A4} are equivalent; for joint IOs in geosocial space \textbf{A3}--\textbf{A4} are equivalent; and for joint fuzzy IOs \textbf{A1}--\textbf{A4} are all different. For illustration purposes, individuals $k$ with a larger distance from $i$ than $j$, depicted by shaded circles instead of solid circles, are counted with weight $0.5$ towards fuzzy IOs.
\textbf{B1}--\textbf{B3} depict the empirical relation in California between $\sci$ and geographical distance, marginal IOs and fuzzy marginal IOs, respectively.
\textbf{C1}--\textbf{C4} illustrate joint geosocial IOs between two individuals $i$ and $j$ at a fixed geographical distance $d^{\GS}$ for categorical social space with categories blue and red, under the assumption of social homophily. Here, $d^{\GS}\kappa$ denotes the social-group-pair-specific geographical range over which individuals of a particular group count as IOs; the red and blue boundaries show this range for red and blue individuals, respectively. In \textbf{C1} and \textbf{C2}, both $i$ and $j$ are red, and the number of IOs between $i$ and $j$ (from the perspective of $i$) is given by all red individuals who are no further in geographical space from $i$ than $j$ is, plus all blue individuals in a smaller geographical area. On the other hand, if $i$ is red and $j$ is blue, \textbf{C3} and \textbf{C4}, then the IOs are given by all blue individuals who are no further than $j$, plus all red individuals in a larger area.
\textbf{D1} shows boxplots of model-specific \emph{pointwise PSIS-LOO log predictive density} ($\mathrm{loo}_i$) \cite{vehtari2017practical} differences $\mathrm{loo}_i - \mathrm{loo}_i^{\text{baseline}}$, where the baseline is the full geosocial IO model with race and geosocial interactions. Each $\mathrm{loo}_i$ corresponds to a region pair $(r,r')$ in California.
\textbf{D2} shows these differences across all states (separate model fits per state); the ordering of the models is consistent across all states: IO models fit the data better than distance-based models (both when only geographical space is considered, and when geosocial space is considered), and the joint IO models fit the data better than a hybrid model based on IOs in geographical space and an SBM in social space.}
\label{fig:IOs}
\end{figure}

Intervening opportunities for an individual $j$ from the perspective of an individual $i$ are defined as the number of individuals $k$ who are at most as distant from $i$ as $j$ is. This unrealistically assumes that there is a sharp boundary which separates individuals who do and do not count as IOs. As a remedy, we use fuzzy IOs \cite{afandizadeh2012fuzzy}, where individuals $k$ at a larger distance from $i$ than $j$ still constitute IOs, although with a weight that declines quickly with distance. For the SCI data, fuzzy marginal IOs visibly track the empirical distance--connectivity relationship better than geographical distance or standard marginal IOs alone, see Fig.~\ref{fig:IOs} \textbf{B1}--\textbf{B3}.

Panels \textbf{A1}--\textbf{A4} of Fig.~\ref{fig:IOs} illustrate the distinction between joint distance, hybrid geographical IOs and social distance, joint IOs in geosocial space, and joint fuzzy IOs in geosocial space by exemplifying that different allocations of individuals in geosocial space can be equivalent under one framework, but distinguishable under another.   

We corroborate the advantage of fuzzy IOs by a grid optimization over the fuzzy-decay parameters, showing that at small distances fuzzy IOs can halve the mean squared error relative to a standard IO baseline, see Extended Data Fig.~\ref{fig:mse_fuzzy_rank_dist}.
We then show by extensive model comparison across 40 US states and multiple model specifications, using pointwise PSIS-LOO log predictive densities \cite{vehtari2017practical}, that joint IO models built upon fuzzy IOs consistently provide a better fit to the SCI data than hybrid models or distance-based models, see Fig.~\ref{fig:IOs}, panel \textbf{D2}; panel \textbf{D1} shows model comparison including additional models for the largest state in the dataset, California.
Thus, the SCI is better explained by opportunities in joint geosocial space than by a factorization of a geographical distance or geographical IO component with a social dissimilarity component.

\subsection{Inferred connectivity reveals strong group homophily along race, education, and income}

Bayesian inference under the joint IO model yields posterior estimates of latent sociodemographic distances and marginal group-to-group connection probabilities.
The inferred distances indicate clear homophily patterns by race, education, and income (Fig.~\ref{fig:marg_conn_hom_prior}\textbf{a},\textbf{b}). White individuals are distant from Black, Latino, and Other groups; individuals with college degrees are distant from those without; and the highest income group is distant from the other income groups. 
The inferred distances are mirrored by the marginal group-to-group connection probabilities: groups with larger inferred distances have lower marginal probabilities of connection (Fig.~\ref{fig:marg_conn_hom_prior}\textbf{c}). Panels \textbf{d}--\textbf{f} show marginal connectivity under the counterfactual models $M_{HB}$, $M_{AB}$, and $M_{RI}$ for social, geographical, and traditional geographical segregation, respectively.

Homophily is more pronounced under the social segregation counterfactual $M_{HB}$ than under the traditional geographical segregation counterfactual $M_{RI}$; the geographical segregation counterfactual $M_{AB}$ exhibits even less homophily. Using individually resolved social network data rather than ARD, ref.~\cite{kazmina2024socio} similarly finds stronger segregation in their social than in their spatial network. The qualitative patterns for education and income are similar under the counterfactual models for social and geographical segregation, consistent with findings for income in ref.~\cite{xu2019quantifying}, also based on individually resolved social network data rather than ARD. Conceptually, the social-network segregation considered in refs.~\cite{kazmina2024socio,xu2019quantifying} corresponds to our total segregation rather than to our social segregation, while the geographical segregation in ref.~\cite{kazmina2024socio} corresponds to our traditional geographical segregation. We note that the inferred marginal connectivities are relatively robust to alternative model specifications; see Extended Data Fig.~\ref{fig:marg_conn_4models}.

\begin{figure}[htbp]
    \centering
    \includegraphics[width=0.9\textwidth]{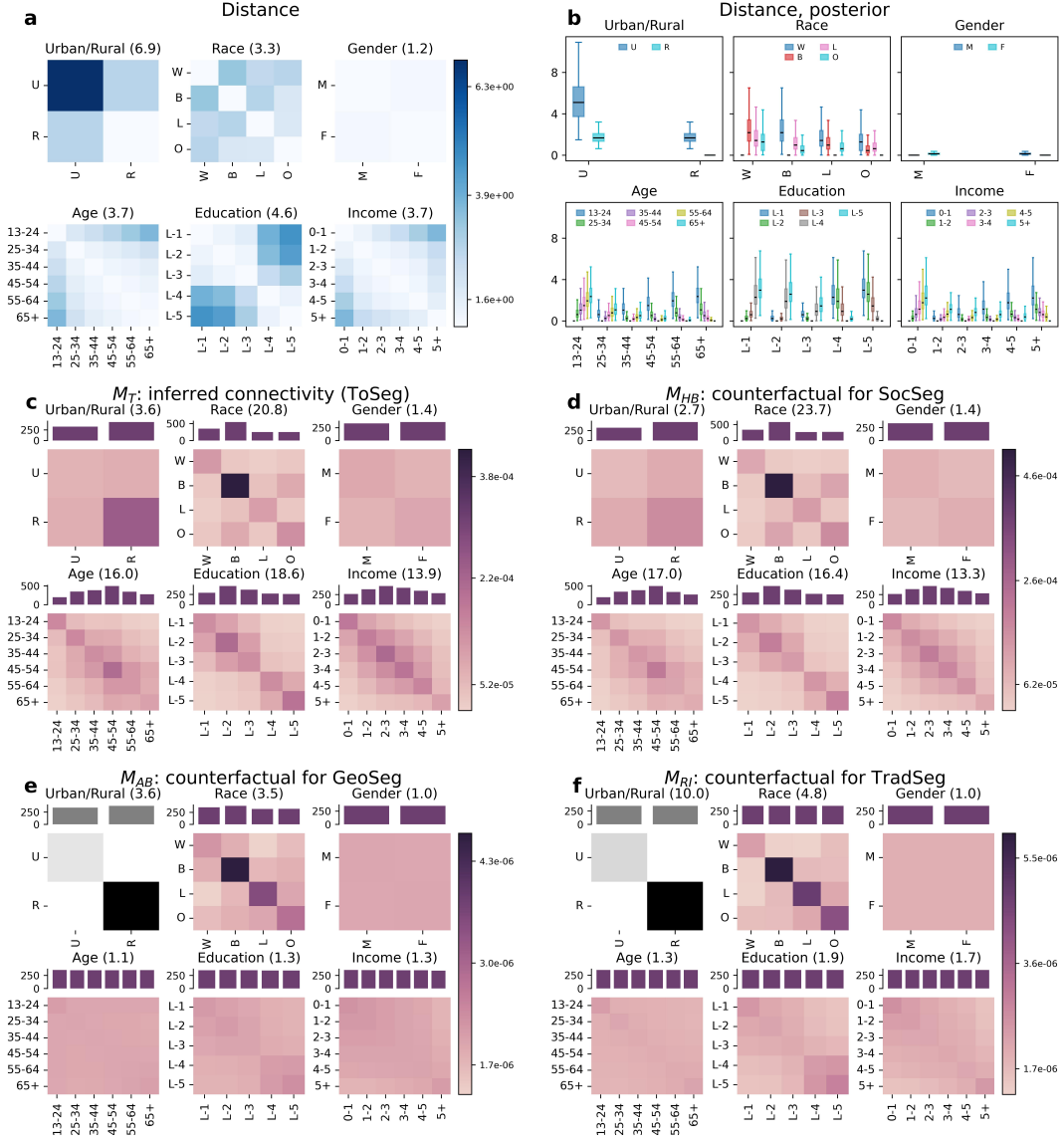}  
    \caption{\emph{Inference yields homophily by race, education and income.} \textbf{a}, Inferred latent social distances for Urban/Rural, Race, Gender, Age, Education, and Income. Heatmaps show posterior means averaged across US states using exposure weights. Numbers in brackets give the ratio of the largest to smallest heatmap entry. \textbf{b}, Posterior distributions of the same distances, obtained by pooling draws from the state-level posteriors. \textbf{c}--\textbf{f}, Posterior mean estimates of marginal connection probabilities under \textbf{c} the inferred network model, and the counterfactual models for \textbf{d} social, \textbf{e} geographical, and \textbf{f} traditional geographical segregation; see Fig.~\ref{fig:tot_trad_soc_geo_seg_w_formulas} for model illustrations. 
    Patterns for geographical and traditional geographical segregation are very similar, with somewhat larger group differences under traditional geographical segregation. Social segregation shows much more pronounced homophily, particularly by race, age, education, and income. White individuals (W) are especially segregated from Black (B), Latino (L), and Other (O) groups; individuals with college degrees (L-4 and L-5) are segregated from those without; and the affluent (income/poverty ratio > 5) are segregated from others. These patterns are strongest under social segregation but remain visible under the geographical variants. Estimated average degrees (average number of friends) of social groups, shown as bar plots above the heatmaps, should be interpreted with caution because they depend on Facebook usage estimates and may therefore inherit any bias in these usage estimates.}
    \label{fig:marg_conn_hom_prior}
\end{figure}

\subsection{Social segregation predominates over geographical segregation}
We now apply the proposed segregation framework to compare total, social, geographical, and traditional geographical segregation by race, education, and income. Across all three variables, social segregation predominates over geographical segregation (Fig.~\ref{fig:seg_bar_plots_reg_group_decomp}\textbf{A}). Total segregation is only slightly larger than social segregation, whereas social segregation is two to three times larger than traditional geographical segregation, and traditional geographical segregation is itself about twice as large as geographical segregation induced by the SCI network. Thus, most segregation in the inferred geosocial network reflects group-based connectivity preferences rather than the geographical distribution of groups alone.

\begin{figure}[htbp]
    \centering
    \includegraphics[width=1.0\textwidth]{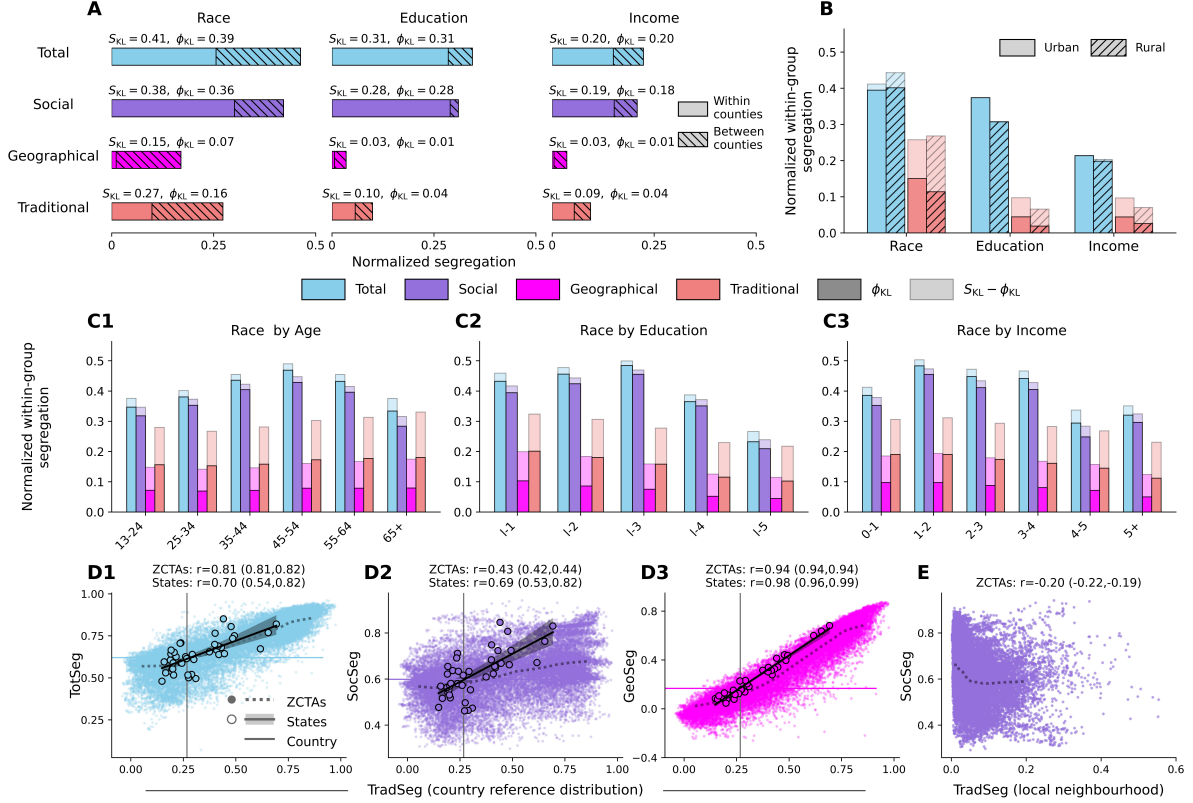} 
    \caption{\emph{Social segregation predominates over geographical segregation.}
    \textbf{A}, Total, social, geographical, and traditional geographical segregation by race, education, and income. Bar lengths show normalized segregation as measured by the KL-divergence $S_{\mathrm{KL}}$, with numerical values for $S_{\mathrm{KL}}$ and the corresponding homophily index $\phi_{\mathrm{KL}}$ shown above the bars. For all three variables, total segregation is slightly stronger than social segregation, which is two to three times stronger than traditional geographical segregation, which in turn is twice as strong as geographical segregation induced by the $\sci$ network. Social segregation indices are largely driven by homophily: the values of $\phi_{\mathrm{KL}}$ are close to $S_{\mathrm{KL}}$. For geographical and traditional geographical segregation, $\phi_{\mathrm{KL}}$ is only about half the size of $S_{\mathrm{KL}}$. We also decompose the segregation indices into within and between-county components (indicated by the hatched area of the bars). Social segregation appears relatively consistent across regions, as indicated by the predominance of within-county segregation, whereas geographical segregation has a larger share of between-county segregation. 
    \textbf{B}, Segregation within urban and rural areas (hatched bars) for race, education and income. Solid shading shows the homophily index $\phi_{\mathrm{KL}}$, and light shading the remaining component, $S_{\mathrm{KL}}-\phi_{\mathrm{KL}}$. While for race, there is little difference between within urban and rural segregation, for education and income there is less traditional geographical segregation within rural than within urban areas. For education, less within-rural segregation is also visible for total segregation. 
    \textbf{C1}--\textbf{C3}, Race segregation within age, education and income groups. Bars show normalized within-group segregation, again decomposed into $\phi_{\mathrm{KL}}$ and $S_{\mathrm{KL}}-\phi_{\mathrm{KL}}$. For all four segregation types, segregation is weaker within college-educated groups (\cat{l-4} and \cat{l-5}), and within relatively affluent groups (\cat{4-5} and \cat{5+}). For traditional geographical segregation, there is an increase of segregation with increasing age; the results for social and total segregation for age are less clear, with a peak among middle-aged groups.
    \textbf{D1}--\textbf{D3}, At the state level, traditional geographical segregation, measured by $\phi_{\mathrm{TV}}$ and relative to the country reference distribution, is moderately positively correlated with total and social segregation, and very strongly positively correlated with geographical segregation. Similar patterns are observed at the ZCTA level when applying the same country-reference framework locally, by comparing the group-specific connectivity of each focal ZCTA with the country-wide reference distribution. Solid reference lines indicate the country-level values of $\phi_{\mathrm{TV}},$ which reduce to relative diversity for traditional geographical segregation and nominal network assortativity for total segregation. \textbf{E}, In contrast, defining local traditional geographical segregation by computing segregation within neighbourhoods of 10 geographically adjacent ZCTAs around each focal ZCTA yields no, or even weakly negative, correlation with social segregation.
    }
    \label{fig:seg_bar_plots_reg_group_decomp}
\end{figure}

The homophily decomposition shows that social segregation is almost entirely homophilous: for race, education, and income, $\phi_{\mathrm{KL}}$ is close to the corresponding normalized $S_{\mathrm{KL}}$ value. 
By contrast, for geographical and traditional geographical segregation, $\phi_{\mathrm{KL}}$ is only about half as large as $S_{\mathrm{KL}}$. This indicates that geographical concentration contributes to segregation not only through homophilous separation between groups, but also through heterophilous deviations from random mixing.
County-level decompositions, enabled by the regional decomposability of $S_{\mathrm{KL}}$, further show that social segregation is mostly within counties, whereas geographical and traditional geographical segregation contain a larger between-county component. Social segregation is therefore not driven primarily by broad between-county variation.

We next apply the same decomposition logic to urban and rural areas by comparing the within-urban and within-rural components of segregation (Fig.~\ref{fig:seg_bar_plots_reg_group_decomp}\textbf{B}). For race, within-urban and within-rural segregation are similar. For education and income, however, traditional geographical segregation is weaker within rural than within urban areas. For education, weaker within-rural segregation is also visible for total segregation.\footnote{Because of how sociodemographic space enters our models, the within-urban/rural decomposition cannot reasonably be applied to social or geographical segregation; see Methods.}

We further decompose race segregation into within-age, within-education, and within-income components (Fig.~\ref{fig:seg_bar_plots_reg_group_decomp}\textbf{C1}--\textbf{C3}). Across all four segregation types, race segregation is lower within college-educated groups and within relatively affluent groups. For traditional geographical segregation, race segregation increases with age, while the corresponding patterns for social and total segregation are less monotonic and peak among middle-aged groups.

Finally, we compare traditional geographical segregation with total, social, and geographical segregation at state and ZCTA level using a country-reference version of $\phi_{\mathrm{TV}}$ (Fig.~\ref{fig:seg_bar_plots_reg_group_decomp}\textbf{D1}--\textbf{D3}). Broadly speaking, the ZCTA-level statistic measures whether individuals in a focal ZCTA are more or less likely to connect to other racial groups than expected under marginal connectivity at the country level; the exact definition is given in Methods. For traditional geographical segregation, this reduces to an isolation-type measure: high values indicate that individuals in the focal ZCTA are less exposed to other racial groups than expected under country-level marginal connectivity. At state level, the corresponding statistic summarizes the same comparison across ZCTAs within each state.
\footnote{The ZCTA-level country-reference construction is similar to the local Divergence Index of ref.~\cite{roberto2018spatial}, which compares the composition of a spatially weighted local environment with the overall population composition; here, the focal distribution is instead the composition of the ZCTA. At state level, our statistic is closely related to relative diversity, but uses country-level rather than state-level marginal connectivity as the reference, to enable a consistent comparison with the ZCTA-level quantities.}
Under this country-reference construction, traditional geographical segregation is moderately positively correlated with total and social segregation, and very strongly positively correlated with geographical segregation, both across states and across ZCTAs.

This association changes when traditional geographical segregation is instead measured locally by computing relative diversity within neighbourhoods of geographically adjacent ZCTAs around each focal ZCTA (Fig.~\ref{fig:seg_bar_plots_reg_group_decomp}\textbf{E}). Under this neighbourhood-based definition, its correlation with social segregation disappears or becomes weakly negative. This suggests that traditional segregation within small neighbourhoods around each focal ZCTA primarily captures local racial diversity, and that higher local diversity may even correspond to lower levels of same-race connectivity preferences, as captured by social segregation.

Taken together, these findings suggest that social segregation not only predominates over geographical segregation, consistent with findings in ref.~\cite{kazmina2024socio}, but is also qualitatively distinct from traditional geographical segregation. Geographical segregation induced by the SCI network, on the other hand, closely tracks traditional geographical segregation, although it is lower in magnitude. Thus, the results emphasize that a complete description of segregation requires resolving total segregation into social and geographical segregation. While traditional geographical segregation captures patterns very similar to geographical segregation induced by an actual geosocial network, it cannot by itself quantify the extent to which geographical concentrations of different social groups actually affect connectivity in that network.

\subsection{Social connectivity patterns vary with geographical distance}

We consider two related questions: first, whether geographical homophily differs across sociodemographic groups; and second, whether social homophily itself varies with geographical distance.

To address the first question, we select ZCTAs that are predominantly composed of one sociodemographic group and label them by their predominant group. We then plot the empirical $\sci$ for pairs of ZCTAs with the same predominant group against geographical distance (Fig.~\ref{fig:geo_blau_inter_dist}\textbf{B1}--\textbf{B4}). These descriptive patterns suggest that the decay of $\sci$ with geographical distance is steeper for low-education and low-income ZCTAs than for high-education and high-income ZCTAs (panels \textbf{B3}, \textbf{B4}). Panel \textbf{B1} suggests a steeper decay for rural than urban ZCTAs, while panel \textbf{B2} suggests a steeper decay for predominantly White than predominantly non-White ZCTAs. These descriptive findings should be interpreted cautiously, because the ZCTAs used in this comparison are selected to have atypically homogeneous sociodemographic composition.

\begin{figure}[htbp]
    \centering
    \includegraphics[width=\textwidth]{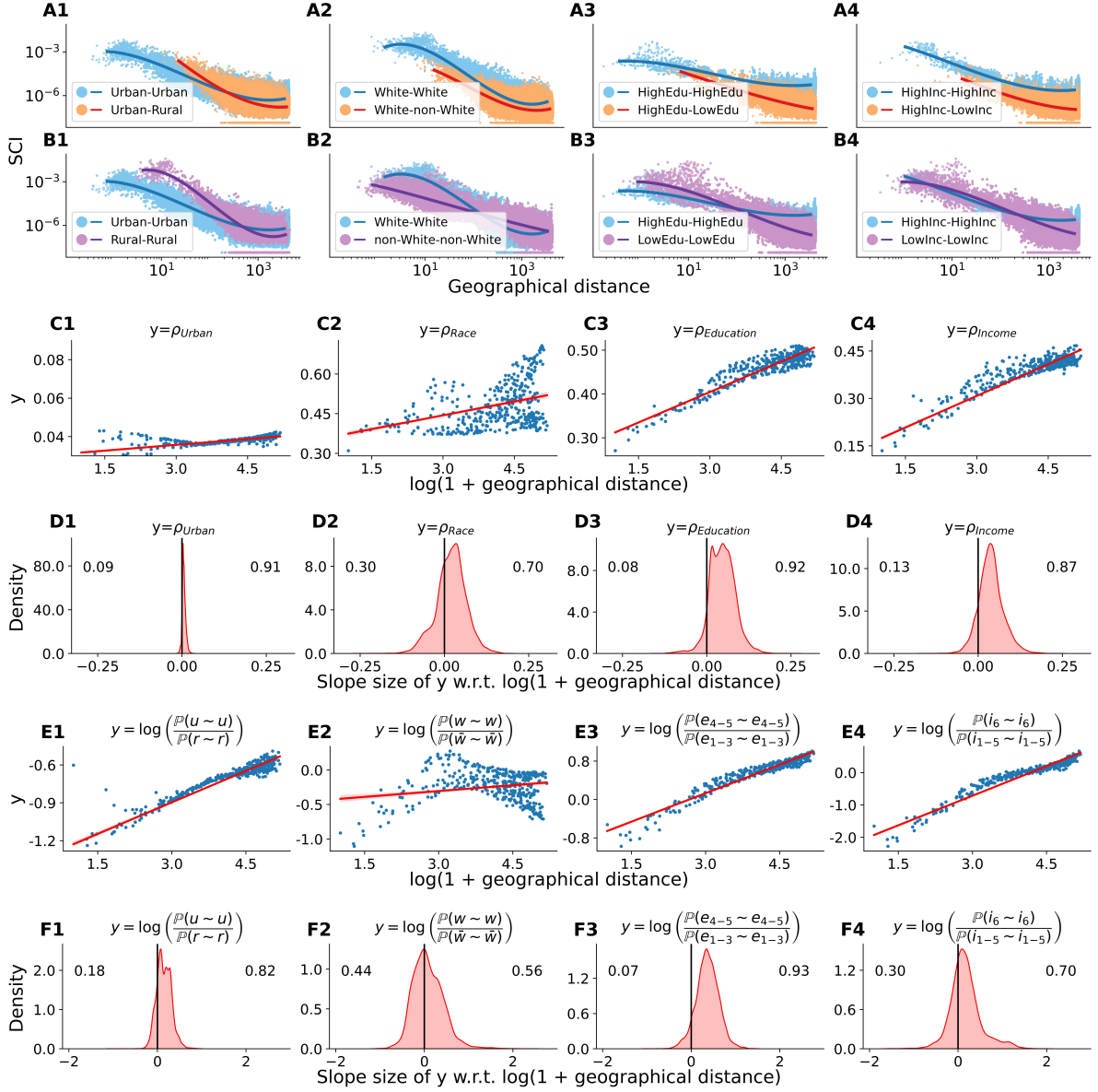}  
    \caption{\emph{Social connectivity patterns interact with geographical distance.} \textbf{A1}--\textbf{A4}, \textbf{B1}--\textbf{B4} Empirical relation between $\sci$ and geographical distance (log--log scale) for pairs of ZCTAs predominantly composed of the indicated groups.
    The solid lines show smooth trends obtained by binning the $\log$ distance, taking the bin-wise median of $\log \sci,$ and fitting a cubic spline to the binned values.
    \textbf{C1}--\textbf{C4} Scatterplots of ZCTA-pair-specific nominal network assortativity $\rho$ against geographical distance from source ZCTA 30306 (Atlanta, Georgia) to all other ZCTAs, together with the fitted slope of $\rho$ with respect to distance. All probabilities are estimated from the social segregation model, to control for any patterns induced by geographical segregation.
    \textbf{D1}--\textbf{D4} Distribution of slopes as in \textbf{C1}--\textbf{C4}, now computed for all source ZCTAs; i.e., each ZCTA is in turn treated as the source, and the slope of $\rho$ with respect to the distance to all other ZCTAs is calculated and included in the distribution.
    \textbf{E1}--\textbf{E4} Scatterplots of ZCTA-pair-specific $\log$ ratios of connection probabilities, as estimated from the social segregation model. $\PP{u\sim u}/\PP{r\sim r}$ denotes the ratio of the probability that an urban individual connects to another urban individual relative to the probability that a rural individual connects to another rural individual; $w,\bar{w}$ denote white and non-white, $e_{1\!-\!3}, e_{4\!-\!5}$ non-college vs college degrees, and $i_{6}, i_{1\!-\!5}$ high-income and low/middle-income groups. 
    \textbf{F1}--\textbf{F4} Distribution of slopes as in \textbf{E1}--\textbf{E4}, now computed for all source ZCTAs.
    }
    \label{fig:geo_blau_inter_dist}
\end{figure}   

We corroborate these patterns using connection probability estimates from the social segregation model, which removes patterns attributable to geographical segregation. We take one ZCTA at a time as the source ZCTA and compare its estimated within-group connection probabilities to other ZCTAs as a function of geographical distance. For each source ZCTA and each target ZCTA, we compute log ratios comparing urban--urban with rural--rural, White--White with non-White--non-White, college--college with non-college--non-college, and high-income--high-income with lower-income--lower-income connection probabilities. We then fit, for each source ZCTA, the slope of each log ratio with respect to the distance to the target ZCTAs. Panels \textbf{E1}--\textbf{E4} illustrate this calculation for ZCTA 30306. Panels \textbf{F1}--\textbf{F4} show the distribution of the fitted slopes when each ZCTA is treated as the source in turn. For education, panel \textbf{F3}, the majority of slopes are positive, corroborating the descriptive pattern in panel \textbf{B3}. For income, panel \textbf{F4} has a heavy right tail, suggesting that in some regions the most affluent are considerably less geographically homophilous than lower-income groups, although this pattern does not hold for all ZCTAs.

We then ask whether social homophily varies with geographical distance. In the empirical SCI, same-group ZCTA pairs have higher connectivity than between-group ZCTA pairs for race, education, and income, indicating homophily (Fig.~\ref{fig:geo_blau_inter_dist}\textbf{A1}--\textbf{A4}). Moreover, the gap between same-group and between-group connectivity widens with geographical distance, suggesting that social homophily increases with distance. We assess this pattern using the social segregation model by computing ZCTA-pair-specific nominal assortativity for urban/rural status, race, education, and income. 
Panels \textbf{C1}--\textbf{C4} show, for ZCTA 30306 (Atlanta, Georgia) as the source ZCTA, that ZCTA-pair-specific assortativity increases with geographical distance for race, education, and income, while urban/rural assortativity changes little. We again fit slopes to these data; panels \textbf{D1}--\textbf{D4} show the corresponding slope distributions when each ZCTA is treated as the source ZCTA in turn.
For the majority of ZCTAs, race, education, and income assortativity increase with geographical distance, whereas the slope distribution for urban/rural status is tightly centered around zero.

\section{Discussion}

We introduced a unifying framework for measuring segregation in geosocial networks that distinguishes total, social, geographical, and traditional geographical segregation. We treat segregation as a deviation from degree-preserving random mixing in the distribution of links over geosocial coordinates, and use counterfactual network models to define social, geographical, and traditional geographical segregation. We decompose our segregation indices into homophilous and heterophilous contributions, and define corresponding homophily indices. Several popular geographical segregation measures such as the Dissimilarity, Relative Diversity, and Theil indices, as well as social-network measures such as network modularity or, when normalized,  nominal network assortativity, can be recovered as special cases under this framework, including the corresponding normalization constants. Furthermore, we generalize the group and regional decomposition properties of the Theil index to arbitrary geosocial networks. Different generators $f$ for the $f$-divergence behave differently when quantifying homophily; while we established some properties (see Methods), a thorough investigation of advantages and disadvantages of different approaches remains future work. KL-divergence based measures have regional and group decomposition properties, while TV-distance yields network modularity as the corresponding homophily measure, whose linear properties can be advantageous.

We then showed how privacy-preserving aggregated relational data, together with detailed joint geosocial census or user counts, enable inference of geosocial connectivity by utilizing the regional variations of sociodemographic proportions. For the statistical inference, we proposed a network model based on intervening opportunities in joint geosocial space, and showed that it consistently fits the SCI data better than simpler models which treat geography and social space as separate building blocks; again suggesting that geographical concentrations and social homophily interact and should be studied jointly. The inferred geosocial connectivity allowed us to investigate total, social, geographical, and traditional geographical segregation for the FB network in the US. 

The conceptual distinction between the different segregation types matters because two societies can exhibit the same level of total segregation while differing substantially in whether this segregation is driven mainly by group-based connectivity preferences or by geographically structured opportunities. Existing notions of social segregation are often entangled with geographical segregation, and the interplay of social and geographical segregation can only be studied once both are well defined and isolated. We found that social segregation predominates over geographical and traditional geographical segregation for race, education, and income. Traditional geographical segregation captures patterns that are qualitatively similar to our geographical segregation as induced by the SCI network, but it is larger in magnitude, and therefore does not quantify the extent to which geographical concentration of groups affects connectivity in the actual network.
Conversely, social segregation is also qualitatively different from traditional geographical segregation, although some patterns persist across all four segregation types, including the separation between college-educated and non-college-educated individuals, and between the very affluent and all other income groups.

We also find that social homophily increases with increasing geographical distance. Taken together with the finding that geographical race isolation in a ZCTA is associated with higher social race homophily of individuals in that ZCTA, this suggests the hypothesis that geographical segregation matters beyond geographically restricting opportunities to connect to other groups: it might also shape same-group connectivity preferences, and missing opportunities to connect to other groups locally might not be compensated by opportunities at larger geographical distances because of the increased homophily at larger distances. 
While most of the results about the interplay between social and geographical segregation and between social homophily and geographical distance may seem intuitive, they are neither logically implied by the conceptual framework nor mechanically enforced by the geosocial network model underlying the inference, and therefore constitute genuine empirical findings that merit further investigation.

The inferred connection probabilities between different sociodemographic groups are latent, since the observed SCI is purely regional; multiple data-generating processes could in principle be compatible with the same region--region connection volumes while implying different patterns of group-to-group connectivity. We addressed this through model comparison, robustness checks, and counterfactual specifications, but unobserved confounding cannot be ruled out as long as sociodemographic coordinates of links are not observed. In addition, the inferred group-specific connection probabilities depend on estimated local sociodemographic composition and Facebook usage rates. We illustrate how confounders or biased usage estimates can bias the inferred connectivity in Extended Data Fig.~\ref{fig:toy_ard_cf_fb_bias}. 
While several of our qualitative findings are already independently supported by results from individual-level network data in other countries \cite{xu2019quantifying,kazmina2024socio}, future work should further validate and calibrate the models used to infer geosocial connectivity using richer data sources like more granular ARD with observed connection volumes between sociodemographic groups, or individual-level network data with known geosocial coordinates of egos and alters. Finally, we note that the SCI captures Facebook friendship ties, and the extent to which these represent other forms of social connection may vary across groups and contexts \cite{dunbar2016onlineoffline}.

While the $f$-divergence framework yields a unification of several pre-existing segregation measures, we also note some of its limitations. It is best suited for categorical sociodemographic variables, as it does not take any ordering of the values of variables with additional structure into consideration. For Total Variation, this reflects the fact that the Total Variation distance corresponds to the Wasserstein-1 metric when the underlying coordinate space is equipped with the discrete metric, which only distinguishes whether two coordinates are identical or distinct. Many commonly used $f$-divergences are quantitatively related to Total Variation through well-established divergence inequalities \cite{polyanskiy2025information}. For ordinal or continuous social variables, a natural extension would be to replace the discrete metric with one that respects the additional structure of the variable, and to use the corresponding Wasserstein-1 distance between the edge distributions $\sigma^M$ and $\sigma^N$ as a segregation measure.
A similar approach could be explored when categorical sociodemographic space is given by the product space of many categorical sociodemographic variables: the discrete metric could be replaced by a Hamming distance or by a learned social distance that better reflects the cumulative structure of group differences than treating all pairs of cross-categories as either distinct or equal. Nevertheless, the simpler approach based on $f$-divergences has the advantage of analytical tractability, and we consider it a natural starting point.

Two avenues for future research stand out. First, it would be informative to replicate the U.S. analysis in other countries and continents, with the long-term goal of producing a global map of total, geographical, and social segregation. A strength of our approach is that it is compatible with publicly available aggregated data such as the SCI and census data, rather than requiring proprietary individual-level network data, making such extensions feasible in principle. Second, naturally arising causal questions cannot be answered with cross-sectional data alone. In particular, existing social ties may themselves shape geographical proximity: longitudinal data suggest that social networks influence residential mobility, with movers preferentially relocating to neighbourhoods where they already have more social contacts \cite{buchel2020residentialmobility}. Our framework quantifies geographical, social, and total segregation at a countrywide scale, but it does not establish the causal processes through which these patterns arise: whether geographical segregation strengthens social segregation, whether social homophily reinforces geographical segregation, or whether both co-evolve through shared institutional and behavioural processes. Answering these questions will require longitudinal or quasi-experimental data. The present results show why such work would be fruitful: total segregation in social networks cannot be understood from geography or social preferences alone, but requires measuring how the two combine to shape who connects to whom.

\newpage
\section{Methods}
\noindent\textit{Note: The current Supplementary Information contains a complete description of the segregation framework, including technical details and proofs. Further supplementary material for the intervening opportunities model and the empirical analysis will be added in a subsequent arXiv version.}

\subsection{A unifying segregation framework}

\subsubsection{Traditional geographical segregation as an f-divergence}
\label{sec:methods_trad_seg_as_fdiv}
Let $\RS$ denote a set of regions and $\BS$ a set of sociodemographic groups; the notation $\BS$ refers to the concept of \emph{Blau space} \cite{blau1977inequality,blau1994structural,mcpherson1983ecology}. We write $\pi_{rb}$ for the proportion of the population in region $r\in\RS$ and group $b\in\BS$, with marginals $\pi_r=\sum_b\pi_{rb}$ and $\pi_b=\sum_r\pi_{rb}$, and conditional proportions $\pi_{b\mid r}=\pi_{rb}/\pi_r$. Several widely used measures of traditional geographical segregation can be written as $f$-divergences between the joint population distribution $(\pi_{rb})_{r,b}$ and the product of its regional and sociodemographic marginals $(\pi_r\pi_b)_{r,b}$.

For two discrete probability distributions $P$ and $Q$, and a convex generator $f\colon[0,\infty)\to(-\infty,\infty]$ satisfying $f(1)=0$, the $f$-divergence of $P$ from $Q$ is
\begin{equation*}
    \DF{P}{Q} = \sum_x Q(x) f\left(\frac{P(x)}{Q(x)}\right).
\end{equation*}
Accordingly, we define a family of unnormalized ``traditional geographical'' segregation measures, indexed by the generator $f$, by
\begin{align}
    \mathcal{D}_{f,\pi} &\coloneqq
    \DF{(\pi_{rb})_{r\in\RS,b\in\BS}}{(\pi_r\pi_b)_{r\in\RS,b\in\BS}}
    = \sum_{r\in\RS}\pi_r \sum_{b\in\BS}\pi_b f\left(\frac{\pi_{b\mid r}}{\pi_b}\right).
    \label{eq:traditional_seg_f_divergence}
\end{align}

Different choices of $f$ recover established traditional geographical segregation measures: $f(x)=x\log x$ yields the unnormalized Theil Information Theory index, $f(x)=\tfrac{1}{2}|x-1|$ the unnormalized multigroup Dissimilarity index, and $f(x)=x^2-1$ the squared coefficient of variation \cite{reardon2002measures}.

This common representation motivates our use of $f$-divergences to measure segregation in geosocial networks. Rather than comparing the population distribution $(\pi_{rb})$ with the product of its marginals directly, we extend the $f$-divergence approach to the distribution of geosocial coordinates among connected pairs of individuals under a network model, and compare this with the corresponding distribution under suitable reference models. As shown below, this network formulation recovers \eqref{eq:traditional_seg_f_divergence} under the regional isolation model, and therefore contains the associated traditional geographical segregation indices as exact special cases.

\subsubsection{Geosocial network models and edge-coordinate distributions}

Let $\CS_1$ and $\CS_2$ denote two, potentially different, coordinate spaces. In the following, a geosocial network model $M$ specifies the marginal coordinate distributions $(\pi_{c_1}^M)_{c_1\in\CS_1}$ and $(\pi_{c_2}^M)_{c_2\in\CS_2}$, together with the conditional probability that two individuals with coordinates $c_1$ and $c_2$ are connected,
\begin{equation*}
    \PS{c_1\sim c_2}{}{M}
    \coloneqq
    \PPS{i\sim j\mid C_1^M(i)=c_1,C_2^M(j)=c_2}{}{M}.
\end{equation*}
Here, $C_1^M(i)$ with values in $\CS_1$ and $C_2^M(j)$ with values in $\CS_2$  denote the coordinates of two randomly selected individuals $i$ and $j$ under model $M$ and $\{i\sim j\}$ denotes the event that they share a link.
The networks considered here are undirected, so the underlying connection probability is symmetric when both individuals are represented in the same coordinate space. We nevertheless label the two individuals as first and second, as the two coordinate spaces might differ as detailed below.

The marginal coordinate distributions and conditional connection probabilities together determine the joint distribution of the coordinates of a random pair of individuals and their link status. We define the corresponding conditional coordinate distribution given a link by
\begin{equation*}
\sigma^M(c_1,c_2)
\coloneqq
\PPS{C_1^M(i)=c_1,C_2^M(j)=c_2\mid i\sim j}{}{M}.
\end{equation*}
For ease of exposition, we formulate the segregation measures below using $\sigma^M$. In Supplementary Information, we also develop the corresponding theory for the full joint coordinate--link distribution and show that, for sparse networks, the normalized segregation measures resulting from the two formulations are approximately equal. Indeed, for the inferred Facebook network considered in our empirical analysis, the two formulations agree well beyond the numerical precision relevant to our analyses.

In the empirical analyses, we use the \emph{geo-egocentric geosocial coordinate space}
\begin{equation}
    \CS_1=\RS\times\BS,
    \qquad
    \CS_2=\BS.
    \label{eq:def_geo_egocentric_coord_space}
\end{equation}
The geographical coordinate of the second individual is therefore marginalized out. In particular, with marginal connectivity $\bar{P}^M\coloneqq\PPS{i\sim j}{}{M}$,
the corresponding conditional coordinate distribution given a link is
\begin{equation}
    \sigma^M(r_1b_1,b_2) = \frac{\pi_{r_1b_1}^M\pi_{b_2}^M\PS{r_1b_1\sim b_2}{}{M}}{\bar{P}^M}.
    \label{eq:def_sigma_geo_egocentric}
\end{equation}
By retaining the region and sociodemographic group of the first individual but marginalizing over the region of the second, segregation depends on which sociodemographic groups individuals in each region connect to, rather than on where those connections are geographically located. We therefore consider this geo-egocentric space the appropriate coordinate space for measuring segregation, rather than the full geosocial coordinates of both individuals.

\subsubsection{Representation of segregation via counterfactual and reference models}
\label{sec:segregation_via_counterfactual_models}

We distinguish total, social, geographical and traditional geographical segregation by associating each with a network model. Let $M$ denote the geosocial network model of interest, which may be the true data-generating model or an estimate thereof. Total segregation is defined directly from $M$. The remaining three forms of segregation are defined through counterfactual models that selectively retain or remove the geographical and sociodemographic mechanisms represented by $M$.

Although segregation is measured in the geo-egocentric coordinate space \eqref{eq:def_geo_egocentric_coord_space}, the counterfactual geosocial network models below are most naturally defined using the full geosocial coordinates of both individuals. The regional coordinate of the second individual is subsequently marginalized out when constructing $\sigma$.

\textbf{Social segregation.} The \emph{homogeneous sociodemographic space model} $M_{HB}=M_{HB}(M)$, where $HB$ refers to homogeneous \emph{Blau space}, counterfactually reallocates individuals such that every region has the same sociodemographic composition as the overall population, while retaining the connection probabilities conditional on individuals' geosocial coordinates:
\begin{equation*}
    \begin{aligned}
        \pi_{rb}^{M_{HB}}&=\pi_r^M\pi_b^M,\\
        \PS{r_1b_1\sim r_2b_2}{}{M_{HB}}
        &=\PS{r_1b_1\sim r_2b_2}{}{M}.
    \end{aligned}
\end{equation*}
Thus, $M_{HB}$ removes the geographical concentration of sociodemographic groups while retaining group-specific connectivity patterns, including their variation across geographical space.

\textbf{Geographical segregation.} The \emph{social ambiphily model} $M_{AB}=M_{AB}(M)$, where $AB$ refers to ambiphily in \emph{Blau space}, retains the observed joint distribution of regions and sociodemographic groups, but removes the direct dependence of connectivity on sociodemographic group membership, conditional on geographical coordinates:
\begin{equation*}
\begin{aligned}
    \pi_{rb}^{M_{AB}}&=\pi_{rb}^M,\\
    \PS{r_1b_1\sim r_2b_2}{}{M_{AB}}
    &\coloneqq
    \sum_{b_1^\prime\in\BS}\sum_{b_2^\prime\in\BS}
    \pi_{b_1^\prime}^M\pi_{b_2^\prime}^M
    \PS{r_1b_1^\prime\sim r_2b_2^\prime}{}{M}.
\end{aligned}
\end{equation*}
Averaging using the marginal sociodemographic distribution $(\pi_b^M)_{b\in\BS}$ rather than the local distributions $(\pi_{b\mid r}^M)_{b\in\BS}$ prevents the resulting region-pair connectivity from being weighted by local sociodemographic composition. The resulting connection probability depends on $r_1$ and $r_2$, but not on $b_1$ or $b_2$.

\textbf{Traditional geographical segregation.} The \emph{regional isolation model} $M_{RI}$ retains the observed population distribution but assumes that individuals connect only to individuals in the same predefined region:
\begin{equation*}
    \begin{aligned}
        \pi_{rb}^{M_{RI}}&=\pi_{rb}^M,\\
        \PS{r_1b_1\sim r_2b_2}{}{M_{RI}}
        &=
        \begin{cases}
            \displaystyle
            \frac{\bar{P}^M}{\pi_r^M}, & r_1=r_2=r,\\[0.75em]
            0, & r_1\neq r_2.
        \end{cases}
    \end{aligned}
\end{equation*}
The scaling by $\pi_r^M$ gives all individuals the same marginal connection probability $\bar{P}^M$ and preserves the overall density of the network. The model is therefore sensitive only to the distribution of sociodemographic groups across the chosen regions. 

\textbf{Random mixing and maximal segregation.} For each model $M'\in\left\{M,M_{HB},M_{AB},M_{RI}\right\}$, we define a random-mixing null model and a maximal-segregation reference model directly on the geo-egocentric coordinate space \eqref{eq:def_geo_egocentric_coord_space} used to define the segregation measures below.
We denote the row and column marginals of connectivity under model $M'$ by
\begin{equation*}
\begin{aligned}
    \bar{P}_{r_1b_1}^{M'}
    &\coloneqq \sum_{b_2\in\BS} \pi_{b_2}^{M'} \PS{r_1b_1\sim b_2}{}{M'},\\
    \bar{P}_{b_2}^{M'} 
    &\coloneqq \sum_{r_1\in\RS} \sum_{b_1\in\BS} \pi_{r_1b_1}^{M'}\PS{r_1b_1\sim b_2}{}{M'}.
\end{aligned}
\end{equation*}
The \emph{degree configuration model} $M_{DC}(M')$ \cite{chung2002average,newman2006modularity} preserves these marginal connectivities but otherwise forms links by random mixing:
\begin{equation*}
\begin{aligned}
    \pi^{M_{DC}(M')}
    &=\pi^{M'},\\
    \PS{r_1b_1\sim b_2}{}{M_{DC}(M')}
    &=\frac{\bar{P}_{r_1b_1}^{M'}\bar{P}_{b_2}^{M'}}{\bar{P}^{M'}}.
\end{aligned}
\end{equation*}
It therefore provides a no-segregation null model that controls for differences in expected degree across geosocial and social coordinates.

The \emph{degree-constrained social isolation model} $M_{DSI}(M')$ also preserves the coordinate distribution and marginal connectivities of $M'$, but allocates all connectivity within sociodemographic groups:
\begin{equation*}
\begin{aligned}
    \pi^{M_{DSI}(M')} &=\pi^{M'},\\
    \PS{r_1b_1\sim b_2}{}{M_{DSI}(M')}
    &= \begin{cases}
        \displaystyle\frac{\bar{P}_{r_1b_1}^{M'}}{\pi_{b_1}^{M'}},
        & b_1=b_2,\\[0.75em]
        0,
        & b_1\neq b_2.
    \end{cases}
\end{aligned}
\end{equation*}
Thus, $M_{DC}(M')$ and $M_{DSI}(M')$ preserve the same degree constraints as $M'$, but represent respectively random mixing and complete separation between sociodemographic groups.\footnote{These reference models define valid connection probabilities under mild sparsity conditions typically satisfied in large-scale geosocial networks; precise conditions are given in Supplementary Information.}

\subsubsection{Divergence-based segregation indices}
For any segregation-defining model $M'$ and any generator $f$, we define the unnormalized segregation measure as the $f$-divergence of its conditional coordinate distribution $\sigma^{M'}$, defined in \eqref{eq:def_sigma_geo_egocentric}, from the corresponding distribution under the degree configuration model:
\begin{equation*}
    \mathcal{D}_{f,\sigma}(M')
    \coloneqq
    \DF{\sigma^{M'}}
       {\sigma^{M_{DC}(M')}}.
\end{equation*}
The divergence therefore measures the extent to which the distribution induced
by $M'$ departs from degree-constrained random mixing. 

To place segregation measures on a common scale, we normalize this divergence
by its value under the corresponding degree-constrained social isolation model $M_{DSI}(M')$.
Because $M_{DSI}(M')$ preserves the coordinate distribution and marginal connectivities of $M'$, it has the same associated degree configuration model as $M'$, and we define the normalized segregation index by
\begin{equation}
    S_{f,\sigma}(M')
    \coloneqq
    \frac{\mathcal{D}_{f,\sigma}(M')}{\mathcal{D}_{f,\sigma}\!\left(M_{DSI}(M')\right)}.
    \label{eq:def_normalized_segregation_methods}
\end{equation}

Under the mild conditions stated in Supplementary Lemma~\ref{lemma:normalization_of_indices}, $M_{DSI}(M')$ maximizes the $f$-divergence from $M_{DC}(M')$ among models with the same coordinate
distribution and marginal connectivities as $M'$, so that $0\leq S_{f,\sigma}(M')\leq1$, with zero corresponding to random mixing and one to maximal segregation.

The choices $M,$ $M_{HB}$, $M_{AB},$ and $M_{RI}$ for $M'$ define total, social, geographical, and traditional geographical segregation, respectively. The generator $f$ determines how deviations from the random-mixing model are
weighted. In the empirical analyses, we consider the
Kullback--Leibler divergence, generated by $f(x)=x\log x$, and Total Variation,
generated by $f(x)=\tfrac{1}{2}|x-1|$.

\subsubsection{Homophilous and heterophilous contributions}

The segregation index $S_{f,\sigma}(M')$ measures the magnitude of departures from degree-constrained random mixing, but does not distinguish whether these departures are homophilous or heterophilous. We call a departure homophilous if, under $M'$ relative to $M_{DC}(M')$, it increases within-group connectivity or decreases between-group connectivity, and heterophilous if it decreases within-group connectivity or increases between-group connectivity.

For the $f$-divergence of $\sigma^{M'}$ from $\sigma^{M_{DC}(M')}$, we assign the contribution of each coordinate pair to either a homophilous or heterophilous component according to the direction of its connectivity departure from $M_{DC}(M')$. 
The precise construction and a natural derivation are provided in Supplementary Information. After normalization by the same maximal-segregation reference as in \eqref{eq:def_normalized_segregation_methods}, this yields the normalized non-negative homophilous and heterophilous contributions $\alpha$ and $\beta$, respectively, with $S_{f,\sigma}(M')=\alpha_{f,\sigma}(M')+\beta_{f,\sigma}(M').$ We define the signed homophily measure $\phi$ by
\begin{equation*}
\phi_{f,\sigma}(M')\coloneqq\alpha_{f,\sigma}(M')-\beta_{f,\sigma}(M').
\end{equation*}
Since $-S_{f,\sigma}(M')\leq\phi_{f,\sigma}(M')\leq S_{f,\sigma}(M')$, $\phi$ takes values in $[-1,1]$: positive values indicate predominantly homophilous departures from random mixing and negative values predominantly heterophilous departures.

\subsubsection{Recovery of established segregation measures}
The regional isolation model $M_{RI}$ connects our geosocial framework directly to established measures of traditional geographical segregation. Under $M_{RI}$, connectivity depends only on whether two individuals occupy the same region, so the conditional coordinate distribution is determined entirely by regional sociodemographic composition. Supplementary Lemma~\ref{lemma:relation_between_geo_and_social_f_divergence} establishes that, for any generator $f$,
\begin{equation*}
\mathcal{D}_{f,\sigma}(M_{RI})=\mathcal{D}_{f,\pi},
\end{equation*}
where $\mathcal{D}_{f,\pi}$ is the traditional geographical segregation measure defined in \eqref{eq:traditional_seg_f_divergence}. More generally, the same construction recovers spatially weighted (or egocentric) geographical segregation measures \cite{reardon2004measures} when connectivity is determined by a geographical proximity kernel.

Supplementary Lemma~\ref{lemma:normalization_of_indices} further establishes that, for the generators corresponding to established normalized traditional segregation indices, the normalization induced by $M_{DSI}(M_{RI})$ equals the conventional normalization. Consequently, $S_{f,\sigma}(M_{RI})$ recovers the multigroup Dissimilarity index for $f(x)=\tfrac{1}{2}|x-1|$, the Theil Information Theory index for $f(x)=x\log x$, and the conventionally normalized squared coefficient of variation for $f(x)=x^2-1$ \cite{reardon2002measures}.

For Total Variation, the unnormalized signed homophily measure reduces to excess within-group edge mass relative to degree-constrained random mixing, which is the quantity underlying network modularity \cite{newman2003mixing,newman2006modularity}: for $g(x)=\tfrac{1}{2}|x-1|$, $\phi_{g,\sigma}$ recovers nominal assortativity, while the corresponding unnormalized signed homophily measure recovers network modularity. Under $M_{RI}$, $\phi_{g,\sigma}(M_{RI})$ recovers the Relative Diversity index \cite{reardon2002measures}. These relationships are established in Supplementary Lemma~\ref{lemma:homophily_quantification}, Supplementary Corollary~\ref{corollary:normalized_homophily_quantification}, and Supplementary Lemma~\ref{lemma:recovery_of_relative_diversity}.

\subsubsection{Regional and sociodemographic group decompositions}

The traditional Theil Information Theory index admits regional and sociodemographic group decompositions \cite{reardon2002measures}. The chain-rule property of the Kullback--Leibler divergence allows us to extend both decompositions to geosocial networks.

Let $h(x)=x\log x$. For a partition $\mathcal{S}$ of $\RS$ into coarser regions $s$, Supplementary Lemma~\ref{lemma:regional_decomposition_of_H} establishes
\begin{equation*}
\mathcal{D}_{h,\sigma}(M')
= \mathcal{D}_{h,\sigma}^{\mathcal{S}\leftrightarrow}(M')
+ \sum_s w_s^{M'}\mathcal{D}_{h,\sigma}^{(s)}(M').
\end{equation*}
Here, $w_s^{M'}$ is the share of edges whose first endpoint belongs to coarse region $s$ under $M'$. The between-region component $\mathcal{D}_{h,\sigma}^{\mathcal{S}\leftrightarrow}(M')$ is the Kullback--Leibler divergence between the marginal distributions of $(s,b_2)$ induced by $\sigma^{M'}$ and $\sigma^{M_{DC}(M')}$, and therefore measures segregation associated with differences between the coarse regions $s$. The within-region component $\mathcal{D}_{h,\sigma}^{(s)}(M')$ measures the remaining segregation associated with finer regional and sociodemographic variation within $s$. 

For the sociodemographic group decomposition, let $\mathcal{A}$ partition the sociodemographic groups $\BS$ into coarser groups $a$. Supplementary Lemma~\ref{lemma:group_decomposition_of_H} establishes
\begin{equation*}
\mathcal{D}_{h,\sigma}(M')
= \mathcal{D}_{h,\sigma}^{\mathcal{A}\leftrightarrow}(M')
+ \sum_a w_a^{M'}\mathcal{D}_{h,\sigma}^{(a)}(M').
\end{equation*}
Here, $w_a^{M'}$ is the share of edges whose second endpoint belongs to coarse sociodemographic group $a$ under $M'$. The between-group component $\mathcal{D}_{h,\sigma}^{\mathcal{A}\leftrightarrow}(M')$ is the Kullback--Leibler divergence of $\sigma^{M'}$ from the degree-constrained random-mixing null when replacing the finer groups $b$ by their coarse groups $a$, whereas $\mathcal{D}_{h,\sigma}^{(a)}(M')$ measures the remaining segregation associated with finer sociodemographic distinctions, conditional on the second coordinate belonging to coarse group $a$. Under the regional isolation model $M_{RI}$, both decompositions reduce component by component to the corresponding regional and sociodemographic group decompositions of the traditional Theil Information Theory index \cite{reardon2002measures}.

To obtain normalized decompositions, components are normalized by the corresponding components obtained by applying the same decomposition to the maximal-segregation reference $M_{DSI}(M')$, except for the between-region component, which is normalized by an entropy upper bound on the corresponding divergence. Under $M_{RI}$, these normalizations reduce to those of the traditional normalized Theil decompositions; precise formulations are given in Supplementary Lemma~\ref{lemma:normalization_decompositions}.

\subsection{Data sources and construction of the geosocial population}
\subsubsection{Data sources}

We conduct our empirical analysis at the level of US ZIP Code Tabulation Areas (ZCTAs) and characterize the population along six sociodemographic variables: Urban, Race, Gender, Age, Education, and Income. To reconstruct their joint distribution at ZCTA level, we combine the American Community Survey (ACS) 2017--2021 5-year estimates \cite{acs_5y_2017_2021_techdoc}, the 2020 Decennial Census Demographic and Housing Characteristics File (DHC) \cite{census_2020_dhc_techdoc}, and the Current Population Survey Annual Social and Economic Supplement (CPS ASEC) 2022 \cite{census_mdat_cps_asec2022}. The ACS and DHC provide overlapping ZCTA-level sociodemographic counts, while CPS ASEC provides joint sociodemographic information at state level.

To estimate Facebook (FB) usage across sociodemographic groups, we combine individual-level data from the Pew Research Center Core Trends Survey 2021 \cite{pew_2021_core_trends} and American Trends Panel Wave 112 \cite{pew_atp_wave_112} with age- and gender-specific reach estimates from Meta's Marketing API \cite{meta_reach_estimate_guide,meta_ad_account_reach_estimate_ref}. Network connectivity is measured using Meta's Social Connectedness Index (SCI) \cite{bailey2018sci}, which provides aggregated information on Facebook friendships between geographical regions.

\subsubsection{Region--group cells and census reconstruction}
The Facebook network is restricted to the population aged 13 years or older. We characterize the population by Urban (\cat{Urban}, \cat{Rural}), Race (\cat{White}, \cat{Black}, \cat{Latino}, \cat{Other}), Gender (\cat{Male}, \cat{Female}), Age (\cat{13--24}, \cat{25--34}, \cat{35--44}, \cat{45--54}, \cat{55--64}, \cat{65+}), Education (\cat{less than high school}, \cat{high school}, \cat{some college or associate degree}, \cat{Bachelor's degree}, \cat{graduate or professional degree}), and Income, defined by the ratio of income to the national poverty line (\cat{0--1}, \cat{1--2}, \cat{2--3}, \cat{3--4}, \cat{4--5}, \cat{5+}). Each sociodemographic group $b\in\BS$ is a cross-category of these six variables. For a ZCTA $r\in\RS$ and group $b\in\BS$, let $N_{rb}$ denote the estimated number of individuals in the corresponding region--group cell.

No single Census dataset provides the full joint distribution of these variables at ZCTA level. We therefore reconstruct the joint distribution of Race, Gender, Age, Education, and Income using iterative proportional fitting (IPF) \cite{deming1940least}. For each ZCTA, we construct overlapping ZCTA-level margins for Race--Gender--Age, Gender--Age--Education, and Age--Income from the ACS and DHC data described above. We initialize the joint distribution using the corresponding state-level Race--Gender--Age--Education--Income distribution from CPS ASEC, and iteratively adjust it until the reconstructed distribution matches the available ZCTA-level margins. Differences in category definitions and population universes across the underlying datasets are harmonized before applying IPF; the source-specific recoding, harmonization procedures, and numerical implementation of IPF are described in Supplementary Information.

Finally, we combine the reconstructed Race--Gender--Age--Education--Income distribution with the observed Urban/Rural proportions within each ZCTA. Because joint ZCTA-level data linking Urban/Rural status to the remaining five variables are not available in the data used here, we assume that Urban/Rural status is independent of the remaining five sociodemographic variables within each ZCTA and factorize the corresponding distributions. Most ZCTAs are predominantly either Urban or Rural, limiting the impact of this factorization. The resulting region--group counts $N_{rb}$ define the estimated geosocial distribution of the general population used to construct intervening opportunities and, after adjustment for Facebook usage as described below, the corresponding Facebook-user population.

\subsubsection{Facebook-user estimates}

The census counts $N_{rb}$ describe the general population, whereas not everyone aged 13 years or older uses Facebook. We estimate Facebook usage across geosocial groups by combining the Pew Research Center Core Trends Survey 2021 and American Trends Panel Wave 112, which record individual-level Facebook use together with sociodemographic characteristics and coarse regional information, state- and age--gender-specific reach estimates obtained from Meta's Marketing API, and the census estimates described above. We model individual Facebook usage as a Bernoulli random variable conditional on Urban, Race, Gender, Age, Education, Income, and US Census division, including Age--Gender and Age--Race interactions.

We do not treat the Meta API reach estimates as direct observations of the true numbers of Facebook users. Instead, we model them as noisy observations of latent age--gender-specific Facebook-user counts, allowing for age--gender-specific true-positive and true-negative rates. We fit a Bayesian model jointly to the Pew survey responses and API reach estimates. This construction extends multilevel regression and poststratification (MRP) \cite{gelman1997poststratification}: the survey model estimates Facebook-usage probabilities across sociodemographic cells, which are poststratified using census cell counts, while the API reach estimates provide an additional noisy aggregate measurement layer. The resulting Facebook-usage estimates across sociodemographic groups and states, together with the inferred age--gender-specific true-positive and true-negative rates of the API reach estimates, are shown in Fig.~\ref{fig:fb_usage_post}. Details of the survey harmonization, API measurement model, prior distributions, and likelihood are provided in Supplementary Information.

For each state $s$ and sociodemographic group $b$, let $\hat{p}_{s,b}^{FB}$ denote the resulting estimated Facebook usage rate. Assuming that, conditional on the sociodemographic group, Facebook usage is constant across ZCTAs within a state, we estimate the number of Facebook users in a ZCTA $r$ and group $b$ as
\begin{equation*}
    \hat{N}_{rb}^{FB}=\hat{p}_{s,b}^{FB}N_{rb},
\end{equation*}
where $s$ denotes the state containing $r$. 

\subsubsection{Social Connectedness Index and aggregated relational data}
The Facebook Social Connectedness Index (SCI) provides normalized aggregate Facebook friendship volumes between pairs of geographical regions \cite{bailey2018sci}. The underlying region--region friendship counts constitute aggregated relational data (ARD) at the regional level, and we use the geosocial network models described below to infer how these counts are distributed across sociodemographic group pairs.

Up to a global scaling constant, the SCI equals the probability that a randomly selected pair of distinct Facebook users from the corresponding regions are connected. We estimate this constant by calibrating the average degree implied by the SCI to an average Facebook degree of $350$, consistent with a Pew Research Center estimate of $338$ friends for US adult Facebook users \cite{smith2014facebook}. Combining the rescaled SCI with the estimated Facebook-user populations, we reconstruct region--region friendship counts $\widehat{\ard}_{r_1r_2}$.

Because the released SCI incorporates noise, truncation, and rounding, we treat $\widehat{\ard}_{r_1r_2}$ as a noisy observation of the latent friendship count $\ard_{r_1r_2}$. The corresponding observation model and its numerical approximation, as well as the details of SCI preprocessing and ARD reconstruction, are described in Supplementary Information.

\subsubsection{Geographical distance}

For distinct ZCTAs $r_1$ and $r_2$, we define geographical distance $d^P(r_1,r_2)$ as the geodesic distance between their centroids using US Census Gazetteer coordinates \cite{census_2021_gaz_zcta}. For within-ZCTA pairs, we use the expected distance between two randomly located individuals under an equal-area circular approximation of the ZCTA \cite{solomon1978geometric}, capped at $80\%$ of the distance to the nearest other ZCTA; details are given in Supplementary Information.

\subsection{Geosocial network model and intervening opportunities}
\subsubsection{Distance-based and hybrid geosocial models}
\label{sec:distance_hybrid_models}

We model connectivity as a function of individuals' geographical and sociodemographic coordinates. Geographical distance is a strong predictor of connectivity in large-scale social networks \cite{takhteyev2012geography,scellato2011socio,bailey2020social}; in the SCI data, connection probability decreases approximately as a power law with geographical distance, with an exponent close to $2$. Since individual connection probabilities are small, this motivates a model that is linear in log geographical distance on the logit scale.

Let $\Theta_{b_1b_2}$ denote a symmetric group-pair effect and $d^P(r_1,r_2)$ the geographical distance between regions $r_1$ and $r_2$. As a simple starting point, we consider
\begin{equation}
\PPS{r_1b_1\sim r_2b_2}{}{M_{\mathrm{dist}}}
=\operatorname{logit}^{-1}\left(c_0-c_1\log\left(1+d^P(r_1,r_2)\right)-\Theta_{b_1b_2}\right).
\label{eq:distance_model_connection_probability}
\end{equation}
We refer to this as the \emph{distance-based model}, as the model can be interpreted in terms of distance in joint geosocial space: writing $\Theta_{b_1b_2}=c_2\log\left(1+d^B(b_1,b_2)\right)$, where $d^B$ denotes a latent sociodemographic dissimilarity, the linear predictor depends on the joint geosocial distance
\begin{equation}
d(r_1b_1,r_2b_2)
=c_1\log\left(1+d^P(r_1,r_2)\right)+c_2\log\left(1+d^B(b_1,b_2)\right).
\label{eq:joint_geosocial_distance}
\end{equation}
We use \emph{distance} in a generalized sense, without requiring $d^B$ or $d$ to satisfy the axioms of a metric.

A direct mapping from geographical distance to connection probability can predict substantially larger expected degrees in regions with high population density than in regions with low population density, contrary to empirically observed degree distributions. We address these density-induced degree differences by allowing the intercept and distance-decay coefficient to vary across regions through regional random effects.

A more principled way to address these density-induced degree artefacts is to replace geographical distance by intervening opportunities (IOs) \cite{stouffer1940intervening,liben2005geographic,kumar2006navigating,backstrom2010find,simini2012universal,kotsubo2021kernel}. IO models characterize the separation between individuals $i$ and $j$ by the number of individuals $k$ that are geographically closer to $i$ than $j$ is. At regional resolution, we define
\begin{equation}
R_{r_1\to r_2}^{\mathrm{marg}}\coloneqq\sum_{r\in\RS}N_r\,\mathbb{1}_{\{d^P(r_1,r)\leq d^P(r_1,r_2)\}},
\qquad
R_{r_1r_2}^{\mathrm{marg}}\coloneqq\sqrt{R_{r_1\to r_2}^{\mathrm{marg}}R_{r_2\to r_1}^{\mathrm{marg}}},
\label{eq:marginal_geographical_ios}
\end{equation}
where $N_r$ is the population size of region $r$. We refer to $R_{r_1r_2}^{\mathrm{marg}}$ as the \emph{marginal geographical IOs}. The same geographical distance therefore corresponds to more IOs in densely populated regions than in sparsely populated regions.

We obtain a \emph{hybrid model} by replacing geographical distance in \eqref{eq:distance_model_connection_probability} with marginal geographical IOs while retaining the group-pair effects $\Theta_{b_1b_2}$:
\begin{equation}
\PPS{r_1b_1\sim r_2b_2}{}{M_{\mathrm{hybrid}}}
=\operatorname{logit}^{-1}\left(c_0-c_1\log\left(1+R_{r_1r_2}^{\mathrm{marg}}\right)-\Theta_{b_1b_2}\right).
\label{eq:hybrid_model_connection_probability}
\end{equation}

This accounts for geographical population density but continues to model geographical and sociodemographic differences through separate components. Marginal IOs depend on the total local population and therefore do not account for variation in the local density of different sociodemographic groups. If a group is locally rare and connectivity to other groups is weak, the hybrid model can consequently predict unrealistically low expected degrees for members of that group. We address this limitation below by extending the notion of IOs to joint geosocial space.

\subsubsection{Joint intervening opportunities in geosocial space}
\label{sec:joint_ios_and_joint_io_model}
IOs in joint geosocial space are most naturally defined at the level of individuals. Let $p(i)$ and $b(i)$ denote the geographical and sociodemographic coordinates of individual $i$, respectively, write $d^P(i,j)\equiv d^P(p(i),p(j))$, and let $d(i,j)\equiv c_1\log(1+d^P(i,j))+c_2\log(1+d^B(b(i),b(j)))$ denote their joint geosocial distance, corresponding to \eqref{eq:joint_geosocial_distance}. The directional joint IOs from individual $i$ to individual $j$ count, across all sociodemographic groups $b$, all individuals $k$ in group $b$ whose joint geosocial distance from $i$ is no larger than that of $j$:
\begin{equation*}
r_{i\to j}\coloneqq\sum_{b\in\BS}\sum_{k:b(k)=b}\mathbb{1}_{\{d(i,k)\leq d(i,j)\}}.
\end{equation*}

For an individual $k$ in sociodemographic group $b$, the joint-distance comparison can equivalently be written as
\begin{equation*}
d(i,k)\leq d(i,j)
\;\Longleftrightarrow\;
d^P(i,k)\leq\left(1+d^P(i,j)\right)\kappa_{i\to j}(b)-1,
\end{equation*}
where $\kappa_{i\to j}(b)\coloneqq\left[(1+d^B(b(i),b(j)))/(1+d^B(b(i),b))\right]^{c_2/c_1}$ is a group-$b$-specific rescaling factor. Thus, for each sociodemographic group $b$, joint geosocial distance determines the geographical range around $i$ over which members of that group count as intervening opportunities (Fig.~\ref{fig:IOs}, \textbf{C1}--\textbf{C4}). Groups that are sociodemographically closer to $i$ than the target individual $j$ contribute over a larger geographical range, whereas more distant groups contribute over a smaller range. Hence, joint IOs can be constructed from group-specific geographical IOs, with the relevant geographical ranges determined by the latent sociodemographic distances $d^B$.

Individual geographical coordinates are not available to us. We therefore retain the individual-level definition above but approximate its evaluation using regional group-specific population counts and geographical centroids. These provide a discrete grid of observed values of the group-specific opportunity counts, corresponding to the cumulative number of group members as the geographical range around an individual increases. The socially rescaled ranges implied by $\kappa_{i\to j}(b)$ generally fall between these grid points. Directly interpolating the corresponding group- and region-specific opportunity functions during model fitting would be computationally very expensive, so we instead use a local quadratic approximation to how the cumulative counts change under this rescaling. The precise construction, symmetrization, and computational implementation are described in Supplementary Information.

This yields symmetric joint geosocial IOs $\widehat R_{r_1b_1,r_2b_2}$ between region--group cells, which we map to connection probabilities as
\begin{equation}
\PPS{r_1b_1\sim r_2b_2}{}{M_{\mathrm{IO}}}
=\operatorname{logit}^{-1}\left(c_0-c_1\log\widehat R_{r_1b_1,r_2b_2}\right).
\label{eq:joint_io_model_connection_probability}
\end{equation}

This \emph{joint IO model} contains the hybrid and distance-based models as special cases. As shown in Supplementary Information, if sociodemographic composition is homogeneous across regions, the joint IO construction reduces to marginal geographical IOs combined with a sociodemographic group-pair component, yielding the hybrid model. If, in addition, geographical population density is globally uniform, marginal IOs become proportional to squared geographical distance, yielding the distance-based model.

\subsubsection{Fuzzy intervening opportunities}
\label{sec:fuzzy_ios}
Standard IOs impose a sharp boundary between individuals who count as intervening opportunities and those who do not, an issue that is accentuated at regional resolution because all individuals within a region are assigned the same representative geographical distance. We therefore use \emph{fuzzy IOs} \cite{afandizadeh2012fuzzy,kotsubo2021kernel}, which retain full weight within the standard-IO boundary and assign decreasing weights beyond it. Following ref.~\cite{kotsubo2021kernel}, we use an exponential decay, which we additionally truncate after the next $k$ regions beyond the standard-IO boundary. We select the distance-decay scale $\nu$ and truncation parameter $k$ before fitting the full network models by comparing simple predictive models relating SCI to marginal fuzzy IOs across a grid of parameter values (Extended Data Fig.~\ref{fig:mse_fuzzy_rank_dist}), and use $\nu=4\,\mathrm{km}$ and $k=15$ throughout. The fuzzy weights are used to construct both the marginal geographical IOs entering the hybrid model and the group-specific geographical opportunity counts entering the joint IO model; their precise construction is described in Supplementary Information.

\subsubsection{ARD likelihood and conditional geosocial connectivity}

Throughout this subsection, we assume a geosocial network model $M$ is fixed, and all probability distributions and expectations are implicitly understood to be conditional on its model parameters.

For a pair of regions $r_1,r_2$, let $\ard_{r_1r_2b_1b_2}$ denote the unobserved number of Facebook friendships between group $b_1$ in region $r_1$ and group $b_2$ in region $r_2$. These group-specific friendship volumes sum to the regional ARD, $\ard_{r_1r_2}=\sum_{b_1,b_2\in\BS}\ard_{r_1r_2b_1b_2}$. Correspondingly, their expected values satisfy $\mu_{r_1r_2}^{M}=\sum_{b_1,b_2\in\BS}\mu_{r_1r_2b_1b_2}^{M}$.

Let $E_{r_1r_2b_1b_2}^{FB}$ denote the number of possible Facebook-user pairs between the corresponding region--group cells. For a geosocial network model $M$, the expected group-specific friendship volume is
\begin{equation*}
\mu_{r_1r_2b_1b_2}^{M}=E_{r_1r_2b_1b_2}^{FB}\PPS{r_1b_1\sim r_2b_2}{}{M}.
\end{equation*}

To allow friendship volumes to vary more than would be implied by conditionally independent edges, we model the group-specific friendship volumes as independent Negative Binomial random variables. We use a parameterization under which their sum for a given region pair is itself Negative Binomial with mean $\mu_{r_1r_2}^{M}$; this regional Negative Binomial distribution defines the likelihood for the observed $\ard_{r_1r_2}$. Conditional on the regional friendship volume, the group-specific volumes then follow a Dirichlet--multinomial distribution. The precise Negative Binomial parameterization and resulting conditional distribution are given in Supplementary Information. In particular, their conditional means are
\begin{equation*}
\EE^M\left[\ard_{r_1r_2b_1b_2}\mid\ard_{r_1r_2}=m\right]
=m\frac{\mu_{r_1r_2b_1b_2}^{M}}{\mu_{r_1r_2}^{M}}.
\end{equation*}

Assuming that individual edges, viewed as Bernoulli random variables, are exchangeable among pairs of Facebook users with the same region--group coordinates, for $\ard_{r_1r_2}=m$ the corresponding conditional connection probability is
\begin{equation*}
\PPS{r_1b_1\sim r_2b_2\mid \ard_{r_1r_2}=m}{}{M}
=\frac{m}{\mu_{r_1r_2}^{M}}\PPS{r_1b_1\sim r_2b_2}{}{M}.
\end{equation*}

In our empirical analysis, conditioning on the observed regional ARD fixes the total friendship volume between each pair of regions at its observed value, while the fitted geosocial network model determines how this volume is distributed across sociodemographic group pairs. This avoids relying on the unconditional model prediction $\mu_{r_1r_2}^{M}$ for the overall connectivity of a particular region pair. The conditional setup does not depend on the specific functional form of the connection probabilities of $M$ and can, with appropriate modifications, also be applied to finer observed ARD that reports some, but not all, sociodemographic coordinates of individuals.

\subsubsection{Factorization of sociodemographic space}
\label{sec:factorization_of_sociodemographic_space}

The six sociodemographic variables Urban, Race, Gender, Age, Education, and Income define $2\times4\times2\times6\times5\times6=2880$ cross-categories, corresponding to more than $4.1$ million distinct unordered sociodemographic group pairs. Direct inference over this full sociodemographic space is computationally very expensive. We therefore factorize the geosocial network models into lower-dimensional groups of sociodemographic variables. The resulting approximation and its implementation for the joint IO, hybrid, and distance-based models are described in Supplementary Information.

For model comparison of these three models, we use the factorization $\{\{U,G\},\{R,A\},\{E,I\}\}$, where U, G, R, A, E, and I denote Urban, Gender, Race, Age, Education, and Income, respectively. For our final estimates, we fit models using three alternative factorizations, $\{\{U,G\},\{R,A\},\{E,I\}\}$, $\{\{U,G\},\{R,E\},\{A,I\}\}$, and $\{\{U,G\},\{R,I\},\{A,E\}\}$. Marginal connection probabilities are similar across the three factorizations (Extended Data Fig.~\ref{fig:marg_conn_different_factorizations}). To minimize the dependence of our results on any particular factorization, we combine their posterior predictive distributions using predictive stacking \cite{yao_using_2018}, with stacking weights estimated from leave-one-out predictive densities using PSIS-LOO \cite{vehtari2017practical}.

\subsubsection{Model extensions}
\label{sec:model_extensions}
We consider the following model extensions; details are provided in Supplementary Information.

\noindent\textbf{Geosocial interaction.} In the baseline joint IO model, the relative weighting of sociodemographic and geographical distance is governed by the constant ratio $q\coloneqq c_2/c_1$, see \eqref{eq:joint_geosocial_distance} and Section~\ref{sec:joint_ios_and_joint_io_model}. We relax this assumption by modelling $q$ as a positive function of geographical distance.

\noindent\textbf{Race and urban interactions.} In the baseline model, the sociodemographic distance $d^B$ is constructed from lower-dimensional variable-specific distances such that $\log(1+d^B)$ is additive across sociodemographic variables. We relax this assumption by allowing the magnitude of race distance to vary with age, education, or income, and by allowing the magnitude of race, education, or income distances to vary with urban/rural status.

\noindent\textbf{Local traditional geographical segregation as predictor.} To analyze whether local traditional geographical segregation covaries systematically with local social segregation, we allow the magnitude of race, education, and income distances to depend locally on two measures of local traditional geographical segregation. The first compares the sociodemographic composition of the focal ZCTA with the countrywide distribution. The second measures geographical segregation within a local neighbourhood formed by the focal ZCTA and 10 geographically adjacent ZCTAs.

\subsection{Inference and model comparison}
\label{sec:inference_and_model_comparison}

We fit the Bayesian geosocial network models separately to 40 US state-level analysis units, with several smaller US states combined and analyzed jointly. For computational tractability, we restrict each state-level fit to at most $100{,}000$ region pairs, retaining pairs at shorter geographical distances and keeping the same number of destination ZCTAs for each focal ZCTA. All models are implemented in NumPyro \cite{phan2019composable}. Models without sociodemographic structure are first fitted using Hamiltonian Monte Carlo (HMC), and the resulting posterior estimates are used to initialize the corresponding models with sociodemographic structure, which are fitted using stochastic variational inference (SVI).

We compare different geosocial network models using pointwise Pareto-smoothed importance-sampling leave-one-out cross-validation (PSIS-LOO) \cite{vehtari2017practical}, and report differences in their expected log predictive densities. Posterior predictive checks of regional ARD, residual variance, and residual spatial autocorrelation provide additional assessments of model fit, see Extended Data Fig.~\ref{fig:ppc-combined-prior-moran-cv}.

For the final joint IO model used for our segregation estimates, we share information across state-level analysis units. Ideally, information about social distances would be pooled across states in a joint hierarchical model, but fitting such a model is computationally prohibitive. We therefore use a multistage approximation: for each analysis unit, prior distributions are estimated from posterior distributions obtained for the remaining 39 units before refitting the model for that unit. Related multistage approaches have been used to approximate computationally demanding hierarchical Bayesian models \cite{johnson2022greater,lunn2013fully}.

Further details on the state-level analysis units, region-pair selection, prior distributions, inference settings and diagnostics, and multistage approximation are provided in Supplementary Information.

\subsection{Construction of connectivity and segregation estimates}
\label{sec:construction_of_segregation_estimates}
From the final model fits, we obtain posterior estimates of the conditional geosocial connectivities $\PPS{r_1b_1\sim r_2b_2\mid\ard_{r_1r_2}=m}{}{M}$, from which we construct the social, geographical, and traditional geographical segregation models $M_{HB}$, $M_{AB}$, and $M_{RI}$ as defined in Section~\ref{sec:segregation_via_counterfactual_models}; the fitted model $M$ itself defines total segregation. For each of these four segregation-defining models $M'\in\{M,M_{HB},M_{AB},M_{RI}\}$, by marginalizing over ZCTA $r_2$ and over sociodemographic variables not entering a given analysis, we construct the corresponding geo-egocentric edge distributions $\sigma^{M'}(r_1a_1,a_2)$ for Race, Education, and Income separately, and jointly for Race with Age, Education, or Income. We use these distributions to compute the segregation measures and their regional and sociodemographic group decompositions as shown in Fig.~\ref{fig:seg_bar_plots_reg_group_decomp}.

By additionally marginalizing over the focal ZCTA $r_1$, we obtain the countrywide marginal group-to-group connection probabilities shown in Fig.~\ref{fig:marg_conn_hom_prior}.

\section*{Data availability}
All data used in this study were publicly available at the time of access. U.S. Census data were obtained from the American Community Survey, the 2020 Decennial Census Demographic and Housing Characteristics File, and the Current Population Survey Annual Social and Economic Supplement \cite{acs_5y_2017_2021_techdoc,census_2020_dhc_techdoc,census_mdat_cps_asec2022}; Facebook-usage survey data were obtained from the Pew Research Center \cite{pew_2021_core_trends,pew_atp_wave_112}; Facebook reach estimates were obtained through Meta's Marketing API \cite{meta_reach_estimate_guide,meta_ad_account_reach_estimate_ref}; and Social Connectedness Index data were obtained from the publicly released SCI dataset \cite{bailey2018sci}.

\section*{Code availability}
The code used for this study will be made publicly available in a future version of this preprint.




\printbibliography

@article{mcpherson2001birds,
  title={Birds of a feather: Homophily in social networks},
  author={McPherson, Miller and Smith-Lovin, Lynn and Cook, James M},
  journal={Annual review of sociology},
  volume={27},
  number={1},
  pages={415--444},
  year={2001},
  publisher={Annual Reviews 4139 El Camino Way, PO Box 10139, Palo Alto, CA 94303-0139, USA}
}

@book{newman2018networks,
  title={Networks},
  author={Newman, Mark},
  year={2018},
  publisher={Oxford university press}
}

@article{kotsubo2021kernel,
  title={Kernel-based formulation of intervening opportunities for spatial interaction modelling},
  author={Kotsubo, Masaki and Nakaya, Tomoki},
  journal={Scientific reports},
  volume={11},
  number={1},
  pages={950},
  year={2021},
  publisher={Nature Publishing Group UK London}
}

@article{afandizadeh2012fuzzy,
  title={A fuzzy intervening opportunity model to predict home-based shopping trips},
  author={Afandizadeh, Shahriar and Hamedani, Seyed Mehdi Yadi},
  journal={Canadian Journal of Civil Engineering},
  volume={39},
  number={2},
  pages={203--222},
  year={2012},
  publisher={NRC Research Press}
}

@article{bailey2020social,
  title={Social connectedness in urban areas},
  author={Bailey, Michael and Farrell, Patrick and Kuchler, Theresa and Stroebel, Johannes},
  journal={Journal of Urban Economics},
  volume={118},
  pages={103264},
  year={2020},
  publisher={Elsevier}
}

@article{vehtari2017practical,
  title={Practical Bayesian model evaluation using leave-one-out cross-validation and WAIC},
  author={Vehtari, Aki and Gelman, Andrew and Gabry, Jonah},
  journal={Statistics and computing},
  volume={27},
  number={5},
  pages={1413--1432},
  year={2017},
  publisher={Springer}
}

@article{reardon2002measures,
  title={Measures of multigroup segregation},
  author={Reardon, Sean F and Firebaugh, Glenn},
  journal={Sociological methodology},
  volume={32},
  number={1},
  pages={33--67},
  year={2002},
  publisher={Wiley Online Library}
}

@article{reardon2004measures,
  title={Measures of spatial segregation},
  author={Reardon, Sean F and O’Sullivan, David},
  journal={Sociological methodology},
  volume={34},
  number={1},
  pages={121--162},
  year={2004},
  publisher={Wiley Online Library}
}

@article{lee2008beyond,
  title={Beyond the census tract: Patterns and determinants of racial segregation at multiple geographic scales},
  author={Lee, Barrett A and Reardon, Sean F and Firebaugh, Glenn and Farrell, Chad R and Matthews, Stephen A and O'Sullivan, David},
  journal={American sociological review},
  volume={73},
  number={5},
  pages={766--791},
  year={2008},
  publisher={Sage Publications Sage CA: Los Angeles, CA}
}

@article{roberto2018spatial,
  title={The spatial proximity and connectivity method for measuring and analyzing residential segregation},
  author={Roberto, Elizabeth},
  journal={Sociological Methodology},
  volume={48},
  number={1},
  pages={182--224},
  year={2018},
  doi={10.1177/0081175018796871},
  publisher={SAGE Publications Sage CA: Los Angeles, CA}
}

@article{kazmina2024socio,
  title={Socio-economic segregation in a population-scale social network},
  author={Kazmina, Yuliia and Heemskerk, Eelke M and Bok{\'a}nyi, Eszter and Takes, Frank W},
  journal={Social Networks},
  volume={78},
  pages={279--291},
  year={2024},
  publisher={Elsevier}
}

@article{newman2003mixing,
  title={Mixing patterns in networks},
  author={Newman, Mark EJ},
  journal={Physical review E},
  volume={67},
  number={2},
  pages={026126},
  year={2003},
  publisher={APS}
}

@article{bojanowski_measuring_2014,
	title = {Measuring segregation in social networks},
	volume = {39},
	issn = {0378-8733},
	url = {https://www.sciencedirect.com/science/article/pii/S0378873314000239},
	doi = {10.1016/j.socnet.2014.04.001},
	urldate = {2024-12-11},
	journal = {Social Networks},
	author = {Bojanowski, Michał and Corten, Rense},
	month = oct,
	year = {2014},
	pages = {14--32},
}

@article{xu2019quantifying,
  title={Quantifying segregation in an integrated urban physical-social space},
  author={Xu, Yang and Belyi, Alexander and Santi, Paolo and Ratti, Carlo},
  journal={Journal of the Royal Society Interface},
  volume={16},
  number={160},
  pages={20190536},
  year={2019},
  publisher={The Royal Society}
}

@article{chetty2022social,
  title={Social capital I: measurement and associations with economic mobility},
  author={Chetty, Raj and Jackson, Matthew O and Kuchler, Theresa and Stroebel, Johannes and Hendren, Nathaniel and Fluegge, Robert B and Gong, Sara and Gonzalez, Federico and Grondin, Armelle and Jacob, Matthew and others},
  journal={Nature},
  volume={608},
  number={7921},
  pages={108--121},
  year={2022},
  publisher={Nature Publishing Group UK London}
}

@article{liben2005geographic,
  title={Geographic routing in social networks},
  author={Liben-Nowell, David and Novak, Jasmine and Kumar, Ravi and Raghavan, Prabhakar and Tomkins, Andrew},
  journal={Proceedings of the National Academy of Sciences},
  volume={102},
  number={33},
  pages={11623--11628},
  year={2005},
  publisher={National Academy of Sciences}
}

@inproceedings{kumar2006navigating,
  title={Navigating low-dimensional and hierarchical population networks},
  author={Kumar, Ravi and Liben-Nowell, David and Tomkins, Andrew},
  booktitle={European Symposium on Algorithms},
  pages={480--491},
  year={2006},
  organization={Springer}
}

@inproceedings{backstrom2010find,
  title={Find me if you can: improving geographical prediction with social and spatial proximity},
  author={Backstrom, Lars and Sun, Eric and Marlow, Cameron},
  booktitle={Proceedings of the 19th international conference on World wide web},
  pages={61--70},
  year={2010}
}

@article{stouffer1940intervening,
  title={Intervening opportunities: a theory relating mobility and distance},
  author={Stouffer, Samuel A},
  journal={American sociological review},
  volume={5},
  number={6},
  pages={845--867},
  year={1940},
  publisher={JSTOR}
}

@article{simini2012universal,
  title={A universal model for mobility and migration patterns},
  author={Simini, Filippo and Gonz{\'a}lez, Marta C and Maritan, Amos and Barab{\'a}si, Albert-L{\'a}szl{\'o}},
  journal={Nature},
  volume={484},
  number={7392},
  pages={96--100},
  year={2012},
  publisher={Nature Publishing Group UK London}
}

@article{takhteyev2012geography,
  title={Geography of Twitter networks},
  author={Takhteyev, Yuri and Gruzd, Anatoliy and Wellman, Barry},
  journal={Social networks},
  volume={34},
  number={1},
  pages={73--81},
  year={2012},
  publisher={Elsevier}
}

@inproceedings{scellato2011socio,
  title={Socio-spatial properties of online location-based social networks},
  author={Scellato, Salvatore and Noulas, Anastasios and Lambiotte, Renaud and Mascolo, Cecilia},
  booktitle={Proceedings of the international AAAI conference on web and social media},
  volume={5},
  number={1},
  pages={329--336},
  year={2011}
}

@article{wong1997spatial,
  title={Spatial dependency of segregation indices},
  author={Wong, David WS},
  journal={Canadian Geographer/Le Geographe Canadien},
  volume={41},
  number={2},
  pages={128--136},
  year={1997},
  publisher={Wiley Online Library}
}

@article{openshaw1984modifiable,
  title={The modifiable areal unit problem},
  author={Openshaw, Stan},
  journal={Concepts and techniques in modern geography},
  year={1984},
  publisher={GeoBooks}
}

@article{bailey2018sci,
  title={{Social Connectedness: Measurement, Determinants, and Effects}},
  author={Bailey, Michael and Cao, Rachel and Kuchler, Theresa and Stroebel, Johannes and Wong, Arlene},
  journal={Journal of Economic Perspectives},
  volume={32},
  number={3},
  pages={259--80},
  year={2018},
  doi={10.1257/jep.32.3.259}
}

@article{moody2001race,
  title={Race, school integration, and friendship segregation in America},
  author={Moody, James},
  journal={American journal of Sociology},
  volume={107},
  number={3},
  pages={679--716},
  year={2001},
  publisher={The University of Chicago Press}
}

@article{holland1983stochastic,
  title={Stochastic blockmodels: First steps},
  author={Holland, Paul W and Laskey, Kathryn Blackmond and Leinhardt, Samuel},
  journal={Social networks},
  volume={5},
  number={2},
  pages={109--137},
  year={1983},
  publisher={Elsevier}
}

@article{karrer2011stochastic,
  title={Stochastic blockmodels and community structure in networks},
  author={Karrer, Brian and Newman, Mark EJ},
  journal={Physical Review E—Statistical, Nonlinear, and Soft Matter Physics},
  volume={83},
  number={1},
  pages={016107},
  year={2011},
  publisher={APS}
}

@article{newman2006modularity,
  title={Modularity and community structure in networks},
  author={Newman, Mark EJ},
  journal={Proceedings of the national academy of sciences},
  volume={103},
  number={23},
  pages={8577--8582},
  year={2006},
  publisher={National Academy of Sciences}
}

@article{chung2002average,
  title={The average distances in random graphs with given expected degrees},
  author={Chung, Fan and Lu, Linyuan},
  journal={Proceedings of the National Academy of Sciences},
  volume={99},
  number={25},
  pages={15879--15882},
  year={2002},
  publisher={National Academy of Sciences}
}

@article{mccormick2012latent,
  title={Latent demographic profile estimation in hard-to-reach groups},
  author={McCormick, Tyler H and Zheng, Tian},
  journal={The annals of applied statistics},
  volume={6},
  number={4},
  pages={1795},
  year={2012}
}

@article{maltiel2015estimating,
  title={Estimating population size using the network scale up method},
  author={Maltiel, Rachael and Raftery, Adrian E and McCormick, Tyler H and Baraff, Aaron J},
  journal={The annals of applied statistics},
  volume={9},
  number={3},
  pages={1247},
  year={2015}
}

@article{feehan2016generalizing,
  title={Generalizing the network scale-up method: a new estimator for the size of hidden populations},
  author={Feehan, Dennis M and Salganik, Matthew J},
  journal={Sociological methodology},
  volume={46},
  number={1},
  pages={153--186},
  year={2016},
  publisher={Sage Publications Sage CA: Los Angeles, CA}
}

@book{polyanskiy2025information,
  title={Information theory: From coding to learning},
  author={Polyanskiy, Yury and Wu, Yihong},
  year={2025},
  publisher={Cambridge university press}
}

@book{blau1977inequality,
  title={Inequality and Heterogeneity: A Primitive Theory of Social Structure},
  author={Blau, Peter M.},
  year={1977},
  publisher={Free Press},
  address={New York}
}

@book{blau1994structural,
  title={Structural contexts of opportunities},
  author={Blau, Peter M},
  year={1994},
  publisher={University of Chicago Press}
}

@article{mcpherson1983ecology,
  title={An ecology of affiliation},
  author={McPherson, Miller},
  journal={American Sociological Review},
  pages={519--532},
  year={1983},
  publisher={JSTOR}
}

@article{yao_using_2018,
	title = {Using {Stacking} to {Average} {Bayesian} {Predictive} {Distributions} (with {Discussion})},
	volume = {13},
	issn = {1936-0975, 1931-6690},
	url = {https://projecteuclid.org/journals/bayesian-analysis/volume-13/issue-3/Using-Stacking-to-Average-Bayesian-Predictive-Distributions-with-Discussion/10.1214/17-BA1091.full},
	doi = {10.1214/17-BA1091},
	number = {3},
	urldate = {2025-09-18},
	journal = {Bayesian Analysis},
	author = {Yao, Yuling and Vehtari, Aki and Simpson, Daniel and Gelman, Andrew},
	month = sep,
	year = {2018},
	note = {Publisher: International Society for Bayesian Analysis},
	pages = {917--1007},
}

@article{borgonovo2025convexity,
  title={Convexity and measures of statistical association},
  author={Borgonovo, Emanuele and Figalli, Alessio and Ghosal, Promit and Plischke, Elmar and Savar{\'e}, Giuseppe},
  journal={Journal of the Royal Statistical Society Series B: Statistical Methodology},
  volume={87},
  number={4},
  pages={1281--1304},
  year={2025},
  publisher={Oxford University Press UK}
}

@article{deming1940least,
  title={On a least squares adjustment of a sampled frequency table when the expected marginal totals are known},
  author={Deming, W Edwards and Stephan, Frederick F},
  journal={The Annals of Mathematical Statistics},
  volume={11},
  number={4},
  pages={427--444},
  year={1940},
  publisher={JSTOR}
}

@misc{acs_5y_2017_2021_techdoc,
  author = {{U.S. Census Bureau}},
  title  = {2017--2021 American Community Survey (ACS) 5-Year Estimates: Technical Documentation and User Notes},
  year   = {2022},
  url    = {https://www.census.gov/programs-surveys/acs/technical-documentation/table-and-geography-changes/2021/5-year.html},
  urldate= {2025-10-03}
}

@misc{census_2020_dhc_techdoc,
  author = {{U.S. Census Bureau}},
  title  = {2020 Census Demographic and Housing Characteristics (DHC) File: Technical Documentation},
  year   = {2023},
  url    = {https://www2.census.gov/programs-surveys/decennial/2020/technical-documentation/complete-tech-docs/demographic-and-housing-characteristics-file-and-demographic-profile/2020census-demographic-and-housing-characteristics-file-and-demographic-profile-techdoc.pdf},
  urldate= {2025-10-03}
}

@dataset{census_mdat_cps_asec2022,
  author       = {{U.S. Census Bureau}},
  title        = {CPS ASEC 2022 microdata: custom tabulation via MDAT},
  year         = {2022},
  howpublished = {Microdata Access Tool (MDAT)},
  url          = {https://data.census.gov/app/mdat/CPSASEC2022/},
  urldate      = {2023-08-18},
  note         = {Person weight: MARSUPWT; geography: 50 states + DC.}
}

@dataset{pew_2021_core_trends,
  author       = {{Pew Research Center}},
  title        = {2021 Core Trends Survey},
  year         = {2021},
  howpublished = {Dataset, Pew Research Center},
  address      = {Washington, DC},
  url          = {https://www.pewresearch.org/internet/dataset/2021-core-trends-survey/},
  urldate      = {2023-09-06},
  note         = {Field dates: Jan.~25--Feb.~8,~2021. Dataset page links reports such as \emph{Social Media Use in 2021} and \emph{Mobile Technology and Home Broadband 2021}.}
}

@dataset{pew_atp_wave_112,
  author       = {{Pew Research Center}},
  title        = {American Trends Panel (ATP), Wave 112},
  year         = {2022},
  howpublished = {Dataset, Pew Research Center},
  address      = {Washington, DC},
  url          = {https://www.pewresearch.org/internet/dataset/american-trends-panel-wave-112/},
  urldate      = {2023-10-10},
  note         = {Field dates: July~18--Aug.~21,~2022; topic: social media use (update). Topline: total $N{=}12{,}147$, MOE $\pm1.4$~pp.}
}

@misc{smith2014facebook,
  author       = {Smith, Aaron},
  title        = {What People Like and Dislike About Facebook},
  year         = {2014},
  month        = feb,
  day          = {3},
  organization = {Pew Research Center},
  url          = {https://www.pewresearch.org/short-reads/2014/02/03/what-people-like-dislike-about-facebook/},
  note         = {Accessed 13 August 2026}
}

@misc{meta_reach_estimate_guide,
  author = {{Meta for Developers}},
  title  = {Reach Estimate API --- Marketing API},
  year   = {2023},
  note   = {Accessed 2023-09-19},
  url    = {https://developers.facebook.com/docs/marketing-api/audiences/guides/reach-estimate/}
}

@misc{meta_ad_account_reach_estimate_ref,
  author = {{Meta for Developers}},
  title  = {Ad Account Reach Estimate --- Marketing API Reference},
  year   = {2023},
  note   = {Accessed 2023-09-19},
  url    = {https://developers.facebook.com/docs/marketing-api/reference/ad-account/reachestimate/}
}

@dataset{census_2021_gaz_zcta,
  author    = {{U.S. Census Bureau}},
  title     = {2021 Gazetteer Files — ZIP Code Tabulation Areas (ZCTAs): 2021\_Gaz\_zcta\_national.txt},
  year      = {2021},
  publisher = {U.S. Census Bureau},
  url       = {https://www2.census.gov/geo/docs/maps-data/data/gazetteer/2021_Gazetteer/2021_Gaz_zcta_national.zip},
  urldate   = {2022-02-04}
}

@book{solomon1978geometric,
  title={Geometric probability},
  author={Solomon, Herbert},
  year={1978},
  publisher={SIAM}
}

@article{johnson2022greater,
  title={Greater Than the Sum of its Parts: Computationally Flexible Bayesian Hierarchical Modeling},
  author={Johnson, Devin S. and Brost, Brian M. and Hooten, Mevin B.},
  journal={Journal of Agricultural, Biological and Environmental Statistics},
  volume={27},
  pages={382--400},
  year={2022},
  doi={10.1007/s13253-021-00485-9}
}

@article{lunn2013fully,
  title={Fully Bayesian Hierarchical Modelling in Two Stages, with Application to Meta-Analysis},
  author={Lunn, David and Barrett, Jessica and Sweeting, Michael and Thompson, Simon},
  journal={Journal of the Royal Statistical Society: Series C (Applied Statistics)},
  volume={62},
  number={4},
  pages={551--572},
  year={2013},
  doi={10.1111/rssc.12007}
}

@article{gelman1997poststratification,
  title={Poststratification into Many Categories Using Hierarchical Logistic Regression},
  author={Gelman, Andrew and Little, Thomas C.},
  journal={Survey Methodology},
  volume={23},
  number={2},
  pages={127--135},
  year={1997}
}

@article{phan2019composable,
  title={Composable Effects for Flexible and Accelerated Probabilistic Programming in {NumPyro}},
  author={Phan, Du and Pradhan, Neeraj and Jankowiak, Martin},
  journal={arXiv preprint arXiv:1912.11554},
  year={2019}
}

@article{schelling1971segregation,
author = {Thomas C. Schelling},
title = {Dynamic models of segregation},
journal = {The Journal of Mathematical Sociology},
volume = {1},
number = {2},
pages = {143--186},
year = {1971},
publisher = {Routledge},
doi = {10.1080/0022250X.1971.9989794},


URL = { 
    
        https://doi.org/10.1080/0022250X.1971.9989794
    
    

},
eprint = { 
    
        https://doi.org/10.1080/0022250X.1971.9989794
    
    

}
}

@article{massey1988segregation,
    author = {Massey, Douglas S. and Denton, Nancy A.},
    title = {The Dimensions of Residential Segregation},
    journal = {Social Forces},
    volume = {67},
    number = {2},
    pages = {281-315},
    year = {1988},
    month = {12},
    issn = {0037-7732},
    doi = {10.1093/sf/67.2.281},
    url = {https://doi.org/10.1093/sf/67.2.281},
    eprint = {https://academic.oup.com/sf/article-pdf/67/2/281/6514769/67-2-281.pdf},
}

@article{duncan1955segregation,
 ISSN = {00031224},
 URL = {http://www.jstor.org/stable/2088328},
 author = {Otis Dudley Duncan and Beverly Duncan},
 journal = {American Sociological Review},
 number = {2},
 pages = {210--217},
 publisher = {[American Sociological Association, Sage Publications, Inc.]},
 title = {A Methodological Analysis of Segregation Indexes},
 urldate = {2026-09-05},
 volume = {20},
 year = {1955}
}

@article{dunbar2016onlineoffline,
    author = {Dunbar, R. I. M.},
    title = {Do online social media cut through the constraints that limit the size of offline social networks?},
    journal = {Royal Society Open Science},
    volume = {3},
    number = {1},
    pages = {150292},
    year = {2016},
    month = {01},
    issn = {2054-5703},
    doi = {10.1098/rsos.150292},
    url = {https://doi.org/10.1098/rsos.150292},
    eprint = {https://royalsocietypublishing.org/rsos/article-pdf/doi/10.1098/rsos.150292/224013/rsos.150292.pdf},
}

@article{buchel2020residentialmobility,
title = {Calling from the outside: The role of networks in residential mobility},
journal = {Journal of Urban Economics},
volume = {119},
pages = {103277},
year = {2020},
issn = {0094-1190},
doi = {10.1016/j.jue.2020.103277},
url = {https://www.sciencedirect.com/science/article/pii/S0094119020300486},
author = {Konstantin Büchel and Maximilian V. Ehrlich and Diego Puga and Elisabet Viladecans-Marsal}
}

\newpage
\section{Extended Data}
\setcounter{figure}{0}
{
\captionsetup[figure]{name=Extended Data Fig.}

\begin{figure}[htbp]
    \centering
    \includegraphics[width=0.9\textwidth]{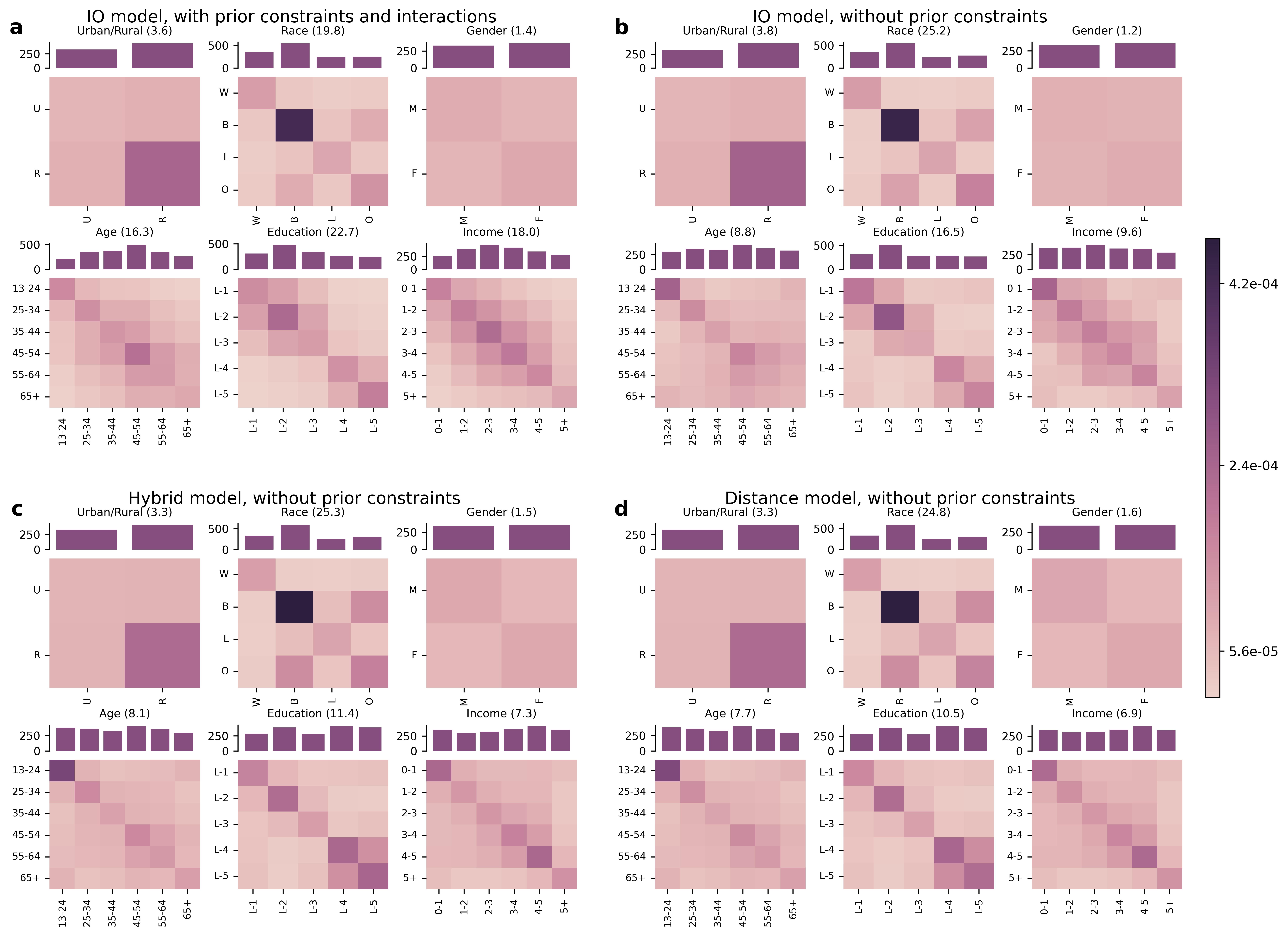}  
    \caption{\emph{Marginal connectivity patterns are qualitatively similar across different model fits.} \textbf{a}, Marginal connectivity between various sociodemographic groups as inferred from a joint geosocial intervening opportunities model with homophily constraints built into the priors (except for the variable \emph{Urban}). \textbf{b}, Marginal connectivity as inferred from a joint geosocial intervening opportunities model with agnostic priors, i.e., no homophily constraints. \textbf{c}, Marginal connectivity as inferred from a hybrid model, where geographical space is modeled via intervening opportunities, and social space as a linear logit model. \textbf{d}, Marginal connectivity as inferred from a purely distance-based model (essentially implemented as a linear logit model). Overall, the qualitative patterns are similar across these four different model fits. However, there is substantial variation in the estimated marginal degrees of various age, education, and income groups --- the inferred degrees should be interpreted with caution. Furthermore, if age connectivity is inferred without homophily constraints, \textbf{b}--\textbf{d}, there seems to be a slight increase in connectivity between the youngest and oldest age groups compared to a model with an ordinal homophily constraint for age as depicted in \textbf{a}; to the extent that this higher young-old connectivity corresponds to the ground truth, model \textbf{a} is slightly misspecified.}
    \label{fig:marg_conn_4models}
\end{figure}

\begin{figure}[htbp]
    \centering 
    \includegraphics[width=0.9\textwidth]{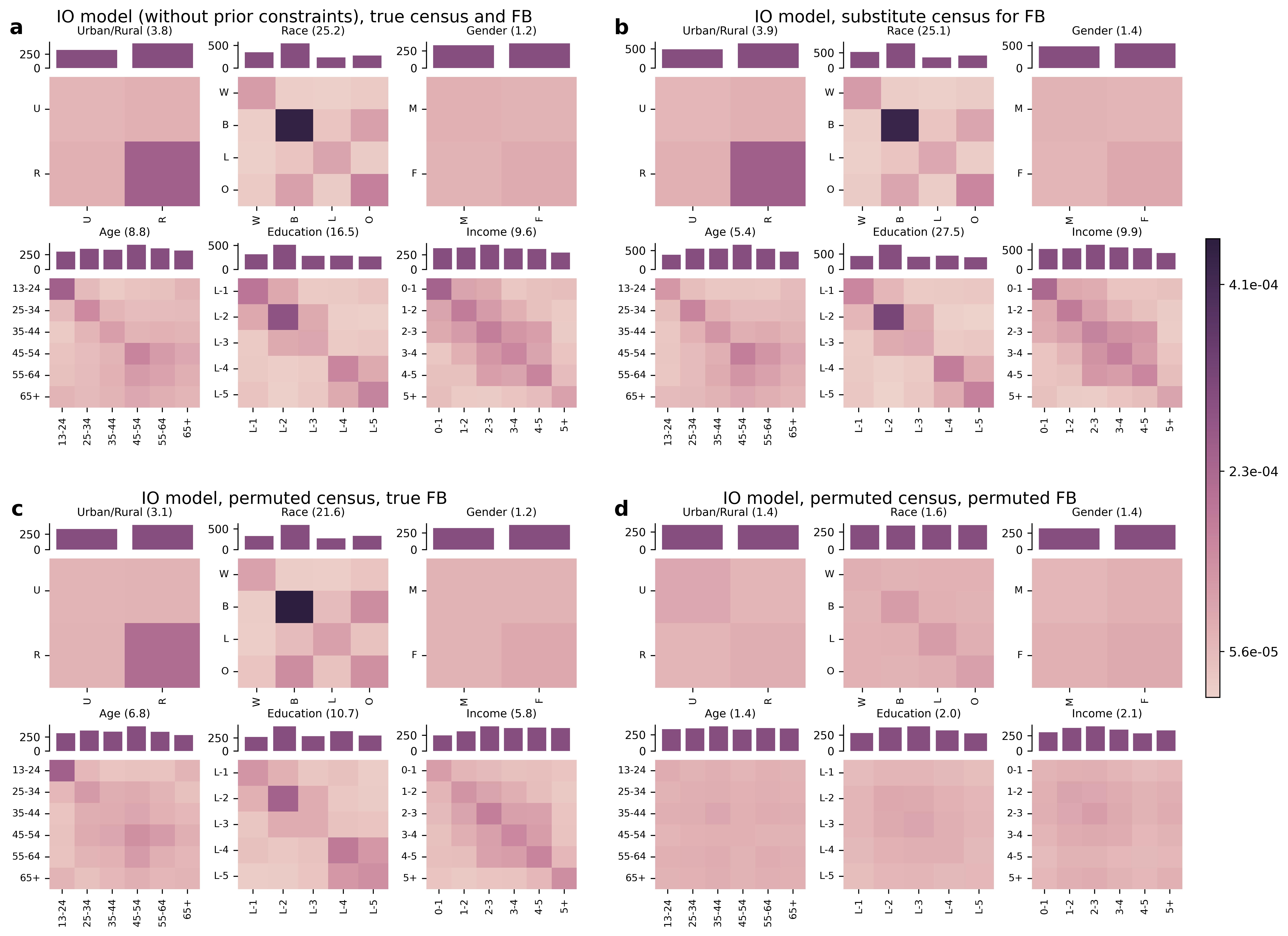}  
    \caption{\emph{Marginal connectivity patterns remain similar when substituting census counts for Facebook user counts.} \textbf{a}, Marginal connectivity between various sociodemographic groups as inferred from a joint geosocial intervening opportunities model with agnostic priors, i.e., without homophily constraints. Intervening opportunities are derived from census counts; the exposure, i.e., the number of possible links between pairs of geosocial groups, is given by the product of estimated FB users in those groups.  \textbf{b}, If exposure is computed as the product of census counts (instead of estimated FB user counts), the inferred connectivity patterns remain qualitatively similar, except for age; the greater sensitivity of the age results to estimated FB usage is not surprising given that our usage estimates vary most strongly with age.
    \textbf{c}, If intervening opportunities are derived from permuted census counts, i.e., proportions of social groups per ZCTA are randomly permuted, but exposure is still computed with unpermuted estimated FB counts, homophily decreases substantially for education, income and age; the number in brackets next to each variable indicates the ratio of the largest to the smallest entry in the depicted heatmaps --- compared to  \textbf{a}, these ratios are substantially smaller for education and income in \textbf{c}, and age connectivity patterns appear qualitatively different. \textbf{d}, If both census and FB counts are randomly permuted over ZCTAs, there is only signal for the decay in connectivity with geographical distance, but no signal for social connectivity left. Consequently, marginal connectivities for race still exhibit some slight degree of homophily (as after permutation, geographical segregation still exists), whereas inferred age, education, and income connectivity patterns appear mostly random. 
    Pairwise model comparisons for each of the 40 US states (we merged several smaller states into larger analysis units) confirm that the model fitted to the unperturbed census and FB data in \textbf{a} always provides a better fit to the SCI data than a model fitted to any of the perturbed data in \textbf{b}, \textbf{c} or \textbf{d}, i.e., the weakened or changed social connectivity patterns go hand in hand with less explanatory power.}
    \label{fig:marg_conn_4models_placebo}
\end{figure}

\begin{figure}[htbp]
    \centering
    \includegraphics[width=0.9\textwidth]{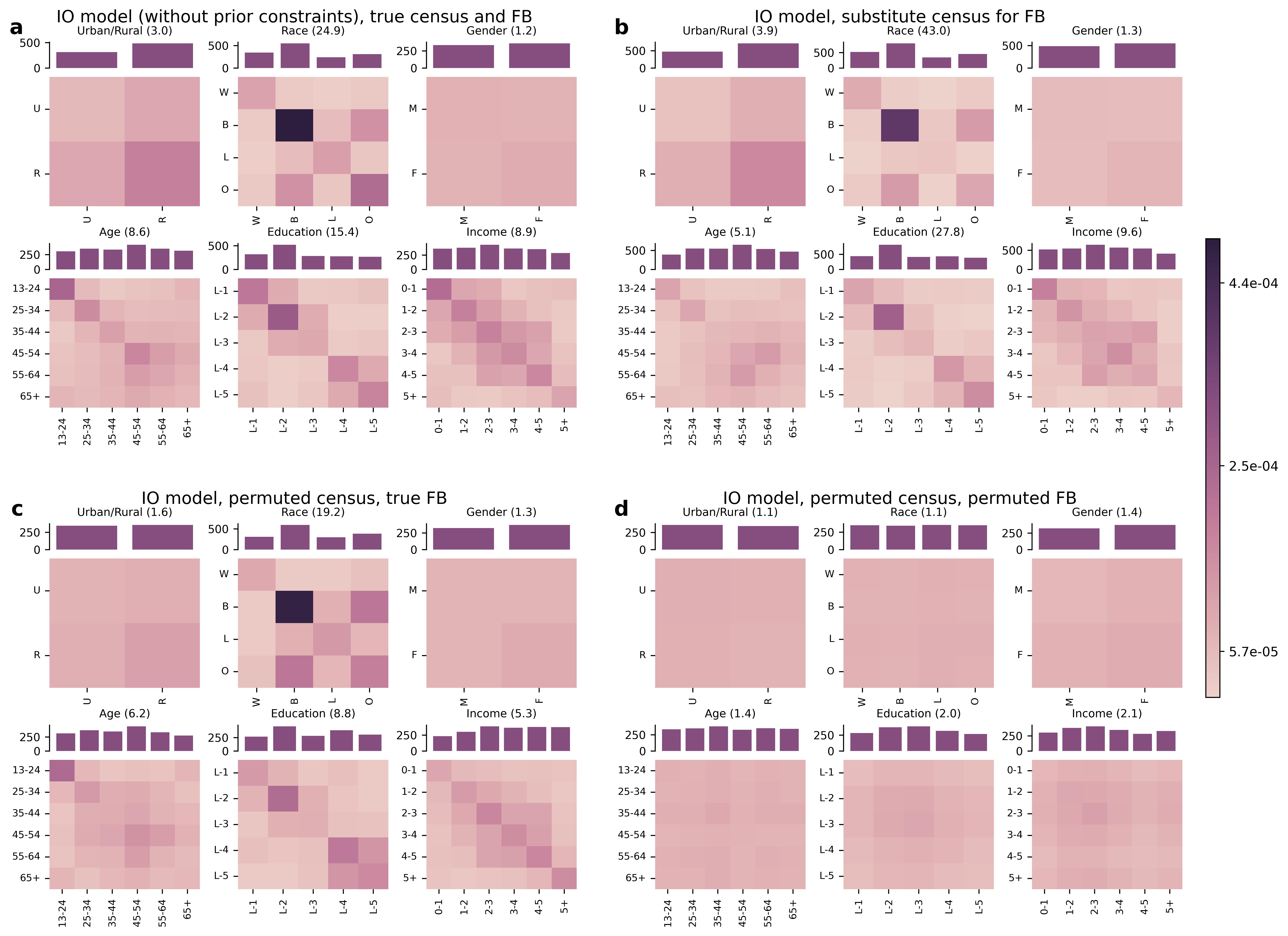}  
    \caption{\emph{Effects of perturbing census and/or Facebook counts become even clearer when considering the implied social segregation models.}
    Marginal connectivity under the social segregation model $M_{HB}(M)$, where $M$ is the corresponding fitted model from Fig.~\ref{fig:marg_conn_4models_placebo}: \textbf{a}, unperturbed census and estimated FB user counts; \textbf{b}, census counts used instead of estimated FB user counts to compute exposure; \textbf{c}, census group proportions randomly permuted across ZCTAs, while exposure is computed from unpermuted estimated FB user counts; \textbf{d}, both census and estimated FB user counts randomly permuted across ZCTAs.}
    \label{fig:marg_conn_soc_seg_4models_placebo}
\end{figure}

\begin{figure}[htbp]
    \centering
    \includegraphics[width=0.9\textwidth]{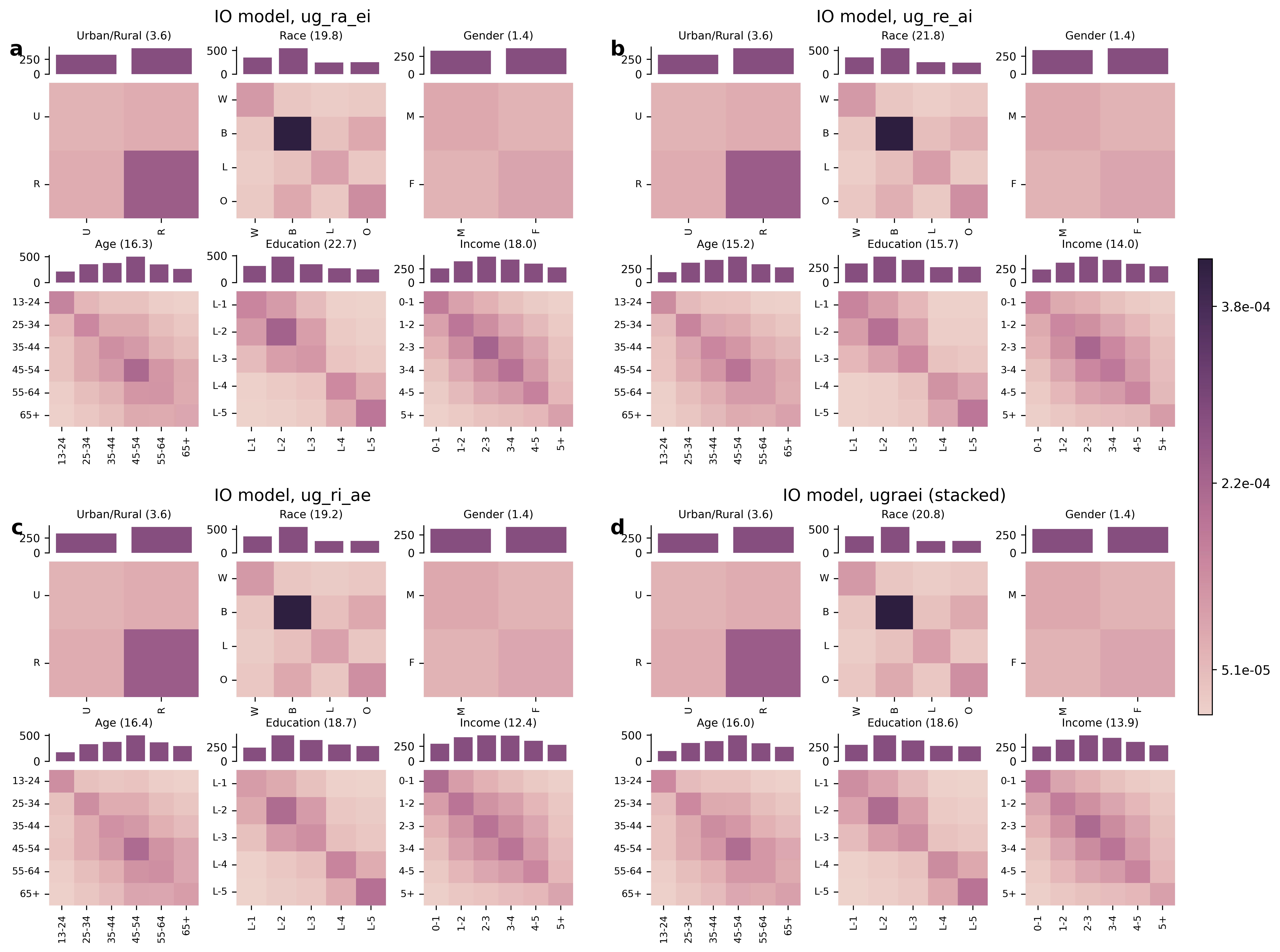}  
    \caption{\emph{Marginal connectivity is similar across alternative factorizations of sociodemographic space.} \textbf{a--c}, Marginal connection probabilities under the three factorizations $\{\{U,G\},\{R,A\},\{E,I\}\}$, $\{\{U,G\},\{R,E\},\{A,I\}\}$, and $\{\{U,G\},\{R,I\},\{A,E\}\}$, respectively. Here, U,G,R,A,E,I denote the variables urban, gender, race, age, education, and income, respectively. \textbf{d}, Marginal connection probabilities obtained by predictive stacking of the models in \textbf{a--c}.}
    \label{fig:marg_conn_different_factorizations}
\end{figure}

\begin{figure}[htbp]
    \centering
    \includegraphics[width=\textwidth]{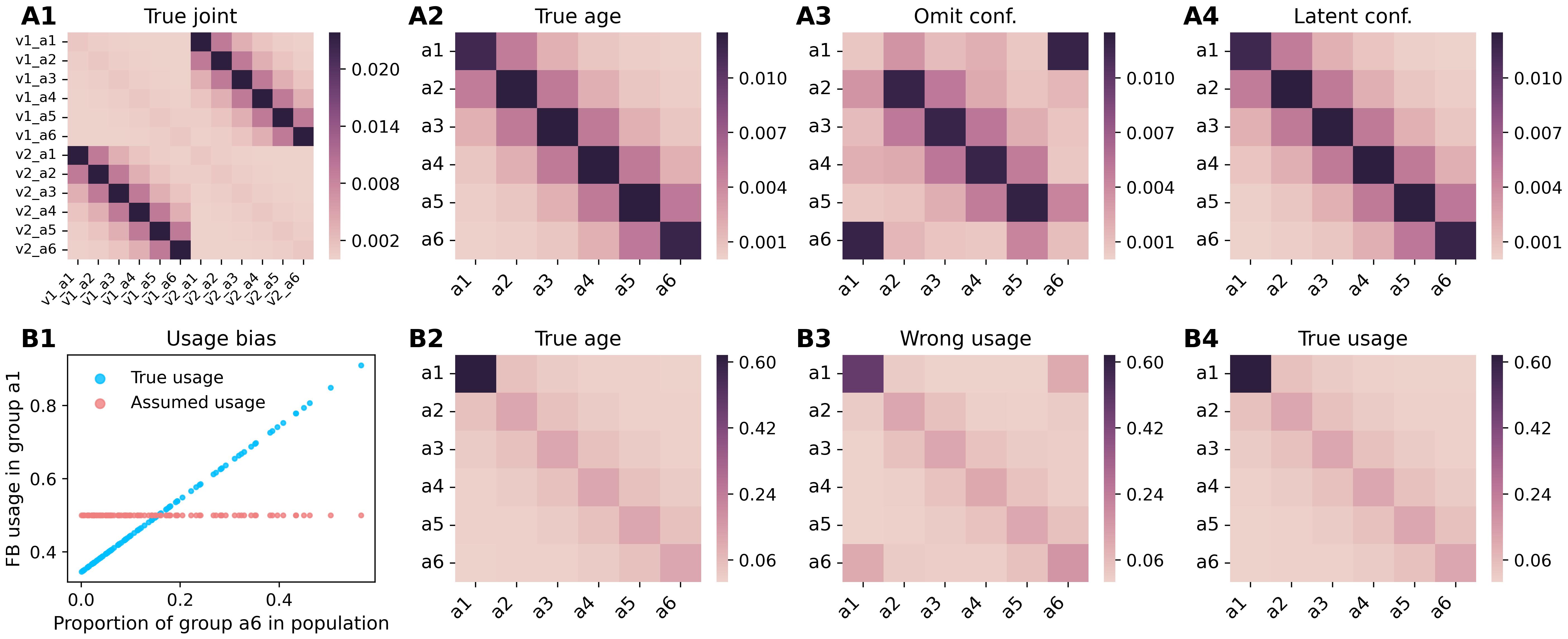} 
\caption{\emph{Toy examples illustrating sensitivity of inferred age connectivity to omitted structure and biased usage weights.}
\textbf{A1--A4}, omitted-confounder example.
\textbf{A1}, true joint connectivity between confounder--age groups.
\textbf{A2}, corresponding true age-marginal connectivity after marginalizing over the confounder.
\textbf{A3}, age connectivity inferred by a model fitted without the confounder, showing distortion relative to the age-marginal target.
\textbf{A4}, age connectivity inferred by a model with an additional latent binary dimension; in this toy example, including a latent dimension recovers the age-marginal target well.
\textbf{B1--B4}, biased Facebook usage example.
\textbf{B1}, data-generating and misspecified Facebook usage for group \(a1\) across regions, plotted against the regional share of group \(a6\); the misspecified usage is constant while the data-generating usage varies systematically with regional composition.
\textbf{B2}, true age connectivity under the data-generating process.
\textbf{B3}, age connectivity inferred when the model is fitted using misspecified usage weights.
\textbf{B4}, age connectivity inferred when the model is fitted using the true usage weights.}
    \label{fig:toy_ard_cf_fb_bias}
\end{figure}

\begin{figure}[htbp]
    \centering
    \includegraphics[width=\textwidth]{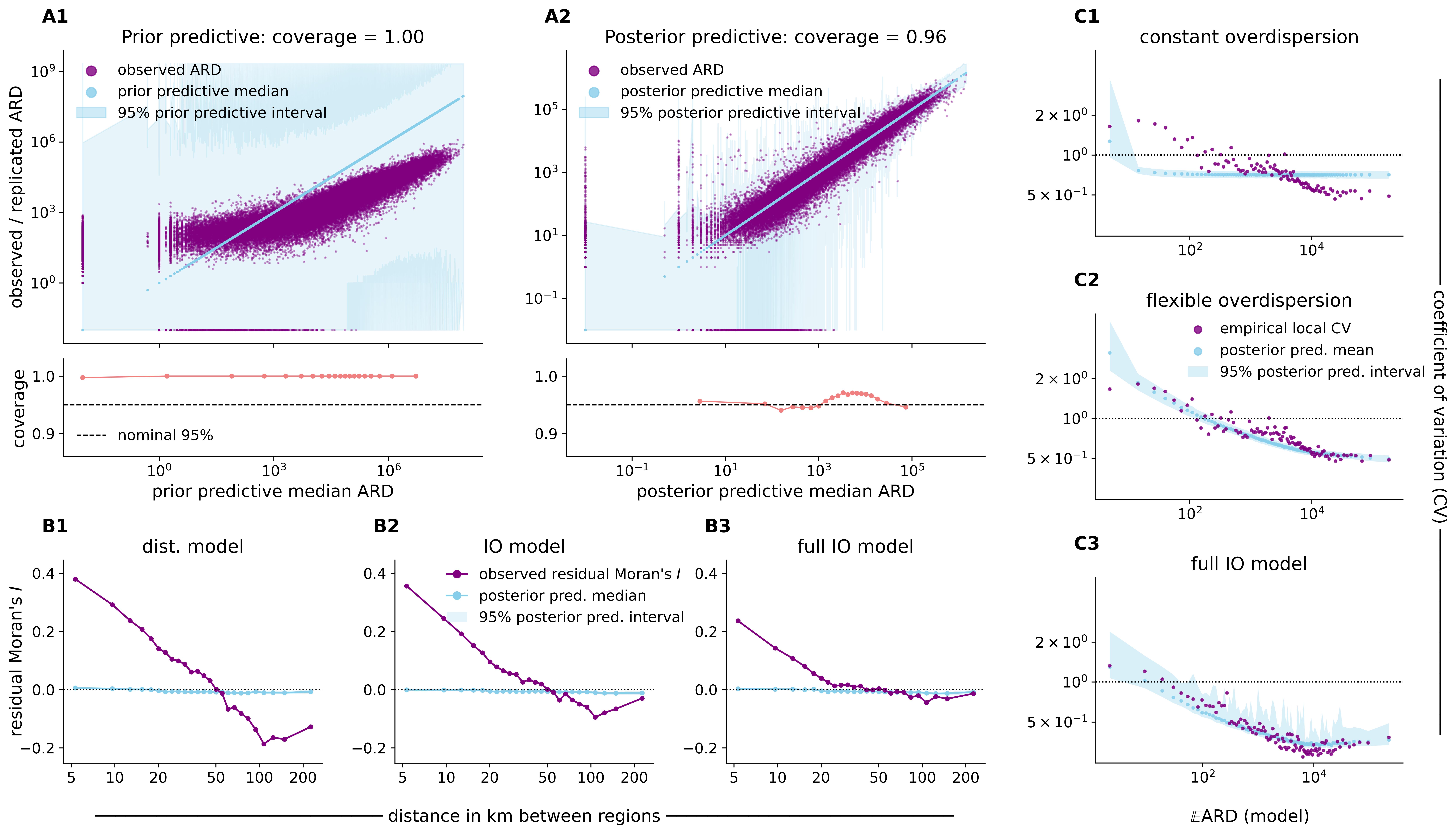}  
\caption{\emph{Posterior predictive checks show accurate marginal ARD predictions but incomplete capture of residual spatial autocorrelation.} All checks are shown for California, the largest state in the analysis. \textbf{A1}, Prior predictive check for aggregated relational data (ARD). Purple points show empirical ARD estimates plotted against the prior predictive median; blue points show the prior predictive median and the shaded region shows the pointwise 95\% prior predictive interval. The lower panel shows binned pointwise coverage of the 95\% prior predictive interval across quantiles of the prior predictive median, with the dashed line indicating nominal 95\% coverage. \textbf{A2}, Corresponding posterior predictive check, with empirical ARD estimates compared to the posterior predictive median and pointwise 95\% posterior predictive interval; the lower panel shows binned pointwise posterior predictive coverage. \textbf{B1}, Posterior predictive check for residual spatial autocorrelation under the distance-only model. Curves show residual Moran's \(I\) as a function of distance between regions; purple denotes the observed residual Moran's \(I\), blue denotes the posterior predictive median and the shaded region denotes the 95\% posterior predictive interval. \textbf{B2}, Corresponding residual Moran's \(I\) check for the intervening-opportunities (IO) model. \textbf{B3}, Corresponding residual Moran's \(I\) check for the full IO model. \textbf{C1}, Posterior predictive check for local ARD variability under the model with constant overdispersion. Points show the empirical local coefficient of variation (CV) and the posterior predictive median local CV for dyads ordered by model-implied expected ARD; the shaded region shows the 95\% posterior predictive interval. \textbf{C2}, Corresponding local CV check for the model with flexible overdispersion. \textbf{C3}, Corresponding local CV check for the full IO model.} \label{fig:ppc-combined-prior-moran-cv}
\end{figure}

\begin{figure}[htbp]
    \centering
    \includegraphics[width=\textwidth]{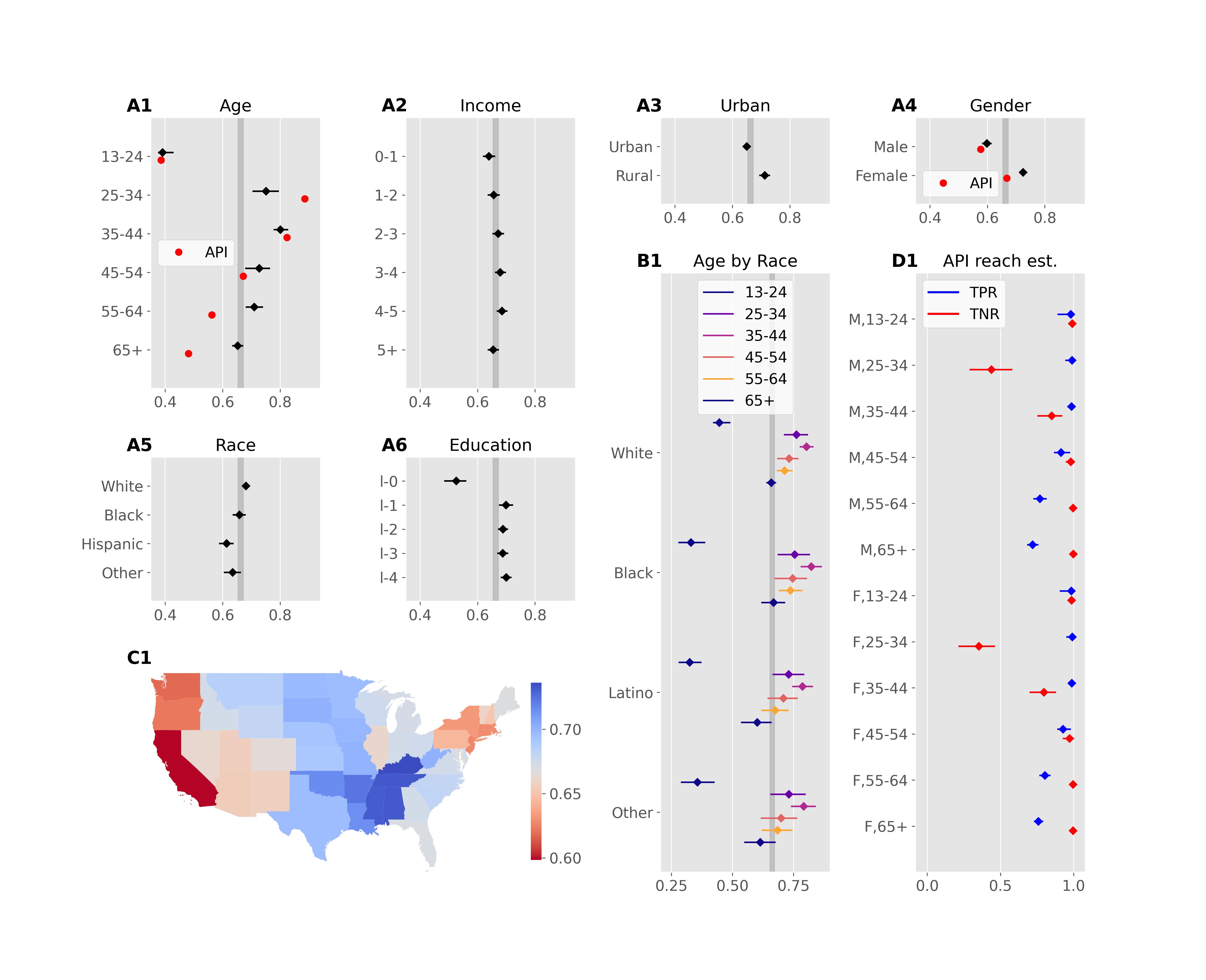} 
    \caption{\emph{Facebook (FB) usage is highest among ages \cat{25-34} and \cat{35-44}.}
    \textbf{A1}--\textbf{A6}, Posterior predictive 95\% credible intervals (black lines) for FB usage across sociodemographic groups; medians shown as black diamonds. In \textbf{A1} and \textbf{A4}, red circles mark usage implied by API reach estimates --- these overestimate usage for ages \cat{25-34} and \cat{35-44} and underestimate usage for older ages.
    \textbf{B1}, Age-by-race: usage in \cat{13-24} is low for all races but comparatively higher for \cat{White}.
    \textbf{C1}, Posterior mean FB usage by US state: lowest on the West Coast and in the Northeast.
    \textbf{D1}, 95\% posterior credible intervals for the expected true positive rates (TPR) and true negative rates (TNR) of the API reach estimates by age–gender groups. The TNR is below 0.5 for age \cat{25-34} (large false positive rate); this age group is consequently overrepresented in the API reach estimates. Conversely, the TPR is comparatively low for groups \cat{55-64} and \cat{65+}; these groups are underrepresented in the API reach estimates.}
    \label{fig:fb_usage_post}
\end{figure}

\begin{figure}[htbp]
    \centering
    \includegraphics[width=1.0\textwidth]{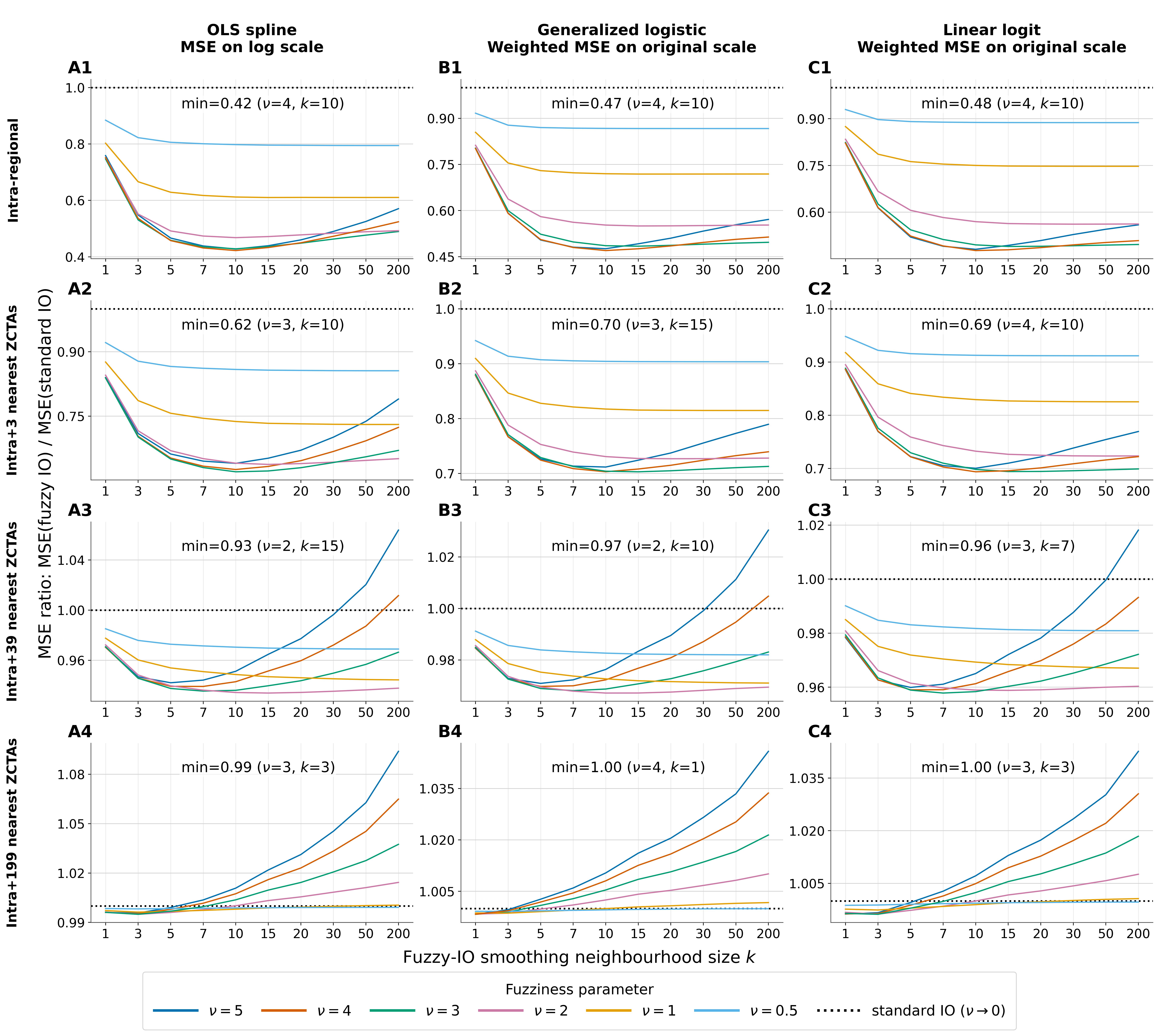} 
    \caption{\emph{Fuzzy IOs explain $\sci$ better than standard IOs.}
For each $(\nu,k)$ on a grid, we compute the mean squared error (MSE) of three simple predictive models of $\sci$ given marginal \emph{fuzzy} IOs, and report it relative to the MSE obtained with \emph{standard} IOs (baseline $\nu=0,$ or, equivalently, $k=0$).
\textbf{A1}--\textbf{A4}, $\sci$ is modeled as log-normal with mean given by a quadratic spline in $\log$ fuzzy IOs. Data subsets: \textbf{A1} intra-regional links only; \textbf{A2} intra-regional plus links to the three nearest ZCTAs per source ZCTA; \textbf{A3} intra-regional plus links to the 39 nearest ZCTAs; \textbf{A4} intra-regional plus links to the 199 nearest ZCTAs.
\textbf{B1}--\textbf{B4}, $\sci$ modeled with mean given by a generalized logistic function of the marginal fuzzy IOs and variance inversely proportional to the square of the fuzzy IOs.
\textbf{C1}--\textbf{C4}, $\sci$ modeled with mean given by a logistic function of the marginal fuzzy IOs and variance inversely proportional to the square of the fuzzy IOs.
For all three predictive models, fuzzy IOs substantially reduce MSE at small geographical distances and have negligible effect at large distances. The minimizing $(\nu,k)$ pair for each data subset and model is reported within each panel; overall, considering the importance of small distances, i.e., that the majority of links are local, a tuning $\nu\in\{3\mathrm{km},4\mathrm{km}\}$ and $k\in\{7,10,15\}$ seems advantageous.}
    \label{fig:mse_fuzzy_rank_dist}
\end{figure}

}

\ifincludesupplement

\clearpage
\newpage



\renewcommand{\thesection}{\arabic{section}}
\renewcommand{\thesubsection}{\thesection.\arabic{subsection}}
\renewcommand{\thesubsubsection}{\thesubsection.\arabic{subsubsection}}

\renewcommand*{\theHsection}{S.\arabic{section}}
\renewcommand*{\theHsubsection}{S.\arabic{section}.\arabic{subsection}}
\renewcommand*{\theHsubsubsection}{S.\arabic{section}.\arabic{subsection}.\arabic{subsubsection}}

\renewcommand*{\sectionformat}{S\thesection\enskip}
\renewcommand*{\subsectionformat}{S\thesubsection\enskip}
\renewcommand*{\subsubsectionformat}{S\thesubsubsection\enskip}

\RedeclareSectionCommand[tocnumwidth=3em]{subsection}
\setcounter{tocdepth}{2}


\numberwithin{equation}{section}
\renewcommand{\theequation}{S\arabic{section}.\arabic{equation}}


\numberwithin{figure}{section}
\numberwithin{table}{section}

\renewcommand{\figurename}{Supplementary Figure}
\renewcommand{\tablename}{Supplementary Table}

\renewcommand{\thefigure}{\arabic{section}.\arabic{figure}}
\renewcommand{\thetable}{\arabic{section}.\arabic{table}}

\newcommand{\suppnoteref}[1]{Supplementary Note~\ref{#1}}
\newcommand{\suppeqref}[1]{Supplementary Eq.~\eqref{#1}}
\newcommand{\suppfigref}[1]{Supplementary Fig.~\ref{#1}}
\newcommand{\supptableref}[1]{Supplementary Table~\ref{#1}}

\setcounter{section}{0}
\setcounter{subsection}{0}
\setcounter{subsubsection}{0}
\setcounter{equation}{0}
\setcounter{figure}{0}
\setcounter{table}{0}

\addpart{Supplementary Information}

\etocsetnexttocdepth{subsection}
{
\let\oldnumberline\numberline
\renewcommand{\numberline}[1]{\oldnumberline{S#1}}
\localtableofcontents
}

\clearpage
%
%


\section{Segregation framework}
\label{sec:supp_seg_framework}
The supplementary material for our segregation measurement framework is organized as follows. We begin by expressing traditional and spatially weighted geographical segregation measures as $f$-divergences and interpreting them in terms of the expected sociodemographic composition of individuals' network connections (section~\ref{sec:trad_geo_seg}).
We then define counterfactual network models to distinguish social, geographical, and traditional geographical segregation from total segregation (section~\ref{sec:segregation_via_counterfactual_and_reference_models}) and we define no-segregation and maximum-segregation reference models. In section~\ref{sec:div_based_seg_indices}, we define our segregation measures in geosocial networks through $f$-divergences between a distribution induced by a network model of interest and distributions induced by the no-segregation and maximum-segregation reference models. We show that the maximum-segregation model indeed maximizes the $f$-divergence from the no-segregation null over a suitable class of network models (Lemma~\ref{lemma:normalization_of_indices}).
Next, we show that the proposed measures recover established geographical segregation measures, including the Dissimilarity and Theil indices, as exact special cases (Lemma~\ref{lemma:relation_between_geo_and_social_f_divergence} and Corollary~\ref{corollary:recovery_normalized_traditional_indices} in section~\ref{sec:recovery_of_trad_seg}). 
Together, these results motivate our recommendation that segregation in geosocial networks should be measured by considering the geosocial coordinates of a random individual $i$ and the social coordinates only of a second random individual $j$, rather than keeping the full coordinates of both individuals (section~\ref{sec:choice_of_coord_space}). We then complete the description of the $f$-divergence measures by establishing asymptotic relations in the sparse and dense network regimes (Lemma~\ref{lemma:T_D_F_to_S_D_F} and Corollary~\ref{corollary:normalized_tau_sigma} in section~\ref{sec:sigma_tau_under_sparsity}).
We next investigate how homophily and heterophily contribute to the proposed segregation measures, define a family of homophily measures in section~\ref{sec:homophily}, and establish some of their properties in Lemma~\ref{lemma:homophily_quantification} and Lemma~\ref{lemma:homophily_quantification_2}. The proposed homophily measure reduces to network modularity/nominal assortativity when the generator of the $f$-divergence is the Total Variation distance (Lemma~\ref{lemma:homophily_quantification} and Corollary~\ref{corollary:normalized_homophily_quantification}), and contains Relative diversity, a traditional geographical segregation measure, as a special case (Lemma~\ref{lemma:recovery_of_relative_diversity}). In sections~\ref{sec:reg_decomp}, \ref{sec:group_decomp}, and \ref{sec:normalization_decomp}, we generalize the regional and group decompositions of the Theil index to segregation measures on geosocial networks (Lemmas~\ref{lemma:regional_decomposition_of_H}, \ref{lemma:group_decomposition_of_H}) and derive normalizations for these decompositions consistent with the standard normalizations for the decompositions of the traditional Theil index (Lemma~\ref{lemma:normalization_decompositions}). Table~\ref{tab:overview_theoretical_results} provides a compact overview of the main theoretical results and where they are stated.
Proofs are provided in section~\ref{sec:proofs}.

\begin{longtable}{@{}p{0.34\textwidth}p{0.62\textwidth}@{}}
\caption{\emph{Overview of the main theoretical results of our segregation measurement framework.}}
\label{tab:overview_theoretical_results}\\
\toprule
\textbf{Stated in} & \textbf{Result} \\
\midrule
\endfirsthead

\toprule
\textbf{Stated in} & \textbf{Result} \\
\midrule
\endhead

Lemma~\ref{lemma:normalization_of_indices}
& Normalization of measures: establishes that the degree-constrained social-isolation model maximizes the $f$-divergence from the degree-configuration null over the relevant class of network models. \\

Lemma~\ref{lemma:relation_between_geo_and_social_f_divergence} and Corollary~\ref{corollary:recovery_normalized_traditional_indices}
& Recovery of established geographical segregation measures, including the Dissimilarity and Theil indices, as special cases.\\

Lemma~\ref{lemma:T_D_F_to_S_D_F} and Corollary~\ref{corollary:normalized_tau_sigma}
& Relate segregation measures based on the joint coordinate--link distribution $\tau$ to those based on the conditional coordinate distribution $\sigma$.\\

Lemma~\ref{lemma:homophily_quantification} and Corollary~\ref{corollary:normalized_homophily_quantification}
& Establish properties of the homophily decomposition and show that, for Total Variation distance, our signed homophily measure recovers network modularity and nominal assortativity.\\

Lemma~\ref{lemma:homophily_quantification_2} and Example~\ref{example:homophily quantification}
& Characterize how different $f$-divergence generators weight departures from the random-mixing null and how this affects homophily quantification.\\

Lemma~\ref{lemma:recovery_of_relative_diversity}
& Shows that Relative Diversity is recovered as a special case of our signed homophily measure.\\

Lemmas~\ref{lemma:regional_decomposition_of_H} and \ref{lemma:group_decomposition_of_H}
& Extend the regional and sociodemographic group decompositions of the Theil index to segregation measures on geosocial networks.\\

Lemma~\ref{lemma:normalization_decompositions}
& Provides normalizations for the regional and group decompositions that recover the standard normalizations of the traditional Theil-index decompositions.\\

Lemma~\ref{lemma:coordinate_coarsening}
& Shows that proposed segregation measures cannot increase under deterministic coarsening of the coordinate space.\\

\bottomrule
\end{longtable}

\subsection{Traditional geographical segregation}
\label{sec:trad_geo_seg}
Let $\RS$ denote a set of geographical regions and $\BS$ a set of sociodemographic groups, such that the regions, and likewise the groups, are mutually exclusive and collectively exhaustive. Let $\pi_{rb}$ denote the proportion of the population in region $r\in\RS$ and sociodemographic group $b\in\BS$, with corresponding marginal proportions $\pi_r=\sum_b\pi_{rb}$ and $\pi_b=\sum_r\pi_{rb}$, and let $\pi_{b\mid r}=\pi_{rb}/\pi_r$ denote the corresponding conditional distribution. Several widely used traditional geographical segregation indices \cite{reardon2002measures} can be written as an $f$-divergence\footnote{\cite{reardon2002measures} list these indices as disproportionality-based measures, with the same functional form as in \eqref{eq:traditional_segregation_indices_as_f_divergence}.} between the product of the marginal distributions $(\pi_r)_{r\in\RS}\otimes (\pi_b)_{b\in\BS}$ and the joint distribution $(\pi_{rb})_{r\in\RS,b\in\BS}$ over regions $\RS$ and sociodemographic groups $\BS$
\begin{equation}
\begin{aligned}
    \mathcal{D}_{f,\pi} & \coloneqq \DF{(\pi_{rb})}{ (\pi_r)\otimes (\pi_b)} \\
    & = \sum_r \pi_r\sum_{b}\pi_{b}
    f\left(\frac{\pi_{b\mid r}}{\pi_{b}}\right),
\end{aligned}
\label{eq:traditional_segregation_indices_as_f_divergence}
\end{equation}
where for two probability measures $P$ and $Q$ such that $P$ is absolutely continuous with respect to $Q$,\footnote{In the discrete case, this just means $Q(x) = 0\Rightarrow P(x) = 0$} and a convex function $f\colon [0,\infty)\to (-\infty, \infty],$ with $f(x)<\infty,\,\forall x>0,$ $f(1) = 0,$ and $f(0) = \lim_{x\to0} f(x),$ the $f$-divergence of $P$ from $Q$ is given by
\begin{equation*}
    \DF{P}{Q} = \EE_Q f\left(\frac{dP}{dQ}\right),
\end{equation*}
which in the discrete case is given by
\begin{equation*}
    \DF{P}{Q} = \sum_x Q(x) f\left(\frac{P(x)}{Q(x)}\right).
\end{equation*}
We call a function $f$ with the above properties a generator for the corresponding $f$-divergence.

With $f(x) = x\log x,$ we obtain the unnormalized Theil's Information Theory index $H,$ with $f(x) = \tfrac{1}{2}|x-1|$ the unnormalized Dissimilarity index $D$, and with $f(x) = x^2-1$ the unnormalized Pearson's $\chi^2$ divergence, called the squared coefficient of variation in \cite{reardon2002measures}.

To address the dependence of traditional segregation measures on pre-defined geographical regions and corresponding shortcomings such as the checkerboard problem and the modifiable areal unit problem discussed in the Introduction, \cite{reardon2004measures,lee2008beyond} propose an egocentric, or spatially weighted, view of segregation.
Let $\GS$ denote the set of geographical coordinates occupied by individuals in the population. For a given symmetric kernel
\begin{equation}
    \phi\colon\GS\times\GS\to[0,\infty),
\end{equation}
the spatially weighted proportion of group $b$ at geographical coordinate $p$ is
\begin{equation}
    \widetilde{\pi}_{b\mid p}
    \coloneqq
    \frac{\sum_{p'\in\GS}\pi_{b\mid p'}\pi_{p'}\phi(p,p')}
    {\sum_{p'\in\GS}\pi_{p'}\phi(p,p')}.
    \label{eq:def_tilde_pi_discrete}
\end{equation}
Here, $\pi_{p'}$ denotes the proportion of the population at geographical coordinate $p'$ and $\pi_{b\mid p'}$ the proportion of that population belonging to group $b$. If there is exactly one individual at geographical coordinate $p$, then $\pi_{b\mid p}=1$ for that individual's group and $0$ otherwise.

In their continuous-space formulation, \cite{reardon2004measures} express the same construction in terms of spatially weighted population-density surfaces; we use the discrete analogue, which is more convenient for the network formulation below.

By substituting $\widetilde{\pi}_{b\mid p}$ for $\pi_{b\mid r}$, \cite{reardon2004measures} define spatially weighted versions of traditional region-based geographical segregation indices. For example, their unnormalized spatially weighted Dissimilarity index is
\begin{equation}
    \widetilde D=\frac12\sum_p\sum_b\pi_p\left|\widetilde\pi_{b\mid p}-\pi_b\right|.
\end{equation}
Writing $\widetilde\pi_{pb}\coloneqq\pi_p\widetilde\pi_{b\mid p}$, this is precisely the Total Variation distance between $(\widetilde\pi_{pb})$ and $(\pi_p)\otimes(\pi_b)$. Their unnormalized spatially weighted Information Theory index is
\begin{equation}
    \widetilde H=E-\sum_p\pi_p\widetilde E_p,
\end{equation}
where $E=-\sum_b\pi_b\log\pi_b$ and $\widetilde E_p=-\sum_b\widetilde\pi_{b\mid p}\log\widetilde\pi_{b\mid p}$. This coincides with the Kullback--Leibler divergence
\begin{equation}
    \DKL{(\widetilde\pi_{pb})}{(\pi_p)\otimes(\pi_b)}
\end{equation}
whenever the spatial weighting preserves the marginal group proportions,
\begin{equation}
    \sum_p\pi_p\widetilde\pi_{b\mid p}=\pi_b,\qquad b\in\BS.
    \label{eq:preservation_of_pi_b_marginal}
\end{equation}
For a symmetric kernel $\phi$, a sufficient condition for \eqref{eq:preservation_of_pi_b_marginal} is that its population-weighted row sums are constant,
\begin{equation}
    \sum_{p'\in\GS}\pi_{p'}\phi(p,p')\equiv\bar\phi,\qquad p\in\GS.
\label{eq:constant_weighted_row_sums_phi}
\end{equation}
We consider it natural that spatial weighting preserves the marginal group proportions \eqref{eq:preservation_of_pi_b_marginal} and therefore restrict attention in the following to symmetric kernels satisfying the sufficient condition \eqref{eq:constant_weighted_row_sums_phi}.
Motivated by the above examples, we define a family of spatially weighted geographical segregation measures via $f$-divergences by
\begin{equation}
    \mathcal D_{f,\widetilde\pi}
    \coloneqq
    \DF{(\widetilde\pi_{pb})}{(\pi_p)\otimes(\pi_b)}.
\label{eq:traditional_egocentric_segregation_indices_as_f_divergence}
\end{equation}

We note that if
\begin{equation}
\phi(p,q)=
\begin{cases}
    \bar\phi/\pi_{r(p)}, & r(p)=r(q),\\
    0, & \text{otherwise},
\end{cases}
\label{eq:def_RI_model1}
\end{equation}
where $r(p)$ denotes the region containing geographical coordinate $p$, then $\widetilde\pi_{b\mid p}=\pi_{b\mid r(p)}$, and hence the spatially weighted segregation measure $\mathcal D_{f,\widetilde\pi}$ reduces to the standard region-based measure $\mathcal D_{f,\pi}$. 

\subsection{Traditional and spatially weighted geographical segregation as connectivity models}
We denote by $\PP{i\sim j \mid p(i) = p_1, b(i) = b_1, p(j) = p_2, b(j) = b_2}$ the probability of the event that two randomly chosen individuals $i$ and $j$ connect, in symbols $i\sim j,$ given their geographical coordinates $p(i),p(j)$ and their sociodemographic coordinates $b(i),b(j).$\footnote{For ease of exposition, we allow $i=j$, corresponding to drawing the two individuals independently from the population. Excluding self-pairs changes the resulting probabilities by terms of order $1/N$, where $N$ is the population size.}

Because $\widetilde\pi_{b\mid p}$ in \eqref{eq:def_tilde_pi_discrete} is invariant under a common positive scaling of $\phi$, we can assume w.l.o.g. that $0\leq\phi(p,p')\leq1$ for all $p,p'\in\GS$. Therefore, and because of assumption \eqref{eq:constant_weighted_row_sums_phi}, we can interpret $\phi$ as a symmetric geospatial connectivity kernel with constant expected degrees for all individuals and  we can interpret the spatially weighted segregation measures as being built upon a connectivity model $M_{\phi}$ where the probability that randomly chosen individuals $i,j$ are connected is solely determined by their geographical locations
\begin{equation}
\begin{aligned}
    &\PPS{i\sim j \mid p(i) = p_1, p(j)=p_2, b(i) = b_1, b(j) = b_2}{}{M_{\phi}}\\
    & = \PPS{i\sim j \mid p(i) = p_1, p(j) = p_2}{}{M_{\phi}}\\
    & = \phi(p_1,p_2),
\end{aligned}
\label{eq:def_geospatial_model_phi}
\end{equation}
which is compared to the Erd\H{o}s--R\'{e}nyi random graph ${M_{ER}}$ as a \emph{Null} or \emph{reference model}, where
\begin{align}
    \PPS{i\sim j \mid p(i)=p_1, p(j)=p_2, b(i) = b_1, b(j) = b_2}{}{M_{ER}} \equiv \bar{P}.
\end{align}
Here, $p(i)\in\GS$ and $b(i)\in\BS$ denote the geographical and sociodemographic coordinates of individual $i$, and $\bar{P} \equiv \PPS{i\sim j}{}{M_{\phi}} = \bar{\phi}$ is the marginal connection probability. 
Under $M_{\phi},$ for any individual $i$ at geographical location $p$ the fraction of expected number of friends in group $b$ out of the total number of friends of $i$ is given by
\begin{align}
    \frac{ \EE_{M_{\phi}}\#\text{friends of i in group }b }{ \EE_{M_{\phi}}\#\text{friends of i}}  = \tilde{\pi}_{b\mid p} 
\end{align}
and under $M_{ER}$ this fraction is given by
\begin{align}\label{eq:expected_faction_of_edges_under_ER_model}
     \frac{ \EE_{M_{ER}}\#\text{friends of i in group b}}{ \EE_{M_{ER}}\#\text{friends of i}}
     = \pi_{b} .
\end{align}
We call the connectivity model corresponding to the region-based kernel in \eqref{eq:def_RI_model1} the regional isolation model $M_{RI}$, under which the above-described fraction of edges of any individual $i$ in region $r$ is given by

\begin{align}\label{eq:expected_fraction_of_edges_under_RI_model}
    \frac{ \EE_{M_{RI}}\#\text{friends of i in group }b }{ \EE_{M_{RI}}\#\text{friends of i}}  = \pi_{b\mid r}. 
\end{align}

To summarize, the traditional segregation indices and their spatially weighted extensions can be understood as comparing the fractions of expected number of edges ending in group $b$ under model $M_{\phi}$ and null model $M_{ER}$. As we will show in the following, a more general approach is given by comparing the fractions of edges between pairs of coordinates $(c_1,c_2)$ out of the total number of edges in the network under a network model $M$ and a suitable null reference model. Depending on the choice of $M,$ such an approach allows to investigate geographical segregation in isolation, or social segregation in isolation, or total segregation.

\subsection{Coordinate spaces and geosocial network models}
\label{sec:coordinate_spaces_and_geosocial_networks}
Before we define our proposed segregation indices, we first establish some notation.
Let $\CS_1$ and $\CS_2$ be one of the following
\begin{enumerate}
    \item \emph{Full geosocial coordinate space:} $\CS_1=\GS\times\BS$ and $\CS_2=\GS\times\BS$, or, coarsened to regional geographical resolution, $\CS_1=\RS\times\BS$ and $\CS_2=\RS\times\BS$;
    \item \emph{Geo-egocentric geosocial coordinate space:} $\CS_1=\GS\times\BS$ and $\CS_2=\BS$, or, coarsened to regional geographical resolution, $\CS_1=\RS\times\BS$ and $\CS_2=\BS$;
    \item \emph{Social coordinate space:} $\CS_1=\BS$ and $\CS_2=\BS$.
\end{enumerate}
Regional geographical coordinates $r\in\RS$ are obtained by assigning each geographical coordinate $p\in\GS$ its containing region $r(p)\in\RS$.

We consider two individuals $i$ and $j$ sampled independently and uniformly from the population underlying a specific geosocial network. Let $C_1$ and $C_2$ denote the random variables (R.V.) recording their coordinates in $\CS_1$ and $\CS_2$, respectively, and let the Bernoulli R.V. $A$ indicate whether they are connected, with $A=1$ if $i\sim j$ and $A=0$ otherwise.

In the following, a geosocial network model $M$ specifies the marginal coordinate distributions $(\pi_{c_1}^M)_{c_1\in\CS_1}$ and $(\pi_{c_2}^M)_{c_2\in\CS_2},$ i.e., the distributions of $C_1$ and $C_2$ under model $M,$ together with the conditional probability that two individuals with coordinates $c_1$ and $c_2$ are connected, which we denote by
\begin{equation*}
    \PS{c_1\sim c_2}{}{M}  \coloneqq \PPS{A = 1 \mid C_1=c_1,C_2=c_2}{}{M}.
\end{equation*}
Equivalently, a geosocial network model is described by the joint distribution of the triplet $C_1, C_2, A$ under $M$
\begin{equation}
    \begin{aligned}
        \tau^M(c_1,c_2,e) & \coloneqq \PPS{C_1=c_1, C_2=c_2, A=e}{}{M}\\
        & = \begin{cases}
        \PS{c_1\sim c_2}{}{M}\pi_{c_1}^M\pi_{c_2}^M,& e = 1\\
        \PS{c_1\nsim c_2}{}{M}\pi_{c_1}^M\pi_{c_2}^M,& e = 0
        \end{cases} .
    \end{aligned}
    \label{eq:def_tau}
\end{equation}
We denote the marginal connection probability under $M$ by
\begin{equation}
    \bar{P}^M \coloneqq \PPS{A=1}{}{M}=\sum_{c_1,c_2}\pi_{c_1}^M\pi_{c_2}^M\PS{c_1\sim c_2}{}{M}.
\end{equation}
Conditioning the joint distribution of $C_1$ and $C_2$ on the presence or absence of a link yields the conditional coordinate distributions
\begin{equation}
    \begin{aligned}
        \sigma^M(c_1,c_2) & \coloneqq \PPS{C_1=c_1, C_2=c_2\mid A=1}{}{M}\\
        & = \frac{\pi_{c_1}^M\pi_{c_2}^M\PS{c_1\sim c_2}{}{M}}{\bar{P}^M}
    \end{aligned}
    \label{eq:def_sigma}
\end{equation}
and
\begin{equation}
    \begin{aligned}
        \xi^M(c_1,c_2) & \coloneqq \PPS{C_1=c_1, C_2=c_2\mid A=0}{}{M}\\
        & = \frac{\pi_{c_1}^M\pi_{c_2}^M\PS{c_1\nsim c_2}{}{M}}{1-\bar{P}^M} .
    \end{aligned}
    \label{eq:def_xi}
\end{equation}
In the following, ``conditional coordinate distribution'' refers to $\sigma^M$ unless conditioning on the absence of a link is stated explicitly. If $c$ records both geographical and social coordinates, we write $c=rb$ for regional geographical coordinates and $c=pb$ for point-level geographical coordinates, and correspondingly $\PS{c_1\sim c_2}{}{M}=\PS{r_1b_1\sim r_2b_2}{}{M}$ or $\PS{p_1b_1\sim p_2b_2}{}{M}$. Furthermore, if we want to emphasize the coordinate spaces on which $\sigma^M$ is defined, we write $\sigma_{\CS_1,\CS_2}^M$.

\subsection{Representation of segregation via counterfactual and reference models}
\label{sec:segregation_via_counterfactual_and_reference_models}

We distinguish total, social, geographical and traditional geographical segregation by associating each with a network model. Let $M$ denote a geosocial network model describing the geosocial network of interest, e.g., empirically or statistically obtained estimates of the true connectivity $\PS{r_1 b_1 \sim r_2 b_2}{}{}$ and the true coordinate distribution $(\pi_{rb})_{r,b},$ or, in an idealized setting, the true data generating process. For ease of exposition, we will refer to $M$ as the observed connectivity patterns. Total segregation is defined directly from $M.$ Social and geographical segregation are defined by counterfactual models relative to $P^M$ and $\pi^M,$ whereas traditional geographical segregation is defined relative to $\pi^M$ and independently of $P^M$ (except that we set its marginal connectivity $\bar{P}$ to equal the one implied by $M$).

\textbf{Social segregation.} The \emph{homogeneous sociodemographic space model} $M_{HB}=M_{HB}(M)$ counterfactually reallocates individuals such that every region has the same sociodemographic composition as the overall population, while retaining the connection probabilities conditional on individuals' geosocial coordinates:
\begin{equation}
    \begin{aligned}
        \pi_{rb}^{M_{HB}}&=\pi_r^M\pi_b^M,\\
        \PS{r_1b_1\sim r_2b_2}{}{M_{HB}}
        &=\PS{r_1b_1\sim r_2b_2}{}{M}.
    \end{aligned}
    \label{eq:def_M_HB}
\end{equation}
Thus, $M_{HB}$ removes the geographical concentration of sociodemographic groups while retaining group-specific connectivity patterns, including their variation across geographical space.

\textbf{Geographical segregation.} The \emph{social ambiphily model} $M_{AB}=M_{AB}(M)$ retains the observed joint distribution of regions and sociodemographic groups, but removes the direct dependence of connectivity on sociodemographic group membership, conditional on geographical coordinates:
\begin{equation}
\begin{aligned}
    \pi_{rb}^{M_{AB}}&=\pi_{rb}^M,\\
    \PS{r_1b_1\sim r_2b_2}{}{M_{AB}}
    &\coloneqq
    \sum_{b_1^\prime\in\BS}\sum_{b_2^\prime\in\BS}
    \pi_{b_1^\prime}^M\pi_{b_2^\prime}^M
    \PS{r_1b_1^\prime\sim r_2b_2^\prime}{}{M}.
\end{aligned}
\label{eq:def_M_BA}
\end{equation}
Averaging using the marginal sociodemographic distribution $(\pi_b^M)_{b\in\BS}$ rather than the local distributions $(\pi_{b\mid r}^M)_{b\in\BS}$ prevents the resulting region-pair connectivity from being weighted by local sociodemographic composition. The resulting connection probability depends on $r_1$ and $r_2$, but not on $b_1$ or $b_2$.

\textbf{Traditional geographical segregation.} The \emph{regional isolation model} $M_{RI}$ retains the observed population distribution but assumes that individuals connect only to individuals in the same predefined region:
\begin{equation}
    \begin{aligned}
        \pi_{rb}^{M_{RI}}&=\pi_{rb}^M,\\
        \PS{r_1b_1\sim r_2b_2}{}{M_{RI}}
        &= \frac{\bar{P}^M}{\pi_{r_1}^M} \mathbb{1}_{\left\{ r_1 = r_2 \right\}} .
    \end{aligned}
    \label{eq:def_M_RI}
\end{equation}
The scaling by $\pi_r^M$ gives all individuals the same marginal connection probability $\bar{P}^M$ and preserves the overall density of the network. The model is therefore sensitive only to the distribution of sociodemographic groups across the chosen regions. 
\begin{remark}
The social ambiphily model $M_{AB}$ and the regional isolation model $M_{RI}$ are both special cases of the class of connectivity models $M_\phi$ introduced in \eqref{eq:def_geospatial_model_phi}, for which, conditional on geographical coordinates, connectivity is independent of sociodemographic coordinates of individuals. $M_{AB}$ and $M_{RI}$ represent different notions of segregation: the connectivity in $M_{AB}(M)$ is constructed from $M$ and retains the geographical pattern of connectivity after removing sociodemographic preferences in tie formation. By contrast, $M_{RI}$ imposes complete isolation between predefined regions and therefore captures traditional residential segregation independently of the geographical connectivity patterns in any actual geosocial network. Therefore, $M_{AB}$ and $M_{RI}$ complement each other, rather than being strict alternatives to measure segregation. Importantly, $M_{AB}$ is network specific, and offline friendship networks will in general exhibit different geographical connectivity patterns than online friendship networks.  Different geographical connectivity kernels $\phi$ could capture additional patterns of geographical segregation. However, we will not explore alternative kernels $\phi$ in this study.
\end{remark}
\begin{remark}
    We say that a model is admissible if it defines valid connection probabilities in $[0,1].$ For the regional isolation model defined in \eqref{eq:def_M_RI}, this implies
    \begin{equation}
        \bar{P}^M\leq \pi_{r}^M, \quad \forall r \in \RS 
        \label{eq:M_RI_admissibility_condition}
    \end{equation}
    For geosocial networks, the marginal connectivity $\bar{P}^M$ typically scales as the reciprocal of the size of the population underlying the network and we will therefore assume that \eqref{eq:M_RI_admissibility_condition} holds.   
\end{remark}
\begin{remark}
    Applying $M_{HB}$ and $M_{AB}$ to a network model $M'$ commute and we denote the correspondingly obtained model by $M_{HB-AB}$
    \begin{equation}
        M_{HB-AB}(M) \coloneqq M_{HB}(M_{AB}(M)) = M_{AB}(M_{HB}(M)).
    \end{equation}
\end{remark}
Thus, $M_{HB-AB}$ constitutes a counterfactual model in which both geographical and social segregation have been removed and as such could define a no-segregation reference model. However, we will adopt a different approach and define a no-segregation reference model $M_{DC}$ by random mixing, as described below.

\textbf{Random mixing and maximal segregation.} 
For each model $M'\in\left\{M,M_{HB},M_{AB},M_{RI}\right\}$, we define a random-mixing null model and a maximal-segregation reference model. The name maximal-segregation reference will be justified in Lemma~\ref{lemma:normalization_of_indices}. We denote the marginal connectivity of geosocial coordinate $rb$ under $M'$ by
\begin{equation}
\begin{aligned}
        \bar P_{rb}^{M'}
    & = \sum_{r'\in\RS}\sum_{b'\in\BS}
    \pi_{r'b'}^{M'}\PS{rb\sim r'b'}{}{M'},
\end{aligned}
    \label{eq:def_geosocial_marginal_connectivity}
\end{equation}
and the marginal connectivity of sociodemographic group $b$ by
\begin{equation}
\begin{aligned}
    \bar P_b^{M'} 
    & = \sum_{r\in\RS}\pi_{r\mid b}^{M'}\bar P_{rb}^{M'}.
\end{aligned}
    \label{eq:def_social_marginal_connectivity}
\end{equation}

The \emph{degree configuration model} $M_{DC}(M')$ \cite{chung2002average,newman2006modularity} preserves the marginal geosocial connectivities, but otherwise forms links by random mixing:
\begin{equation}
\begin{aligned}
    \pi_{rb}^{M_{DC}(M')}&=\pi_{rb}^{M'},\\
    \PS{r_1b_1\sim r_2b_2}{}{M_{DC}(M')}
    &=\frac{\bar P_{r_1b_1}^{M'}\bar P_{r_2b_2}^{M'}}{\bar P^{M'}}.
\end{aligned}
\label{eq:def_M_DC}
\end{equation}
It therefore provides a no-segregation null model that controls for differences in expected degrees across geosocial coordinates. 

The \emph{degree-constrained social isolation model} $M_{DSI}(M')$ also preserves the coordinate distribution and marginal connectivities of $M'$, but allocates all connectivity within sociodemographic groups:
\begin{equation}
\begin{aligned}
    \pi_{rb}^{M_{DSI}(M')}&=\pi_{rb}^{M'},\\
    \PS{r_1b_1\sim r_2b_2}{}{M_{DSI}(M')}
    &= \frac{\bar P_{r_1b_1}^{M'}\bar P_{r_2b_2}^{M'}}{\pi_{b_1}^{M'}\bar P_{b_1}^{M'}} \mathbb{1}_{\left\{b_1 = b_2\right\}} .
\end{aligned}
\label{eq:def_M_DSI}
\end{equation}
Thus, $M_{DC}(M')$ and $M_{DSI}(M')$ preserve the same geosocial marginal connectivities, or degree constraints, as $M'$, but represent respectively random mixing and complete separation between sociodemographic groups.

When all geosocial coordinates have the same marginal connectivity, $\bar P_{rb}^{M'}\equiv\bar P^{M'}$, $M_{DC}(M')$ reduces to the Erd\H{o}s--R\'enyi model $M_{ER}$ introduced above,
\begin{equation}
    \PS{r_1b_1\sim r_2b_2}{}{M_{ER}}\equiv\bar P^{M'}.
    \label{eq:def_M_ER_reference}
\end{equation}
The corresponding \emph{social isolation model} $M_{SI}$ has connection probabilities
\begin{equation}
    \PS{r_1b_1\sim r_2b_2}{}{M_{SI}}
    = \frac{\bar P^{M'}}{\pi_{b_1}^{M'}} \mathbb{1}_{\left\{b_1 = b_2\right\}}.
    \label{eq:def_M_SI}
\end{equation}

All models above are formulated at regional geographical resolution; they can analogously be formulated with finer geographical coordinates, i.e., with geosocial coordinates $pb$ in $\GS\times\BS$. The regional coordinates are obtained by the coarsening $p\mapsto r(p)$.

\begin{remark}
\label{remark:reference_model_admissibility}
We formulated the models with full geosocial coordinates for both individuals, i.e., with coordinate space $\CS_1=\RS\times\BS$, $\CS_2=\RS\times\BS;$ marginalization over $r_2$ yields the geo-egocentric formulation, and further marginalization over $r_1$ yields the social-space formulation.
Whenever the reference models $M_{DC}$ and $M_{DSI}$ as defined above are admissible, by which we mean they define valid connection probabilities in $[0,1]$ for all coordinate pairs, marginalizing over the region of the second individual yields
\begin{equation}
    \PS{r_1b_1\sim b_2}{}{M_{DC}(M')}
    =\frac{\bar P_{r_1b_1}^{M'}\bar P_{b_2}^{M'}}{\bar P^{M'}}
    \label{eq:def_M_DC_geo_egocentric}
\end{equation}
and
\begin{equation}
    \PS{r_1b_1\sim b_2}{}{M_{DSI}(M')}
    = \frac{\bar P_{r_1b_1}^{M'}}{\pi_{b_1}^{M'}} \mathbb{1}_{\left\{b_1 = b_2\right\}} .
    \label{eq:def_M_DSI_geo_egocentric}
\end{equation}
These expressions can be used to define $M_{DC}(M')$ and $M_{DSI}(M')$ directly on the geo-egocentric coordinate space $\CS_1=\RS\times\BS$, $\CS_2=\BS$, even when the corresponding full-space constructions do not define admissible connection probabilities: at full geosocial resolution, preserving the marginal connectivity of every geosocial coordinate can require conditional connection probabilities exceeding one. The corresponding admissibility conditions are substantially less restrictive after marginalizing the second individual's geographical coordinate. For $M_{DSI}$, we require that 
\begin{equation}
    \bar{P}_{r b}^{M'} \leq \pi_b^{M'}
\label{eq:M_DSI_admissibility_condition}
\end{equation}
for all groups $b$ and regions $r$. Typically, average connection probabilities scale like the reciprocal of the population size, such that for large-scale geosocial networks  \eqref{eq:M_DSI_admissibility_condition} will typically hold. For $M_{DC}$ we require
\begin{equation}
    \bar P_{r_1b_1}^{M'}\bar P_{b_2}^{M'} \leq \bar P^{M'}.
\label{eq:M_DC_admissibility_condition}
\end{equation}
If we assume that $\bar{P}_b^{M'} / \bar{P}^{M'}$ is bounded from above, which is a reasonable assumption for categorical sociodemographic space $\BS,$ then \eqref{eq:M_DC_admissibility_condition} will typically hold as well in large-scale geosocial networks. Indeed, in the empirical application considered in Chapter~3, all connection probabilities implied by the geo-egocentric reference models remain below one.
\end{remark}


\subsection{Divergence-based segregation indices}
\label{sec:div_based_seg_indices}
For any segregation-defining model $M'$ and any of its associated distributions $\rho\in\{\tau,\sigma, \xi\}$, where $\tau$ denotes the joint coordinate--link distribution defined in \eqref{eq:def_tau}, $\sigma$ the coordinate distribution conditional on a link defined in \eqref{eq:def_sigma} and $\xi$ the coordinate distribution conditional on the absence of a link defined in \eqref{eq:def_xi}, we define the unnormalized segregation measure as the $f$-divergence from the corresponding degree configuration model on the geo-egocentric geosocial coordinate space:
\begin{equation}
    \mathcal{D}_{f,\rho}(M')
    \coloneqq
    \DF{\rho_{\RS\times \BS, \BS}^{M'}}{\rho_{\RS\times \BS, \BS}^{M_{DC}(M')}}.
    \label{eq:def_unnormalized_segregation}
\end{equation}
The divergence therefore measures the extent to which the distribution induced
by $M'$ departs from degree-constrained random mixing. It equals zero whenever
$\rho^{M'}=\rho^{M_{DC}(M')}$.

To place segregation measures on a common scale, we normalize this divergence
by its value under the corresponding degree-constrained social isolation model $M_{DSI}(M')$.
Because $M_{DSI}(M')$ preserves the coordinate distribution and marginal connectivities of $M'$, it has the same associated degree configuration model as $M'$, and we define the normalized segregation index by
\begin{equation}
    S_{f,\rho}(M')
    \coloneqq
    \frac{\mathcal{D}_{f,\rho}(M')}{\mathcal{D}_{f,\rho}\!\left(M_{DSI}(M')\right)}.
    \label{eq:def_normalized_segregation}
\end{equation}

The choices $M,$ $M_{HB}$, $M_{AB},$ and $M_{RI}$ for $M'$ define total, social, geographical, and traditional geographical segregation, respectively. The generator $f$ determines how deviations from the random-mixing model are weighted. 

Both $\mathcal{D}_{f,\rho}$ and $S_{f,\rho}$ are defined on the geo-egocentric geosocial coordinate space and with reference null model $M_{DC}(M').$ When interested in other coordinate spaces, or other reference null models $N,$ we write
\begin{equation}
    \mathcal{D}_{f,\rho}\left(M',N;\CS_1, \CS_2\right)
    \coloneqq
    \DF{\rho_{\CS_1, \CS_2}^{M'}}{\rho_{\CS_1, \CS_2}^{N}}.
    \label{eq:def_unnormalized_segregation_general_space_and_null}
\end{equation}

For the conditional coordinate distribution $\rho=\sigma$ on the geo-egocentric geosocial coordinate space, under the mild admissibility conditions stated in Lemma~\ref{lemma:normalization_of_indices}, $M_{DSI}(M')$ maximizes the $f$-divergence from $M_{DC}(M')$ among models with the same marginal coordinate distribution and marginal connectivities as $M'$, so that $0\leq S_{f,\sigma}(M')\leq1$, with zero corresponding to random mixing and one to maximal segregation subject to these constraints. By Corollary~\ref{corollary:normalization_social_space}, the analogous $M_{DC}$--$M_{DSI}$ normalization result holds on the social coordinate space. Lemma~\ref{lemma:T_D_F_to_S_D_F} justifies why, for large-scale geosocial networks, \eqref{eq:def_normalized_segregation} also defines a natural normalization for the joint coordinate--link distribution $\rho = \tau$-based measures. \eqref{eq:def_normalized_segregation} will in general not be in $[0,1]$ for $\rho=\xi$-based measures, however, for convenience, we define its normalized version nevertheless by \eqref{eq:def_normalized_segregation}.

On full geosocial space, $\DF{\rho_{\RS\times \BS, \RS\times \BS}^{M}}{\rho_{\RS\times \BS, \RS\times \BS}^{M_{DC}(M)}}$ does, in our opinion, not define an adequate segregation measure for most use cases: consider the case of a society where every geosocial group $r_1 b_1$ has the same connectivity, marginalized over regions $r_2$, to every social group $b_2$, but this connectivity mass is very non-uniformly distributed over regions, e.g., because individuals prefer to connect to geographically close-by regions. Hence, the resulting $\sigma^{M}$ or $\tau^{M}$ would be very different from $\sigma^{M_{DC}(M)}$ and $\tau^{M_{DC}(M)},$ with a large value of the corresponding divergence, even though we would argue that there is no segregation at all: indeed, the corresponding geo-egocentric measures $\mathcal{D}_{f,\rho}(M)$ and $S_{f,\rho}(M)$ will be zero. A natural candidate for a reference null model on full geosocial space is given by $M_{HB-AB}$, which removes both direct dependence of connectivity on social coordinates, and indirect dependence of connectivity on social coordinates via geographical connectivity patterns, while it maintains the broader geographical connectivity patterns. However, $M_{DSI}$ does not provide a useful maximum segregation reference relative to $M_{HB-AB}$, and we leave it to future research what constitutes good no-segregation and maximum-segregation reference models on full geosocial coordinate space.        

We recall from Remark~\ref{remark:reference_model_admissibility} that a reference model is admissible if its implied connection probabilities lie in $[0,1]$ for all coordinate pairs.  The corresponding admissibility conditions for $M_{DC}$ and $M_{DSI}$ on geo-egocentric geosocial coordinate space are given in \eqref{eq:M_DC_admissibility_condition} and \eqref{eq:M_DSI_admissibility_condition}.

\begin{lemma}\label{lemma:normalization_of_indices}
For a given segregation-defining model $M'$, let
\begin{equation}
    \mathcal M(M')=\left\{\widetilde M:\pi^{\widetilde M}=\pi^{M'},\;\bar P_{rb}^{\widetilde M}=\bar P_{rb}^{M'}\ \forall r,b\right\}.
\end{equation}
Assuming that the geo-egocentric reference models $M_{DC}(M')$ and $M_{DSI}(M')$ are admissible, we have
\begin{equation*}
    M_{DC}(M')\in\mathcal M(M'),\qquad M_{DSI}(M')\in\mathcal M(M')=\mathcal M\!\left(M_{DC}(M')\right),
\end{equation*}
and, for every generator $f$,
\begin{equation}
    M_{DSI}(M')\in\argmax_{\widetilde M\in\mathcal M(M')}\mathcal D_{f,\sigma}(\widetilde M).
    \label{eq:M_DSI_in_argmax}
\end{equation}
In particular, if $\bar P_{rb}^{M'}=\bar P^{M'}$ for all $r,b$, i.e., $M_{DC}(M')=M_{ER}$ and $M_{DSI}(M')=M_{SI}$, then
\begin{equation}
    M_{SI}\in\argmax_{\widetilde M\in\mathcal M(M')}\mathcal D_{f,\sigma}(\widetilde M)
    \label{eq:M_SI_in_argmax}
\end{equation}
and
\begin{equation}
    \mathcal D_{f,\sigma}(M_{SI})
    =
    \begin{cases}
        -\displaystyle\sum_b\pi_b^{M'}\log\pi_b^{M'},&f(x)=x\log x,\\[0.75em]
        \displaystyle\sum_b\pi_b^{M'}(1-\pi_b^{M'}),&f(x)=\tfrac12|x-1|,\\[0.75em]
        |\BS|-1,&f(x)=x^2-1.
    \end{cases}
    \label{eq:normalization_const_H_D_C}
\end{equation}
Thus, for these generators, the normalization constants correspond to those of the traditional Theil Information Theory, Dissimilarity, and Pearson $\chi^2$ (squared coefficient of variation) segregation indices \cite{reardon2002measures}.
\end{lemma}

\begin{corollary}
\label{corollary:normalization_social_space}
For a given segregation-defining model $M'$, let
\begin{equation}
    \mathcal M_{\BS}(M')
    =\left\{\widetilde M: \pi_b^{\widetilde M}=\pi_b^{M'},\;
    \bar P_b^{\widetilde M}=\bar P_b^{M'}\ \forall b\right\}.
\end{equation}
Assuming that $M_{DC}(M')$ and $M_{DSI}(M')$ are admissible, for every generator $f$,
\begin{equation}
    M_{DSI}(M') \in \argmax_{\widetilde M\in\mathcal M_{\BS}(M')}
    \mathcal D_{f,\sigma} \left(\widetilde M,M_{DC}(M');\BS,\BS\right).
    \label{eq:M_DSI_in_argmax_BB}
\end{equation}
\end{corollary}


\subsection{Recovery of traditional geographical segregation measures as special cases}
\label{sec:recovery_of_trad_seg}
In section~\ref{sec:trad_geo_seg} we showed how traditional geographical segregation measures can be interpreted as being built upon simple geosocial network models. We are now in a position to show that our suggested geosocial segregation measures defined in \eqref{eq:def_unnormalized_segregation} and \eqref{eq:def_normalized_segregation} contain traditional geographical segregation measures as exact special cases. 

\begin{lemma}\label{lemma:relation_between_geo_and_social_f_divergence}
Let $M_\phi$ be the geospatial connectivity model defined in \eqref{eq:def_geospatial_model_phi}, for a symmetric kernel $\phi$ satisfying the constant-degree condition \eqref{eq:constant_weighted_row_sums_phi}, and let $M_{RI}$ be the regional isolation model. Then, for every generator $f$, the following hold.

\begin{enumerate}
    \item[(i)] On the geo-egocentric coordinate space,
    \begin{equation}
        \mathcal D_{f,\sigma}\left(M_\phi,M_{DC}(M_\phi);\GS\times\BS,\BS\right)
        = \mathcal D_{f,\widetilde\pi},
        \label{eq:equivalence_of_indices_under_spatial_kernel}
    \end{equation}
    and, in particular,
    \begin{equation}
        \mathcal D_{f,\sigma}\left(M_{RI},M_{DC}(M_{RI});\RS\times\BS,\BS\right)
        = \mathcal D_{f,\pi},
        \label{eq:equivalence_of_indices_under_regional_isolation}
    \end{equation}
    where $\mathcal D_{f,\widetilde\pi}$ and $\mathcal D_{f,\pi}$ are defined in \eqref{eq:traditional_egocentric_segregation_indices_as_f_divergence} and \eqref{eq:traditional_segregation_indices_as_f_divergence}, respectively.

    \item[(ii)] Our definition of $\mathcal D_{f,\sigma}$ in \eqref{eq:def_unnormalized_segregation} uses $M_{DC}$ as the reference model. If instead we use $M_{HB}(M_\phi)$ as the reference model, the corresponding divergences on the geo-egocentric and full geosocial coordinate spaces generally do not coincide with one another or with $\mathcal D_{f,\widetilde\pi}$. In particular,
    \begin{align}
        D_f\!\left(
        \sigma_{\GS\times\BS,\BS}^{M_\phi}
        \middle\| \sigma_{\GS\times\BS,\BS}^{M_{HB}(M_\phi)}
        \right)
        &\neq
        D_f\!\left(
        \sigma_{\GS\times\BS,\GS\times\BS}^{M_\phi}
        \middle\| \sigma_{\GS\times\BS,\GS\times\BS}^{M_{HB}(M_\phi)}
        \right),
        \label{eq:lemma_ineq1}\\
        D_f\!\left(
        \sigma_{\GS\times\BS,\BS}^{M_\phi}
        \middle\| \sigma_{\GS\times\BS,\BS}^{M_{HB}(M_\phi)}
        \right)
        &\neq \mathcal D_{f,\widetilde\pi},
        \label{eq:lemma_ineq2}\\
        D_f\!\left(
        \sigma_{\GS\times\BS,\GS\times\BS}^{M_\phi}
        \middle\| \sigma_{\GS\times\BS,\GS\times\BS}^{M_{HB}(M_\phi)}
        \right)
        &\neq \mathcal D_{f,\widetilde\pi}.
        \label{eq:lemma_ineq3}
    \end{align}
    For the regional isolation model, however,
    \begin{equation}
    \begin{aligned}
        D_f\!\left(
        \sigma_{\RS\times\BS,\BS}^{M_{RI}}
        \middle\| \sigma_{\RS\times\BS,\BS}^{M_{HB}(M_{RI})}
        \right)
        &= D_f\!\left(
        \sigma_{\RS\times\BS,\RS\times\BS}^{M_{RI}}
        \middle\| \sigma_{\RS\times\BS,\RS\times\BS}^{M_{HB}(M_{RI})}
        \right)\\
        &= \sum_{r\in\RS}\sum_{b_1,b_2\in\BS} \pi_r\pi_{b_1}\pi_{b_2}
        f\left(
        \frac{\pi_{b_1\mid r}\pi_{b_2\mid r}}{\pi_{b_1}\pi_{b_2}}
        \right).
    \end{aligned}
    \label{eq:equivalence_of_indices_under_HB_phi}
    \end{equation}
    Hence, for $g(x)=x\log x$ or $g(x)=-\log x$,
    \begin{equation}
    \begin{aligned}
        D_g\!\left(
        \sigma_{\RS\times\BS,\BS}^{M_{RI}}
        \middle\| \sigma_{\RS\times\BS,\BS}^{M_{HB}(M_{RI})}
        \right)
        &= D_g\!\left(
        \sigma_{\RS\times\BS,\RS\times\BS}^{M_{RI}}
        \middle\| \sigma_{\RS\times\BS,\RS\times\BS}^{M_{HB}(M_{RI})}
        \right)\\
        &=2\mathcal D_{g,\pi}.
    \end{aligned}
    \label{eq:equivalence_for_xlogx}
    \end{equation}

    \item[(iii)] On the social coordinate space, geographical segregation remains measurable:
    \begin{equation}
    \begin{aligned}
        D_f\!\left(\sigma_{\BS,\BS}^{M_{RI}}\middle\|\sigma_{\BS,\BS}^{M_{DC}(M_{RI})}\right)
        &=D_f\!\left(\sigma_{\BS,\BS}^{M_{RI}}\middle\|
        \sigma_{\BS,\BS}^{M_{HB}(M_{RI})}\right)\\
        &=\sum_{b_1,b_2\in\BS}\pi_{b_1}\pi_{b_2}
        f\left(
        \sum_{r\in\RS}\pi_r\frac{\pi_{b_1\mid r}\pi_{b_2\mid r}}{\pi_{b_1}\pi_{b_2}}
        \right).
    \end{aligned}
    \label{eq:S_under_blau_space_only}
    \end{equation}
\end{enumerate}
\end{lemma}

\begin{corollary}
\label{corollary:recovery_normalized_traditional_indices}
Under the regional isolation model $M_{RI}$,
\begin{equation}
\begin{aligned}
    S_{x\log x,\sigma}(M_{RI}) &= \underline H,\\
    S_{\frac12|x-1|,\sigma}(M_{RI}) &= \underline D,\\
    S_{x^2-1,\sigma}(M_{RI}) &= \underline{\chi^2},
\end{aligned}
\end{equation}
where $\underline H$, $\underline D$, and $\underline{\chi^2}$ denote the normalized Theil Information Theory, Dissimilarity, and squared coefficient of variation segregation indices, respectively.
\end{corollary}

\begin{remark}
Lemma~\ref{lemma:relation_between_geo_and_social_f_divergence} says that under \emph{geo-egocentric geosocial space} the $f$-divergence based segregation indices based upon connection probabilities contain the traditional geographical segregation indices as the special cases when the regional isolation model $M_{RI}$ is compared to the null model $M_{ER}$ (Erd\H{o}s--R\'enyi random graph). The traditional egocentric (or spatially weighted) segregation indices for a kernel $\phi$ are the special cases when $M_{\phi}$ is compared to $M_{ER}.$ Corollary~\ref{corollary:recovery_normalized_traditional_indices} shows that this also holds for the normalized variants of the traditional Theil, Dissimilarity and $\chi^2$ indices. To our knowledge, the normalizations in the geographical segregation literature \cite{reardon2002measures} are formulated in an index-specific manner, while our approach as formulated in Lemma~\ref{lemma:normalization_of_indices} provides a general normalization approach.

Up to the factor $2$, the recovery of $\mathcal{D}_{f,\pi}$ by our framework also holds when the reference null is the homogeneous sociodemographic space model $M_{HB}(M_{RI})$ and the function $f(x)$ is given by $x\log x$ or $-\log x$, see \eqref{eq:equivalence_for_xlogx}. However, for general $f,$ using $M_{HB}(M_{RI})$ as the reference null does not recover the traditional geographical measures $\mathcal D_{f,\pi}$. This is noteworthy as $\mathcal D_{f,\pi}$ compares the joint $\pi_{rb}$ to the product of its marginals $\pi_r\pi_b$ and $M_{HB}(M_{RI})$ also substitutes $\pi_r \pi_b$ for the joint $\pi_{rb}$. Note that in our segregation framework, $M_{HB}(M)$ constitutes social segregation when compared to its random mixing null $M_{DC}(M_{HB}(M)).$ When applied to the regional isolation model, $M_{HB}(M_{RI})$ removes the geographical contribution to segregation and keeps social segregation of $M_{RI},$ but $M_{RI}$ does not exhibit social segregation and it is no surprise $M_{HB}$ can constitute a reference null model relative to a purely geographically segregated model. However, with the exception of Lemma~\ref{lemma:relation_between_geo_and_social_f_divergence}, we will not consider it as a reference null model, and solely restrict its use to counterfactually model social segregation. 

While $M_{DC}(M_{RI})$ does not constitute a useful null model for full geosocial coordinate space, as it does not control for geographical homophily, $M_{HB}(M_{RI})$ does control for geographical homophily: when used as a null relative to $M_{RI}$ on full geosocial space, the resulting divergence coincides with the one on geo-egocentric geosocial space, \eqref{eq:equivalence_of_indices_under_HB_phi}. While restriction to social space only \eqref{eq:S_under_blau_space_only} yields a different expression for the divergence from $M_{RI}$ to either $M_{DC}(M_{RI})$ or $M_{HB}(M_{RI})$---on social coordinate space they coincide as null models---it is still possible to meaningfully measure geographical segregation on social space.
\end{remark}

\subsection{Choice of coordinate space and information loss}
\label{sec:choice_of_coord_space}
We have now established several reasons why geo-egocentric geosocial coordinate space provides a natural setting for measuring segregation in geosocial networks:
\begin{enumerate}
    \item[(i)] The degree-configuration model and the degree-constrained social isolation model are well suited reference models for geo-egocentric geosocial coordinate space and allow a principled normalization.
    \item[(ii)] Any measure on geo-egocentric geosocial coordinate space is insensitive to how quickly or irregularly connectivity decays or varies with geographical distance, as long as the total connectivity from geosocial coordinates to social groups stays the same, whereas any divergence-based measure defined on full geosocial coordinate space will be highly sensitive to variations of connectivity with geographical distance relative to a null model. Crucially, measures on geo-egocentric geosocial coordinate space remain sensitive to how social connectivity itself varies across regions, while measures on social coordinate space lose this information.
    \item[(iii)] As shown in Lemma~\ref{lemma:relation_between_geo_and_social_f_divergence}, traditional geographical segregation measures can be recovered as exact special cases when segregation is measured on geo-egocentric geosocial coordinate space.
    \item[(iv)] Compared to full geosocial coordinate space, geo-egocentric geosocial coordinate space is less restrictive when formulating degree-constrained reference models, see remark~\ref{remark:reference_model_admissibility}.
\end{enumerate}
Thus, we formulate our segregation framework for geo-egocentric geosocial coordinate space. However, the information loss is real. In Lemma~\ref{lemma:coordinate_coarsening} we formally establish that any geographical coordinate coarsening cannot increase the values of our $\sigma$-based segregation measures, and we expect such coarsening to typically decrease the values of our segregation measures. Furthermore, while for many use cases the additional information contained in full geosocial coordinate space might be irrelevant, there are certainly questions which can only be answered on full geosocial coordinate space. E.g., one might be interested in how geographical connectivity patterns vary across different social groups, or whether the preference to befriend members of one's own group increases or decreases as the geographical distance to individuals increases.

\subsection{Conditional coordinate and joint coordinate--link distributions}
\label{sec:sigma_tau_under_sparsity}
The joint coordinate--link distribution $\tau$ contains the full information about coordinates and connectivity represented by a network model as defined in section~\ref{sec:coordinate_spaces_and_geosocial_networks}, whereas the conditional coordinate distribution $\sigma$ gives rise to the natural normalization in Lemma~\ref{lemma:normalization_of_indices} and the recovery of traditional geographical segregation measures in Lemma~\ref{lemma:relation_between_geo_and_social_f_divergence} and Corollary~\ref{corollary:recovery_normalized_traditional_indices}. As we will show in Section~\ref{sec:homophily}, the joint coordinate--link distribution also allows a natural homophily decomposition of the $\tau$-based segregation measures via conditional $f$-divergences.

Which distribution should be used to measure segregation?
The following Lemma and Corollary show that in the sparse-network limit of vanishing marginal connectivity $\bar P\to0$, for three-times continuously differentiable generators $f$, the rescaled unnormalized measure $\mathcal D_{f,\tau} / \bar{P}$ and $\mathcal D_{f,\sigma}$ are asymptotically equal, and the normalized measures $S_{f,\tau}$ and $S_{f,\sigma}$ are asymptotically equal.
For the Total Variation distance, we obtain exact equality of the normalized measures. Large-scale geosocial networks are typically sparse, as the average degree (average number of friends) typically does not, or only slowly, increase with population size; rather, the marginal connectivity $\bar P$ scales approximately as the reciprocal of the population size. Furthermore, except for the Total Variation distance with generator $f(x)=\frac12|x-1|$, most commonly used $f$-divergences have three-times continuously differentiable generators. Hence, for the applications we have in mind, there is little practical difference between using $\sigma$ and $\tau$, and we choose the simpler distribution $\sigma$ as our default.

\begin{lemma}\label{lemma:T_D_F_to_S_D_F}
\begin{enumerate}
Let $M$ and $N$ be two network models   .
\item[(i)] Assume that $\pi^M = \pi^N = \pi$ and $\bar{P}^M = \bar{P}^N = \bar{P}.$
Then for any generator $f$
\begin{equation}
    \mathcal D_{f,\tau}(M,N;\CS_1,\CS_2)
    =\bar{P}\mathcal D_{f,\sigma}(M,N;\CS_1,\CS_2)
    +(1-\bar{P})\mathcal D_{f,\xi}(M,N;\CS_1,\CS_2).
    \label{eq:T_D_f_eq_sum_of_S_D_fs}
\end{equation}

In particular, for $g(x) = \frac{1}{2}|x-1|$ it holds that
\begin{equation}
    \mathcal D_{g,\tau}(M,N;\CS_1,\CS_2)
    =2\bar{P}\mathcal D_{g,\sigma}(M,N;\CS_1,\CS_2)
    =2(1-\bar{P})\mathcal D_{g,\xi}(M,N;\CS_1,\CS_2).
    \label{eq:T_D_f_is_twice_dissimilarity}
\end{equation}

Allowing $\bar{P}^M\neq \bar{P}^N,$ for $h(x) = x \log x$
\begin{equation}
\begin{aligned}
    \mathcal D_{h,\tau}(M,N;\CS_1,\CS_2)
    &=\bar{P}^M\mathcal D_{h,\sigma}(M,N;\CS_1,\CS_2)
    +(1-\bar{P}^M)\mathcal D_{h,\xi}(M,N;\CS_1,\CS_2)\\
    &\quad+\DKL{\operatorname{Bernoulli}(\bar{P}^M)}{\operatorname{Bernoulli}(\bar{P}^N)}.
\end{aligned}
\label{eq:T_D_f_and H}
\end{equation}

\item[(ii)] 
Let there be a sequence of models $M_i,$ $N_i,$ such that for all $i$ it holds that $\pi^{M_i} = \pi^{N_i} = \pi,$ $\bar{P}^{M_i} = \bar{P}^{N_i} = \bar{P}_i,$ and $\bar{P}_i\to 0$ as $i \to \infty$. Furthermore assume
\begin{align}
    \frac{\PS{c_1\sim c_2}{}{M_i}}{\bar{P}_i} &= \frac{\PS{c_1\sim c_2}{}{M_1}}{\bar{P}_1},\quad \forall c_1,c_2,\;\forall i,\label{eq:constant_ratio_of_Pmi_to_barP}\\
    \frac{\PS{c_1\sim c_2}{}{N_i}}{\bar{P}_i} &= \frac{\PS{c_1\sim c_2}{}{N_1}}{\bar{P}_1},\quad \forall c_1,c_2,\;\forall i.\label{eq:constant_ratio_of_Pni_to_barP}
\end{align}
If $f$ is three-times continuously differentiable in a neighborhood of $1$, then
\begin{align}
    \frac{\mathcal D_{f,\tau}(M_i,N_i;\CS_1,\CS_2)}{\bar{P}_i}
    = \mathcal D_{f,\sigma}(M_1,N_1;\CS_1,\CS_2)
    + \mathcal{O}(\bar P_i), \quad i\to\infty.
    \label{eq:T_D_f_to_S_D_f}
\end{align}

\item[(iii)] 
Conversely, let there be a sequence of models $M_i,$ $N_i,$ such that for all $i$ it holds that $\pi^{M_i} = \pi^{N_i} = \pi,$ $\bar{P}^{M_i} = \bar{P}^{N_i} =\bar{P}_i,$ and $\bar{P}_i\to 1$ as $i \to \infty$. Furthermore assume 
\begin{align}
    \frac{1-\PS{c_1\sim c_2}{}{M_i}}{1-\bar{P}_i} &=  \frac{1-\PS{c_1\sim c_2}{}{M_1}}{1-\bar{P}_1} ,\quad \forall c_1,c_2,\;\forall i,\\
    \frac{1-\PS{c_1\sim c_2}{}{N_i}}{1-\bar{P}_i} &=  \frac{1-\PS{c_1\sim c_2}{}{N_1}}{1-\bar{P}_1} ,\quad \forall c_1,c_2,\;\forall i.
\end{align}
If $f$ is three-times continuously differentiable in a neighborhood of $1$, then
\begin{align}
    \frac{\mathcal D_{f,\tau}(M_i,N_i;\CS_1,\CS_2)}{1-\bar{P}_i}
    = \mathcal D_{f,\xi}(M_1,N_1;\CS_1,\CS_2) 
    + \mathcal{O}(1 - \bar P_i), \quad i\to\infty.
    \label{eq:T_D_f_to_S_D_f_nsim}
\end{align}

\item[(iv)] In addition to the asymptotic equivalence described in (ii) and (iii), for the generator $f(x)=0.5(x-1)^2$ of the $\chi^2$ divergence, there is exact equivalence when $N=M_{ER}$ the Erd\H{o}s--R\'enyi random graph null, i.e., $P^N\equiv \bar{P},$ in the sense that
\begin{align}
    \mathcal D_{f,\tau}(M,M_{ER};\CS_1,\CS_2)
    &= \frac{1}{2\bar{P}(1-\bar{P})}\sum_{c_1,c_2}\pi_{c_1}\pi_{c_2}\left(\PS{c_1\sim c_2}{}{M} - \bar{P}\right)^2 ,\\
    \mathcal D_{f,\sigma}(M,M_{ER};\CS_1,\CS_2)
    &= \frac{1-\bar{P}}{\bar{P}}\mathcal D_{f,\tau}(M,M_{ER};\CS_1,\CS_2),\label{eq:simga_chi_2_equiv_tau}\\
    \mathcal D_{f,\xi}(M,M_{ER};\CS_1,\CS_2)
    &= \frac{\bar{P}}{1-\bar{P}}\mathcal D_{f,\tau}(M,M_{ER};\CS_1,\CS_2).\label{eq:xi_chi_2_equiv_tau}
\end{align}

\end{enumerate}
\end{lemma}

\begin{corollary}
\label{corollary:normalized_tau_sigma}
For the normalized measures defined in \eqref{eq:def_normalized_segregation} we obtain the following relations.
\begin{enumerate}
\item[(i)] Let $M_i$ be a sequence of models such that $\pi^{M_i}=\pi$, $\bar P^{M_i}=\bar P_i\to0$, and \eqref{eq:constant_ratio_of_Pmi_to_barP} holds. If $f$ is three-times continuously differentiable in a neighborhood of $1$, then
\begin{equation}
    S_{f,\tau}(M_i) = S_{f,\sigma}(M_1)
    + \mathcal{O}(\bar P_i), \quad i\to\infty.
    \label{eq:asymptotic_equality_norm_sigma_tau}
\end{equation}

    \item[(ii)] For $g(x)=\frac12|x-1|$,
    \begin{equation}
        S_{g,\tau}(M)=S_{g,\sigma}(M)=S_{g,\xi}(M).
        \label{eq:equality_norm_sigma_tau_tv}
    \end{equation}

    \item[(iii)] For $h(x)=0.5(x-1)^2$, if $M_{DC}(M)=M_{ER}$, then
    \begin{equation}
        S_{h,\tau}(M)=S_{h,\sigma}(M)=S_{h,\xi}(M).
        \label{eq:equality_norm_sigma_tau_chi2}
    \end{equation}
\end{enumerate}
\end{corollary}

The relations in Lemma~\ref{lemma:T_D_F_to_S_D_F} and Corollary~\ref{corollary:normalized_tau_sigma} are illustrated in the first two rows of Fig.~\ref{fig:hom_quant_demo}, which compares the $\sigma$-, $\tau$-, and $\xi$-based measures when the marginal connectivity $\bar P$ ranges across the interval $(0,1)$ for several generators $f$.

\begin{figure}[htbp]
    \centering
    \includegraphics[width=1.0\textwidth]{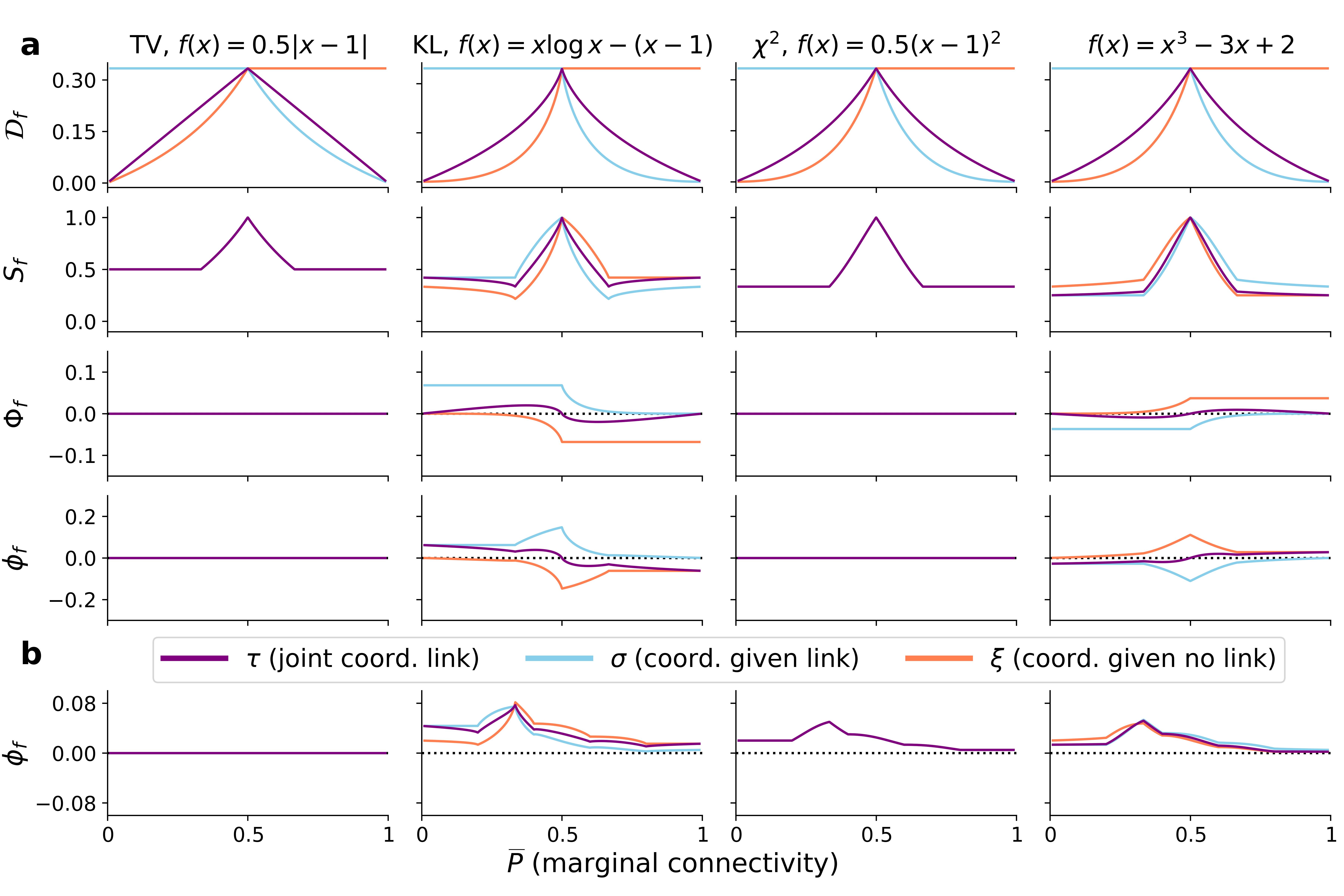}  
    \caption{\emph{Homophily quantification depends on generator $f.$} \textbf{a} The induced $\sigma^M$ and $\sigma^N$ are as in example~\ref{example:homophily quantification} (i), with $\bar{P}$ varied from $0$ to $1$ and $\epsilon = \min\{\bar{P}, 1-\bar{P}\} / (9\bar{P}).$ As predicted by our theory, for TV distance and the $\chi^2$ divergence, 
    the normalized measures $S_{f}$ and $\phi_{f}$ based on $\sigma$, $\tau$, and $\xi$ all coincide exactly, while for the KL-divergence and the cubic generator $f(x)=x^3-3x+2,$ the normalized measures based on $\sigma$ and $\xi$ approach the ones based on $\tau$ when $\bar{P}$ approaches $0$ and $1$ respectively. The TV- and $\chi^2$-based homophily quantifications $\phi$ are zero, whereas for the KL-divergence $\phi_{f,\sigma}$ is positive and for the cubic generator $\phi_{f,\sigma}$ is negative; note that $f''(x) = 6x$ is increasing on $(0,\infty)$ for the cubic generator and $f''(x)=1/x$ is decreasing for the KL divergence. \textbf{b} The normalized homophily quantification $\phi_f$ for example~\ref{example:homophily quantification} (ii) with $\epsilon = \min\{\bar{P}, (1-\bar{P})/2\} / (25\bar{P}).$ For the KL divergence, both terms $A$ and $B$ in \eqref{eq:example_hom_AB_eq} are positive, for the $\chi^2$ divergence $A>0$ and $B=0,$ while for the cubic generator $A>0>B,$ but $A>|B|.$ These relations correspond to $\phi_{f,\sigma}$ being largest for the KL divergence, slightly smaller for the $\chi^2$ divergence, and even smaller, but still positive, for the cubic generator. For the TV distance, $A=B=0,$ and hence $\phi_{f,\rho}=0$ for $\rho\in\{\tau,\sigma,\xi\}$ and $f(x)=\frac12|x-1|$.}  
    \label{fig:hom_quant_demo}
\end{figure}

\subsection{Homophily, heterophily, and network modularity}\label{sec:homophily}
The segregation indices $S_{f,\tau}(M')$ or $S_{f,\sigma}(M')$ measure the magnitude of departures from degree-constrained random mixing, but do not distinguish whether these departures are homophilous or heterophilous. We call a departure homophilous if, under $M'$ relative to $M_{DC}(M')$, it increases within-group connectivity or decreases between-group connectivity, and heterophilous if it decreases within-group connectivity or increases between-group connectivity. The measures based on the joint coordinate--link distribution $\tau$ allow a natural decomposition, described below, into homophilous and heterophilous contributions via conditional $f$-divergences. We propose a signed homophily measure by subtracting the heterophilous contribution from the homophilous contribution. We also present a decomposition of the measures based on the conditional coordinate distribution $\sigma$ into homophilous and heterophilous contributions and show in Lemma ~\ref{lemma:homophily_quantification} that for three times continuously differentiable generators $f$ the $\tau$- and $\sigma$ based decompositions are asymptotically equivalent in the sparse-network limit $\bar P\to 0$, and for Total Variation distance the normalized versions coincide exactly, thus extending the result of Lemma~\ref{lemma:T_D_F_to_S_D_F}. Furthermore, we show that for Total Variation distance, the proposed signed homophily measure recovers Network modularity, or when normalized, nominal assortativity. 

 In this section, for a given segregation-defining model $M'$, let $N = M_{DC}(M')$ denote its degree configuration model. We define the homophilous, heterophilous and ambiphilous sets of coordinate pairs as
\begin{align}
    \mathcal{I}_{\mathrm{hom}}(M')
    &\coloneqq
    \left\{
        (rb,b): \PS{rb\sim b}{}{M'} > \PS{rb\sim b}{}{N}
    \right\}
    \nonumber\\
    &\quad\cup
    \left\{
        (rb_1,b_2):b_1\neq b_2,\, \PS{rb_1\sim b_2}{}{M'} < \PS{rb_1\sim b_2}{}{N}
    \right\},
    \label{eq:def_homophilous_coordinate_set}\\
    \mathcal{I}_{\mathrm{het}}(M')
    &\coloneqq
    \left\{
        (rb,b): \PS{rb\sim b}{}{M'} < \PS{rb\sim b}{}{N}
    \right\}
    \nonumber\\
    &\quad\cup
    \left\{
        (rb_1,b_2): b_1\neq b_2,\, \PS{rb_1\sim b_2}{}{M'} > \PS{rb_1\sim b_2}{}{N}
    \right\},
    \label{eq:def_heterophilous_coordinate_set}\\
    \mathcal{I}_{\mathrm{amb}}(M')
    &\coloneqq
    \left\{
        (rb_1,b_2):\PS{rb_1\sim b_2}{}{M'} = \PS{rb_1\sim b_2}{}{N}
    \right\}.
    \label{eq:def_ambiphilous_coordinate_set}
\end{align}

The $f$-divergence based on the joint coordinate--link distribution $\tau$ admits a natural decomposition into these two types of deviation. Since $M'$ and $N$ have the same (unconditional) coordinate distributions $\pi^{M'}=\pi^N$, they assign the same marginal probability mass to each of the sets in
\eqref{eq:def_homophilous_coordinate_set}--\eqref{eq:def_ambiphilous_coordinate_set}. For any set $\mathcal{I}$ of coordinate pairs, let
\begin{equation}
    \pi^{M'}(\mathcal{I})
    \coloneqq
    \sum_{(rb_1,b_2)\in\mathcal{I}}
    \pi_{rb_1}^{M'}\pi_{b_2}^{M'}.
\end{equation}
The conditional decomposition property of $f$-divergences (see \cite{polyanskiy2025information}) then gives
\begin{equation}
\begin{aligned}
    & \mathcal{D}_{f,\tau}(M')\\
    &= \pi^{M'}\!\left(\mathcal{I}_{\mathrm{hom}}(M')\right)\\
    &\quad\quad \cdot\DF{\tau^{M'}\!\left(rb_1,b_2,e\mid (rb_1,b_2)\in\mathcal{I}_{\mathrm{hom}}(M')\right)}
    {\tau^{N}\!\left(rb_1,b_2,e\mid (rb_1,b_2)\in\mathcal{I}_{\mathrm{hom}}(M')\right)} \\
    &\quad + \pi^{M'}\!\left(\mathcal{I}_{\mathrm{het}}(M')\right)\\
    &\quad\quad\cdot\DF{\tau^{M'}\!\left(rb_1,b_2,e\mid (rb_1,b_2)\in\mathcal{I}_{\mathrm{het}}(M')\right)}
    {\tau^{N}\!\left(rb_1,b_2,e\mid (rb_1,b_2)\in\mathcal{I}_{\mathrm{het}}(M')\right)} \\
    &\eqqcolon
    \mathcal{A}_{f,\tau}(M') + \mathcal{B}_{f,\tau}(M').
    \label{eq:tau_homophily_decomposition}
\end{aligned}
\end{equation}

The ambiphilous set makes no contribution because the conditional connection probabilities under $M'$ and $N$ coincide there. We say that $\mathcal{A}_{f,\tau}(M')$ and $\mathcal{B}_{f,\tau}(M')$ quantify the unnormalized homophilous and heterophilous contributions, respectively.
We define a corresponding signed homophily measure by
\begin{equation}
    \Phi_{f,\tau}(M') \coloneqq \mathcal{A}_{f,\tau}(M') - \mathcal{B}_{f,\tau}(M').
    \label{eq:def_unnormalized_tau_homophily}
\end{equation}

Normalizing by the divergence of the maximal-segregation reference from the null $N,$ i.e., by the same normalization constant used to obtain $S_{f,\tau}(M')$, we define
\begin{equation}
    \begin{aligned}
        \alpha_{f,\tau}(M') &\coloneqq \frac{\mathcal{A}_{f,\tau}(M')}{\mathcal{D}_{f,\tau}\!\left(M_{DSI}(M')\right)},\\
        \beta_{f,\tau}(M') &\coloneqq \frac{\mathcal{B}_{f,\tau}(M')}{\mathcal{D}_{f,\tau}\!\left(M_{DSI}(M')\right)},\\
        \phi_{f,\tau}(M') &\coloneqq \alpha_{f,\tau}(M')-\beta_{f,\tau}(M').
    \end{aligned}
    \label{eq:def_normalized_tau_homophily}
\end{equation}
It follows that
\begin{equation}
    S_{f,\tau}(M') = \alpha_{f,\tau}(M') + \beta_{f,\tau}(M'),
    \label{eq:tau_segregation_homophily_relation}
\end{equation}
and
\begin{equation}
    -S_{f,\tau}(M') \leq \phi_{f,\tau}(M') \leq S_{f,\tau}(M').
\end{equation}
If every deviation from random mixing is homophilous, then
$\phi_{f,\tau}(M')=S_{f,\tau}(M')$; if every deviation is heterophilous, then
$\phi_{f,\tau}(M')=-S_{f,\tau}(M')$.

\subsubsection{Positive and unique generators}
The conditional-divergence decomposition used above to split $\mathcal{D}_{f,\tau}(M')$ into homophilous and heterophilous contributions is not available for $\mathcal{D}_{f,\sigma}(M')$, because conditioning on the presence of a link generally changes the probability mass assigned to the homophilous and heterophilous coordinate sets. Instead, we decompose $\mathcal{D}_{f,\sigma}(M')$ by grouping its individual summands
\begin{equation}
    u_{f,\sigma}^{M'}(rb_1,b_2)\coloneqq \sigma^{N}(rb_1,b_2)f\left(\frac{\sigma^{M'}(rb_1,b_2)}{\sigma^{N}(rb_1,b_2)}\right).
    \label{eq:def_sigma_cellwise_divergence}
\end{equation}
according to whether the corresponding coordinate pair belongs to $\mathcal{I}_{\mathrm{hom}}(M')$ or $\mathcal{I}_{\mathrm{het}}(M')$. However, we need to take care of a technical detail: the generator of an $f$-divergence is not unique. Two generators $f$ and $\tilde f$ generate the same $f$-divergence if and only if $\tilde f(x)=f(x)+c(x-1)$ for some $c\in\mathbb{R}$ \cite{polyanskiy2025information}. We denote the equivalence class of all generators defining the same $f$-divergence by
\begin{equation}
    [f]\coloneqq\left\{f(x)+c(x-1):c\in\mathbb{R}\right\}.
\end{equation}
Although all $f$ in $[f]$ generate the same overall $f$-divergence, the individual summands \eqref{eq:def_sigma_cellwise_divergence} generally depend on the specific representative $f$.
However, when defining homophilous and heterophilous contributions via individual summands, we need the summands to be uniquely determined by the $f$-divergence and non-negative. For differentiable generators, we therefore choose the unique representative $\tilde{f}$ of $[f]$ satisfying $\tilde{f}'(1)=0$: for any $f\in[f]$, this representative is given by $\tilde f(x)=f(x)-f'(1)(x-1)$. Since $\tilde f$ is convex and $\tilde f(1)=0$ --- these properties are shared by all generators --- the additional property $\tilde f'(1)=0$ implies that $\tilde f(x)\geq0$ for all $x\geq0$. Moreover, this representative is unique: assume $g\in [f]$ satisfies $g'(1)=0.$ Because $g\in[f]$, we can write $g(x) = \tilde{f}(x) + c(x-1)$ for some $c\in\mathbb{R}$ and it follows that $g'(1)=c$ and hence $c=0.$

For non-differentiable generators, we generalize this approach by choosing the unique representative of $[f]$ whose subdifferential at $1$ is centered at zero: we choose any $f\in[f],$ determine its subdifferential $\partial f(1)=[a,b]$, and define
\begin{equation}
\tilde f(x)\coloneqq f(x)-\frac{a+b}{2}(x-1),
\end{equation}
such that
\begin{equation}
\partial\tilde f(1)=\left[-\frac{b-a}{2},\frac{b-a}{2}\right].
\label{eq:sub_diff_centered_at_zero}
\end{equation}
Since adding $c(x-1)$ to $f(x)$ shifts both endpoints of its subdifferential $\partial f(1)$ by $c,$ the representative $\tilde f$ is independent of the initially chosen $f$ and hence unique. Moreover, because $0\in\partial\tilde f(1)$ and $\tilde f(1)=0$, it follows by convexity that $\tilde f(x)\geq0$ for all $x\geq0$. For the Total Variation distance, the standard generator $g(x)=\frac12|x-1|$ already satisfies \eqref{eq:sub_diff_centered_at_zero}, since $\partial g(1)=[-\frac{1}{2},\frac{1}{2}]$. 
We henceforth always choose the representative defined in \eqref{eq:sub_diff_centered_at_zero} --- for ease of notation we denote it by $f$ --- and define the non-negative contribution of coordinate pair $(rb_1,b_2)$ to the $f$-divergence by \eqref{eq:def_sigma_cellwise_divergence}. 

Because $M'$ and $N$ have the same (unconditional) coordinate distribution $\pi$ and marginal connectivity $\bar P$, the sign of $\sigma^{M'}(rb_1,b_2)-\sigma^{N}(rb_1,b_2)$ agrees with the sign of the corresponding difference in connection probabilities. We therefore define the unnormalized homophilous and heterophilous contributions by
\begin{align}
    \mathcal{A}_{f,\sigma}(M')&\coloneqq\sum_{(rb_1,b_2)\in\mathcal{I}_{\mathrm{hom}}(M')} u_{f,\sigma}^{M'}(rb_1,b_2),\label{eq:def_sigma_homophilous_contribution}\\
    \mathcal{B}_{f,\sigma}(M')&\coloneqq\sum_{(rb_1,b_2)\in\mathcal{I}_{\mathrm{het}}(M')} u_{f,\sigma}^{M'}(rb_1,b_2),\label{eq:def_sigma_heterophilous_contribution}
\end{align}
respectively. Their sum gives the overall divergence
\begin{equation}
    \mathcal{D}_{f,\sigma}(M')=\mathcal{A}_{f,\sigma}(M')+\mathcal{B}_{f,\sigma}(M').
    \label{eq:sigma_homophily_decomposition}
\end{equation}
In analogy with \eqref{eq:def_unnormalized_tau_homophily} we define the unnormalized signed homophily measure by
\begin{equation}
    \Phi_{f,\sigma}(M')\coloneqq\mathcal{A}_{f,\sigma}(M')-\mathcal{B}_{f,\sigma}(M').
\label{eq:def_unnormalized_sigma_homophily}
\end{equation}
Normalizing by the divergence of the maximal-segregation reference $M_{DSI}$ from $N$, we define
\begin{equation}
\begin{aligned}
    \alpha_{f,\sigma}(M')&\coloneqq\frac{\mathcal{A}_{f,\sigma}(M')}{\mathcal{D}_{f,\sigma}(M_{DSI}(M'))},\\
    \beta_{f,\sigma}(M')& \coloneqq\frac{\mathcal{B}_{f,\sigma}(M')}{\mathcal{D}_{f,\sigma}(M_{DSI}(M'))},\\
    \phi_{f,\sigma}(M')&\coloneqq\alpha_{f,\sigma}(M')-\beta_{f,\sigma}(M').
\end{aligned}
\label{eq:def_sigma_homophily}
\end{equation}
Thus,
\begin{equation}
    S_{f,\sigma}(M')=\alpha_{f,\sigma}(M')+\beta_{f,\sigma}(M')
\end{equation}
and
\begin{equation}
    -1\leq-S_{f,\sigma}(M')\leq\phi_{f,\sigma}(M')\leq S_{f,\sigma}(M')\leq 1.
    \label{eq:sigma_segregation_homophily_relation}
\end{equation}

\subsubsection{Network modularity}
Perhaps the most popular established measure of homophily in social networks is network modularity $Q$, or when normalized, the nominal assortativity coefficient $\rho_A$ \cite{newman2003mixing}. The signed homophily measures $\Phi$ and $\phi$ introduced above include network modularity and nominal assortativity as special cases. We first express these established measures in our model notation and then state the corresponding recovery results.

Let $e_{b_1b_2}$ denote the fraction of edges in an observed network connecting a node in group $b_1$ to a node in group $b_2$, and let $s_b\coloneqq\sum_{b'}e_{bb'}$. For our undirected networks, $e_{b_1b_2}=e_{b_2b_1}$. Network modularity and the nominal assortativity coefficient compare the observed intra-group edge-mass to the expected intra-group edge-mass under degree-constrained random mixing (the degree configuration model) and, for an undirected network, are given by
\begin{equation}
\begin{aligned}
    Q & = \sum_b \left(e_{bb} - s_b^2\right), \\
    \rho_A & = \frac{Q}{ 1 - \sum_b s_b^2} .
\end{aligned}
\label{eq:def_net_mod_ass}
\end{equation}

We rewrite this in our notation with $\hat M$ denoting the maximum-likelihood fit of a Stochastic Block Model \cite{holland1983stochastic} with predefined groups fitted to the observed number of edges between groups\footnote{\label{fn:mle_SBM} Suppose edges are independent Bernoulli random variables whose success probabilities depend only on the group memberships of their nodes (individuals). Let $\hat M$ denote the maximum-likelihood fit of a Stochastic Block Model with predefined groups fitted by maximum likelihood to the observed numbers of edges between groups; then its expected group-mixing fractions satisfy $\sigma^{\hat M}(b,b')=e_{bb'}$ and the implied fractions of edge endpoints belonging to each group are $s_b=\sum_{b'}e_{bb'}$. For $M' = \hat M$, our degree configuration model $M_{DC}(M')$ defined in \eqref{eq:def_M_DC} coincides with the null model used in the definition of network modularity \cite{chung2002average,newman2006modularity}, since $\sigma_{\BS,\BS}^{M_{DC}(\hat M)}(b,b')=s_bs_{b'}$.}
and $M_{DC}(\hat M)$ its degree configuration model as
\begin{align}
    Q & = \sum_b \left(\sigma_{\BS,\BS}^{\hat{M}}(b,b) - \sigma_{\BS,\BS}^{M_{DC}(\hat M)}(b,b)\right)
    \label{eq:network_mod_our_notation}
\end{align}
from which we arrive at our more general definition of network modularity for arbitrary models $M$
\begin{align}\label{eq:def_modularity_general}
    Q(M,M_{DC}(M);\BS,\BS) \coloneqq \sum_b \left(\sigma_{\BS,\BS}^M(b,b) - \sigma_{\BS,\BS}^{M_{DC}(M)}(b,b)\right) .
\end{align}
For geo-egocentric geosocial coordinate space, we define
\begin{equation}
\begin{aligned}
    Q(M) & \coloneqq Q(M,M_{DC}(M);\RS\times\BS,\BS) \\
    & \coloneqq \sum_r\sum_b \left(\sigma_{\RS\times\BS,\BS}^M(rb,b) - \sigma_{\RS\times\BS,\BS}^{M_{DC}(M)}(rb,b)\right)
\end{aligned}
\end{equation}
and note that $Q(M,M_{DC}(M);\RS\times\BS,\BS) = Q(M,M_{DC}(M);\BS,\BS)$ since 
$\sum_r \sigma_{\RS\times\BS,\BS}^M(rb,b)=\sigma_{\BS,\BS}^M(b,b).$
We also define the normalized version\footnote{As for $Q(M)$, we could again equivalently define \eqref{eq:def_general_rho_A} on social coordinate space, as both the denominator and the numerator are equal for geo-egocentric geosocial coordinate space and social coordinate space.}
\begin{equation}
\begin{aligned}
\rho_A(M) 
&\coloneqq \frac{Q(M)}{\mathcal{D}_{g,\sigma}(M_{DSI}(M),M_{DC}(M);\RS\times\BS,\BS)}.
\end{aligned}
\label{eq:def_general_rho_A}
\end{equation}

For Total Variation distance, i.e., the $f$-divergence with generator $g(x)=\frac{1}{2}|x-1|$, our signed unnormalized homophily measure $\Phi_{g, \sigma}(M)$ as defined in \eqref{eq:def_unnormalized_sigma_homophily} equals $Q(M)$ and the normalized signed measure $\phi_{g, \sigma}(M)$ equals $\rho_A(M).$ Furthermore, the normalization in \eqref{eq:def_general_rho_A} indeed corresponds to the normalization in \eqref{eq:def_net_mod_ass}, i.e.,  $\rho_A(\hat{M})$ equals the traditional nominal assortativity coefficient $\rho_A$, and taken together we have
\begin{equation}
\begin{aligned}
    \Phi_{g, \sigma}(\hat M) & = Q(\hat{M}) = Q , \\
    \phi_{g, \sigma}(\hat M) & = \rho_A(\hat{M}) = \rho_A .
\end{aligned}
\end{equation}

These relations are established in Lemma~\ref{lemma:homophily_quantification} and Corollary~\ref{corollary:normalized_homophily_quantification}. We formulate the Lemma for general null models $N$ satisfying the necessary support conditions, and denote the above defined homophily measures with null $N$ as $\mathcal{A}_{f,\rho}(M,N;\RS_1\times \BS,\BS),$ $\mathcal{B}_{f,\rho}(M,N;\RS_1\times \BS,\BS),$ $\Phi_{f,\rho}(M,N;\RS_1\times \BS,\BS)$ and $Q(M,N;\RS_1\times \BS,\BS).$ 
We recall that for any $f$-divergence, we choose the generator $f$ defined in \eqref{eq:sub_diff_centered_at_zero} and denote it by $f$.

\begin{lemma}[Properties of homophily quantification]\label{lemma:homophily_quantification}
We assume the null $N$ has the same marginal coordinate distribution as $M$ and the same marginal connectivity, i.e., $\pi^N=\pi^M$ and $\bar{P}^N = \bar{P}^M$.
\begin{enumerate}
    \item[(i)] For $g(x)=\frac{1}{2}|x-1|$,
    \begin{align}
        \Phi_{g,\sigma}(M,N;\RS\times\BS,\BS) & = Q(M,N;\RS\times\BS,\BS),\label{eq:alpha_m_beta_eq_Q_for_sigma}\\
        \Phi_{g,\sigma}(M,N;\RS\times\BS,\BS) & = \tfrac{1}{2\bar{P}}\Phi_{g,\tau}(M,N;\RS\times\BS,\BS).\label{eq:phi_tau_eq_2phi_sigma}
    \end{align}
    i.e., up to the multiplicative factor $2\bar P$, the approaches based on the joint coordinate--link distribution $\tau$ and the conditional coordinate distribution $\sigma$ are equivalent when using Total Variation distance.
    
    \item[(ii)] Let there be a sequence of models $M_i$, $N_i$, such that for all $i$ it holds that $\pi^{M_i}=\pi^{N_i}=\pi$, $\bar{P}^{M_i}=\bar{P}^{N_i}=\bar{P}_i$, and $\bar{P}_i\downarrow0$ as $i\to\infty$. Furthermore assume that
    \begin{align}
        \frac{\PS{rb_1\sim b_2}{}{M_i}}{\bar{P}^{M_i}} &= \frac{\PS{rb_1\sim b_2}{}{M_1}}{\bar{P}^{M_1}},\quad \forall rb_1,b_2,\;\forall i,\label{eq:const_ratio_hom_lemma}\\
        \frac{\PS{rb_1\sim b_2}{}{N_i}}{\bar{P}^{N_i}} &= \frac{\PS{rb_1\sim b_2}{}{N_1}}{\bar{P}^{N_1}},\quad \forall rb_1,b_2,\;\forall i. \label{eq:const_ratio_hom_lemma_null}
    \end{align}
    If $f$ is three-times continuously differentiable in a neighborhood of $1$, then
    \begin{align}
        \frac{\mathcal{A}_{f,\tau}(M_i,N_i;\RS\times\BS,\BS)}{\bar{P}_i}
        &=\mathcal{A}_{f,\sigma}(M_1,N_1;\RS\times\BS,\BS)
        +\mathcal O(\bar P_i), \quad i\to\infty,\label{eq:alpha1_to_alphaleq1}\\
        \frac{\mathcal{B}_{f,\tau}(M_i,N_i;\RS\times\BS,\BS)}{\bar{P}_i}
        &=\mathcal{B}_{f,\sigma}(M_1,N_1;\RS\times\BS,\BS)
        +\mathcal O(\bar P_i), \quad i\to\infty,\label{eq:beta1_to_betaleq1}
    \end{align}
and consequently also
\begin{equation}
        \frac{\Phi_{f,\tau}(M_i,N_i;\RS\times\BS,\BS)}{\bar{P}_i}
        =\Phi_{f,\sigma}(M_1,N_1;\RS\times\BS,\BS)+\mathcal O(\bar P_i), \quad i\to\infty.
\label{eq:phi1_to_phileq1}
\end{equation}
\end{enumerate}
\end{lemma}

\begin{corollary}\label{corollary:normalized_homophily_quantification}
\begin{enumerate}
    \item[(i)] 
Under the assumptions of Lemma~\ref{lemma:homophily_quantification}(i)
    \begin{align}
        \phi_{g,\sigma}(M) & = \rho_{A}(M),\label{eq:alpha_m_beta_eq_Q_for_sigma_normalized}\\
        \phi_{g,\sigma}(M) & = \phi_{g,\tau}(M).\label{eq:phi_tau_eq_2phi_sigma_normalized}
    \end{align}
In particular, with $\hat{M}$ the maximum-likelihood fit introduced in \eqref{eq:network_mod_our_notation}
\begin{equation}
    \phi_{g, \sigma}(\hat M) = \rho_A(\hat{M}) = \rho_A .
    \label{eq:recovery_of_network_ass}
\end{equation}
\item[(ii)]
Let $M_i$ be a sequence of models such that $\pi^{M_i}=\pi$ for all $i$, $\bar{P}^{M_i}\downarrow0$ as $i\to\infty$, and such that \eqref{eq:const_ratio_hom_lemma} holds. If $f$ is three-times continuously differentiable in a neighborhood of $1$, then
    \begin{align}
        \alpha_{f,\tau}(M_i)
        &=\alpha_{f,\sigma}(M_1)+\mathcal O(\bar P_i), \quad i\to\infty,\label{eq:alpha1_to_alphaleq1_normalized}\\
        \beta_{f,\tau}(M_i)
        &=\beta_{f,\sigma}(M_1) +\mathcal O(\bar P_i), \quad i\to\infty,\label{eq:beta1_to_betaleq1_normalized}
    \end{align}
and consequently also
\begin{equation}
        \phi_{f,\tau}(M_i)=\phi_{f,\sigma}(M_1) +\mathcal O(\bar P_i), \quad i\to\infty.
\label{eq:phi1_to_phileq1_normalized}
\end{equation}
\end{enumerate}
\end{corollary}

The following lemma examines how our homophily quantification depends on the choice of $f$-divergence. We recall our convention that, for any $f$-divergence, we choose as generator the representative $f$ satisfying \eqref{eq:sub_diff_centered_at_zero}.
\begin{lemma}[Properties of homophily quantification 2]\label{lemma:homophily_quantification_2} 
Let $q\in(0,1).$ For any strictly convex generator $f$
\begin{align}
    q f\left(\frac{q + 2\epsilon }{q}\right) & > 2q f\left(\frac{q +\epsilon }{q}\right),\quad q+2\epsilon\in[0,1],\;\epsilon\neq 0 .\label{eq:lemma_hom_q_ineq_0}
\end{align}
For any $f$ twice differentiable at $1$ with $f''(1)>0$ and sufficiently small \(\epsilon>0\),
\begin{align}
    q f\left(\frac{q + 2\epsilon }{q}\right) & > 2q f\left(\frac{q -\epsilon }{q}\right),\label{eq:lemma_hom_q_ineq_1}\\
    q f\left(\frac{q - 2\epsilon }{q}\right) & > 2q f\left(\frac{q +\epsilon }{q}\right).\label{eq:lemma_hom_q_ineq_2}
\end{align}

If $f''$ is strictly increasing on $(0,\infty),$ then, in addition to \eqref{eq:lemma_hom_q_ineq_1} and \eqref{eq:lemma_hom_q_ineq_2},
\begin{align}
    q f\left(\frac{q - \epsilon }{q}\right) & < q f\left(\frac{q +\epsilon }{q}\right),\quad 0<\epsilon<\min\{q,1-q\},\label{eq:lemma_hom_q_ineq_3}\\
\end{align}
conversely, if $f''$ is strictly decreasing on $(0,\infty),$ then in addition to \eqref{eq:lemma_hom_q_ineq_1} and \eqref{eq:lemma_hom_q_ineq_2},
\begin{align}
    q f\left(\frac{q - \epsilon }{q}\right) & > q f\left(\frac{q +\epsilon }{q}\right),\quad 0<\epsilon<\min\{q,1-q\}.\label{eq:lemma_hom_q_ineq_4}\\
\end{align}
Finally, if $f$ is strictly convex and symmetric around $1,$ then \eqref{eq:lemma_hom_q_ineq_1} and \eqref{eq:lemma_hom_q_ineq_2} hold for all $0<\epsilon<\min\{q/2, (1-q)/2\},$ and
\begin{align}
    q f\left(\frac{q - \epsilon }{q}\right) & = q f\left(\frac{q +\epsilon }{q}\right),\quad  0<\epsilon<\min\{q,1-q\}.\label{eq:lemma_hom_q_ineq_5}
\end{align}
\end{lemma}


\begin{remark}
    Lemma~\ref{lemma:homophily_quantification_2} sheds some light on how the homophily quantification depends on the properties of the generator $f.$ Our generalized version $Q(M)$ of network modularity is the sum of the signed deviations of the within-group edge masses under $M$ from the corresponding within-group edge masses under the null $M_{DC}(M),$ i.e., it does not matter how the total surplus or lack of within group edge mass is distributed over different groups, and how this surplus or lack corresponds to a lack or surplus in edge mass for various between group pairs. This corresponds to the following observation: assume $\sigma^{M_{DC}(M)}(r_1b_1,b_2)\equiv q$ is constant, and consider deviations of $\sigma^M(r_1 b_1, b_2)$ from $q,$ then the Total Variation contribution of these deviations does not depend on their sign, and several small deviations have the same contribution as one large deviation which equals the sum of the small deviations. That is, for TV distance there is equality in \eqref{eq:lemma_hom_q_ineq_0}, \eqref{eq:lemma_hom_q_ineq_1}, \eqref{eq:lemma_hom_q_ineq_2}, and \eqref{eq:lemma_hom_q_ineq_5}. For general $f,$ this is not the case, and negative and positive deviations can have different contributions to the $f$-divergence, and one large deviation will in general not have the same contribution as several small deviations whose sum equals the large deviation. Note that for the generators $f$ of the KL-divergence, the Jeffreys divergence, the Le Cam divergence, and the squared Hellinger distance, it holds that $f''(1)>0$ and $f''$ is strictly decreasing in $(0,\infty),$ hence for these generators large deviations from the degree configuration null contribute more than several small deviations, and a lack of edge mass relative to this null contributes more than an equally sized surplus. The implications for the homophily quantification $\phi_{f,\sigma}$ are illustrated in the following example.
\end{remark}

\begin{example}\label{example:homophily quantification} We illustrate the implications of Lemma~\ref{lemma:homophily_quantification_2} for homophily quantification with two simple examples, and for ease of exposition consider social coordinate space $\BS\times\BS$ (marginalized over region space $\RS$) and we suppress the model and coordinate-space arguments of the homophily measures, e.g., we write $\mathcal{A}_{f,\sigma}$ instead of $\mathcal{A}_{f,\sigma}(M,M_{DC}(M);\BS,\BS)$.
    \begin{enumerate}
        \item[(i)]    Let \begin{align}
    \bigg(\sigma^M(b_1,b_2)\bigg)_{b_1,b_2\in\BS}
        =\frac{1}{9} + \begin{pmatrix}
        -\epsilon & \epsilon & 0 \\
         \epsilon &    0 & -\epsilon  \\
         0 &     -\epsilon & \epsilon  
        \end{pmatrix}, \quad 0<\epsilon\leq\frac{1}{9} \label{eq:example_hom_pert_matrix9}
    \end{align}
    and $\sigma^{M_{DC}(M)}(b_1,b_2)\equiv q\coloneqq \frac{1}{9}.$ This implies, by allocating the surplus of edge masses on the diagonal and the lack on the off-diagonal to $\mathcal{A}$, and the lack on the diagonal and the surplus on the off-diagonal to $\mathcal{B},$ that
    \begin{align}
        \mathcal{A}_{f,\sigma} &= q f\left(\frac{q + \epsilon }{q}\right) + q f\left(\frac{q - \epsilon }{q}\right) + q f\left(\frac{q - \epsilon }{q}\right),\\
        \mathcal{B}_{f,\sigma} &= q f\left(\frac{q - \epsilon }{q}\right) + q f\left(\frac{q + \epsilon }{q}\right) + q f\left(\frac{q + \epsilon }{q}\right) .
    \end{align}
    
    While $Q(M) = 0,$ for any generator $f$ with strictly decreasing $f''$ by \eqref{eq:lemma_hom_q_ineq_4} $\Phi_{f,\sigma}=\mathcal{A}_{f,\sigma}-\mathcal{B}_{f,\sigma} > 0.$ Conversely, for any generator $f$ with strictly increasing $f'',$ by \eqref{eq:lemma_hom_q_ineq_3} we have $\Phi_{f,\sigma}<0.$

    The situation is depicted in panel~\textbf{a} of Fig.~\ref{fig:hom_quant_demo}, where the perturbation pattern of \eqref{eq:example_hom_pert_matrix9} is applied at the level of connection probabilities $P^M$ with perturbation size $\delta = \min\{\bar{P}, 1-\bar{P}\}$ and $P^{M_{DC}(M)}\equiv\bar{P}, \pi = (1/3,1/3,1/3)^T.$ The induced edge distribution is $\sigma^M = P^M / (9\bar{P}),$ so $\sigma^M$ is as in \eqref{eq:example_hom_pert_matrix9} with $\epsilon = \delta / (9\bar{P}).$
    
        \item[(ii)] Let \begin{align}
    \bigg(\sigma^M(b_1,b_2)\bigg)_{b_1,b_2\in\BS}
        =\frac{1}{25} + \begin{pmatrix}
        -\epsilon & \epsilon & 0 & 0 & 0\\
         \epsilon & -\epsilon & 0 & 0 & 0 \\
         0 & 0 & 0 & \epsilon & -\epsilon \\
         0 & 0 & \epsilon & 0 & - \epsilon\\
         0 & 0 & -\epsilon & -\epsilon & 2\epsilon 
        \end{pmatrix},\quad 0<\epsilon \leq\frac{1}{25} \label{eq:example_hom_pert_matrix25}
    \end{align}
    and $\sigma^{M_{DC}(M)}(b_1,b_2)\equiv q\coloneqq \frac{1}{25}.$ While again $Q(M) = 0,$ for general generator $f$
    \begin{align}
        \mathcal{A}_{f,\sigma} &= q f\left(\frac{q + 2\epsilon }{q}\right) + 4 q f\left(\frac{q - \epsilon }{q}\right), \\
        \mathcal{B}_{f,\sigma} &= 2q f\left(\frac{q - \epsilon }{q}\right) + 4 q f\left(\frac{q + \epsilon }{q}\right), \\
        \Phi_{f,\sigma} & =  \underbrace{q f\left(\frac{q + 2\epsilon }{q}\right) - 2q f\left(\frac{q + \epsilon }{q}\right)}_{X} 
        + \underbrace{2 q f\left(\frac{q - \epsilon }{q}\right) - 2q f\left(\frac{q + \epsilon }{q}\right)}_{Y}.\label{eq:example_hom_AB_eq}
    \end{align}
    If $\epsilon>0$ sufficiently small and $f''$ is strictly decreasing (e.g., $f(x) =x\log x -(x-1)$ for the KL-divergence), then in \eqref{eq:example_hom_AB_eq} $X>0$ by \eqref{eq:lemma_hom_q_ineq_0} and the implied strict convexity of $f,$ and $Y>0$ by \eqref{eq:lemma_hom_q_ineq_4}, i.e., $\Phi_{f,\sigma}>0.$ If $f$ is strictly convex and symmetric around 1 (e.g., $f(x) = 0.5(x-1)^2$ for the $\chi^2$ divergence), then by \eqref{eq:lemma_hom_q_ineq_5} $Y=0$ and $\Phi_{f,\sigma}>0$ for all $0<\epsilon<\frac{1}{25}.$

    The situation is depicted in panel~\textbf{b} of Fig.~\ref{fig:hom_quant_demo}, where the perturbation pattern of \eqref{eq:example_hom_pert_matrix25} is applied at the level of connection probabilities $P^M$ with perturbation size $\delta = \min\{\bar{P}, (1-\bar{P})/2\}$ and $P^{M_{DC}(M)}\equiv\bar{P}, \pi_b = 1/5,b=1,\dots,5.$ The induced edge distribution is $\sigma^M = P^M / (25\bar{P}),$ so $\sigma^M$ is as in \eqref{eq:example_hom_pert_matrix25} with $\epsilon = \delta / (25\bar{P}).$
    \end{enumerate}
\end{example}

\begin{remark}
For any models $M_1,M_2,M_3$, $Q(M_3,M_1;\RS\times\BS,\BS)$ admits the model decomposition
\begin{equation}
    Q(M_3,M_1;\RS\times\BS,\BS)
    =Q(M_3,M_2;\RS\times\BS,\BS)+Q(M_2,M_1;\RS\times\BS,\BS),
    \label{eq:Q_model_decomposition}
\end{equation}
or more generally, for models $M_1,\dots,M_n$,
\begin{equation}
    Q(M_n,M_1;\RS\times\BS,\BS)
    =\sum_{i=1}^{n-1}Q(M_{i+1},M_i;\RS\times\BS,\BS).
\end{equation}
Indeed,
\begin{equation}
\begin{aligned}
    Q(M_3,M_1;\RS\times\BS,\BS)
    &=\sum_{r,b}\Big(
    \big(\sigma^{M_3}(rb,b)-\sigma^{M_2}(rb,b)\big)
    +\big(\sigma^{M_2}(rb,b)-\sigma^{M_1}(rb,b)\big)\Big)\\
    &=Q(M_3,M_2;\RS\times\BS,\BS)+Q(M_2,M_1;\RS\times\BS,\BS).
\end{aligned}
\end{equation}
In particular, this opens up the possibility of additively decomposing total segregation into social and geographical segregation by choosing suitable intermediate models. Such a decomposition would, however, require a different approach to preserving expected degrees than the one used in our present construction, and we leave this for future work.
\end{remark}

We end our discussion of homophily by building a bridge to traditional geographical segregation. The Relative Diversity index $R$ \cite{reardon2002measures}, together with its normalized version $\underline R$, is recovered by our signed homophily measure $\Phi_{f,\sigma}$ when the $f$-divergence is the Total Variation distance and the network model is the regional isolation model $M_{RI}$. Put another way, Relative Diversity equals ``Network Modularity'' when the edge distribution is obtained from $M_{RI}.$ This is the content of the following Lemma.

\begin{lemma}[Recovery of Relative Diversity]\label{lemma:recovery_of_relative_diversity}
Let $g(x)=\frac12|x-1|$ and let $M_{RI}$ be the regional isolation model defined in \eqref{eq:def_M_RI}. Then
\begin{equation}
    \Phi_{g,\sigma}(M_{RI})=Q(M_{RI})
    =\sum_r\sum_b\pi_r\left(\pi_{b\mid r}-\pi_b\right)^2
    \eqqcolon R,
    \label{eq:Q_contains_R_as_special_case}
\end{equation}
and
\begin{equation}
    \phi_{g,\sigma}(M_{RI})=\rho_A(M_{RI})
    = \frac{R}{\sum_b\pi_b(1-\pi_b)}
    \eqqcolon\underline R,
    \label{eq:rho_A_contains_R_as_special_case}
\end{equation}
where $R$ is the Relative Diversity index and \eqref{eq:rho_A_contains_R_as_special_case} its standard normalized formulation as given in \cite{reardon2002measures}. 
\end{lemma}

\begin{remark}
Lemma~\ref{lemma:recovery_of_relative_diversity} implies that traditional geographical segregation is always net homophilous when homophily is quantified by network modularity, since by \eqref{eq:Q_contains_R_as_special_case}
\begin{equation}
    Q(M_{RI})
    =\sum_r\sum_b\pi_r\left(\pi_{b\mid r}-\pi_b\right)^2
    \geq 0,
\end{equation}
with equality if and only if regional sociodemographic compositions equal the marginal sociodemographic composition. Thus, traditional geographical segregation of sociodemographic groups can generate no net heterophily relative to degree-constrained random mixing under this measure. In our notation, with $g(x)=\tfrac{1}{2}|x-1|$ the generator for the Total Variation distance,
\begin{equation}
\begin{aligned}
        \underline R &= \alpha_{g,\sigma}(M_{RI}) - \beta_{g,\sigma}(M_{RI}),\\
        \underline D &= \alpha_{g,\sigma}(M_{RI}) + \beta_{g,\sigma}(M_{RI}),
\end{aligned}
\end{equation}
where $\underline D$ is the (normalized) Dissimilarity index. Therefore, the Dissimilarity index counts both homophilous and heterophilous deviations from random mixing towards segregation, whereas Relative Diversity subtracts the heterophilous from the homophilous contributions.
This suggests that, for general generators $f$, $\alpha_{f,\sigma}$ could itself constitute a meaningful measure of homophilous segregation. Rather than adding heterophilous contributions to, or subtracting them from, the homophilous contributions, it quantifies the magnitude of homophilous departures from random mixing directly.
\end{remark}

\subsection{Regional decomposability}\label{sec:reg_decomp}
Consider a set of (super) regions $\mathcal{S}$ such that $\bigcup_{s\in \mathcal{S}} s = \mathcal{R}$ is a partition of the set of regions $\mathcal{R} = \{r\colon r\in \mathcal{R}\},$ i.e., each $s$ is a set containing regions $\{r\colon r\in s\}$ and $\forall s\neq s^\prime\colon s\cap s^\prime = \emptyset$. 
The well known regional decomposition \cite{reardon2002measures} of the traditional unnormalized Theil Information Theory index $H$ into a between component and a weighted sum (with weights $\pi_s$) over within components for regions $s$ can be written as 
\begin{align}\label{eq:standard_regional_decomposition_of_hatH}
    H & = \sum_s \pi_s\sum_b \pi_b \frac{\pi_{b\mid s}}{\pi_b} \log\left(\frac{\pi_{b\mid s}}{\pi_b} \right)
    + \sum_{s}\pi_s \left(
\sum_{r\in s}\pi_{r\mid s}\sum_b\pi_{b\mid s}\frac{\pi_{b\mid r}}{\pi_{b\mid s}} \log\left(\frac{\pi_{b\mid r}}{\pi_{b\mid s}} \right)
    \right).
\end{align}
For ease of notation, throughout the following we write $N\coloneqq M_{DC}(M)$ for the degree configuration null model, $\pi\coloneqq\pi^M=\pi^N$, and $\bar P\coloneqq\bar P^M=\bar P^N$. Let $h(x) = x \log x$ be the generator for our KL-divergence based segregation measure $\mathcal D_{h,\sigma}(M).$
We now establish a regional decomposition of $\mathcal D_{h,\sigma}(M),$ whose between component compares the connectivity of random individuals in regions $s$ (irrespective of their finer region $r$ and their sociodemographic membership) to sociodemographic groups in the country, under a model $M$ relative to the null $N.$ The within component of region $s$ looks at the additional impact of the finer region $r$ and the sociodemographic membership on the connectivity to sociodemographic groups in the country, under model $M$ relative to $N.$ 
For $M=M_{RI}$, for which $N=M_{DC}(M_{RI})=M_{ER}$ (the Erd\H{o}s--R\'enyi random graph model), this approach recovers the traditional decomposition \eqref{eq:standard_regional_decomposition_of_hatH}.
\begin{lemma}\label{lemma:regional_decomposition_of_H}
We have the following decomposition of $\mathcal D_{h,\sigma}(M)$ into a between component $\mathcal D_{h,\sigma}^{\mathcal{S}\leftrightarrow}(M)$ measuring segregation between regions $s$ and within components $\mathcal D_{h,\sigma}^{(s)}(M)$ measuring segregation within regions $s$
\begin{align}
    & \mathcal D_{h,\sigma}(M) = \mathcal D_{h,\sigma}^{\mathcal{S}\leftrightarrow}(M) + \sum_s \pi_s\frac{\bar{P}_s^M}{\bar{P}}\mathcal D_{h,\sigma}^{(s)}(M)\label{eq:regional_decomposability_sb_rbb}
\end{align}
with
\begin{align}
    \mathcal D_{h,\sigma}^{\mathcal{S}\leftrightarrow}(M)&\coloneqq \DKL{\sigma^M(s,b_2)}{\sigma^N(s,b_2)}, \label{eq:regional_decomposability_sb_rbb_asKL_div_between}\\
    \mathcal D_{h,\sigma}^{(s)}(M)&\coloneqq \DKL{\sigma^M(rb_1, b_2\mid s)}{\sigma^{N\mid M}( rb_1, b_2\mid s)}\label{eq:regional_decomposability_sb_rbb_asKL_div_within},
\end{align}
where
\begin{align}
    \sigma^M( rb_1, b_2\mid s) & = \pi_{b_2}\frac{\pi_{r\mid s} \pi_{b_1\mid r} \PS{rb_1\sim b_2}{}{M}}{\bar{P}_s^M},\label{eq:def_sgima_M_within_s}\\
    \sigma^{N\mid M}( rb_1, b_2\mid s) &\coloneqq \pi_{b_2}\frac{\PS{s\sim b_2}{}{M}}{\bar{P}_s^M}\frac{\pi_{r\mid s} \pi_{b_1\mid r} \PS{rb_1\sim b_2}{}{N}}{\PS{s\sim b_2}{}{N}}, \label{eq:def_sgima_N_given_M_within_s}
\end{align}
with
\begin{equation}
        \bar{P}_s^M \coloneqq \sum_{r\in s}\pi_{r\mid s}\sum_{b_1,b_2}\pi_{b_1\mid r}\pi_{b_2}\PS{rb_1\sim b_2}{}{M}
\end{equation}
and for $M^\prime = M,N$
\begin{align}
    \PS{s\sim b_2}{}{M^\prime} & = \sum_{r\in s}\sum_{b_1}\pi_{r\mid s} \pi_{b_1\mid r}\PS{rb_1\sim b_2}{}{M^\prime},\\
    \sigma^{M^\prime}(s,b_2) & = \frac{\pi_s \pi_{b_2} \PS{s\sim b_2}{}{M^\prime}}{\bar{P}}.
\end{align}

Here $\sigma^{N\mid M}$ is the distribution obtained by keeping the $(s,b_2)$-marginal structure of $M,$ while using the conditional distribution of $rb_1$ given $(s,b_2)$ induced by $N.$

In particular, for $M=M_{RI}$ the regional isolation model we obtain
\begin{align}
    \mathcal D_{h,\sigma}^{\mathcal{S}\leftrightarrow}(M_{RI})
    & = \sum_s \pi_s\sum_b \pi_b \frac{\pi_{b\mid s}}{\pi_b} \log\left(\frac{\pi_{b\mid s}}{\pi_b} \right),\label{eq:between_component_equals_traditional_between_component}\\
    \mathcal D_{h,\sigma}^{(s)}(M_{RI})
    & = \sum_{r\in s}\pi_{r\mid s}\sum_b\pi_{b\mid s}\frac{\pi_{b\mid r}}{\pi_{b\mid s}} \log\left(\frac{\pi_{b\mid r}}{\pi_{b\mid s}} \right),\label{eq:within_component_equals_traditional_within_component}\\
    \pi_s\frac{\bar{P}_s^{M_{RI}}}{\bar{P}} & = \pi_s\label{eq:barPs_equals_barP_under_MRI},
\end{align}
which equals exactly the regional decomposition \eqref{eq:standard_regional_decomposition_of_hatH} of the traditional segregation index $H.$
\end{lemma}

It is worth mentioning that the decomposition \eqref{eq:regional_decomposability_sb_rbb} is still meaningful if $\mathcal{S}  = \mathcal{R},$ in which case the between component considers the impact of membership in region $r$ on connectivity, and the within component the additional impact of group membership $b_1$ given region $r$ on connectivity.

We note that alternative regional decompositions of $\mathcal D_{h,\sigma}(M)$ into a between and region $s$ specific within components are possible. At a first glance, one might naively assume that the desired decomposition should consider the variation of $\PS{sb_1\sim b_2}{}{M}$ around $\bar{P}$, relative to $\PS{sb_1\sim b_2}{}{N}$ and $\bar{P}$, for the between component, and the within component for region $s$ should consider the variation of $\PS{rb_1\sim b_2}{}{M}$ around $\bar{P}_s^M$, relative to $\PS{rb_1\sim b_2}{}{N}$ and $\bar{P}_s^M$ (note that $\bar{P}_s^{N} = \bar{P}_s^M$). However, it is obvious that in addition to the between and the within components, such an approach would have to include a third non-positive term. For example, it is easy to construct examples where under such an approach the between component equals the weighted sum over the within components, and such that the ``third term'' again has the same absolute value but with a negative sign. We would consider such a decomposition meaningless.

However, the following alternative approach could be considered: the between component considers $\PS{sb_1\sim b_2}{}{M}$ relative to $\PS{sb_1\sim b_2}{}{N}$, while the within component for region $s$ only looks at the variation of $\PS{rb_1\sim b_2}{}{M}$ around $\PS{sb_1\sim b_2}{}{M}$, relative to $\PS{rb_1\sim b_2}{}{N}$ and $\PS{sb_1\sim b_2}{}{N}$. In this case, the majority of the information is contained in the between component. While we find this approach less appealing, as it puts most of the information about connectivity in the between component, as opposed to the approach outlined in Lemma~\ref{lemma:regional_decomposition_of_H}, which distributes information more equally, the ultimate choice of a decomposition should depend on the research question at hand.

\subsection{Group decomposability}\label{sec:group_decomp}
Consider a set $\mathcal{A}$ such that $\bigcup_{a\in \mathcal{A}} a = \mathcal{B}$ is a partition of sociodemographic space $\mathcal{B}$. The traditional Theil index $H$ admits the group decomposition into a between component and a weighted sum (with weights $\pi_a$) over within components for groups $a$
\begin{align}
    H & = \sum_a \sum_r \pi_{r} \pi_a \frac{\pi_{a\mid r}}{\pi_a} \log \frac{\pi_{a\mid r}}{\pi_a}
    + \sum_a \pi_a \left(\sum_r \pi_{r\mid a} \sum_{b\in a}\pi_{b\mid a} \frac{\pi_{b\mid r,a}}{\pi_{b\mid a}}
    \log \frac{\pi_{b\mid r,a}}{\pi_{b\mid a}}\right).\label{eq:standard_group_decomposition_of_hatH}
\end{align}
We establish a group decomposition of $\mathcal D_{h,\sigma}(M),$ whose between component is given by simply considering the coarser coordinates $a$ instead of $b,$ 
and the within component for group $a_2$ looks at the additional impact of the finer sociodemographic coordinates $b_1,b_2$, conditional on the second coordinate belonging to coarse group $a_2$.
For $M=M_{RI}$ the regional isolation model, for which $N=M_{DC}(M_{RI})=M_{ER}$ (the Erd\H{o}s--R\'enyi random graph model), this approach recovers the traditional decomposition \eqref{eq:standard_group_decomposition_of_hatH}.

\begin{lemma}\label{lemma:group_decomposition_of_H}
We have the following decomposition of $\mathcal D_{h,\sigma}(M)$ into a between component $\mathcal D_{h,\sigma}^{\mathcal{A} \leftrightarrow}(M)$ measuring segregation between sociodemographic groups $a$ and within components $\mathcal D_{h,\sigma}^{(a)}(M)$ measuring segregation within sociodemographic group $a$
\begin{align}
    & \mathcal D_{h,\sigma}(M) = \mathcal D_{h,\sigma}^{\mathcal{A} \leftrightarrow}(M) + \sum_a \pi_a\frac{\bar{P}_a^M}{\bar{P}}\mathcal D_{h,\sigma}^{(a)}(M)\label{eq:blau_group_decomposability_raa_rbb}
\end{align}
and pair $(a_1,a_2)$ specific decomposition
\begin{align}
    & \mathcal D_{h,\sigma}(M) = \mathcal D_{h,\sigma}^{\mathcal{A} \leftrightarrow}(M) + \sum_{a_1}\sum_{a_2} \pi_{a_1}\pi_{a_2}\frac{\PS{a_1\sim a_2}{}{M}}{\bar{P}}\mathcal D_{h,\sigma}^{(a_1,a_2)}(M)\label{eq:blau_pair_decomposability_raa_rbb}
\end{align}
with
\begin{align}
    \mathcal D_{h,\sigma}^{\mathcal{A} \leftrightarrow}(M)&\coloneqq \DKL{\sigma^M(ra_1,a_2)}{\sigma^N( ra_1, a_2)}, \label{eq:group_decomposability_raa_rbb_asKL_div_between}\\
    \mathcal D_{h,\sigma}^{(a_2)}(M)&\coloneqq \DKL{\sigma^M(rb_1,b_2\mid a_2)}{\sigma^{N\mid M}( rb_1, b_2\mid a_2)},\label{eq:group_decomposability_raa_rbb_asKL_div_within_a_2}\\
    \mathcal D_{h,\sigma}^{(a_1,a_2)}(M)&\coloneqq \DKL{\sigma^M( rb_1, b_2\mid a_1, a_2)}{\sigma^{N\mid M}(rb_1,b_2\mid a_1,a_2)}\label{eq:group_decomposability_raa_rbb_asKL_div_within_a_1a_2},
\end{align}
where for $M^\prime = M,N$
\begin{align}
    \sigma^{M^\prime}( ra_1, a_2) = \pi_{a_1}\pi_{r\mid a_1} \pi_{a_2}\frac{\PS{ra_1\sim a_2}{}{M^\prime}}{\bar{P}}\label{eq:sigma_between_groups} 
\end{align}
and
\begin{align}
    \sigma^M( rb_1, b_2\mid a_1,a_2) &= \pi_{r\mid a_1} \pi_{b_1\mid r,a_1}\pi_{b_2\mid a_2}\frac{\PS{rb_1\sim b_2}{}{M}}{\PS{a_1\sim a_2}{}{M}},\\
    \sigma^{N\mid M}( rb_1, b_2\mid a_1,a_2) & \coloneqq \pi_{r\mid a_1} \frac{\PS{ra_1\sim a_2}{}{M}}{\PS{a_1\sim a_2}{}{M}} \nonumber\\
    &\hspace{3em}\cdot\pi_{b_1\mid r,a_1}\pi_{b_2\mid a_2}\frac{\PS{rb_1\sim b_2}{}{N}}{\PS{ra_1\sim a_2}{}{N}}
\end{align}
\begin{align}
    \sigma^M(rb_1,b_2\mid a_2) &= \pi_{a_1}\pi_{r\mid a_1} \pi_{b_1\mid r,a_1}\pi_{b_2\mid a_2}\frac{\PS{rb_1\sim b_2}{}{M}}{\bar{P}_{a_2}^M},\label{eq:sigma_M_group_within}\\
    \sigma^{N\mid M}( rb_1, b_2\mid a_2) & \coloneqq \pi_{a_1}\pi_{r\mid a_1} \frac{\PS{ra_1\sim a_2}{}{M}}{\bar{P}_{a_2}^M} \nonumber\\
    &\hspace{3em}\cdot\pi_{b_1\mid r,a_1}\pi_{b_2\mid a_2}\frac{\PS{rb_1\sim b_2}{}{N}}{\PS{ra_1\sim a_2}{}{N}}\label{eq:sigma_N_given_M_group_within}
\end{align}
\begin{align}
    \bar{P}_{a}^M \coloneqq \sum_{b}\sum_r \PS{rb\sim a}{}{M}\pi_r\pi_{b\mid r} .
\end{align}
In particular, with $M=M_{RI}$ the regional isolation model, we have
\begin{align}
    \mathcal D_{h,\sigma}^{\mathcal{A} \leftrightarrow}(M_{RI})& = \sum_a \sum_r \pi_{r} \pi_a \frac{\pi_{a\mid r}}{\pi_a} \log \frac{\pi_{a\mid r}}{\pi_a},\label{eq:between_group_component_of_triaditional_H_recovered_for_MRI_MER}\\
    \mathcal D_{h,\sigma}^{(a)}(M_{RI})&= \sum_r \pi_{r\mid a} \sum_{b\in a}\pi_{b\mid a} \frac{\pi_{b\mid r,a}}{\pi_{b\mid a}}
    \log \frac{\pi_{b\mid r,a}}{\pi_{b\mid a}},\label{eq:within_group_component_of_triaditional_H_recovered_for_MRI_MER}\\
    \pi_a\frac{\bar{P}_a^{M_{RI}}}{\bar{P}}& = \pi_a,\label{eq:within_group_component_weights_of_triaditional_H_recovered_for_MRI_MER}
\end{align}
which equals exactly the group decomposition \eqref{eq:standard_group_decomposition_of_hatH} of the traditional segregation index $H.$
\end{lemma}

\subsection{Normalization of regional and group decompositions}\label{sec:normalization_decomp}
For the group decomposition of the traditional Theil index, the between (coarse) groups normalization constant is given by the entropy
\begin{align}
    \sum_a \pi_a \log \frac{1}{\pi_a}\label{eq:between_a_normalization_trad_theil},
\end{align}
and the within group $a$ normalization constants by
\begin{align}
    \sum_{b\in a} \pi_{b\mid a}\log \frac{1}{\pi_{b\mid a}}\label{eq:within_a_normalization_trad_theil} .
\end{align}
These normalizations arise naturally by decomposing the total entropy $H$: e.g., for coarse groups $a$ and finer groups $b$,
\begin{equation}
\begin{aligned}
    H(B) & = -\sum_b \pi_b \log \pi_b = -\sum_a \sum_{b\in a}\pi_a \pi_{b\mid a} \log\left(\pi_a \pi_{b\mid a}\right) 
    = - \sum_a \pi_a \log \pi_a -\sum_a\pi_a \sum_{b\in a} \pi_{b\mid a} \log \pi_{b\mid a}\\
    & = H(A) +\sum_a \pi_a H(B\mid A = a).
\end{aligned}
\end{equation}
Thus the normalization constant, i.e., the entropy $H(B),$ itself decomposes by the same chain-rule approach as the KL divergence of $\pi_{rb}$ from $\pi_r \otimes \pi_b,$ i.e., the unnormalized traditional Theil index. This motivates defining the within and between normalizations for our group decomposition by applying the KL chain rule directly to our global normalization constant $\DKL{\sigma^{M_{DSI}(M)}}{\sigma^{M_{DC}(M)}}$, see Lemma~\ref{lemma:normalization_of_indices}. 

For the regional decomposition of the traditional Theil index, the between-super-regions normalization constant is given by the entropy of $\pi_b$
\begin{align}
    -\sum_b \pi_b \log \pi_b \label{eq:between_s_normalization_trad_theil}
\end{align}
and the within region $s$  normalization constants by the entropies
\begin{align}
    -\sum_b \pi_{b\mid s} \log \pi_{b\mid s} = H(B\mid S=s). \label{eq:within_s_normalization_trad_theil} 
\end{align}

Now 
\begin{align}
H(B) = I(S;B) + H(B|S) = I(S;B) + \EE_{\pi_s}[H(B|S = s)],\label{eq:regional_decomp_normalization_const_trad}
\end{align}
with $H(B|S = s)$ the normalization constant for the within $s$ component, and $I(S;B)$ the mutual information of $(S,B)$ under $\pi_{sb}.$

This suggests the following approach for our regional decomposition as given in Lemma~\ref{lemma:regional_decomposition_of_H}: We write shorthand $M_{DC} = M_{DC}(M)$ and $M_{DSI}=M_{DSI}(M)$ and note that since $M_{DSI}$ preserves the geosocial marginal connectivities of $M$, we have $M_{DC}(M_{DSI})=M_{DC}$. We further observe that $\sigma^{M_{DC}}(s, b_2) = \sigma^{M_{DC}}(s)\sigma^{M_{DC}}(b_2) = \sigma^{M}(s)\sigma^{M}(b_2).$ If we now let random variables $S$ and $B$ have joint distribution $\sigma^{M_{DSI}}(s, b_2)$ with marginals $\sigma^{M_{DSI}}(s) = \sigma^{M}(s)$ and $\sigma^{M_{DSI}}(b_2)=\sigma^{M}(b_2),$ we can rewrite the global normalization constant $\DKL{\sigma^{M_{DSI}}}{\sigma^{M_{DC}}}$ by applying the regional decomposition of Lemma~\ref{lemma:regional_decomposition_of_H}
\begin{equation}
\begin{aligned}\label{eq:regional_decomp_normalization_const}
    &\DKL{\sigma^{M_{DSI}}(r_1b_1, b_2)}{\sigma^{M_{DC}}(r_1b_1, b_2)} \\
    &=  \DKL{\sigma^{M_{DSI}}(s,b_2)}{\sigma^{M_{DC}}(s,b_2)} \\
    & + \sum_{s}\sigma^{M_{DSI}}(s)
    \DKL{\sigma^{M_{DSI}}(rb_1, b_2\mid s)}{\sigma^{M_{DC}\mid M_{DSI}}(rb_1, b_2\mid s)}\\
    & = I(S;B) + \EE_{\sigma^{M_{DSI}}(s)} \DKL{\sigma^{M_{DSI}}(rb_1, b_2\mid s)}{\sigma^{M_{DC}\mid M_{DSI}}(rb_1, b_2\mid s)}
\end{aligned}
\end{equation}
where $S$ and $B$ have joint distribution $\sigma^{M_{DSI}}(s,b_2).$ The observation that \eqref{eq:regional_decomp_normalization_const} and \eqref{eq:regional_decomp_normalization_const_trad} have the same structure as a sum of $I(S;B)$ (with $(S,B)$ having distribution $\sigma^{M_{DSI}}(s,b)$ and $\pi_{sb}$ respectively) and an expectation over ``within" components motivates our definition of the within normalization constants for our regional decomposition given by 
\begin{align}
    \DKL{\sigma^{M_{DSI}}(rb_1, b_2\mid s)}{\sigma^{M_{DC}\mid M_{DSI}}(rb_1, b_2\mid s)} ,\quad s\in\mathcal{S} . \label{eq:def_within_normalization_reg_decomp}
\end{align}
However, $I(S;B)$ with joint distribution $\pi_{sb}$ is not a useful normalization for the traditional Theil between component (as it equals the between component and hence always gives a normalized between component of $1$), nor a useful normalization for the between component $\DKL{\sigma^{M}(s,b_2)}{\sigma^{M_{DC}}(s,b_2)}$ of our regional decomposition (as e.g. for $M=M_{RI}$ the regional isolation model, we again obtain a normalized between component of $1$). However, since for any model $M'$ in $\mathcal{M}_{sb}\coloneqq\{M'\colon \sigma^{M'}(s) = \sigma^{M_{DC}}(s),\sigma^{M'}(b_2) = \sigma^{M_{DC}}(b_2)\}$
\begin{align}
    I_{\sigma^{M'}}(S;B) = H_{\sigma^{M'}}(B) - H_{\sigma^{M'}}(B \mid S) \leq H_{\sigma^{M'}}(B) = H_{\sigma^{M_{DC}}}(B)
\end{align}
we have the upper bound
\begin{align}
    \max_{M^\prime \in \mathcal{M}_{sb}} I_{\sigma^{M'}}(S;B) \leq H_{\sigma^{M_{DC}}}(B) = -\sum_{b}\pi_b \frac{\bar{P}_b^M}{\bar{P}}\log\left(\pi_b \frac{\bar{P}_b^M}{\bar{P}}\right), \label{eq:def_between_normalization_reg_decomp}
\end{align}
which we take as the normalization constant for the between component of our regional decomposition.

\begin{lemma}\label{lemma:normalization_decompositions}
    For the group decomposition of $\DKL{\sigma^M}{\sigma^{M_{DC}}}$ given in Lemma~\ref{lemma:group_decomposition_of_H}, define the within and between normalization constants by applying the same group decomposition to $\DKL{\sigma^{M_{DSI}}}{\sigma^{M_{DC}}},$ i.e., the between normalization constant by
    \begin{align}\label{eq:between_group_normalization_const}
        \DKL{\sigma^{M_{DSI}}(r_1 a_1, a_2)}{\sigma^{M_{DC}}(r_1 a_1, a_2)}
    \end{align}
and the within $a_2$ normalization constant by
\begin{align}\label{eq:within_group_normalization_const}
    \DKL{\sigma^{M_{DSI}}(r_1 b_1, b_2\mid a_2)}{\sigma^{M_{DC}\mid M_{DSI}}(r_1 b_1, b_2\mid a_2)}.
\end{align}
The following properties hold:
\begin{enumerate}

    \item[(i)] 
    The between normalization constant \eqref{eq:between_group_normalization_const} is the KL-divergence of the edge distribution over $(r_1 a_1, a_2)$ induced by the connectivity model
\begin{align}
    \PS{r_1 a_1 \sim a_2}{}{} = \mathbb{1}_{ \{a_1=a_2\} }\frac{1}{\pi_{a_2}} \bar{P}_{r_1 a_1}^M\label{eq:between_group_si_model}
\end{align}
from the edge distribution induced by the connectivity model
\begin{align}
    \PS{r_1 a_1 \sim a_2}{}{} = \bar{P}_{r_1 a_1}^M \bar{P}_{a_2}^M \frac{1}{\bar{P}},\label{eq:between_group_dc_model}
\end{align}
i.e., the usual degree constrained isolation model and degree configuration model when defined directly on the coarse coordinates $(\RS\times\mathcal{A})\times\mathcal{A}.$

\item[(ii)]
The within $a_2$ normalization constant \eqref{eq:within_group_normalization_const} is the KL-divergence of the edge distribution over $(r_1b_1,b_2)\in(\RS\times\BS)\times\{b\in a_2\}$
 induced by the connectivity model
 \begin{align}
     \PS{r_1 b_1\sim b_2}{}{} = \mathbb{1}_{ \{b_1=b_2\} } \bar{P}_{r_1 b_1}^M \frac{1}{\pi_{b_2\mid a_2}}\label{eq:within_group_si_model}
 \end{align}
 from the edge distribution induced by the connectivity model
\begin{align}
    \PS{r_1 b_1\sim b_2}{}{}  = \bar{P}_{r_1 b_1}^M \bar{P}_{b_2}^M \frac{1}{\bar{P}_{a_2}^M}.\label{eq:within_group_dc_model}
\end{align}

\item[(iii)]
When the underlying model is the regional isolation model, i.e., $M = M_{RI},$  $M_{DC} = M_{ER}$ and $M_{DSI}=M_{SI},$ we recover 
the traditional normalization constants \eqref{eq:between_a_normalization_trad_theil} and \eqref{eq:within_a_normalization_trad_theil} of the group decomposition of the Theil index.

\item[(iv)]
Likewise, the normalization constants \eqref{eq:def_between_normalization_reg_decomp} and \eqref{eq:def_within_normalization_reg_decomp}
recover the traditional normalization constants \eqref{eq:between_s_normalization_trad_theil} and \eqref{eq:within_s_normalization_trad_theil} of the regional decomposition of the Theil index when the underlying model is the regional isolation model.
\end{enumerate}

\end{lemma}

\begin{remark}
\begin{enumerate}
    \item[(i)]
The connectivity models \eqref{eq:within_group_si_model} and \eqref{eq:within_group_dc_model} differ from the models obtained by defining $M_{DSI}$ and $M_{DC}$ in the usual way directly on $(\RS\times\{b\in a_2\})\times \{b\in a_2 \},$ i.e., from
\begin{align}\label{eq:SI_within_a2_a2_directly}
    \PS{r_1 b_1\sim b_2}{}{} = \mathbb{1}_{ \{b_1=b_2\} } \PS{r_1 b_1 \sim a_2}{}{M}\frac{1}{\pi_{b_2\mid a_2}} 
\end{align}
and 
\begin{align}\label{eq:DC_within_a2_a2_directly}
    \PS{r_1 b_1\sim b_2}{}{}  = \PS{r_1 b_1 \sim a_2}{}{M}\PS{a_2 \sim b_2}{}{M} \frac{1}{\PS{a_2\sim a_2}{}{M}}.
\end{align}
Consequently, the within normalization constants in Lemma~\ref{lemma:normalization_decompositions} should be interpreted as conditional components of the global normalization constant
\begin{align*}
    \DKL{\sigma^{M_{DSI}}}{\sigma^{M_{DC}}},    
\end{align*}
rather than as normalization constants obtained by solving a new, restricted normalization problem inside each coarse group, as in general the normalization constants built upon \eqref{eq:SI_within_a2_a2_directly} and \eqref{eq:DC_within_a2_a2_directly} will differ from the ones in Lemma~\ref{lemma:normalization_decompositions}.

\item[(ii)] Due to the nature of the global social isolation model $M_{DSI}$, $\sigma^{M_{DSI}}(\cdot\mid a_2)$ and $\sigma^{M_{DC}\mid M_{DSI}}(\cdot\mid a_2)$ only have support on $\{a_1=a_2\},$ while in general $\sigma^M(\cdot\mid a_2)$ and $\sigma^{M_{DC}\mid M}(\cdot\mid a_2)$ as given by \eqref{eq:sigma_M_group_within} and \eqref{eq:sigma_N_given_M_group_within} have support on all of $(\RS\times\BS)\times \{b\in a_2\}.$ While one might wish to simply apply the normalization strategy of Lemma~\ref{lemma:normalization_of_indices} to models defined on $(\RS\times\BS)\times \{b\in a_2\}$ with marginals given by $\sigma^M(\cdot\mid a_2),$ it is not obvious what the considered class of models $\mathcal{M}$ should be: the core ingredient in the proof of Lemma~\ref{lemma:normalization_of_indices} is that $M_{DC}$ and $M$ have the same average degrees (i.e., the same marginal distributions of $b_2$ and of $r_1b_1$). But this does not hold for $\sigma^M(\cdot\mid a_2)$ and $\sigma^{M_{DC}\mid M}(\cdot\mid a_2)$ anymore: while the marginals of $b_2\in a_2$ are still equal, the marginals of $r_1b_1$ will in general not be equal. However, by noting that by construction both $\sigma^M(\cdot\mid a_2)$ and $\sigma^{M_{DC}\mid M}(\cdot\mid a_2)$ also have the same $r_1a_1$ marginal conditional on $a_2,$
one could define a suitable class of models by requiring that the marginals of $b_2$ and of $r_1a_1$
are preserved and then find a maximizing model in this class for the KL-divergence from $\sigma^{M_{DC}\mid M}(\cdot\mid a_2).$ While conceptually valid, such an approach does not recover the traditional normalization constants of the Theil index when the underlying model is the regional isolation model $M_{RI}$ (details are omitted). For comparability with the traditional Theil index, we therefore adopt the normalization strategy of Lemma~\ref{lemma:normalization_decompositions}.
\end{enumerate}
\end{remark}

\subsection{Additional results for the segregation framework}
\label{sec:supp_additional_results}

\begin{lemma}[Coordinate coarsening and segregation]
\label{lemma:coordinate_coarsening}
Let $M'$ be a segregation-defining model and suppose that the corresponding reference models are admissible. We consider the three coordinate spaces defined in section~\ref{sec:coordinate_spaces_and_geosocial_networks}, which are obtained from one another by deterministic coarsening of geographical coordinates, and we consider finer geographical coordinates $p$ and regional coordinates $r$. We have the following relations.

\begin{enumerate}
    \item[(i)] For every generator $f$ and every $\rho\in\{\tau,\sigma,\xi\}$, coordinate coarsening cannot increase the $f$-divergence of $\rho^{M'}$ from $\rho^{M_{DC}(M')}$. In particular,
    \begin{equation}
    \begin{aligned}
        \mathcal{D}_{f,\rho}\left(M',M_{DC}(M');\BS, \BS\right)
        &\leq \mathcal{D}_{f,\rho}\left(M',M_{DC}(M');\RS\times \BS, \BS\right)\\
        &\leq \mathcal{D}_{f,\rho}\left(M',M_{DC}(M');\RS\times \BS, \RS\times \BS\right)
    \end{aligned}
    \label{eq:coordinate_coarsening_unnormalized}
    \end{equation}
and 
    \begin{equation}
    \begin{aligned}
        \mathcal{D}_{f,\rho}\left(M',M_{DC}(M');\RS\times \BS, \BS\right)\leq 
        \mathcal{D}_{f,\rho}\left(M',M_{DC}(M');\GS\times \BS, \BS\right) .
    \end{aligned}
    \label{eq:coordinate_coarsening_unnormalized_geoegocentric}
    \end{equation}

\item[(ii)] For $\rho=\sigma$, the same inequalities hold after normalization by the divergence of $M_{DSI}(M')$ from $M_{DC}(M')$. In particular,
\begin{equation}
\begin{aligned}
&\frac{\mathcal D_{f,\sigma}\left(M',M_{DC}(M');\BS,\BS\right)}
{\mathcal D_{f,\sigma}\left(M_{DSI}(M'),M_{DC}(M');\BS,\BS\right)}\\
\leq\;&\frac{\mathcal D_{f,\sigma}\left(M',M_{DC}(M');\RS\times\BS,\BS\right)}
{\mathcal D_{f,\sigma}\left(M_{DSI}(M'),M_{DC}(M');\RS\times\BS,\BS\right)}= S_{f,\sigma}(M')\\
\leq\;&
\frac{\mathcal D_{f,\sigma}\left(M',M_{DC}(M');\RS\times\BS,\RS\times\BS\right)}
{\mathcal D_{f,\sigma}\left(M_{DSI}(M'),M_{DC}(M');\RS\times\BS,\RS\times\BS\right)}
\end{aligned}
\label{eq:coordinate_coarsening_normalized_sigma}
\end{equation}
and
\begin{equation}
\frac{\mathcal D_{f,\sigma}\left(M',M_{DC}(M');\RS\times\BS,\BS\right)}
{\mathcal D_{f,\sigma}\left(M_{DSI}(M'),M_{DC}(M');\RS\times\BS,\BS\right)}
\leq
\frac{\mathcal D_{f,\sigma}\left(M',M_{DC}(M');\GS\times\BS,\BS\right)}
{\mathcal D_{f,\sigma}\left(M_{DSI}(M'),M_{DC}(M');\GS\times\BS,\BS\right)}.
\label{eq:coordinate_coarsening_normalized_sigma_geoegocentric}
\end{equation}
\end{enumerate}
\end{lemma}

\begin{remark}
\label{remark:coarsening_tau_xi}
The monotonicity result in part~(ii) of Lemma~\ref{lemma:coordinate_coarsening} is specific to $\rho=\sigma$. For $\rho=\tau$ and $\rho=\xi$, the corresponding inequalities after normalization by the divergence of $\rho^{M_{DSI}(M')}$ from $\rho^{M_{DC}(M')}$ do not hold in general: in the following Example~\ref{example:counterexample_coordinate_coarsening} geographical coordinate coarsening decreases the divergence of $\rho^{M_{DSI}(M')}$ from $\rho^{M_{DC}(M')}$ proportionally more than it decreases the divergence of $\rho^{M'}$ from $\rho^{M_{DC}(M')}.$ This cannot happen for the $\sigma$-based measures, as the divergence of $\sigma^{M_{DSI}(M')}$ from $\sigma^{M_{DC}(M')}$ is invariant under coordinate coarsening, as shown in part~(ii) of the proof of Lemma~\ref{lemma:coordinate_coarsening}.
\end{remark}

\begin{example}[Counterexample for Remark~\ref{remark:coarsening_tau_xi}]
\label{example:counterexample_coordinate_coarsening}
Let $\RS=\BS=\{1,2\}$ and
\begin{equation}
    \pi_{rb}^{M'}=\frac14,
    \qquad r,b\in\{1,2\}.
\end{equation}
We order the geosocial coordinates as $(1,1),(1,2),(2,1),(2,2)$ and let the symmetric matrix of connection probabilities under $M'$ be
\begin{equation}
    \left(
    \PS{c_1\sim c_2}{}{M'}
    \right)_{c_1,c_2}
    =
    \begin{pmatrix}
        0.1 & 0.4 & 0.1 & 0.2\\
        0.4 & 0.4 & 0.4 & 0.2\\
        0.1 & 0.4 & 0.1 & 0.2\\
        0.2 & 0.2 & 0.2 & 0.1
    \end{pmatrix}.
\label{eq:counterexample_coarsening_connection_matrix}
\end{equation}
The resulting marginal connection probability is $\bar P^{M'} = 0.23125$, and both $M_{DC}(M')$ and $M_{DSI}(M')$ are admissible. For the Kullback--Leibler generator $f(x)=x\log x$ we obtain
\begin{equation}
    \frac{\mathcal D_{f,\xi}\left(M',M_{DC}(M');\RS\times\BS,\RS\times\BS\right)}
    {\mathcal D_{f,\xi}\left(M_{DSI}(M'),M_{DC}(M');\RS\times\BS,\RS\times\BS\right)}
    \approx 0.0911,
\end{equation}
whereas marginalizing the region of the second individual gives
\begin{equation}
    \frac{\mathcal D_{f,\xi}\left(M',M_{DC}(M');\RS\times\BS,\BS\right)}
    {\mathcal D_{f,\xi}\left(M_{DSI}(M'),M_{DC}(M');\RS\times\BS,\BS\right)}
    \approx 0.0997.
\end{equation}
Thus, the DSI-scaled $\xi$-based measure increases under coordinate coarsening. For the same example,
\begin{equation}
    \frac{\mathcal D_{f,\tau}\left(M',M_{DC}(M');\RS\times\BS,\RS\times\BS\right)}
    {\mathcal D_{f,\tau}\left(M_{DSI}(M'),M_{DC}(M');\RS\times\BS,\RS\times\BS\right)}
    \approx 0.0809,
\end{equation}
while
\begin{equation}
    \frac{\mathcal D_{f,\tau}\left(M',M_{DC}(M');\RS\times\BS,\BS\right)}
    {\mathcal D_{f,\tau}\left(M_{DSI}(M'),M_{DC}(M');\RS\times\BS,\BS\right)}
    \approx 0.0821.
\end{equation}
Thus, the DSI-scaled $\tau$-based measure can also increase under coordinate coarsening.
\end{example}

\section{Proofs for the segregation framework}
\label{sec:proofs}

\subsection{Normalization of measures}

\begin{proof}[Proof of Lemma~\ref{lemma:normalization_of_indices}]
The proof follows an argument similar to the maximal statistical association argument in \cite{borgonovo2025convexity}: association between two R.V.s $X$ and $Y$ becomes large when the conditional distribution $\mathbb{P}_{Y\mid X=x}$ becomes very different from the marginal distribution $\mathbb{P}_Y$, and association, i.e., the average separation between conditional $\mathbb{P}_{Y\mid X=x}$ and marginal $\mathbb{P}_Y,$ becomes maximal when $Y$ is a deterministic function of $X$. In our use case, we are asking about the maximal separation from the random mixing null $\sigma^{M_{DC}}(r_1b_1,b_2)$ where the first coordinate $r_1b_1$ is independent of the second coordinate $b_2$. Perhaps not surprisingly, the maximal separation is given when the second coordinate is a deterministic function of the first coordinate.

For ease of notation, we write shorthand $N=M_{DC}(M')$ for the degree configuration null model and $M_{DSI}=M_{DSI}(M')$. The set $\mathcal M(M')$ is an equivalence class, and $N\in\mathcal M(M')$ by definition of the degree configuration model. Hence $\mathcal M(M')=\mathcal M(N)$, and every $\widetilde M\in\mathcal M(M')$ has the same degree configuration null $N$, since there is only one degree configuration model in each equivalence class. Additionally, $M_{DSI}\in\mathcal M(M')$: by definition, $\pi^{M_{DSI}}=\pi^{M'}$, and
\begin{equation}
\begin{aligned}
    \bar P_{r_1b_1}^{M_{DSI}} & =\sum_{b_2}\pi_{b_2}^{M'}\PS{r_1b_1\sim b_2}{}{M_{DSI}}
    =\sum_{b_2}\pi_{b_2}^{M'}\mathbb{1}_{\{b_1=b_2\}}\frac{\bar P_{r_1b_1}^{M'}}{\pi_{b_2}^{M'}}
    =\bar P_{r_1b_1}^{M'}.
\end{aligned}
\end{equation}
Hence $M_{DSI}$ satisfies the defining constraints of $\mathcal M(M')$.
We also note that for any $\widetilde M\in\mathcal M(M')$, the equality $\pi^{\widetilde M}=\pi^{M'}$ and the constraints $\bar P_{rb}^{\widetilde M}=\bar P_{rb}^{M'}$ imply
\begin{equation*}
\begin{aligned}
    \bar P_b^{\widetilde M}
    &=\sum_r\pi_{r\mid b}^{\widetilde M}\bar P_{rb}^{\widetilde M}
    =\sum_r\pi_{r\mid b}^{M'}\bar P_{rb}^{M'}
    =\bar P_b^{M'},\\
    \bar P^{\widetilde M}
    &=\sum_{r,b}\pi_{rb}^{\widetilde M}\bar P_{rb}^{\widetilde M}
    =\sum_{r,b}\pi_{rb}^{M'}\bar P_{rb}^{M'}
    =\bar P^{M'},\\
    \sigma^{\widetilde M}(b)
    &=\frac{\pi_b^{\widetilde M}\bar P_b^{\widetilde M}}{\bar P^{\widetilde M}}
    =\frac{\pi_b^{M'}\bar P_b^{M'}}{\bar P^{M'}} = \sigma^{M'}(b).
\end{aligned}
\end{equation*}

Next we note that for any $\widetilde M\in\mathcal M(M')$ by writing
\begin{align}
    \sigma^{\widetilde M}(r_1b_1,b_2)
    =\sigma^{\widetilde M}(r_1b_1)\sigma^{\widetilde M}(b_2\mid r_1b_1)
\end{align}
and using the degree constraints
\begin{align}
    \sigma^{\widetilde M}(r_1b_1)
    &=\sum_{b_2}\pi_{r_1b_1}^{M'}\pi_{b_2}^{M'}
    \frac{\PS{r_1b_1\sim b_2}{}{\widetilde M}}{\bar P^{M'}}\\
    &=\frac{\pi_{r_1b_1}^{M'}\bar P_{r_1b_1}^{M'}}{\bar P^{M'}}
    =\sigma^N(r_1b_1)
\end{align}
we obtain
\begin{align}
    \sigma^{\widetilde M}(r_1b_1,b_2)
    =\sigma^N(r_1b_1)\sigma^{\widetilde M}(b_2\mid r_1b_1).
\end{align}
In particular, for the degree constrained social isolation model $M_{DSI}$, we obtain
\begin{align}
    \sigma^{M_{DSI}}(r_1b_1,b_2) = \sigma^N(r_1b_1)\mathbb{1}_{\{b_1=b_2\}}
\end{align}
with
\begin{align}
    \left\{\sigma^{M_{DSI}}(b_2\mid r_1b_1)\right\}_{b_2\in\BS}=e_{b_1},
\end{align}
where $e_{b_1}$ denotes the unit vector with a $1$ at position $b_1$ and zeros elsewhere.
We further note that $\sigma^N(b_2\mid r_1b_1)=\sigma^N(b_2),$ because
\begin{equation}
\begin{aligned}
    \sigma^N(r_1b_1,b_2)
    =\pi_{r_1b_1}^{M'}\frac{\bar P_{r_1b_1}^{M'}}{\bar P^{M'}}
    \pi_{b_2}^{M'}\frac{\bar P_{b_2}^{M'}}{\bar P^{M'}}
    =\sigma^N(r_1b_1)\sigma^N(b_2),
\end{aligned}
\end{equation}
and define the function $F_N$ on the $|\BS|-1$ simplex $\Delta^{|\BS|-1}$ by
\begin{align}
    F_N\colon &\Delta^{|\BS|-1}\to\mathbb{R},\quad
    p\mapsto\sum_{b_2}\sigma^N(b_2)f\left(\frac{p_{b_2}}{\sigma^N(b_2)}\right).
\end{align}
$F_N$ is convex by the convexity of $f,$ and hence
\begin{equation}
\begin{aligned}
    F_N\left(\left\{\sigma^{\widetilde M}(b_2\mid r_1b_1)\right\}_{b_2\in\BS}\right)
    &=F_N\left(\sum_{b_2}\sigma^{\widetilde M}(b_2\mid r_1b_1)e_{b_2}\right)\\
    &\leq\sum_{b_2}\sigma^{\widetilde M}(b_2\mid r_1b_1)F_N(e_{b_2}).
\end{aligned}
\end{equation}
For $M_{DSI}$ we obtain equality because
\begin{equation}
\begin{aligned}
    \sum_{b_2}\sigma^{M_{DSI}}(b_2\mid r_1b_1)F_N(e_{b_2})
    =F_N(e_{b_1})
    =F_N\left(\left\{\sigma^{M_{DSI}}(b_2\mid r_1b_1)\right\}_{b_2\in\BS}\right).
\end{aligned}
\end{equation}

Putting everything together, we obtain for the $f$-divergence between $\sigma^{\widetilde M}$ and $\sigma^N$
\begin{equation}
\begin{aligned}
    D_f\left(\sigma^{\widetilde M}\middle\|\sigma^N\right)
    &=\sum_{r_1b_1}\sigma^N(r_1b_1)
    F_N\left(\left\{\sigma^{\widetilde M}(b_2\mid r_1b_1)\right\}_{b_2\in\BS}\right)\\
    &\leq\sum_{r_1b_1}\sigma^N(r_1b_1)
    \sum_{b_2}\sigma^{\widetilde M}(b_2\mid r_1b_1)F_N(e_{b_2})\\
    &=\sum_{b_2}F_N(e_{b_2})
    \sum_{r_1b_1}\sigma^N(r_1b_1)\sigma^{\widetilde M}(b_2\mid r_1b_1)\\
    &=\sum_{b_2}F_N(e_{b_2})
    \sum_{r_1b_1}\sigma^{\widetilde M}(r_1b_1,b_2)\\
    &=\sum_{b_2}\sigma^{\widetilde M}(b_2)F_N(e_{b_2})
    =\sum_{b_2}\sigma^{M_{DSI}}(b_2)F_N(e_{b_2})\\
    &=\sum_{r_1b_1}\sigma^N(r_1b_1)
    \sum_{b_2}\sigma^{M_{DSI}}(b_2\mid r_1b_1)F_N(e_{b_2})\\
    &=\sum_{r_1b_1}\sigma^N(r_1b_1)F_N(e_{b_1})\\
    &=\sum_{r_1b_1}\sigma^N(r_1b_1)
    F_N\left(\left\{\sigma^{M_{DSI}}(b_2\mid r_1b_1)\right\}_{b_2\in\BS}\right)\\
    &=D_f\left(\sigma^{M_{DSI}}\middle\|\sigma^N\right).
\end{aligned}
\end{equation}
This shows that
\begin{equation*}
    M_{DSI}\in\argmax_{\widetilde M\in\mathcal M(M')}
    \mathcal D_{f,\sigma}(\widetilde M).
\end{equation*}
If $\bar P_{rb}^{M'}=\bar P^{M'}$ for all $r,b$, then $M_{DC}(M')=M_{ER}$ and $M_{DSI}=M_{SI}$ and \eqref{eq:M_SI_in_argmax} follows. To show \eqref{eq:normalization_const_H_D_C}, we first observe that
\begin{equation}
\begin{aligned}
    \mathcal D_{f,\sigma}(M_{SI})
    &=\sum_{r,b_1,b_2}\sigma^{M_{ER}}(rb_1,b_2)
    f\left(\frac{\sigma^{M_{SI}}(rb_1,b_2)}{\sigma^{M_{ER}}(rb_1,b_2)}\right)\\
    &=\sum_{r,b_1,b_2}\pi_r^{M'}\pi_{b_1\mid r}^{M'}\pi_{b_2}^{M'}
    f\left(\frac{\mathbb{1}_{\{b_1=b_2\}}}{\pi_{b_2}^{M'}}\right).
\end{aligned}
\end{equation}
Hence,
\begin{equation}
\begin{aligned}
    \mathcal D_{x\log x,\sigma}(M_{SI})
    &=\sum_{r,b_1,b_2}\pi_r^{M'}\pi_{b_1\mid r}^{M'}\pi_{b_2}^{M'}
    \left(\frac{\mathbb{1}_{\{b_1=b_2\}}}{\pi_{b_2}^{M'}}\right)
    \log\left(\frac{\mathbb{1}_{\{b_1=b_2\}}}{\pi_{b_2}^{M'}}\right)\\
    &=-\sum_b\pi_b^{M'}\log(\pi_b^{M'}),\\
    \mathcal D_{\frac12|x-1|,\sigma}(M_{SI})
    &=\frac12\sum_{r,b_1,b_2}\pi_r^{M'}\pi_{b_1\mid r}^{M'}\pi_{b_2}^{M'}
    \left|\mathbb{1}_{\{b_1=b_2\}}\frac{1}{\pi_{b_2}^{M'}}-1\right|\\
    &=\frac12\left(
    \sum_{r,b}\pi_r^{M'}\pi_{b\mid r}^{M'}\pi_b^{M'}
    \left(\frac{1}{\pi_b^{M'}}-1\right)
    +\sum_r\sum_{b_1\neq b_2}\pi_r^{M'}\pi_{b_1\mid r}^{M'}\pi_{b_2}^{M'}
    \right)\\
    &=\frac12\left(
    1+1-2\sum_r\sum_b\pi_r^{M'}\pi_{b\mid r}^{M'}\pi_b^{M'}
    \right)
    =\sum_b\pi_b^{M'}(1-\pi_b^{M'}),\\
    \mathcal D_{x^2-1,\sigma}(M_{SI})
    &=\sum_{r,b_1,b_2}\pi_r^{M'}\pi_{b_1\mid r}^{M'}\pi_{b_2}^{M'}
    \left(\mathbb{1}_{\{b_1=b_2\}}\frac{1}{(\pi_{b_2}^{M'})^2}-1\right)\\
    &=\sum_b(\pi_b^{M'})^2\frac{1}{(\pi_b^{M'})^2}-1
    =|\BS|-1.
\end{aligned}
\label{eq:normalization_const_KL_TV_chi}
\end{equation}
\end{proof}

\begin{proof}[Proof of Corollary~\ref{corollary:normalization_social_space}]
Recall the definitions of $M_{DC}(M')$ and $M_{DSI}(M')$ in \eqref{eq:def_M_DC} and \eqref{eq:def_M_DSI}. By marginalizing over the regional coordinates, we obtain the group-to-group connection probabilities
\begin{equation}
\begin{aligned}
    \PS{b_1\sim b_2}{}{M_{DC}(M')}
    &=\frac{\bar P_{b_1}^{M'}\bar P_{b_2}^{M'}}{\bar P^{M'}},\\
    \PS{b_1\sim b_2}{}{M_{DSI}(M')}
    &=\frac{\bar P_{b_1}^{M'}}{\pi_{b_1}^{M'}}\mathbb{1}_{\{b_1=b_2\}}.
\end{aligned}
\label{eq:M_DC_and_M_DSI_on_BB}
\end{equation}
The proof of Lemma~\ref{lemma:normalization_of_indices} now applies directly by considering the special case of a single region $\RS=\{r\}$ with $\pi_r=1$. Then $rb$ is identified with $b$, $\pi_{rb}=\pi_b$, and $\bar P_{rb}=\bar P_b$, so the constraints defining $\mathcal M(M')$ become precisely those defining $\mathcal M_{\BS}(M')$. Furthermore, under the identification $rb$ with $b$, the degree configuration and degree-constrained social isolation models take the forms in \eqref{eq:M_DC_and_M_DSI_on_BB}. Hence \eqref{eq:M_DSI_in_argmax_BB} follows from \eqref{eq:M_DSI_in_argmax}.
\end{proof}

\subsection{Recovery of traditional geographical segregation measures as special cases}

\begin{proof}[Proof of Lemma~\ref{lemma:relation_between_geo_and_social_f_divergence}]
For ease of notation, throughout this proof we write $\bar P\coloneqq \bar P^{M_\phi}=\bar P^{M_{RI}}$ for their common marginal connection probability.
For part (i), by the constant-degree condition \eqref{eq:constant_weighted_row_sums_phi}, the marginal connectivity of every geosocial coordinate $pb$ under $M_\phi$ is
\begin{equation*}
    \bar P_{pb}^{M_\phi}=\sum_{p',b'}\pi_{p'b'}\phi(p,p')=\sum_{p'}\pi_{p'}\phi(p,p')=\bar P.
\end{equation*}
Hence $M_{DC}(M_\phi)=M_{ER}$.

To show \eqref{eq:equivalence_of_indices_under_spatial_kernel}, we first note that on the geo-egocentric coordinate space
\begin{equation}
\begin{aligned}
    \PS{p_1b_1\sim b_2}{}{M_\phi}
    &=\sum_{p_2}\PS{p_1b_1\sim p_2b_2}{}{M_\phi}\pi_{p_2\mid b_2}
    =\frac{1}{\pi_{b_2}}\sum_{p_2}\pi_{p_2b_2}\phi(p_1,p_2)\\
    &=\frac{\widetilde\pi_{b_2\mid p_1}}{\pi_{b_2}}\bar P,
\end{aligned}
\label{eq:connectivity_under_M_phi_gecentric}
\end{equation}
where the last equality follows from the constant-degree condition and \eqref{eq:def_tilde_pi_discrete}. Therefore,
\begin{equation*}
\begin{aligned}
    \frac{\sigma^{M_\phi}(pb_1,b_2)}{\sigma^{M_{DC}(M_\phi)}(pb_1,b_2)}
    &=\frac{\pi_p\pi_{b_1\mid p}\pi_{b_2}\frac{\widetilde\pi_{b_2\mid p}}{\pi_{b_2}}\frac{\bar P}{\bar P}}
    {\pi_p\pi_{b_1\mid p}\pi_{b_2}\frac{\bar P}{\bar P}}
    =\frac{\widetilde\pi_{b_2\mid p}}{\pi_{b_2}}.
\end{aligned}
\end{equation*}
Consequently,
\begin{equation}
\begin{aligned}
    \mathcal D_{f,\sigma}\left(M_\phi,M_{DC}(M_\phi);\GS\times\BS,\BS\right)
    &=\sum_p\sum_{b_1,b_2}\sigma^{M_{DC}(M_\phi)}(pb_1,b_2)
    f\left(\frac{\sigma^{M_\phi}(pb_1,b_2)}{\sigma^{M_{DC}(M_\phi)}(pb_1,b_2)}\right)\\
    &=\sum_p\sum_{b_1,b_2}\pi_p\pi_{b_1\mid p}\pi_{b_2}
    f\left(\frac{\widetilde\pi_{b_2\mid p}}{\pi_{b_2}}\right)\\
    &=\sum_p\sum_{b_2}\pi_p\pi_{b_2}
    f\left(\frac{\widetilde\pi_{b_2\mid p}}{\pi_{b_2}}\right)
    =\mathcal D_{f,\widetilde\pi},
\end{aligned}
\label{eq:proof_D_f_sigma_equals_D_f_tilde_pi}
\end{equation}
which proves \eqref{eq:equivalence_of_indices_under_spatial_kernel}.

The regional isolation model $M_{RI}$ is the special case of $M_\phi$ corresponding to the region-based kernel in \eqref{eq:def_RI_model1}. If $r$ denotes the region containing $p$, then
\begin{equation}
\begin{aligned}
    \widetilde\pi_{b\mid p}
    &=\frac{1}{\bar P}\sum_{p'}\pi_{p'b}\phi(p,p')
    =\frac{1}{\pi_r}\sum_{p':\,r(p')=r}\pi_{p'b}
    =\frac{\pi_{rb}}{\pi_r}\\
    &=\pi_{b\mid r}.
\end{aligned}
\label{eq:tilde_pi_under_M_RI}
\end{equation}
Hence, by specializing \eqref{eq:proof_D_f_sigma_equals_D_f_tilde_pi} to $M_{RI}$ and grouping the sum over $p$ by regions, we obtain
\begin{equation*}
    \mathcal D_{f,\sigma}\left(M_{RI},M_{DC}(M_{RI});\RS\times\BS,\BS\right)
    =\sum_r\sum_b\pi_r\pi_b f\left(\frac{\pi_{b\mid r}}{\pi_b}\right)
    =\mathcal D_{f,\pi},
\end{equation*}
which proves \eqref{eq:equivalence_of_indices_under_regional_isolation}.


For part (ii), we first prove \eqref{eq:lemma_ineq1}--\eqref{eq:lemma_ineq3} for a general geospatial connectivity model $M_\phi$. For $M_{HB}$ as defined in \eqref{eq:def_M_HB}, we have
\begin{equation*}
    \pi_{pb}^{M_{HB}(M_\phi)}=\pi_p\pi_b,\qquad \PS{p_1b_1\sim p_2b_2}{}{M_{HB}(M_\phi)}=\phi(p_1,p_2).
\end{equation*}
On the geo-egocentric coordinate space,
\begin{equation*}
\begin{aligned}
    \PS{pb_1\sim b_2}{}{M_{HB}(M_\phi)}
    &=\sum_{p_2}\PS{pb_1\sim p_2b_2}{}{M_{HB}(M_\phi)}\pi_{p_2\mid b_2}^{M_{HB}(M_\phi)}\\
    &=\sum_{p_2}\phi(p,p_2)\pi_{p_2}
    =\bar P,
\end{aligned}
\end{equation*}
where the last equality follows from the constant-degree condition. Since this holds for all $p,b_1,b_2$, the marginal connection probability $\bar{P}^{M_{HB}(M_\phi)}$ also equals $\bar P$. It follows that
\begin{equation*}
    \sigma_{\GS\times\BS,\BS}^{M_{HB}(M_\phi)}(pb_1,b_2)
    =\frac{\pi_p\pi_{b_1}\pi_{b_2}\bar P}{\bar P}
    =\pi_p\pi_{b_1}\pi_{b_2}.
\end{equation*}
By \eqref{eq:connectivity_under_M_phi_gecentric},
\begin{equation*}
\begin{aligned}
    \sigma_{\GS\times\BS,\BS}^{M_\phi}(pb_1,b_2)
    &=\frac{\pi_p\pi_{b_1\mid p}\pi_{b_2}}{\bar P}\PS{pb_1\sim b_2}{}{M_\phi}\\
    &=\pi_p\pi_{b_1\mid p}\widetilde\pi_{b_2\mid p}.
\end{aligned}
\end{equation*}
Hence
\begin{equation*}
    \frac{\sigma_{\GS\times\BS,\BS}^{M_\phi}(pb_1,b_2)}
    {\sigma_{\GS\times\BS,\BS}^{M_{HB}(M_\phi)}(pb_1,b_2)}
    =\frac{\pi_{b_1\mid p}\widetilde\pi_{b_2\mid p}}{\pi_{b_1}\pi_{b_2}},
\end{equation*}
and therefore
\begin{equation}
    D_f\!\left(\sigma_{\GS\times\BS,\BS}^{M_\phi}\middle\|\sigma_{\GS\times\BS,\BS}^{M_{HB}(M_\phi)}\right)
    =\sum_p\sum_{b_1,b_2}\pi_p\pi_{b_1}\pi_{b_2}
    f\left(\frac{\pi_{b_1\mid p}\widetilde\pi_{b_2\mid p}}{\pi_{b_1}\pi_{b_2}}\right).
    \label{eq:S_Df_phi_HB_geoegocentric}
\end{equation}
On the full geosocial coordinate space,
\begin{equation*}
    \sigma_{\GS\times\BS,\GS\times\BS}^{M_\phi}(p_1b_1,p_2b_2)
    =\pi_{p_1}\pi_{b_1\mid p_1}\pi_{p_2}\pi_{b_2\mid p_2}\frac{\phi(p_1,p_2)}{\bar P},
\end{equation*}
whereas
\begin{equation*}
    \sigma_{\GS\times\BS,\GS\times\BS}^{M_{HB}(M_\phi)}(p_1b_1,p_2b_2)
    =\pi_{p_1}\pi_{b_1}\pi_{p_2}\pi_{b_2}\frac{\phi(p_1,p_2)}{\bar P}.
\end{equation*}
Consequently, on their common support,
\begin{equation*}
    \frac{\sigma_{\GS\times\BS,\GS\times\BS}^{M_\phi}(p_1b_1,p_2b_2)}
    {\sigma_{\GS\times\BS,\GS\times\BS}^{M_{HB}(M_\phi)}(p_1b_1,p_2b_2)}
    =\frac{\pi_{b_1\mid p_1}\pi_{b_2\mid p_2}}{\pi_{b_1}\pi_{b_2}},
\end{equation*}
and hence
\begin{equation}
\begin{aligned}
    &D_f\!\left(\sigma_{\GS\times\BS,\GS\times\BS}^{M_\phi}\middle\|\sigma_{\GS\times\BS,\GS\times\BS}^{M_{HB}(M_\phi)}\right)\\
    &\qquad=\sum_{p_1,p_2}\sum_{b_1,b_2}\pi_{p_1}\pi_{p_2}\pi_{b_1}\pi_{b_2}\frac{\phi(p_1,p_2)}{\bar P}
    f\left(\frac{\pi_{b_1\mid p_1}\pi_{b_2\mid p_2}}{\pi_{b_1}\pi_{b_2}}\right).
\end{aligned}
\label{eq:S_Df_phi_HB_full_geo_blau}
\end{equation}
Equations~\eqref{eq:S_Df_phi_HB_geoegocentric} and \eqref{eq:S_Df_phi_HB_full_geo_blau} generally differ from one another and from
\begin{equation*}
    \mathcal D_{f,\widetilde\pi}
    =\sum_p\sum_b\pi_p\pi_b f\left(\frac{\widetilde\pi_{b\mid p}}{\pi_b}\right),
\end{equation*}
which proves \eqref{eq:lemma_ineq1}--\eqref{eq:lemma_ineq3}.

We next prove \eqref{eq:equivalence_of_indices_under_HB_phi} and \eqref{eq:equivalence_for_xlogx} for the regional isolation model $M_{RI}$. Since $M_{RI}$ is the special case of $M_\phi$ corresponding to the region-based kernel in \eqref{eq:def_RI_model1}, we can specialize the expressions above. By \eqref{eq:tilde_pi_under_M_RI}, $\widetilde\pi_{b\mid p}=\pi_{b\mid r}$ for the region $r$ containing $p$. Hence, specializing \eqref{eq:S_Df_phi_HB_geoegocentric} and grouping the sum over $p$ by regions gives
\begin{equation}
    D_f\!\left(\sigma_{\RS\times\BS,\BS}^{M_{RI}}\middle\|\sigma_{\RS\times\BS,\BS}^{M_{HB}(M_{RI})}\right)
    =\sum_r\sum_{b_1,b_2}\pi_r\pi_{b_1}\pi_{b_2}
    f\left(\frac{\pi_{b_1\mid r}\pi_{b_2\mid r}}{\pi_{b_1}\pi_{b_2}}\right).
    \label{eq:HB_RI_geoegocentric}
\end{equation}
Similarly, specializing \eqref{eq:S_Df_phi_HB_full_geo_blau} to the region-based kernel in \eqref{eq:def_RI_model1} gives
\begin{equation}
\begin{aligned}
    &D_f\!\left(\sigma_{\RS\times\BS,\RS\times\BS}^{M_{RI}}\middle\|\sigma_{\RS\times\BS,\RS\times\BS}^{M_{HB}(M_{RI})}\right)\\
    &\qquad=\sum_{r_1,r_2}\sum_{b_1,b_2}\pi_{r_1}\pi_{r_2}\pi_{b_1}\pi_{b_2}
    \frac{\mathbb{1}_{\{r_1=r_2\}}}{\pi_{r_1}}
    f\left(\frac{\pi_{b_1\mid r_1}\pi_{b_2\mid r_2}}{\pi_{b_1}\pi_{b_2}}\right)\\
    &\qquad=\sum_r\sum_{b_1,b_2}\pi_r\pi_{b_1}\pi_{b_2}
    f\left(\frac{\pi_{b_1\mid r}\pi_{b_2\mid r}}{\pi_{b_1}\pi_{b_2}}\right).
\end{aligned}
\label{eq:HB_RI_full}
\end{equation}
Comparing \eqref{eq:HB_RI_geoegocentric} and \eqref{eq:HB_RI_full} proves \eqref{eq:equivalence_of_indices_under_HB_phi}.

\eqref{eq:equivalence_for_xlogx} follows now from
\begin{equation*}
    \log\left(\frac{\pi_{b_1\mid r}\pi_{b_2\mid r}}{\pi_{b_1}\pi_{b_2}}\right)
    =\log\left(\frac{\pi_{b_1\mid r}}{\pi_{b_1}}\right)+\log\left(\frac{\pi_{b_2\mid r}}{\pi_{b_2}}\right).
\end{equation*}

For part (iii), to show \eqref{eq:S_under_blau_space_only}, we first note that by \eqref{eq:connectivity_under_M_phi_gecentric} and \eqref{eq:tilde_pi_under_M_RI},
\begin{equation}
\begin{aligned}
    \PS{b_1\sim b_2}{}{M_{RI}}
    &=\sum_r\PS{rb_1\sim b_2}{}{M_{RI}}\pi_{r\mid b_1} =\bar P\sum_r\frac{\pi_{b_2\mid r}}{\pi_{b_2}}\pi_{r\mid b_1}\\
    &=\bar P\sum_r\pi_r\frac{\pi_{b_1\mid r}}{\pi_{b_1}}\frac{\pi_{b_2\mid r}}{\pi_{b_2}}.
\end{aligned}
\label{eq:P_b_sim_b_under_M_RI}
\end{equation}
Since $M_{RI}$ is a special case of $M_\phi$, the result $M_{DC}(M_\phi)=M_{ER}$ established at the beginning of the proof of part (i) gives $M_{DC}(M_{RI})=M_{ER}$. Coarsening to the social coordinate space therefore yields
\begin{equation}
    \sigma_{\BS,\BS}^{M_{DC}(M_{RI})}(b_1,b_2)=\pi_{b_1}\pi_{b_2}.
    \label{eq:sigma_M_DC_bb}
\end{equation}
Moreover, marginalizing the geo-egocentric distribution $\sigma_{\RS\times\BS,\BS}^{M_{HB}(M_{RI})}(rb_1,b_2)$ obtained in part (ii) over $r$ gives
\begin{equation}
    \sigma_{\BS,\BS}^{M_{HB}(M_{RI})}(b_1,b_2)
    =\sum_r\pi_r\pi_{b_1}\pi_{b_2}
    =\pi_{b_1}\pi_{b_2}.
    \label{eq:sigma_M_HB_bb}
\end{equation}
By \eqref{eq:sigma_M_DC_bb} and \eqref{eq:sigma_M_HB_bb}, using \eqref{eq:P_b_sim_b_under_M_RI},
\begin{equation*}
    \frac{\sigma_{\BS,\BS}^{M_{RI}}(b_1,b_2)}
    {\sigma_{\BS,\BS}^{M_{HB}(M_{RI})}(b_1,b_2)}
    =\frac{\sigma_{\BS,\BS}^{M_{RI}}(b_1,b_2)}
    {\sigma_{\BS,\BS}^{M_{DC}(M_{RI})}(b_1,b_2)}
    =\sum_r\pi_r\frac{\pi_{b_1\mid r}}{\pi_{b_1}}\frac{\pi_{b_2\mid r}}{\pi_{b_2}} .
\end{equation*}
It follows that
\begin{equation*}
\begin{aligned}
    D_f\!\left(\sigma_{\BS,\BS}^{M_{RI}}\middle\|\sigma_{\BS,\BS}^{M_{DC}(M_{RI})}\right)
    &=D_f\!\left(\sigma_{\BS,\BS}^{M_{RI}}\middle\|\sigma_{\BS,\BS}^{M_{HB}(M_{RI})}\right)\\
    &=\sum_{b_1,b_2}\pi_{b_1}\pi_{b_2}
    f\left(
    \sum_r\pi_r\frac{\pi_{b_1\mid r}}{\pi_{b_1}}\frac{\pi_{b_2\mid r}}{\pi_{b_2}}
    \right).
\end{aligned}
\end{equation*}
\end{proof}

\begin{proof}[Proof of Corollary~\ref{corollary:recovery_normalized_traditional_indices}]
These results follow directly from Lemma~\ref{lemma:relation_between_geo_and_social_f_divergence} and the normalization constants in \eqref{eq:normalization_const_H_D_C}.
\end{proof}


\subsection{Conditional coordinate and joint coordinate--link distributions}

\begin{proof}[Proof of Lemma~\ref{lemma:T_D_F_to_S_D_F}]
We first note that for any generator $f$ and models $M,N$ with $\pi^M = \pi^N = \pi$ we have
\begin{align}
    \mathcal D_{f,\tau}(M,N;\CS_1,\CS_2)
    &= \sum_{c_1 c_2} \pi_{c_1}\pi_{c_2} \PS{c_1\sim c_2}{}{N} f\left(\frac{\pi_{c_1}\pi_{c_2} \PS{c_1\sim c_2}{}{M}}{\pi_{c_1}\pi_{c_2} \PS{c_1\sim c_2}{}{N}}\right)\nonumber\\
   & + \sum_{c_1 c_2} \pi_{c_1}\pi_{c_2} \left(1 - \PS{c_1\sim c_2}{}{N}\right) f\left(\frac{\pi_{c_1}\pi_{c_2} \left(1 - \PS{c_1\sim c_2}{}{M}\right)}{\pi_{c_1}\pi_{c_2} \left(1 - \PS{c_1\sim c_2}{}{N}\right)}\right)\label{eq:expand_T_D_f_v0}\\
& = \sum_{c_1 c_2} \pi_{c_1}\pi_{c_2} \PS{c_1\sim c_2}{}{N} f\left(\frac{ \PS{c_1\sim c_2}{}{M}}{\PS{c_1\sim c_2}{}{N}}\right)\nonumber\\
   & + \sum_{c_1 c_2} \pi_{c_1}\pi_{c_2} \left(1 - \PS{c_1\sim c_2}{}{N}\right) f\left(\frac{\left(1 - \PS{c_1\sim c_2}{}{M}\right)}{\left(1 - \PS{c_1\sim c_2}{}{N}\right)}\right) \label{eq:expand_T_D_f}
\end{align}
and \eqref{eq:T_D_f_eq_sum_of_S_D_fs} follows from \eqref{eq:expand_T_D_f_v0} and the assumption $\bar{P}^M = \bar{P}^N = \bar{P}.$

With $g(x) = \frac{1}{2}|x-1|,$ \eqref{eq:T_D_f_is_twice_dissimilarity} follows from \eqref{eq:expand_T_D_f} like this
\begin{align}
    \mathcal D_{g,\tau}(M,N;\CS_1,\CS_2) &= \frac{1}{2}\sum_{c_1 c_2} \pi_{c_1}\pi_{c_2} \left|\PS{c_1\sim c_2}{}{M}-\PS{c_1\sim c_2}{}{N}\right| \nonumber\\
    & + \frac{1}{2}\sum_{c_1 c_2} \pi_{c_1}\pi_{c_2} \left|(1-\PS{c_1\sim c_2}{}{M})-(1-\PS{c_1\sim c_2}{}{N})\right|\nonumber\\
    & = 2 \frac{\bar{P}}{2} \sum_{c_1 c_2}  \left|\frac{\pi_{c_1}\pi_{c_2}\PS{c_1\sim c_2}{}{M}}{\bar{P}}-\frac{\pi_{c_1}\pi_{c_2}\PS{c_1\sim c_2}{}{N}}{\bar{P}}\right|\nonumber\\
    & = 2\bar P\mathcal D_{g,\sigma}(M,N;\CS_1,\CS_2)\nonumber\\
    & = 2(1-\bar{P})\mathcal D_{g,\xi}(M,N;\CS_1,\CS_2).
\end{align}
For $h(x) = x \log x,$ allowing $\bar{P}^M\neq \bar{P}^N,$ by writing
\begin{align}
    \logg{\frac{\pi_{c_1}\pi_{c_2} \PS{c_1\sim c_2}{}{M}}{\pi_{c_1}\pi_{c_2} \PS{c_1\sim c_2}{}{N}}}
    & = \logg{
    \frac{\frac{\pi_{c_1}\pi_{c_2} \PS{c_1\sim c_2}{}{M}}{\bar{P}^M}}
    {\frac{\pi_{c_1}\pi_{c_2} \PS{c_1\sim c_2}{}{N}}{\bar{P}^N}}
    }
    + \logg{\frac{\bar{P}^M}{\bar{P}^N}},\\
    \logg{\frac{\pi_{c_1}\pi_{c_2} (1-\PS{c_1\sim c_2}{}{M})}{\pi_{c_1}\pi_{c_2}(1- \PS{c_1\sim c_2}{}{N})}}
    & = \logg{
    \frac{\frac{\pi_{c_1}\pi_{c_2} (1-\PS{c_1\sim c_2}{}{M})}{(1-\bar{P}^M)}}
    {\frac{\pi_{c_1}\pi_{c_2} (1-\PS{c_1\sim c_2}{}{N})}{(1-\bar{P}^N)}}
    }
    + \logg{\frac{(1-\bar{P}^M)}{(1-\bar{P}^N)}}
\end{align}
\eqref{eq:T_D_f_and H} follows from \eqref{eq:expand_T_D_f_v0} by noting that
\begin{align}
&\bar{P}^M\logg{\frac{\bar{P}^M}{\bar{P}^N}} + (1-\bar{P}^M)\logg{\frac{(1-\bar{P}^M)}{(1-\bar{P}^N)}}\nonumber\\
& = \DKL{\operatorname{Bernoulli}(\bar{P}^M)}{\operatorname{Bernoulli}(\bar{P}^N)}.    
\end{align}

To show \eqref{eq:T_D_f_to_S_D_f}, we first note that \eqref{eq:constant_ratio_of_Pmi_to_barP} implies for all $c_1,c_2$ and $i$ that
\begin{align}
    \PS{c_1\sim c_2}{}{M_i} &= \bar{P}_i\frac{\PS{c_1\sim c_2}{}{M_i}}{\bar{P}_i}
    = \bar{P}_i\frac{\PS{c_1\sim c_2}{}{M_1}}{\bar{P}_1} \leq \frac{\bar{P}_i}{\bar{P}_1}\label{eq:upperbound_for_prob_Mi}
\end{align}
and similarly \eqref{eq:constant_ratio_of_Pni_to_barP} implies 
\begin{align}
   \PS{c_1\sim c_2}{}{N_i}  &\leq \frac{\bar{P}_i}{\bar{P}_1}\label{eq:upperbound_for_prob_Ni}
\end{align}
and therefore also
\begin{align}
    \left| \PS{c_1\sim c_2}{}{M_i} - \PS{c_1\sim c_2}{}{N_i}  \right| \leq \frac{\bar{P}_i}{\bar{P}_1}. \label{eq:upperbound_for_prob_diff_Mi_Ni}
\end{align}
Let us arbitrarily order the coordinate pairs into a finite sequence 
\begin{align}
    (c_1^{(1)},c_2^{(1)}), (c_1^{(2)},c_2^{(2)}),\dots,(c_1^{(K)},c_2^{(K)})
\end{align}
where $K = |\CS_1||\CS_2|,$ and write shorthand
\begin{align}
    w_{k} & = \pi_{c_1^{(k)}}\pi_{c_2^{(k)}},\\
    x_{k,i} & = \PS{c_1^{(k)}\sim c_2^{(k)}}{}{M_i},\\
    y_{k,i} & = \PS{c_1^{(k)}\sim c_2^{(k)}}{}{N_i},
\end{align}
such that the second sum in \eqref{eq:expand_T_D_f} for models $M_i$ and $N_i$ can be written as
\begin{align}
    &\sum_{c_1 c_2} \pi_{c_1}\pi_{c_2} \left(1 - \PS{c_1\sim c_2}{}{N_i}\right) f\left(\frac{\left(1 - \PS{c_1\sim c_2}{}{M_i}\right)}{\left(1 - \PS{c_1\sim c_2}{}{N_i}\right)}\right)\nonumber\\
    &= \sum_k w_k (1 - y_{k,i}) f\left(\frac{1-x_{k,i}}{1-y_{k,i}}\right).
\end{align}
Note that by assumption $\sum_k w_k x_{k,i} = \sum_k w_k y_{k,i} = \bar{P}_i$ for all $i$.
Let $h(x,y)\coloneqq (1-y)f\left(\frac{1-x}{1-y}\right).$
The gradient of $h(x,y)$ at $(x,y) = (a,a)$ is given by $(-f^\prime(1), f^\prime(1))^T$ (where we have used that $f(1) = 0$ for any generator $f$) and the Hessian of $h(x,y)$ at $(x,y) = (a,a)$ is given (again using that $f(1) = 0$ for any generator $f$) by
$\frac{f^{\prime\prime}(1)}{1-a}\begin{pmatrix}
    1,& -1\\
    -1,& 1
\end{pmatrix}.$
Since $f$ is three-times continuously differentiable in a neighborhood of $1$, there exists a real number $0<c<1$ such that $f$ is three-times continuously differentiable in $[1-c,\frac{1}{1-c}]$ and we can uniformly upper-bound the absolute value of the third derivative of 
$f(x)$ in $[1-c,\frac{1}{1-c}]$ by a constant $\alpha>0,$ and hence uniformly upper-bound the absolute values of all third derivatives of $h(x,y)$ in $[0,c]\times [0,c]$ by a constant $\beta>0,$ so that a second order Taylor expansion of $h(x_{k,i},y_{k,i})$ around $(x_{k,i},y_{k,i}) = (\bar{P}_i,\bar{P}_i)$ yields for all $i$ large enough such that $\frac{\bar{P}_i}{\bar{P}_1} < c$ and hence, by \eqref{eq:upperbound_for_prob_Mi}, \eqref{eq:upperbound_for_prob_Ni}, $x_{k,i}\leq c,$ $y_{k,i}\leq c,$
\begin{align}
    &\sum_k w_k (1-y_{k,i})f\left(\frac{1-x_{k,i}}{1-y_{k,i}}\right)\nonumber\\
    & = -f^\prime(1)\sum_k w_k (x_{k,i} - y_{k,i}) + \frac{1}{2}\frac{f^{\prime\prime}(1)}{1-\bar{P}_i}\sum_k w_k (x_{k,i} - y_{k,i})^2
    + R_i\nonumber\\
    & = \frac{1}{2}\frac{f^{\prime\prime}(1)}{1-\bar{P}_i}\sum_k w_k (x_{k,i} - y_{k,i})^2
    + R_i,\label{eq:second_order_taylor_f_div}
\end{align}
where in the last equality we used that $\sum_k w_k x_{k,i} = \sum_k w_k y_{k,i}.$ 

By \eqref{eq:upperbound_for_prob_Mi}
\begin{align}
    \left| \PS{c_1\sim c_2}{}{M_i} - \bar{P}_i \right| = \bar{P}_i\left| \frac{\PS{c_1\sim c_2}{}{M_i}}{\bar{P}_i} - 1 \right|\leq \bar{P}_i\max\{\frac{1}{\bar{P}_1}, 1\}
    = \mathcal{O}(\bar{P}_i)
\end{align}
and similarly, by \eqref{eq:upperbound_for_prob_Ni}
\begin{align}
    \left| \PS{c_1\sim c_2}{}{N_i} - \bar{P}_i \right| 
    = \mathcal{O}(\bar{P}_i),
\end{align}
such that the rest term $R_i$ is of order 
\begin{align}
    R_i = \mathcal{O}\left(\beta \max_k\left\{|x_{k,i} - \bar{P}_i|^3 + |y_{k,i} - \bar{P}_i|^3\right\}\right) = \mathcal{O}(\bar{P}_i^3),\label{eq:third_order_rest_term_upperbound}  
\end{align}
and the squared difference $(x_{k,i} - y_{k,i})^2$ is of order
\begin{align}
(x_{k,i} - y_{k,i})^2 = \mathcal{O}(\bar{P}_i^2).\label{eq:upperbound_for_squared_prob_diff_Mi_Ni}   
\end{align}
Plugging \eqref{eq:third_order_rest_term_upperbound}  and \eqref{eq:upperbound_for_squared_prob_diff_Mi_Ni} into \eqref{eq:second_order_taylor_f_div} yields
\begin{align}\label{eq:misc_order_2_term}
    \sum_k w_k (1-y_{k,i})f\left(\frac{1-x_{k,i}}{1-y_{k,i}}\right) = \mathcal{O}\left(\bar{P}_i^{2}\right),
\end{align}
and therefore, by \eqref{eq:expand_T_D_f} and \eqref{eq:misc_order_2_term}, 
\begin{align}
    \frac{\mathcal D_{f,\tau}(M_i,N_i;\CS_1,\CS_2)}{\bar P_i} & =  \mathcal D_{f,\sigma}(M_i,N_i;\CS_1,\CS_2) + \frac{1}{\bar{P}_i} \mathcal{O}(\bar{P}_i^{2})\nonumber\\
    &\to \mathcal D_{f,\sigma}(M_1,N_1;\CS_1,\CS_2)\text{ as }i\to\infty,
\end{align}
since $\mathcal D_{f,\sigma}(M_i,N_i;\CS_1,\CS_2) = \mathcal D_{f,\sigma}(M_1,N_1;\CS_1,\CS_2)$ for all $i,$ and $\bar P_i \to 0$.

Finally, \eqref{eq:T_D_f_to_S_D_f_nsim} follows by a similar argument when reversing the roles of $\PS{c_1\sim c_2}{}{M}$ and $1 - \PS{c_1\sim c_2}{}{M}.$

To show (iv), we note that under $P^N\equiv \bar{P}$ and with generator $f(x)=0.5(x-1)^2$
\begin{equation}
\begin{aligned}
   \mathcal D_{f,\sigma}(M,M_{ER};\CS_1,\CS_2) &= \frac{1}{2}\sum_{c_1 c_2} \frac{\pi_{c_1}\pi_{c_2}\bar{P}}{\bar{P}}\left(\frac{\PS{c_1\sim c_2}{}{M}}{\bar{P}} -1\right)^2\\
    &=  \frac{1}{2\bar{P}^2}\sum_{c_1 c_2} \pi_{c_1}\pi_{c_2}\left(\PS{c_1\sim c_2}{}{M} -\bar{P}\right)^2
\end{aligned}
\end{equation}
and 
\begin{equation}
\begin{aligned}  
    \mathcal D_{f,\xi}(M,M_{ER};\CS_1,\CS_2) & = \frac{1}{2}\sum_{c_1 c_2} \frac{\pi_{c_1}\pi_{c_2}(1-\bar{P})}{1-\bar{P}}\left(\frac{1-\PS{c_1\sim c_2}{}{M}}{1-\bar{P}} -1\right)^2\\
    &=  \frac{1}{2(1-\bar{P})^2}\sum_{c_1 c_2} \pi_{c_1}\pi_{c_2}\left(\PS{c_1\sim c_2}{}{M} -\bar{P}\right)^2
\end{aligned}
\end{equation}
hence by \eqref{eq:T_D_f_eq_sum_of_S_D_fs}
\begin{equation}
\begin{aligned}
    \mathcal D_{f,\tau}(M,M_{ER};\CS_1,\CS_2) &= \frac{1}{2\bar{P}}\sum_{c_1 c_2} \pi_{c_1}\pi_{c_2}\left(\PS{c_1\sim c_2}{}{M} -\bar{P}\right)^2\\
    &+ \frac{1}{2(1-\bar{P})}\sum_{c_1 c_2} \pi_{c_1}\pi_{c_2}\left(\PS{c_1\sim c_2}{}{M} -\bar{P}\right)^2\\
    &= \frac{1}{2\bar{P}(1-\bar{P})}\sum_{c_1 c_2} \pi_{c_1}\pi_{c_2}\left(\PS{c_1\sim c_2}{}{M} -\bar{P}\right)^2,
\end{aligned}
\end{equation}
which completes the proof.
\end{proof}

\begin{proof}[Proof of Corollary~\ref{corollary:normalized_tau_sigma}]
For part (i), \eqref{eq:constant_ratio_of_Pmi_to_barP} and the assumption $\pi^{M_i}=\pi\;\forall i $ imply that the marginal geosocial connectivities satisfy
\begin{equation}
    \frac{\bar P_{rb}^{M_i}}{\bar P_i}
    = \frac{\bar P_{rb}^{M_1}}{\bar P_1},
    \qquad \forall r,b,\;\forall i,
\end{equation}
and 
\begin{equation}
    \frac{\bar P_b^{M_i}}{\bar P_i}
    = \frac{\bar P_b^{M_1}}{\bar P_1},
    \qquad \forall b,\;\forall i.
\end{equation}
It therefore follows from the definitions of $M_{DC}(M_i)$ and $M_{DSI}(M_i)$ that
\begin{equation}
    \frac{\PS{r_1b_1\sim b_2}{}{M_{DC}(M_i)}}{\bar P_i}
    = \frac{\PS{r_1b_1\sim b_2}{}{M_{DC}(M_1)}}{\bar P_1},
\end{equation}
and
\begin{equation}
    \frac{\PS{r_1b_1\sim b_2}{}{M_{DSI}(M_i)}}{\bar P_i}
    = \frac{\PS{r_1b_1\sim b_2}{}{M_{DSI}(M_1)}}{\bar P_1}.
\end{equation}
Thus, the sequences $M_i$ and $M_{DC}(M_i)$, and the sequences $M_{DSI}(M_i)$ and $M_{DC}(M_i)$ satisfy the conditions of Lemma~\ref{lemma:T_D_F_to_S_D_F}(ii).
Applying \eqref{eq:T_D_f_to_S_D_f} to the numerator and denominator of $S_{f,\tau}(M_i)$ therefore gives
\begin{equation}
\begin{aligned}
    S_{f,\tau}(M_i)
    &= \frac{\mathcal D_{f,\sigma}(M_1)+\mathcal O(\bar P_i)}
    {\mathcal D_{f,\sigma}(M_{DSI}(M_1))+\mathcal O(\bar P_i)}\\
    &= \frac{\mathcal D_{f,\sigma}(M_1)}
    {\mathcal D_{f,\sigma}(M_{DSI}(M_1))}
    +\mathcal O(\bar P_i)\\
    &= S_{f,\sigma}(M_1)+\mathcal O(\bar P_i),
    \quad i\to\infty,
\end{aligned}
\end{equation}
which proves part (i).

For part (ii), applying \eqref{eq:T_D_f_is_twice_dissimilarity} to $M$ and $M_{DC}(M)$ gives
\begin{equation}
    \mathcal D_{g,\tau}(M)=2\bar P\,\mathcal D_{g,\sigma}(M) = 2(1 - \bar P)\,\mathcal D_{g,\xi}(M)
\end{equation}
and applying it to $M_{DSI}(M)$ and $M_{DC}(M)$ gives
\begin{equation}
    \mathcal D_{g,\tau}(M_{DSI}(M)) = 2\bar P\,\mathcal D_{g,\sigma}(M_{DSI}(M)) = 2(1- \bar P)\,\mathcal D_{g,\xi}(M_{DSI}(M)).
\end{equation}
Hence
\begin{equation}
    S_{g,\tau}(M)=S_{g,\sigma}(M) = S_{g,\xi}(M).
\end{equation}

Finally, part (iii) follows from Lemma~\ref{lemma:T_D_F_to_S_D_F}(iv) as the multiplicative factors $\frac{1-\bar P}{\bar P}$ and $\frac{\bar P}{1-\bar P}$ cancel upon normalization.
\end{proof}

\subsection{Homophily, heterophily, and network modularity}

\begin{proof}[Proof of lemma~\ref{lemma:homophily_quantification}]
    We first show \eqref{eq:alpha_m_beta_eq_Q_for_sigma}. By assumption $\pi^N=\pi^M$ and $\bar{P}^M=\bar{P}^N$ and hence $\PPS{rb_1\sim b_2}{}{M} > \PPS{rb_1\sim b_2}{}{N}\Leftrightarrow \sigma^M(r b_1, b_2) > \sigma^N(r b_1, b_2),$ and likewise for $<$ instead of $>$. Consequently
\begin{align}
    &2\Phi_{g,\sigma}(M,N;\RS\times\BS,\BS)  = 2\mathcal{A}_{g,\sigma}(M,N;\RS\times\BS,\BS) - 2\mathcal{B}_{g,\sigma}(M,N;\RS\times\BS,\BS) \nonumber\\
    &= \sum_{r,b}\left| \sigma^M(rb,b) - \sigma^N(rb,b)  \right|\mathbb{1}_{\{ \sigma^M(rb,b) \geq \sigma^N(rb,b)\}}\nonumber\\
    &\hspace{2em} + \sum_{r}\sum_{b_1\neq b_2}\left| \sigma^M(rb_1,b_2) - \sigma^N(rb_1,b_2)  \right|\mathbb{1}_{\{ \sigma^M(rb_1,b_2) \leq \sigma^N(rb_1,b_2)\}}\nonumber\\
    &\hspace{2em} - \sum_{r,b}\left| \sigma^M(rb,b) - \sigma^N(rb,b)  \right|\mathbb{1}_{\{ \sigma^M(rb,b) \leq \sigma^N(rb,b)\}}\nonumber\\
    &\hspace{2em} - \sum_{r}\sum_{b_1\neq b_2}\left| \sigma^M(rb_1,b_2) - \sigma^N(rb_1,b_2)  \right|\mathbb{1}_{\{ \sigma^M(rb_1,b_2) \geq \sigma^N(rb_1,b_2)\}}\nonumber\\
    & = \sum_{r,b}\left( \sigma^M(rb,b) - \sigma^N(rb,b)  \right) - \sum_{r}\sum_{b_1\neq b_2}\left( \sigma^M(rb_1,b_2) - \sigma^N(rb_1,b_2)  \right)\nonumber\\
    & = 2 \sum_{r,b}\left( \sigma^M(rb,b) - \sigma^N(rb,b)\right) = 2Q(M,N;\RS\times\BS,\BS),
\end{align}
where we have used that 
\begin{align}
    \sum_{r,b}\left( \sigma^M(rb,b) - \sigma^N(rb,b)  \right) + \sum_r\sum_{b_1\neq b_2}\left( \sigma^M(rb_1,b_2) - \sigma^N(rb_1,b_2)  \right) = 0.
\end{align}

To show \eqref{eq:phi_tau_eq_2phi_sigma}, by the definition of $\mathcal{A}_{g,\tau}(M,N;\RS\times\BS,\BS)$ in \eqref{eq:tau_homophily_decomposition} 
\begin{align}
&\mathcal{A}_{g,\tau}(M,N;\RS\times\BS,\BS)\\
&=\pi(\mathcal I_{\mathrm{hom}})\DF{\tau^M(rb_1,b_2,e\mid (rb_1,b_2)\in\mathcal I_{\mathrm{hom}})}
{\tau^N(rb_1,b_2,e\mid (rb_1,b_2)\in\mathcal I_{\mathrm{hom}})}\nonumber\\
&=\frac12\sum_{(rb_1,b_2)\in\mathcal I_{\mathrm{hom}}}\sum_{e\in\{0,1\}}
\left|\tau^M(rb_1,b_2,e)-\tau^N(rb_1,b_2,e)\right|\nonumber\\
&=\sum_{(rb_1,b_2)\in\mathcal I_{\mathrm{hom}}}\pi_{rb_1}\pi_{b_2}
\left|\PS{rb_1\sim b_2}{}{M}-\PS{rb_1\sim b_2}{}{N}\right|\nonumber\\
&=2\bar{P} \frac{1}{2}\sum_{(rb_1,b_2)\in\mathcal I_{\mathrm{hom}}}
\left|\sigma^M(rb_1,b_2)-\sigma^N(rb_1,b_2)\right|\\
&= 2\bar{P}\mathcal{A}_{g,\sigma}(M,N;\RS\times\BS,\BS).
\end{align}
A similar calculation for $\mathcal{B}$ yields $\mathcal{B}_{g,\tau}(M,N;\RS\times\BS,\BS) = 2\bar{P}\mathcal{B}_{g,\sigma}(M,N;\RS\times\BS,\BS)$ and subtracting $\mathcal{B}$ from $\mathcal{A}$ yields \eqref{eq:phi_tau_eq_2phi_sigma}.

To show \eqref{eq:alpha1_to_alphaleq1}, we proceed as in the proof of lemma~\ref{lemma:T_D_F_to_S_D_F}, adopting the same notation.
Assumptions \eqref{eq:const_ratio_hom_lemma} and \eqref{eq:const_ratio_hom_lemma_null} imply
$\mathcal{I}\coloneqq\mathcal{I}_{\mathrm{hom}}(M_1,N_1)
=\mathcal{I}_{\mathrm{hom}}(M_i,N_i)$ for all $i\geq 1$.
Let $\mathcal{A}_{f,\tau}^{(i)} \coloneqq \mathcal{A}_{f,\tau}\left(M_i, N_i, \RS\times\BS,\BS\right)$ and $\mathcal{A}_{f,\sigma}^{(i)}\coloneqq\mathcal{A}_{f,\sigma}(M_i,N_i;\RS\times\BS,\BS).$ Noting that \eqref{eq:const_ratio_hom_lemma} and \eqref{eq:const_ratio_hom_lemma_null}
imply $\mathcal{A}_{f,\sigma}^{(i)} = \mathcal{A}_{f,\sigma}^{(1)}$ for all $i$, we obtain
    \begin{align}
        \mathcal{A}_{f,\tau}^{(i)}
        & = \sum_{k\in \mathcal{I}} w_k y_{k,i}f\left(\frac{x_{k,i}}{y_{k,i}}\right)
        + \sum_{k\in \mathcal{I}} w_k (1-y_{k,i})f\left(\frac{(1-x_{k,i})}{(1-y_{k,i})}\right)\nonumber\\
        &= \bar{P}_i \mathcal{A}_{f,\sigma}^{(1)} + \sum_{k\in \mathcal{I}} w_k (1-y_{k,i})f\left(\frac{(1-x_{k,i})}{(1-y_{k,i})}\right).
    \end{align}
A Taylor expansion yields
\begin{align}
    &\sum_{k\in \mathcal{I}} w_k (1-y_{k,i})f\left(\frac{(1-x_{k,i})}{(1-y_{k,i})}\right)\nonumber\\
    & = -f^{\prime}(1)\sum_{k\in \mathcal{I}} w_k ( x_{k,i} - y_{k,i} )
    +\frac{1}{2}\frac{f^{\prime\prime}(1)}{1-\bar{P}_i} \sum_{k\in \mathcal{I}} w_k ( x_{k,i} - y_{k,i} )^2 + R_i.
\end{align}
Now while $\sum_{k\in \mathcal{I}} w_k ( x_{k,i} - y_{k,i} )$ will in general not equal $0,$ the first term vanishes nevertheless because $f^{\prime}(1) = 0$ by our convention of choosing the representative $f$ satisfying \eqref{eq:sub_diff_centered_at_zero}. Following similar arguments as below \eqref{eq:second_order_taylor_f_div}, we obtain \eqref{eq:alpha1_to_alphaleq1}. \eqref{eq:beta1_to_betaleq1} follows similarly.
\end{proof}

\begin{proof}[Proof of Corollary~\ref{corollary:normalized_homophily_quantification}]
Equations~\eqref{eq:alpha_m_beta_eq_Q_for_sigma_normalized} and \eqref{eq:phi_tau_eq_2phi_sigma_normalized} follow from Lemma~\ref{lemma:homophily_quantification}(i) and the arguments in the proof of Corollary~\ref{corollary:normalized_tau_sigma}. Likewise, \eqref{eq:alpha1_to_alphaleq1_normalized}--\eqref{eq:phi1_to_phileq1_normalized} follow from Lemma~\ref{lemma:homophily_quantification}(ii) and the arguments in the proof of Corollary~\ref{corollary:normalized_tau_sigma}.

It remains to show \eqref{eq:recovery_of_network_ass}. 
Let $\sigma_{\RS\times\BS}^M(rb)$ and $\sigma_{\BS}^M(b)$ denote the corresponding marginal distributions under $\sigma^M$, with
\begin{equation}
\begin{aligned}
        \sigma_{\RS\times\BS}^M(rb)&=\sum_{r,b'}\sigma_{\RS\times\BS,\BS}^M(rb,b'),\\
        \sigma_{\BS}^M(b)&=\sum_r\sigma_{\RS\times\BS}^M(rb).
\end{aligned}
\end{equation}
For $M_{DSI}(M)$ and $M_{DC}(M)$, we have by \eqref{eq:def_M_DC_geo_egocentric} and \eqref{eq:def_M_DSI_geo_egocentric}
\begin{equation}
\begin{aligned}
    \sigma_{\RS\times\BS,\BS}^{M_{DSI}(M)}(rb_1,b_2)
    &=\sigma_{\RS\times\BS}^M(rb_1)\mathbb{1}_{\{b_1=b_2\}},\\
    \sigma_{\RS\times\BS,\BS}^{M_{DC}(M)}(rb_1,b_2)
    &=\sigma_{\RS\times\BS}^M(rb_1)\sigma_{\BS}^M(b_2).
\end{aligned}
\end{equation}
Hence, for Total Variation distance with generator $g(x)=\frac{1}{2}|x-1|$,
\begin{equation}
\begin{aligned}
&\mathcal{D}_{g,\sigma}\!\left(M_{DSI}(M),M_{DC}(M);\RS\times\BS,\BS\right)\\
&=\frac{1}{2}\sum_{r,b_1,b_2}\sigma_{\RS\times\BS}^M(rb_1)
\left|\mathbb{1}_{\{b_1=b_2\}}-\sigma_{\BS}^M(b_2)\right|\\
&=\frac{1}{2}\sum_{r,b}\sigma_{\RS\times\BS}^M(rb)\left(1-\sigma_{\BS}^M(b)\right)
+\frac{1}{2}\sum_r\sum_{b_1\neq b_2}\sigma_{\RS\times\BS}^M(rb_1)\sigma_{\BS}^M(b_2)\\
& = \frac{1}{2}\left(1 -\sum_b \left(\sigma_{\BS}^M(b)\right)^2 \right)
+ \frac{1}{2}\left(\sum_{b_1, b_2}  \sigma_{\BS}^M(b_1)\sigma_{\BS}^M(b_2) - \sum_b\left(\sigma_{\BS}^M(b)\right)^2 \right)\\
&=1-\sum_b\left(\sigma_{\BS}^M(b)\right)^2.
\end{aligned}
\end{equation}
For the maximum likelihood model $\hat M$ we have $\sigma_{\BS}^{\hat M}(b_2) = s_b$ with $s_b$ as defined in \eqref{eq:def_net_mod_ass}, and \eqref{eq:recovery_of_network_ass} follows from \eqref{eq:alpha_m_beta_eq_Q_for_sigma} and \eqref{eq:def_net_mod_ass} together with $Q = Q(\hat M, M_{DC}(\hat M); \RS\times \BS, \BS).$
\end{proof}


\begin{proof}[Proof of Lemma~\ref{lemma:homophily_quantification_2}]
To show \eqref{eq:lemma_hom_q_ineq_0}, write
\begin{equation}
    \begin{aligned}
        f\left(\frac{q+\epsilon}{q}\right) & =  f\left(\frac{q+2\epsilon + q}{2q}\right)< \frac{1}{2} f\left(\frac{q+2\epsilon}{q}\right)
        + \frac{1}{2} f\left(\frac{q}{q}\right)\\
        &= \frac{1}{2}  f\left(\frac{q+2\epsilon}{q}\right),\label{eq:proof_lemma_hom_convexity}
    \end{aligned}
\end{equation}
where the inequality uses the assumed strict convexity of $f$ and the last equality follows from $f(1)=0.$

A second order Taylor expansion around $1$ yields, using $f(1)=f'(1) = 0$ by our choice of representative,
\begin{equation}
    \begin{aligned}
         f\left(\frac{q-\delta}{q}\right) &= f''(1)\frac{\delta^2}{2 q^2} + o\left(\frac{\delta^2}{q^2}\right)\\
         f\left(\frac{q+2\delta}{q}\right) &= f''(1)\frac{2\delta^2}{q^2} + o\left(\frac{4\delta^2}{q^2}\right)
    \end{aligned}
\end{equation}
and therefore
\begin{align}
    f\left(\frac{q+2\delta}{q}\right) - 2f\left(\frac{q-\delta}{q}\right) = f''(1)\frac{\delta^2}{q^2} + o(\delta^2),
\end{align}
and \eqref{eq:lemma_hom_q_ineq_1} and \eqref{eq:lemma_hom_q_ineq_2} follow from $f''(1)>0.$

If $f''$ is strictly increasing on $(0,\infty)$, integration over $f''$ yields $f'(1)-f'(1-s)< f'(1+s) - f'(1),$ and, since $f'(1)=0$ by our choice of representative, $-f'(1-s)<f'(1+s).$ Integration over $f'$ then yields $f(1-t)<f(1+t).$  This shows \eqref{eq:lemma_hom_q_ineq_3}; \eqref{eq:lemma_hom_q_ineq_4} follows by a similar argument. Note that strict monotonicity of $f''$ together with convexity of $f$ implies that $f''(1)>0,$ such that \eqref{eq:lemma_hom_q_ineq_1} and \eqref{eq:lemma_hom_q_ineq_2} hold in both cases.

Finally, if $f$ is strictly convex and symmetric, then \eqref{eq:lemma_hom_q_ineq_5} follows directly from symmetry around $1$, and \eqref{eq:lemma_hom_q_ineq_1} and \eqref{eq:lemma_hom_q_ineq_2}  follow by \eqref{eq:proof_lemma_hom_convexity} and using \eqref{eq:lemma_hom_q_ineq_5}.

\end{proof}


\begin{proof}[Proof of Lemma~\ref{lemma:recovery_of_relative_diversity}]
Recall that by \eqref{eq:connectivity_under_M_phi_gecentric} and \eqref{eq:tilde_pi_under_M_RI},
\begin{equation}
    \PS{rb_1\sim b_2}{}{M_{RI}}=\frac{\pi_{b_2\mid r}}{\pi_{b_2}}\bar P.
\end{equation}

Moreover, all geosocial coordinates have the same marginal connectivity under $M_{RI}$, and hence $M_{DC}(M_{RI})=M_{ER}$. Therefore,
\begin{equation}
\begin{aligned}
    Q(M_{RI})
    &=\sum_r\sum_b\left(\sigma_{\RS\times\BS,\BS}^{M_{RI}}(rb,b)
    -\sigma_{\RS\times\BS,\BS}^{M_{DC}(M_{RI})}(rb,b)\right)\\
    &=\sum_r\sum_b\left(
    \pi_r\pi_{b\mid r}\pi_b
    \frac{\pi_{b\mid r}\bar P}{\pi_b\bar P}
    -\pi_r\pi_{b\mid r}\pi_b\frac{\bar P}{\bar P}\right)\\
    &=\sum_r\sum_b\pi_r\left(\pi_{b\mid r}^2-\pi_{b\mid r}\pi_b\right)\\
    &=\sum_r\sum_b\pi_r\left(\pi_{b\mid r}-\pi_b\right)^2
    =R,
\end{aligned}
\end{equation}
where in the penultimate equality we have used
\begin{equation}
    \sum_r\pi_r\pi_{b\mid r}\pi_b =\sum_r\pi_r\pi_b^2.
\end{equation}

Furthermore, since all geosocial coordinates have the same marginal connectivity under $M_{RI}$, $M_{DSI}(M_{RI})=M_{SI}$. Hence, by \eqref{eq:normalization_const_H_D_C},
\begin{equation}
    \mathcal D_{g,\sigma}\!\left(M_{DSI}(M_{RI})\right)
    =\sum_b\pi_b(1-\pi_b).
\end{equation}
Combining this with \eqref{eq:Q_contains_R_as_special_case} and \eqref{eq:alpha_m_beta_eq_Q_for_sigma_normalized} gives
\begin{equation}
    \phi_{g,\sigma}(M_{RI})=\rho_A(M_{RI})
    =\frac{R}{\sum_b\pi_b(1-\pi_b)}
    =\underline R,
\end{equation}
which proves \eqref{eq:rho_A_contains_R_as_special_case}.
\end{proof}


\subsection{Regional and group decompositions}

\begin{proof}[Proof of Lemma~\ref{lemma:regional_decomposition_of_H}]
    To show \eqref{eq:regional_decomposability_sb_rbb}, consider random variables $X\sim \sigma^M(rb_1,b_2),$ $X^{(S)}\sim \sigma^M(s,b_2),$ and $Y\sim \sigma^N(rb_1,b_2),$ $Y^{(S)}\sim \sigma^N(s,b_2),$
then because of the chain rule for the KL divergence
\begin{align}
    \mathcal D_{h,\sigma}(M) &= \DKL{P_X}{P_Y} = \DKL{P_{X,X^{(S)}}}{P_{Y,Y^{(S)}}}\nonumber\\
    &= \DKL{P_{X^{(S)}}}{P_{Y^{(S)}}} + \EE_{P_{X^{(S)}}}\DKL{P_{X\mid X^{(S)}}}{P_{Y\mid Y^{(S)}}}\nonumber\\
    & = \DKL{\sigma^M(s,b_2)}{\sigma^N(s,b_2)}\nonumber\\
    &+ \EE_{\sigma^M(s,b_2)}\DKL{\sigma^M( rb_1, b_2\mid s,b_2)}{\sigma^N( rb_1, b_2\mid s,b_2)},
\end{align}
yielding the between component $\DKL{\sigma^M(s,b_2)}{\sigma^N(s,b_2)}$ in \eqref{eq:regional_decomposability_sb_rbb_asKL_div_between} and the sum over within components
\begin{align}
    &\EE_{\sigma^M(s,b_2)}\DKL{\sigma^M( rb_1, b_2\mid s,b_2)}{\sigma^N( rb_1, b_2\mid s,b_2)}\nonumber\\
    & = \sum_s\sum_{b_2}\pi_s\pi_{b_2}\frac{\PS{s\sim b_2}{}{M}}{\bar{P}}\nonumber\\
    &\hspace{5em}\cdot
    \left[
    \sum_{r\in s}\sum_{b_1}\pi_{r\mid s}\pi_{b_1\mid r}\frac{\PS{rb_1\sim b_2}{}{M}}{\PS{s\sim b_2}{}{M}}
    \logg{
    \frac{\pi_{r\mid s}\pi_{b_1\mid r}\frac{\PS{rb_1\sim b_2}{}{M}}{\PS{s\sim b_2}{}{M}}}
    {\pi_{r\mid s}\pi_{b_1\mid r}\frac{\PS{rb_1\sim b_2}{}{N}}{\PS{s\sim b_2}{}{N}}}
    }
    \right]\nonumber\\
    & = \sum_s \pi_s\frac{\bar{P}_s^M}{\bar{P}}\left[
    \sum_{r\in s}\sum_{b_1}\sum_{b_2} \pi_{r\mid s}\pi_{b_1\mid r}\pi_{b_2} \frac{\PS{rb_1\sim b_2}{}{M}}{\bar{P}_s^M}
    \right.\nonumber\\
    &\hspace{5em}\cdot\left.
    \logg{
    \frac{\pi_{b_2}\pi_{r\mid s}\pi_{b_1\mid r}\frac{\PS{rb_1\sim b_2}{}{M}}{\bar{P}_s^M}}
    {\pi_{b_2}\frac{\PS{s\sim b_2}{}{M}}{\bar{P}_s^M}
    \pi_{r\mid s}\pi_{b_1\mid r}\frac{\PS{rb_1\sim b_2}{}{N}}{\PS{s\sim b_2}{}{N}}}
    }
    \right]\label{eq:regional_decomposition_expand}\\
    & = \sum_s \pi_s\frac{\bar{P}_s^M}{\bar{P}} \DKL{\sigma^M(rb_1, b_2\mid s)}{\sigma^{N\mid M}( rb_1, b_2\mid s)}
\end{align}
giving the within terms \eqref{eq:regional_decomposability_sb_rbb_asKL_div_within}.

We now consider the special case $M=M_{RI}$ and first note that $N=M_{DC}(M_{RI})=M_{ER}$. We recall that by \eqref{eq:connectivity_under_M_phi_gecentric} and \eqref{eq:tilde_pi_under_M_RI},
\begin{align}
    \PS{rb_1\sim b_2}{}{M_{RI}} = \frac{\pi_{b_2\mid r}}{\pi_{b_2}}\bar P
\end{align}
and hence also
\begin{equation}
\begin{aligned}
    \PS{s \sim b_2}{}{M_{RI}}
    & = \sum_{r\in s}\pi_{r\mid s} \sum_{b_1}\pi_{b_1\mid r}
    \PS{rb_1\sim b_2}{}{M_{RI}}
    = \sum_{r\in s}\pi_{r\mid s} \sum_{b_1}\pi_{b_1\mid r} \frac{\pi_{b_2\mid r}}{\pi_{b_2}}\bar{P}\\
    & = \frac{\pi_{b_2\mid s}}{\pi_{b_2}}\bar{P} .
\end{aligned}
\label{eq:P_M_RI_s_b2}
\end{equation}
For the null $N=M_{ER}$ we have constant connectivity
\begin{equation}
    \PS{rb_1\sim b_2}{}{N}=\PS{s\sim b_2}{}{N}=\bar P.
    \label{eq:P_M_ER_s_b2}
\end{equation}
Using \eqref{eq:P_M_RI_s_b2} and \eqref{eq:P_M_ER_s_b2}, we now show \eqref{eq:between_component_equals_traditional_between_component}: 
by \eqref{eq:regional_decomposability_sb_rbb_asKL_div_between} the between component for $M=M_{RI}$ is
\begin{align}
    \mathcal D_{h,\sigma}^{\mathcal S\leftrightarrow}(M_{RI})&
    = \sum_s \sum_{b_2} \pi_s \pi_{b_2}
    \frac{\PS{s\sim b_2}{}{M_{RI}}}{\bar{P}}
    \logg{\frac{\PS{s\sim b_2}{}{M_{RI}}}{\bar{P}}}\nonumber\\
    & = \sum_s \sum_{b_2}\pi_s \pi_{b_2}\frac{\pi_{b_2\mid s}}{\pi_{b_2}}\log\left(\frac{\pi_{b_2\mid s}}{\pi_{b_2}}\right).
\end{align}
Finally, for \eqref{eq:within_component_equals_traditional_within_component} and \eqref{eq:barPs_equals_barP_under_MRI}, we note that
\begin{align}
    \bar{P}_s^{M_{RI}} & = \PPS{i\sim j\mid r(i)\in s}{}{M_{RI}} = \sum_{b_2}\pi_{b_2} \PS{s\sim b_2}{}{M_{RI}}
     = \bar{P} \sum_{b_2}\pi_{b_2}\frac{\pi_{b_2\mid s}}{\pi_{b_2}}\nonumber\\
     & = \bar{P},
\end{align}
and \eqref{eq:P_M_ER_s_b2} together with
\eqref{eq:regional_decomposability_sb_rbb_asKL_div_within}, \eqref{eq:def_sgima_M_within_s}, and
\eqref{eq:def_sgima_N_given_M_within_s} yield as within component for region $s$
\begin{align}
    &\mathcal D_{h,\sigma}^{(s)}(M_{RI})\nonumber\\
    & =\sum_{r\in s}\sum_{b_1,b_2}\pi_{r\mid s}\pi_{b_1\mid r}\pi_{b_2}
    \frac{\PS{rb_1\sim b_2}{}{M_{RI}}}{\bar{P}_s^{M_{RI}}}
    \logg{\frac{\PS{rb_1\sim b_2}{}{M_{RI}}}{\PS{s\sim b_2}{}{M_{RI}}}}\nonumber\\
    & = \sum_{r\in s}\sum_{b_1,b_2}\pi_{r\mid s}\pi_{b_1\mid r}\pi_{b_2}
    \frac{\pi_{b_2\mid r}}{\pi_{b_2}}\logg{
    \frac{\frac{\pi_{b_2\mid r}}{\pi_{b_2}}}{\frac{\pi_{b_2\mid s}}{\pi_{b_2}}}
    }\nonumber\\
    & = \sum_{r\in s}\sum_{b_2}\pi_{r\mid s}\pi_{b_2\mid s}
    \frac{\pi_{b_2\mid r}}{\pi_{b_2\mid s}}\logg{
    \frac{\pi_{b_2\mid r}}{\pi_{b_2\mid s}}
    }.
\end{align}
\end{proof}


\begin{proof}[Proof of Lemma~\ref{lemma:group_decomposition_of_H}]
    By the chain rule for the Kullback-Leibler divergence,
    \begin{align}
        &\mathcal D_{h,\sigma}(M) = \DKL{\sigma^M(rb_1,b_2)}{\sigma^N(rb_1,b_2)}\nonumber\\
        &= \DKL{\sigma^M(ra_1,a_2)}{\sigma^N( ra_1, a_2)}\nonumber\\
        &+ \EE_{\sigma^M(ra_1,a_2)}
        \DKL{\sigma^M(rb_1,b_2\mid ra_1,a_2)}{\sigma^N(rb_1,b_2\mid ra_1,a_2)},
    \end{align}
yielding the between component $\DKL{\sigma^M(ra_1,a_2)}{\sigma^N( ra_1, a_2)}$ and the sum over within components
\begin{align}
    &\EE_{\sigma^M(ra_1,a_2)}
        \DKL{\sigma^M(rb_1,b_2\mid ra_1,a_2)}{\sigma^N(rb_1,b_2\mid ra_1,a_2)}\nonumber\\
    & = \sum_{a_2}\sum_{a_1}\sum_{r}\pi_{a_2}\pi_{a_1}\pi_{r\mid a_1}\frac{\PS{ra_1\sim a_2}{}{M}}{\bar{P}}\nonumber\\
    &\hspace{3em}\cdot\left[
    \sum_{b_1,b_2}\pi_{b_1\mid a_1,r}\pi_{b_2\mid a_2}\frac{\PS{rb_1\sim b_2}{}{M}}{\PS{ra_1\sim a_2}{}{M}}
    \logg{
    \frac{\pi_{b_1\mid a_1,r}\pi_{b_2\mid a_2}\frac{\PS{rb_1\sim b_2}{}{M}}{\PS{ra_1\sim a_2}{}{M}}}{\pi_{b_1\mid a_1,r}\pi_{b_2\mid a_2}\frac{\PS{rb_1\sim b_2}{}{N}}{\PS{ra_1\sim a_2}{}{N}}}
    }
    \right] \label{eq:sum_over_within_group_components_expanded}\\
    & = \sum_{a_2}\pi_{a_2}\frac{\bar{P}_{a_2}^M}{\bar{P}}\left[\sum_{a_1}\sum_{r}\pi_{a_1}\pi_{r\mid a_1}\sum_{b_1,b_2}\pi_{b_1\mid a_1,r}\pi_{b_2\mid a_2}\frac{\PS{rb_1\sim b_2}{}{M}}{\bar{P}_{a_2}^M}\right.\nonumber\\
    &\hspace{3em}\left.\cdot
    \logg{
    \frac{\pi_{a_1}\pi_{r\mid a_1}\pi_{b_1\mid a_1,r}\pi_{b_2\mid a_2}\frac{\PS{rb_1\sim b_2}{}{M}}{\bar{P}_{a_2}^M}}
    {\pi_{a_1}\pi_{r\mid a_1}\frac{\PS{ra_1\sim a_2}{}{M}}{\bar{P}_{a_2}^M}\pi_{b_1\mid a_1,r}\pi_{b_2\mid a_2}\frac{\PS{rb_1\sim b_2}{}{N}}{\PS{ra_1\sim a_2}{}{N}}}
    }\right] \label{eq:read_out_group_a_component}\\
    & = \sum_{a_1,a_2}\pi_{a_1}\pi_{a_2}\frac{\PS{a_1\sim a_2}{}{M}}{\bar{P}}\left[\sum_{r}\pi_{r\mid a_1}\sum_{b_1,b_2}\pi_{b_1\mid a_1,r}\pi_{b_2\mid a_2}\frac{\PS{rb_1\sim b_2}{}{M}}{\PS{a_1\sim a_2}{}{M}}\right.\nonumber\\
    &\hspace{3em}\left.\cdot
    \logg{
    \frac{\pi_{r\mid a_1}\pi_{b_1\mid a_1,r}\pi_{b_2\mid a_2}\frac{\PS{rb_1\sim b_2}{}{M}}{\PS{a_1\sim a_2}{}{M}}}
    {\pi_{r\mid a_1}\frac{\PS{ra_1\sim a_2}{}{M}}{\PS{a_1\sim a_2}{}{M}}\pi_{b_1\mid a_1,r}\pi_{b_2\mid a_2}\frac{\PS{rb_1\sim b_2}{}{N}}{\PS{ra_1\sim a_2}{}{N}}}
    }\right]\label{eq:read_out_pair_a1_a2_component}, 
\end{align}
where from \eqref{eq:read_out_group_a_component} we can read out the group specific within components \eqref{eq:group_decomposability_raa_rbb_asKL_div_within_a_2} and from \eqref{eq:read_out_pair_a1_a2_component} the pair specific within components \eqref{eq:group_decomposability_raa_rbb_asKL_div_within_a_1a_2}.

We now consider the special case $M=M_{RI}$ and first note that the null $N=M_{DC}(M_{RI})=M_{ER}$ has constant connectivity
\begin{align}
    \PS{rb_1\sim b_2}{}{N} = \PS{ra_1\sim a_2}{}{N} = \bar P.
        \label{eq:P_M_ER_a1_a2}
\end{align}
To show \eqref{eq:between_group_component_of_triaditional_H_recovered_for_MRI_MER}, recall that by \eqref{eq:connectivity_under_M_phi_gecentric} and \eqref{eq:tilde_pi_under_M_RI}
\begin{align}
    \PS{rb_1\sim b_2}{}{M_{RI}} = \frac{\pi_{b_2\mid r}}{\pi_{b_2}}\bar P
\end{align}
and therefore
\begin{align}
    \PS{ra_1\sim a_2}{}{M_{RI}} & = \sum_{b_1\in a_1}\sum_{b_2\in a_2}\PS{rb_1\sim b_2}{}{M_{RI}}\pi_{b_1\mid r,a_1}\pi_{b_2\mid a_2}\nonumber\\
    & = \sum_{b_1\in a_1}\pi_{b_1\mid r,a_1}\sum_{b_2\in a_2} \frac{\bar{P}\pi_{b_2\mid r}}{\pi_{b_2}} \frac{\pi_{b_2}}{\pi_{a_2}}
     = \frac{\bar{P}}{\pi_{a_2}} \sum_{b_2\in a_2} \pi_{b_2\mid r}\nonumber\\
   & = \frac{\bar{P}}{\pi_{a_2}} \pi_{a_2\mid r} .
   \label{eq:P_M_RI_a1_a2}
\end{align}
Using \eqref{eq:P_M_ER_a1_a2} and \eqref{eq:P_M_RI_a1_a2}, by \eqref{eq:group_decomposability_raa_rbb_asKL_div_between} the between component simplifies to
\begin{align}
    & \mathcal D_{h,\sigma}^{\mathcal A\leftrightarrow}(M_{RI})\nonumber\\
   &\hspace{3em}= \sum_r\sum_{a_1,a_2}\pi_r\pi_{a_1\mid r}\pi_{a_2}
    \frac{\PS{ra_1\sim a_2}{}{M_{RI}}}{\bar{P}}\logg{\frac{\PS{ra_1\sim a_2}{}{M_{RI}}}{\bar{P}}}\nonumber\\
    &\hspace{3em}= \sum_r \sum_{a_2} \pi_r \pi_{a_2} \frac{\pi_{a_2\mid r}}{\pi_{a_2}} \logg{\frac{\pi_{a_2\mid r}}{\pi_{a_2}}}.
\end{align}
Finally, for \eqref{eq:within_group_component_of_triaditional_H_recovered_for_MRI_MER} and  \eqref{eq:within_group_component_weights_of_triaditional_H_recovered_for_MRI_MER}, note that
\begin{align}
    \bar{P}_{a_2}^{M_{RI}} & = \PPS{i\sim j\mid b(j)\in a_2}{}{M_{RI}}
    = \sum_{r}\sum_{a_1}\PS{ra_1\sim a_2}{}{M_{RI}}\pi_{r}\pi_{a_1\mid r}\nonumber\\
    & = \sum_{r}\left(\sum_{a_1}\pi_{a_1\mid r}\right) \frac{\bar{P}}{\pi_{a_2}} \pi_{a_2\mid r} \pi_{r}
    = \bar{P}
\end{align}
and \eqref{eq:P_M_ER_a1_a2} together with \eqref{eq:group_decomposability_raa_rbb_asKL_div_within_a_2}, \eqref{eq:sigma_M_group_within}, and \eqref{eq:sigma_N_given_M_group_within} yield the within component for group $a_2$
\begin{align}
    &\mathcal D_{h,\sigma}^{(a_2)}(M_{RI})\nonumber\\
    &\hspace{3em} = \sum_{a_1}\sum_{r}\pi_{a_1}\pi_{r\mid a_1}\sum_{b_1,b_2}\pi_{b_1\mid r,a_1}\pi_{b_2\mid a_2}
    \frac{\PS{rb_1\sim b_2}{}{M_{RI}}}{\bar{P}_{a_2}^{M_{RI}}}
    \logg{\frac{\PS{rb_1\sim b_2}{}{M_{RI}}}{\PS{ra_1\sim a_2}{}{M_{RI}}}}\nonumber\\
    &\hspace{3em} = \sum_r \pi_r\sum_{a_1}  \left(\pi_{a_1\mid r} \sum_{b_1\in a_1}\pi_{b_1\mid r,a_1}\right)\sum_{b_2\in a_2}\pi_{b_2\mid a_2}
    \frac{\pi_{b_2\mid r}}{\pi_{b_2}}\logg{\frac{\frac{\pi_{b_2\mid r}}{\pi_{b_2}}}{\frac{\pi_{a_2\mid r}}{\pi_{a_2}}}}\nonumber\\
    &\hspace{3em} = \sum_r \pi_r \sum_{b_2\in a_2}\pi_{b_2\mid a_2}
    \frac{\pi_{a_2\mid r}}{\pi_{a_2}}\frac{\frac{\pi_{b_2\mid r}}{\pi_{b_2}}}{\frac{\pi_{a_2\mid r}}{\pi_{a_2}}}\logg{\frac{\frac{\pi_{b_2\mid r}}{\pi_{b_2}}}{\frac{\pi_{a_2\mid r}}{\pi_{a_2}}}}\nonumber\\
    &\hspace{3em} = \sum_r \pi_{r\mid a_2} \sum_{b_2\in a_2} \pi_{b_2\mid a_2}
    \frac{\pi_{b_2\mid r,a_2}}{\pi_{b_2\mid a_2}}\logg{\frac{\pi_{b_2\mid r,a_2}}{\pi_{b_2\mid a_2}}},
\end{align}
where for the last equality we have used that for $b_2\in a_2$ it holds that $\pi_{b_2\mid r,a_2} = \frac{\pi_{b_2\mid r}}{\pi_{a_2\mid r}}$ and $\pi_r \frac{\pi_{a_2\mid r}}{\pi_{a_2}} = \pi_{r\mid a_2}.$ 
\end{proof}

\begin{proof}[Proof of Lemma~\ref{lemma:normalization_decompositions}]
    We apply the group decomposition of Lemma~\ref{lemma:group_decomposition_of_H} to $\DKL{\sigma^{M_{DSI}}}{\sigma^{M_{DC}}},$ where we recall the definitions of the degree constrained social isolation model and
    \begin{align}
        \PS{r_1 b_1 \sim b_2}{}{M_{DSI}} = \bar{P}_{r_1 b_1}^M \frac{1}{\pi_{b_2}} \mathbb{1}_{\{b_1=b_2\}}
    \end{align}
    the degree configuration model
    \begin{align}
        \PS{r_1 b_1 \sim b_2}{}{M_{DC}} = \bar{P}_{r_1 b_1}^M \bar{P}_{b_2}^M \frac{1}{\bar{P}} .
    \end{align}
\begin{enumerate}
    \item[(i)]
To show that the resulting between normalization constant is based on the connectivity models \eqref{eq:between_group_si_model} and \eqref{eq:between_group_dc_model}, we perform the marginalization of the connectivity models to coarser coordinates $(r_1, a_1, a_2)$ as needed for the edge distributions \eqref{eq:sigma_between_groups} 
\begin{equation}
\begin{aligned}
  \PS{r_1 a_1 \sim a_2}{}{M_{DSI}}
 & =\sum_{b_1\in a_1}\sum_{b_2\in a_2} \pi_{b_1\mid r_1,a_1}\pi_{b_2\mid a_2} \PS{r_1 b_1 \sim b_2}{}{M_{DSI}}\\ 
 & = \mathbb{1}_{\{a_1=a_2\}} \sum_{b_1\in a_1} \pi_{b_1\mid r_1,a_1}\pi_{b_1\mid a_1}\bar{P}_{r_1 b_1}^M \frac{1}{\pi_{b_1}} 
 = \mathbb{1}_{\{a_1=a_2\}}  \bar{P}_{r_1 a_1}^M \frac{1}{\pi_{a_1}}
\end{aligned}
\end{equation}
and
\begin{align}
    \PS{r_1 a_1 \sim a_2}{}{M_{DC}} 
    = \sum_{b_1\in a_1}\sum_{b_2\in a_2}\pi_{b_1\mid r_1, a_1}\pi_{b_2\mid a_2}  \bar{P}_{r_1 b_1}^M \bar{P}_{b_2}^M \frac{1}{\bar{P}}
    = \bar{P}_{r_1 a_1}^M \bar{P}_{a_2}^M \frac{1}{\bar{P}},
\end{align}
i.e., marginalization to coarser coordinates $a_1,a_2$ yields the degree constrained social isolation model and the degree configuration model as defined by the coarser coordinates respectively.

\item[(ii)]
For the within $a_2$ component, we obtain by \eqref{eq:sigma_N_given_M_group_within} the $M_{DSI}$ specific adjustment of $M_{DC}$
\begin{equation}
\begin{aligned}\label{eq:sigma_dc_given_dsi_given_a2}
    \sigma^{M_{DC} \mid M_{DSI}}(r_1 b_1, b_2\mid a_2)
    & = \pi_{r_1 a_1} \frac{\PS{r_1 a_1 \sim a_2}{}{M_{DSI}}}{\bar{P}_{a_2}^{M_{DSI}}}\pi_{b_1\mid r_1 a_1}\pi_{b_2\mid a_2}\frac{\PS{r_1 b_1 \sim b_2}{}{M_{DC}}}{\PS{r_1 a_1 \sim a_2}{}{M_{DC}}}\\
    & = \pi_{r_1 \mid a_1}\pi_{a_1} \frac{\mathbb{1}_{\{a_1=a_2\}}  \bar{P}_{r_1 a_1}^M \frac{1}{\pi_{a_1}}}{\bar{P}_{a_2}^M}
    \pi_{b_1\mid r_1 a_1}\pi_{b_2\mid a_2}\frac{\bar{P}_{r_1 b_1}^M \bar{P}_{b_2}^M \frac{1}{\bar{P}}}{\bar{P}_{r_1 a_1}^M \bar{P}_{a_2}^M \frac{1}{\bar{P}}}\\
    & = \pi_{r_1 \mid a_2} \pi_{b_1\mid r_1 a_2}\pi_{b_2\mid a_2}\frac{\mathbb{1}_{\{a_1=a_2\}} }{\bar{P}_{a_2}^M} \frac{\bar{P}_{r_1 b_1}^M \bar{P}_{b_2}^M}{\bar{P}_{a_2}^M },
\end{aligned}
\end{equation}
i.e., the edge distribution on $(\RS\times \BS)\times\{b\in a_2\}$ induced by the connectivity model \eqref{eq:within_group_dc_model} with zero support on $\{a_1\neq a_2\}.$
By \eqref{eq:sigma_M_group_within}, the within $a_2$ edge distribution for $M_{DSI}$ is given by
\begin{equation}
\begin{aligned}\label{eq:sigma_dsi_given_a2}
    \sigma^{M_{DSI}}(r_1 b_1, b_2\mid a_2)
    &= \pi_{r_1 a_1} \pi_{b_1\mid r_1 a_1}\pi_{b_2\mid a_2} \frac{\bar{P}_{r_1 b_1}^M \frac{1}{\pi_{b_2}} \mathbb{1}_{\{b_1=b_2\}}}{\bar{P}_{a_2}^{M_{DSI}}}\\
    & = \pi_{r_1\mid a_1} \pi_{b_1\mid r_1 a_1}\pi_{b_2\mid a_2} \mathbb{1}_{\{b_1=b_2\}}\frac{\bar{P}_{r_1 b_1}^M \frac{1}{\pi_{b_2\mid a_2}}}{\bar{P}_{a_2}^M}, 
\end{aligned}
\end{equation}
i.e., the edge distribution on $(\RS\times \BS)\times\{b\in a_2\}$ induced by the connectivity model \eqref{eq:within_group_si_model} with zero support on $\{a_1\neq a_2\}$ (note that $b_1 =  b_2$ implies $a_1 = a_2$).

\item[(iii)]
For $M=M_{RI}$, i.e., $\bar{P}_{r_1 b_1}^M=\bar{P}_{b_2}^M=\bar{P}$, $M_{DC}=M_{ER}$, and $M_{DSI}=M_{SI}$, part~(i) shows that the between normalization constant is the normalization constant \eqref{eq:normalization_const_H_D_C} in Lemma~\ref{lemma:normalization_of_indices} applied on the coarser coordinate space $(\RS\times\mathcal A)\times\mathcal A$ and therefore recovers the traditional between normalization constant \eqref{eq:between_a_normalization_trad_theil}. For the within normalization constant, \eqref{eq:sigma_dc_given_dsi_given_a2} simplifies to
\begin{align}
    \sigma^{M_{ER} \mid M_{SI}}(r_1 b_1, b_2\mid a_2)
    = \pi_{r_1 a_1}  \frac{\mathbb{1}_{\{a_1=a_2\}}}{\pi_{a_2}} \pi_{b_1\mid r_1 a_1}\pi_{b_2\mid a_2}
\end{align}
    while \eqref{eq:sigma_dsi_given_a2} simplifies to
\begin{align}
     \sigma^{M_{SI}}(r_1 b_1, b_2\mid a_2)
    = \pi_{r_1 a_1} \pi_{b_1\mid r_1 a_1}\pi_{b_2\mid a_2} \frac{\mathbb{1}_{\{b_1=b_2\}}}{\pi_{b_2}}.
\end{align}
Consequently we obtain
\begin{equation}
\begin{aligned}
    &\DKL{\sigma^{M_{SI}}(r_1 b_1, b_2\mid a_2)}{\sigma^{M_{ER} \mid M_{SI}}(r_1 b_1, b_2\mid a_2)}\\
    & = \sum_{r_1}\sum_{a_1}\sum_{b_1\in a_1,b_2\in a_2}
    \pi_{r_1 a_1} \pi_{b_1\mid r_1 a_1}\pi_{b_2\mid a_2} \frac{\mathbb{1}_{\{b_1=b_2\}}}{\pi_{b_2}}
    \log\left(
    \frac{ \frac{\mathbb{1}_{\{b_1=b_2\}}}{\pi_{b_2}}}{ \frac{\mathbb{1}_{\{a_1=a_2\}}}{\pi_{a_2}}}
    \right)\\
    &= \sum_{r}\sum_{b\in a} \pi_a \pi_{r\mid a} \pi_{b\mid r,a}\frac{1}{\pi_{a}}
    \log\left(\frac{1}{\pi_{b\mid a}}\right)\\
    & = \sum_r\sum_{b\in a} \pi_{r\mid a}\pi_{b\mid r,a} \log\left(\frac{1}{\pi_{b\mid a}}\right) .
\end{aligned}
\end{equation}

\item[(iv)]
When $M=M_{RI}$ the regional isolation model, then $\bar{P}_b^M = \bar{P}$ such that \eqref{eq:def_between_normalization_reg_decomp} reduces to \eqref{eq:between_s_normalization_trad_theil}. To show that \eqref{eq:def_within_normalization_reg_decomp} reduces to \eqref{eq:within_s_normalization_trad_theil}, we first note that
\begin{equation}
    \begin{aligned}
    \PS{s\sim b_2}{}{M_{DSI}} & = \sum_{r_1\in s}\sum_{b_1}\pi_{r_1 b_1 \mid s} \PS{r_1 b_1 \sim b_2}{}{M_{DSI}}\\
    & = \sum_{r_1\in s}\sum_{b_1}\pi_{r_1 b_1 \mid s} \mathbb{1}_{\{b_1=b_2\}}\frac{\bar{P}_{r_1 b_1}^M}{\pi_{b_2}}
\end{aligned}
\end{equation}
Hence we obtain by \eqref{eq:def_sgima_M_within_s} and \eqref{eq:def_sgima_N_given_M_within_s}
\begin{align}
    \frac{\sigma^{M_{DSI}}(r_1b_1, b_2\mid s)}{\sigma^{M_{DC}\mid M_{DSI}}(r_1b_1, b_2\mid s)}
    = \frac{ \mathbb{1}_{\{b_1=b_2\}} \bar{P}_{r_1 b_1}^M \frac{1}{\pi_{b_2}}}
    { \left(\sum_{r'\in s}\sum_{b'} \pi_{r'\mid s}\pi_{b' \mid r'} \mathbb{1}_{\{b'=b_2\}} \bar{P}_{r' b'}^M\frac{1}{\pi_{b_2}}\right)\left(\frac{\bar{P}_{r_1 b_1}^M\bar{P}_{b_2}^M}{\bar{P}_{s}^M \bar{P}_{b_2}^M}\right) }
\end{align}
Now when $M=M_{RI},$ then $\bar{P}_{s}^M = \bar{P}_{b_2}^M =  \bar{P}_{r_1 b_1}^M = \bar{P},$ and by observing that $\sum_{r\in s} \pi_{r,b\mid s} = \pi_{b\mid s}\sum_{r\in s}\pi_{r\mid b,s} = \pi_{b\mid s},$ we obtain
\begin{align}
    \frac{\sigma^{M_{SI}}(r_1b_1, b_2\mid s)}{\sigma^{M_{ER}\mid M_{SI}}(r_1b_1, b_2\mid s)}
    = \frac{\mathbb{1}_{\{b_1=b_2\}}}{\pi_{b_{2}\mid s}}.
\end{align}
It follows that
\begin{equation}
    \begin{aligned}
        &\DKL{\sigma^{M_{SI}}(r_1b_1, b_2\mid s)}{\sigma^{M_{ER}\mid M_{SI}}(r_1b_1, b_2\mid s)}\\
        &= \sum_{r_1\in s} \sum_{b_1, b_2} \pi_{r_1 b_1\mid s} \pi_{b_2} \frac{\mathbb{1}_{\{b_1=b_2\}}\bar{P}\frac{1}{\pi_{b_2}}}{\bar{P}}\log\left(
                \frac{\mathbb{1}_{\{b_1=b_2\}}}{\pi_{b_2\mid s}}
        \right)\\
        & = \sum_{r_1\in s} \sum_{b_1} \pi_{r_1 b_1\mid s} \log\left(\frac{1}{\pi_{b_{1}\mid s}}\right)
        = \sum_b \pi_{b\mid s}\log\left(\frac{1}{\pi_{b\mid s}}\right) .
    \end{aligned}
\end{equation}

\end{enumerate}
\end{proof}

\subsection{Additional results}

\begin{proof}[Proof of Lemma~\ref{lemma:coordinate_coarsening}]   
For part (i), it suffices to consider a generic coordinate coarsening $T_1\colon\CS_1\to\widetilde{\CS}_1$ and $T_2\colon\CS_2\to\widetilde{\CS}_2$. For $\rho=\tau$, define 
\begin{equation*}
    T_\tau(c_1,c_2,e)\coloneqq\left(T_1(c_1),T_2(c_2),e\right),
\end{equation*}
and for $\rho\in\{\sigma,\xi\}$ define
\begin{equation*}
    T_\rho(c_1,c_2)\coloneqq\left(T_1(c_1),T_2(c_2)\right).
\end{equation*}
For a probability distribution $P$ on $\mathcal X$ and a deterministic map $T\colon\mathcal X\to\widetilde{\mathcal X}$, let $P_T$ denote the induced distribution on $\widetilde{\mathcal X}$, i.e.,
\begin{equation*}
    P_T(\tilde x)\coloneqq\sum_{x:T(x)=\tilde x}P(x).
\end{equation*}
Accordingly, for any model $N$ and $\rho\in\{\tau,\sigma,\xi\}$, let $\rho_T^N$ denote the distribution induced from $\rho^N$ by the map $T_\rho$. We use a tilde to distinguish distributions $\widetilde\rho^N$ defined directly on the coarsened coordinate spaces $\widetilde{\CS}_1$ and $\widetilde{\CS}_2$.

The data-processing inequality for $f$-divergences \cite{polyanskiy2025information} states that for any probability distributions $P$ and $Q$ on $\mathcal X$ and any deterministic map $T\colon\mathcal X\to\widetilde{\mathcal X}$,
\begin{equation}
    D_f(P_T\|Q_T)\leq D_f(P\|Q) .
    \label{eq:DPI_f_divergence}
\end{equation}
Applied to the two distributions defining $\mathcal D_{f,\rho}\left(M',M_{DC}(M');\CS_1,\CS_2\right)$, this gives
\begin{equation}
    D_f\!\left(\rho_T^{M'}\middle\|\rho_T^{M_{DC}(M')}\right)
    \leq D_f\!\left(\rho^{M'}\middle\|\rho^{M_{DC}(M')}\right).
    \label{eq:DPI_segregation_coarsening}
\end{equation}
To identify the left-hand side with the corresponding divergences on the coarsened coordinate spaces, it therefore remains to show, for each $\rho\in\{\tau,\sigma,\xi\}$, that
\begin{equation}
    \rho_T^{M'}=\widetilde\rho^{M'}
    \label{eq:coarsening_Mprime_distribution}
\end{equation}
and
\begin{equation}
    \rho_T^{M_{DC}(M')}=\widetilde\rho^{M_{DC}(M')},
    \label{eq:coarsening_MDC_distribution}
\end{equation}
where the right-hand sides in \eqref{eq:coarsening_Mprime_distribution} and \eqref{eq:coarsening_MDC_distribution} denote the corresponding distributions defined directly on the coarsened coordinate spaces.

Equation~\eqref{eq:coarsening_Mprime_distribution} follows directly from the definitions in \eqref{eq:def_tau}--\eqref{eq:def_xi}: for $\tau$, it is the distribution of $\left(T_1(C_1),T_2(C_2),A\right)$, while for $\sigma$ and $\xi$ it is the distribution of $\left(T_1(C_1),T_2(C_2)\right)$ conditional on $A=1$ and $A=0$, respectively. We therefore only need to establish \eqref{eq:coarsening_MDC_distribution}.

We first consider $\rho=\tau$. For $k\in\{1,2\}$,
\begin{equation}
    \pi_{\tilde c_k}^{M'}\bar P_{\tilde c_k}^{M'}
    = \PPS{T_k(C_k)=\tilde c_k,A=1}{}{M'}
    = \sum_{c_k:T_k(c_k)=\tilde c_k}
    \pi_{c_k}^{M'}\bar P_{c_k}^{M'}.
    \label{eq:coarsened_marginal_connectivity}
\end{equation}
Hence, for any $\tilde c_1\in\widetilde{\CS}_1$ and $\tilde c_2\in\widetilde{\CS}_2$,
\begin{equation}
\begin{aligned}
\tau_T^{M_{DC}(M')}(\tilde c_1,\tilde c_2,1)
&=\PPS{T_1(C_1)=\tilde c_1,T_2(C_2)=\tilde c_2,A=1}{}{M_{DC}(M')}\\
&=\sum_{\substack{c_1:T_1(c_1)=\tilde c_1\\c_2:T_2(c_2)=\tilde c_2}}
\pi_{c_1}^{M'}\pi_{c_2}^{M'}
\frac{\bar P_{c_1}^{M'}\bar P_{c_2}^{M'}}{\bar P^{M'}}\\
&=\frac{1}{\bar P^{M'}}
\left(\sum_{c_1:T_1(c_1)=\tilde c_1}\pi_{c_1}^{M'}\bar P_{c_1}^{M'}\right)
\left(\sum_{c_2:T_2(c_2)=\tilde c_2}\pi_{c_2}^{M'}\bar P_{c_2}^{M'}\right)\\
&=\pi_{\tilde c_1}^{M'}\pi_{\tilde c_2}^{M'}
\frac{\bar P_{\tilde c_1}^{M'}\bar P_{\tilde c_2}^{M'}}{\bar P^{M'}}\\
&=\widetilde\tau^{M_{DC}(M')}(\tilde c_1,\tilde c_2,1),
\end{aligned}
\label{eq:MDC_coarsening_tau_one}
\end{equation}
where \eqref{eq:coarsened_marginal_connectivity} was used in the penultimate equality.

Both $\tau_T^{M_{DC}(M')}$ and $\widetilde{\tau}^{M_{DC}(M')}$ assign, when marginalized over $e$, mass $\pi_{\tilde c_1}^{M'}\pi_{\tilde c_2}^{M'}$ to the coordinate pair $(\tilde c_1,\tilde c_2).$ Therefore, using \eqref{eq:MDC_coarsening_tau_one}, we obtain for $e=0$
\begin{equation*}
\begin{aligned}
\tau_T^{M_{DC}(M')}(\tilde c_1,\tilde c_2,0)
&=\pi_{\tilde c_1}^{M'}\pi_{\tilde c_2}^{M'}
-\tau_T^{M_{DC}(M')}(\tilde c_1,\tilde c_2,1)\\
&=\pi_{\tilde c_1}^{M'}\pi_{\tilde c_2}^{M'}
-\widetilde\tau^{M_{DC}(M')}(\tilde c_1,\tilde c_2,1)\\
&=\widetilde\tau^{M_{DC}(M')}(\tilde c_1,\tilde c_2,0).
\end{aligned}
\end{equation*}
Thus \eqref{eq:coarsening_MDC_distribution} holds for $\rho=\tau$.

For $\rho=\sigma$, since $M_{DC}(M')$ has marginal connection probability $\bar P^{M'}$ and coordinate coarsening does not change this probability,
\begin{equation*}
\begin{aligned}
\sigma_T^{M_{DC}(M')}(\tilde c_1,\tilde c_2)
&=\frac{\tau_T^{M_{DC}(M')}(\tilde c_1,\tilde c_2,1)}{\bar P^{M'}}=\frac{\widetilde\tau^{M_{DC}(M')}(\tilde c_1,\tilde c_2,1)}{\bar P^{M'}}\\
&=\widetilde\sigma^{M_{DC}(M')}(\tilde c_1,\tilde c_2).
\end{aligned}
\end{equation*}
Likewise,
\begin{equation*}
\begin{aligned}
\xi_T^{M_{DC}(M')}(\tilde c_1,\tilde c_2)
&=\frac{\tau_T^{M_{DC}(M')}(\tilde c_1,\tilde c_2,0)}{1-\bar P^{M'}}=\frac{\widetilde\tau^{M_{DC}(M')}(\tilde c_1,\tilde c_2,0)}{1-\bar P^{M'}}\\
&=\widetilde\xi^{M_{DC}(M')}(\tilde c_1,\tilde c_2).
\end{aligned}
\end{equation*}
Hence \eqref{eq:coarsening_MDC_distribution} holds for every $\rho\in\{\tau,\sigma,\xi\}$.

Using \eqref{eq:coarsening_Mprime_distribution} and \eqref{eq:coarsening_MDC_distribution} in \eqref{eq:DPI_segregation_coarsening} therefore gives
\begin{equation}
    \mathcal D_{f,\rho}\left(M',M_{DC}(M');\widetilde{\CS}_1,\widetilde{\CS}_2\right)
    \leq
    \mathcal D_{f,\rho}\left(M',M_{DC}(M');\CS_1,\CS_2\right),
    \qquad \rho\in\{\tau,\sigma,\xi\}.
\end{equation}
Applying this inequality successively to the coordinate coarsenings in \eqref{eq:coordinate_coarsening_unnormalized} and \eqref{eq:coordinate_coarsening_unnormalized_geoegocentric} proves part (i).

For part (ii), by part (i) the $f$-divergence in the numerator cannot increase under geographical coordinate coarsening. It therefore suffices to show that the denominator
\begin{equation*}
    \mathcal D_{f,\sigma}\left(M_{DSI}(M'),M_{DC}(M');\CS_1,\CS_2\right)
    = D_f\!\left(\sigma^{M_{DSI}(M')}\middle\|\sigma^{M_{DC}(M')}\right)
\end{equation*}
is unchanged by coarsening of geographical coordinates.

For any of the coordinate spaces considered in \eqref{eq:coordinate_coarsening_normalized_sigma} and \eqref{eq:coordinate_coarsening_normalized_sigma_geoegocentric}, let $b_k = b(c_k)$ denote the sociodemographic component of coordinate $c_k$. From the definitions of $M_{DC}(M')$ and $M_{DSI}(M')$,
\begin{equation}
    \frac{\sigma^{M_{DSI}(M')}(c_1,c_2)}{\sigma^{M_{DC}(M')}(c_1,c_2)}
    = \begin{cases}
        \displaystyle\frac{\bar P^{M'}}{\pi_b^{M'}\bar P_b^{M'}},&b_1=b_2=b,\\[0.75em]
        0,&b_1\neq b_2.
    \end{cases}
    \label{eq:ratio_sigma_DSI_DC_coarsening}
\end{equation}
Thus, the likelihood ratio depends only on the sociodemographic coordinates $b_1$ and $b_2$. When evaluating the $f$-divergence, we can therefore first sum the mass of $\sigma^{M_{DC}(M')}$ over all geographical coordinates corresponding to each fixed pair $(b_1,b_2)$. This mass is
\begin{equation}
\begin{aligned}
    \sum_{\substack{c_1:\,b(c_1)=b_1\\c_2:\,b(c_2)=b_2}}
    \sigma^{M_{DC}(M')}(c_1,c_2)
    &=
    \frac{\pi_{b_1}^{M'}\bar P_{b_1}^{M'}}{\bar P^{M'}}
    \frac{\pi_{b_2}^{M'}\bar P_{b_2}^{M'}}{\bar P^{M'}}.
\end{aligned}
\label{eq:sigma_DC_social_pair_mass}
\end{equation}
Consequently,
\begin{equation}
\begin{aligned}
\mathcal D_{f,\sigma}\left(M_{DSI}(M'),M_{DC}(M');\CS_1,\CS_2\right)
& = \sum_{b_1,b_2\in\BS}
\frac{\pi_{b_1}^{M'}\bar P_{b_1}^{M'}}{\bar P^{M'}}
\frac{\pi_{b_2}^{M'}\bar P_{b_2}^{M'}}{\bar P^{M'}}
f\left(\frac{\bar P^{M'}}{\pi_{b_1}^{M'}\bar P_{b_1}^{M'}}\mathbb{1}_{\left\{b_1=b_2\right\}}\right).
\end{aligned}
\label{eq:sigma_DSI_normalization_social_only}
\end{equation}
The right-hand side depends only on $\pi_b^{M'}$, $\bar P_b^{M'}$, and $\bar P^{M'}$, which are unchanged by geographical coordinate coarsening. Hence
\begin{equation}
\begin{aligned}
    \mathcal D_{f,\sigma}\left(M_{DSI}(M'),M_{DC}(M');\BS,\BS\right)
    &=
    \mathcal D_{f,\sigma}\left(M_{DSI}(M'),M_{DC}(M');\RS\times\BS,\BS\right)\\
    &=
    \mathcal D_{f,\sigma}\left(M_{DSI}(M'),M_{DC}(M');\RS\times\BS,\RS\times\BS\right)\\
    &=
    \mathcal D_{f,\sigma}\left(M_{DSI}(M'),M_{DC}(M');\GS\times\BS,\BS\right)\\
    &=
    \mathcal D_{f,\sigma}\left(M_{DSI}(M'),M_{DC}(M');\GS\times\BS,\GS\times\BS\right).
\end{aligned}
\label{eq:sigma_normalization_invariant_to_geo_coarsening}
\end{equation}
\end{proof}

\fi

\end{document}